\pdfoutput=1
\documentclass[11pt]{article}
\usepackage{authblk}
\usepackage[toc,page]{appendix}
\usepackage[top=2cm, bottom=2cm, left=2cm, right=2cm]{geometry}

\usepackage{color}
\usepackage{helvet}         
\usepackage{courier}        
\usepackage{type1cm}        
 
\usepackage{framed}
\usepackage{tikz}
\usepackage{makeidx}         
\usepackage{graphicx}        
\usepackage{multicol}        
\usepackage[bottom]{footmisc}

\usepackage{amsmath}
\usepackage{amssymb}
\usepackage{bbold}
\usepackage{amsthm}
\usepackage{subcaption}
\usepackage{sidecap}
\usepackage{floatrow}
\usepackage{pdflscape}
\usepackage{pdflscape}
\usepackage{comment}
\usepackage[font=small]{caption}
\usepackage{enumitem}
\usepackage{esint}

\usepackage{thmtools}
\usepackage{thm-restate}

\usepackage{scalerel}
\usepackage{hyperref}
\usepackage{wasysym}

\newtheorem{theorem}{Theorem}
\newtheorem{corollary}[theorem]{Corollary}
\newtheorem{lemma}[theorem]{Lemma}
\newtheorem{conjecture}[theorem]{Conjecture}
\newtheorem{proposition}[theorem]{Proposition}

\theoremstyle{remark}

\newtheorem{remark}[theorem]{\bf Remark}
\newtheorem{example}[theorem]{\bf Example}
\newtheorem{definition}[theorem]{\bf Definition}

\numberwithin{theorem}{section}
\numberwithin{corollary}{section}
\numberwithin{definition}{section}
\numberwithin{lemma}{section}
\numberwithin{conjecture}{section}
\numberwithin{proposition}{section}
\numberwithin{question}{section}
\numberwithin{example}{section}
\numberwithin{remark}{section}
\numberwithin{figure}{section}
\numberwithin{equation}{section}

\begin{document}

\title{Multiple SLEs for $\kappa\in (0,8)$:\\ Coulomb gas integrals and pure partition functions}
\bigskip{}
\author[1]{Yu Feng\thanks{yufeng\_proba@163.com}}
\author[3]{Mingchang Liu\thanks{liumc\_prob@163.com}}
\author[4]{Eveliina Peltola\thanks{eveliina.peltola@aalto.fi}}
\author[2]{Hao Wu\thanks{hao.wu.proba@gmail.com}}
\affil[1]{University of Michigan, United States}
\affil[2]{Tsinghua University, China}
\affil[3]{Capital Normal University, China}
\affil[4]{Aalto University, Finland, and University of Bonn, Germany}

\date{}

%
%

\global\long\def\CR{\mathrm{CR}}
\global\long\def\ST{\mathrm{ST}}
\global\long\def\SF{\mathrm{SF}}
\global\long\def\cov{\mathrm{cov}}
\global\long\def\dist{\mathrm{dist}}
\global\long\def\SLE{\mathrm{SLE}}
\global\long\def\hSLE{\mathrm{hSLE}}
\global\long\def\CLE{\mathrm{CLE}}
\global\long\def\GFF{\mathrm{GFF}}
\global\long\def\inte{\mathrm{int}}
\global\long\def\ext{\mathrm{ext}}
\global\long\def\inrad{\mathrm{inrad}}
\global\long\def\outrad{\mathrm{outrad}}
\global\long\def\dimH{\mathrm{dim}}
\global\long\def\capa{\mathrm{cap}}
\global\long\def\diam{\mathrm{diam}}
\global\long\def\sign{\mathrm{sgn}}
\global\long\def\LZ{\mathcal{Z}}

\newcommand{\LHtwo}{\mathcal{H}_{\includegraphics[scale=0.2]{figures-arXiv/link-0}}}
\newcommand{\LGthree}{\LG_{\includegraphics[scale=0.2]{figures-arXiv/link-2}/\vee_1}^{(\kappa)}}

\newcommand{\LHfoura}{\mathcal{H}_{\includegraphics[scale=0.2]{figures-arXiv/link-1}}}
\newcommand{\LHfourb}{\mathcal{H}_{\includegraphics[scale=0.2]{figures-arXiv/link-2}}}

\newcommand{\LFtwo}{\mathcal{F}_{\includegraphics[scale=0.2]{figures-arXiv/link-0}}}

\newcommand{\LFfoura}{\mathcal{F}_{\includegraphics[scale=0.2]{figures-arXiv/link-1}}}
\newcommand{\LFfourb}{\mathcal{F}_{\includegraphics[scale=0.2]{figures-arXiv/link-2}}}
\newcommand{\LFfouraRenorm}{\widehat{\mathcal{F}}_{\includegraphics[scale=0.2]{figures-arXiv/link-1}}}
\newcommand{\LFfourbRenorm}{\widehat{\mathcal{F}}_{\includegraphics[scale=0.2]{figures-arXiv/link-2}}}
\newcommand{\coulombGasHRenorm}{\widehat{\coulombGasH}}

\newcommand{\LZtwo}{\mathcal{Z}_{\includegraphics[scale=0.2]{figures-arXiv/link-0}}}

\newcommand{\LZfoura}{\mathcal{Z}_{\includegraphics[scale=0.2]{figures-arXiv/link-1}}}
\newcommand{\LZfourb}{\mathcal{Z}_{\includegraphics[scale=0.2]{figures-arXiv/link-2}}}

\newcommand{\LZthree}{\mathcal{Z}_{\includegraphics[scale=0.8]{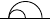}}}

\global\long\def\cst{\mathrm{C}}

\global\long\def\free{\mathrm{free}}
\global\long\def\hF{{}_2\mathrm{F}_1}
\global\long\def\hFgen{{}_3\mathrm{F}_2}
\global\long\def\simple{\mathrm{simple}}
\global\long\def\even{\mathrm{even}}
\global\long\def\odd{\mathrm{odd}}
\global\long\def\st{\mathrm{ST}}
\global\long\def\usf{\mathrm{USF}}
\global\long\def\Leb{\mathrm{Leb}}
\global\long\def\LP{\mathrm{LP}}
\global\long\def\coulomb{\LH}
\global\long\def\auxcoulomb{\hat{\coulomb}}
\global\long\def\coulombGas{\LF}
\global\long\def\coulombnew{\LK}
\global\long\def\coulombLine{\LG}
\global\long\def\kfunc{p}
\global\long\def\OO{\mathcal{O}}

\global\long\def\eps{\epsilon}
\global\long\def\ov{\overline}
\global\long\def\U{\mathbb{U}}
\global\long\def\T{\mathbb{T}}
\global\long\def\HH{\mathbb{H}}
\global\long\def\LA{\mathcal{A}}
\global\long\def\LB{\mathcal{B}}
\global\long\def\LC{\mathcal{C}}
\global\long\def\LD{\mathcal{D}}
\global\long\def\LF{\mathcal{F}}
\global\long\def\LK{\mathcal{K}}
\global\long\def\LE{\mathcal{E}}
\global\long\def\LG{\mathcal{G}}
\global\long\def\LI{\mathcal{I}}
\global\long\def\LJ{\mathcal{J}}
\global\long\def\LL{\mathcal{L}}
\global\long\def\LN{\mathcal{N}}
\global\long\def\LQ{\mathcal{Q}}
\global\long\def\LR{\mathcal{R}}
\global\long\def\LT{\mathcal{T}}
\global\long\def\LS{\mathcal{S}}
\global\long\def\LU{\mathcal{U}}
\global\long\def\LV{\mathcal{V}}
\global\long\def\LW{\mathcal{W}}
\global\long\def\LX{\mathcal{X}}
\global\long\def\LY{\mathcal{Y}}
\global\long\def\PartF{\mathcal{Z}}
\global\long\def\LH{\mathcal{H}}
\global\long\def\LJ{\mathcal{J}}
\global\long\def\R{\mathbb{R}}
\global\long\def\C{\mathbb{C}}
\global\long\def\N{\mathbb{N}}
\global\long\def\Z{\mathbb{Z}}
\global\long\def\E{\mathbb{E}}
\global\long\def\PP{\mathbb{P}}
\global\long\def\QQ{\mathbb{Q}}
\global\long\def\A{\mathbb{A}}
\global\long\def\one{\mathbb{1}}
\global\long\def\bn{\mathbf{n}}
\global\long\def\MR{MR}
\global\long\def\cond{\,|\,}
\global\long\def\la{\langle}
\global\long\def\ra{\rangle}
\global\long\def\tree{\Upsilon}
\global\long\def\prob{\mathbb{P}}
\global\long\def\hm{\mathrm{Hm}}
%

\global\long\def\Im{\operatorname{Im}}
\global\long\def\Re{\operatorname{Re}}

\global\long\def\ud{\mathrm{d}}
\global\long\def\pder#1{\frac{\partial}{\partial#1}}
\global\long\def\pdder#1{\frac{\partial^{2}}{\partial#1^{2}}}
\global\long\def\pddder#1{\frac{\partial^{3}}{\partial#1^{3}}}
\global\long\def\der#1{\frac{\ud}{\ud#1}}

\global\long\def\bZnn{\mathbb{Z}_{\geq 0}}
\global\long\def\bZpos{\mathbb{Z}_{> 0}}
\global\long\def\bZneg{\mathbb{Z}_{< 0}}

\global\long\def\Vfunc{\LG}
\global\long\def\gfunc{g^{(\rr)}}
\global\long\def\hfunc{h^{(\rr)}}

\global\long\def\SimplexInt{\rho}
\global\long\def\CubeInt{\widetilde{\rho}}

\global\long\def\ii{\mathfrak{i}}
\global\long\def\rr{\mathfrak{r}}
\global\long\def\chamber{\mathfrak{X}}
\global\long\def\Wchamber{\mathfrak{W}}

\global\long\def\SimplexIntKappa8{\SimplexInt}


\global\long\def\acycle{\vartheta}
\global\long\def\bcycle{\tilde{\acycle}}

\global\long\def\metric{\mathrm{dist}}

\global\long\def\adj#1{\mathrm{adj}(#1)}

\global\long\def\bs{\boldsymbol}

\global\long\def\edge#1#2{\langle #1,#2 \rangle}
\global\long\def\graph{G}

\newcommand{\conn}{\varsigma}
\newcommand{\realacycle}{\smash{\mathring{\acycle}}}
\newcommand{\realpt}{\smash{\mathring{x}}}
\newcommand{\corrind}{\LC}
\newcommand{\bssymb}{\pi}
\newcommand{\PRCM}{\mu}
\newcommand{\coeff}{p}
\newcommand{\MainConst}{C}

\global\long\def\removeLink{/}

\global\long\def\domainofdef{\mathfrak{U}}
\global\long\def\Smooth_space{C^\infty}
\global\long\def\Test_space{C_c^\infty}
\global\long\def\Distr_space{(\Test_space)^*}

\global\long\def\bs{\boldsymbol}
\global\long\def\cst{\mathrm{C}}

\newcommand{\red}{\textcolor{red}}

\newcommand{\coulombGasH}{\mathcal{H}}
\newcommand{\secondbeta}{\intloop}

\newcommand{\cev}[1]{\reflectbox{\ensuremath{\vec{\reflectbox{\ensuremath{#1}}}}}}

\global\long\def\anticonf{\zeta}
\global\long\def\intloop{\varrho}
\global\long\def\Gloop{\smash{\mathring{\intloop}}}

\global\long\def\SLEmeasure{\mathrm{P}}
\global\long\def\SLEmeasureEx{\mathrm{E}}

\global\long\def\fugacity{\nu}
\global\long\def\meanderMat{\mathcal{M}}
\global\long\def\meanderMatrix{\meanderMat_{\fugacity}}
\global\long\def\meanderMatrixPrime{\meanderMat_{\fugacity(\kappa')}}
\global\long\def\meanderRenorm{\widehat{\mathcal{M}}}

\global\long\def\PartFRenorm{\widehat{\PartF}}
\global\long\def\coulombGasRenorm{\widehat{\coulombGas}}

\global\long\def\hexa{\scalebox{1.3}{\hexagon}}

\global\long\def\np{p}

\global\long\def\FKdual{\mathcal{L}}

\global\long\def\fixedindex{\flat}

\newcommand{\KWleq}{\stackrel{\quad \scalebox{.6}{()}}{\leftarrow}}

\maketitle

\begin{center}
\begin{minipage}{0.9\textwidth}
\abstract{
In this article, we give an explicit relationship of SLE partition functions with Coulomb gas formalism of conformal field theory. We first construct a family of $\SLE_\kappa$ partition functions as Coulomb gas integrals and derive their various properties. In accordance with an interpretation as probabilistic correlations in loop $O(n)$ models, they are always positive when $\kappa\in (8/3,8)$, while they may have zeroes for $\kappa\le 8/3$. They also admit a Frobenius series expansion that matches with the algebraic content from CFT. Moreover, we check that at the first level of fusion, they have logarithmic asymptotic behavior when $\kappa=8/3$ and $\kappa=8$, in accordance with logarithmic minimal models $M(2,1)$ and $M(2,3)$, respectively. 

Second, we construct $\SLE_\kappa$ pure partition functions and show that they are real-analytic in $\kappa\in (0,8)$ and decay to zero as a polynomial of $(8-\kappa)$ as $\kappa\to 8$. We explicitly relate the Coulomb gas integrals and pure partition functions together in terms of the meander matrix. As a by-product, our results yield a construction of global non-simple multiple chordal $\SLE_\kappa$ measures ($\kappa \in (4,8)$) uniquely determined by their re-sampling property.  
}

\bigskip{}



\end{minipage}
\end{center}

\bigskip{}
\setcounter{tocdepth}{2}
\tableofcontents

\allowdisplaybreaks


\newpage
\section{Introduction}
In the present work, 
we are interested in geometric and algebraic properties of natural probability measures on conformally invariant planar curves: 
variants of the celebrated Schramm-Loewner evolution (SLE).
They have become a success story in the description of critical planar interfaces in statistical physics~\cite{Schramm:Scaling_limits_of_LERW_and_UST,
Smirnov-Werner:Critical_exponents_for_two-dimensional_percolation,
LSW:Conformal_invariance_of_planar_LERW_and_UST,
Smirnov:Towards_conformal_invariance_of_2D_lattice_models,
Schramm-Sheffield:Contour_lines_of_2D_discrete_GFF,
CDHKS:Convergence_of_Ising_interfaces_to_SLE},  
sharing an intimate connection to conformal field theory (CFT)~\cite{Bauer-Bernard:Conformal_field_theories_of_SLEs,
Friedrich-Kalkkinen:On_CFT_and_SLE,
BBK:Multiple_SLEs_and_statistical_mechanics_martingales,
Dubedat:SLE_and_Virasoro_representations_localization,
Peltola:Towards_CFT_for_SLEs}, 
and playing an important role in various constructions in planar random geometry: 
properties of random fractals~\cite{LSW:The_dimension_of_the_planar_Brownian_frontier_is_four_thirds,
LSW:Brownian_intersection_exponents1}, 
constructions for models of gravity~\cite{Sheffield:Zipper, 
DMS:Liouville_quantum_gravity_as_mating_of_trees,
MSW:Non-simple_conformal_loop_ensembles_on_Liouville_quantum_gravity_and_the_law_of_CLE_percolation_interfaces, 
AHSY:Conformal_welding_of_quantum_disks_and_multiple_SLE_the_non-simple_case,
Sun-Yu:SLE_partition_functions_via_conformal_welding_of_random_surfaces}, 
and so on.

Concretely, our goal is to highlight the \emph{explicit relationship} of the SLE partition functions with the Coulomb gas formalism of CFT: 
certain natural SLE partition functions have Coulomb gas integral form, 
while others do not seem to admit very simple integral expressions and can alternatively be constructed probabilistically. 
We show that both approaches result in the same objects, which unifies various results presented in the literature, 
e.g.,~\cite{BBK:Multiple_SLEs_and_statistical_mechanics_martingales, 
Dubedat:Euler_integrals_for_commuting_SLEs, 
Kozdron-Lawler:Configurational_measure_on_mutually_avoiding_SLEs,
Lawler:Partition_functions_loop_measure_and_versions_of_SLE,
Flores-Kleban:Solution_space_for_system_of_null-state_PDE3,
Kytola-Peltola:Pure_partition_functions_of_multiple_SLEs,
Peltola-Wu:Global_and_local_multiple_SLEs_and_connection_probabilities_for_level_lines_of_GFF, 
Wu:Convergence_of_the_critical_planar_ising_interfaces_to_hypergeometric_SLE, 
BPW:On_the_uniqueness_of_global_multiple_SLEs,
Zhan:Existence_and_uniqueness_of_nonsimple_multiple_SLE,
AHSY:Conformal_welding_of_quantum_disks_and_multiple_SLE_the_non-simple_case,
Sun-Yu:SLE_partition_functions_via_conformal_welding_of_random_surfaces}.
We interpret them in terms of crossing formulas for critical loop models.
We also analyze for which values of the variance parameter $\kappa$ these are positive --- surprisingly, this is not the whole parameter range $\kappa \in (0,8)$, but agrees with the phase diagram of critical loop models, $\kappa \in (8/3,8)$. 

The equivalence of both constructions gives us access to prove salient properties of the $\SLE_\kappa$ partition functions, which have been out of reach to date:
we show that they admit a particular \emph{Frobenius series} structure and thus satisfy certain specific CFT fusion rules,
their limits solve \emph{higher order BPZ PDEs}, and importantly, they are \emph{real-analytic} functions of the variance parameter $\kappa \in (0,8)$, which enables us to conclude that 
at the algebraically problematic rational values of $\kappa$, the partition functions still behave well in terms of analytic properties, important in applications. 
This is also a step towards proving complex analyticity, which could be regarded as a version of a major open question in the CFT interpretation of these functions. 
In contrast, some of the algebraic properties of the partition functions are drastically different for rational values of $\kappa$.
We show explicitly that they have logarithmic asymptotic behavior when $\kappa=8/3$ and $\kappa=8$, 
in accordance with logarithmic minimal models $M(2,1)$ and $M(2,3)$, respectively. 

\medskip

The introduction to this article presents the big picture, gathers the key results of our work, and summarizes some known results and consequences.
It is organized as follows. 
We first discuss critical loop $O(n)$ models 
and formulate a conjecture concerning general connection probabilities in these models (Conjecture~\ref{conj::On}). 
This formula is known to hold in a few cases, 
and from the theory of $\SLE_\kappa$ it is plausible that it should hold for all $\kappa \in (8/3,8)$.
The same formula should also describe hook-up probabilities for loops in a suitable conformal loop ensemble ($\CLE_\kappa$), in the spirit of~\cite{Miller-Werner:Connection_probabilities_for_conformal_loop_ensembles} 
(see the recent~\cite{AMY:MultipleSLEfromCLE}).

After setting up the conjecture, we turn to discussing the key players in it. First of all, we define the \emph{Coulomb gas integrals} in Section~\ref{subsec::CGI} 
(the denominator in Conjecture~\ref{conj::On})
and briefly relate our definition to CFT constructions. 
We also collect the main properties of these functions: they are examples of multiple $\SLE_\kappa$ partition functions 
and have intricate behavior in terms of certain special values of $\kappa$, that correspond to certain CFT (possibly logarithmic) extended minimal models. 
We show how they also admit a Frobenius series expansion matching with the algebraic content from CFT 
(see Theorem~\ref{thm::CGI}, Propositions~\ref{prop::CGI_odd}~\&~\ref{prop::CGI_8}, and Propositions~\ref{prop::Frobenius_F}~\&~\ref{prop::third_order_PDE} in Section~\ref{subsec:Frob}). 
We also illustrate the logarithmic behavior of the Coulomb gas integrals 
in the special cases of $ \kappa=8/3$ and $\kappa=8$
(see Propositions~\ref{prop::CGI8over3_Frobenius}~\&~\ref{prop::CGI8_Frobenius}, respectively), 
which correspond to the two endpoints of the interval $(8/3,8) \ni \kappa$ relevant to Conjecture~\ref{conj::On}.

Second of all, we recall the definition of multiple $\SLE_\kappa$ \emph{pure partition functions} in Section~\ref{subsec::PPF} (the numerator in Conjecture~\ref{conj::On}).
These have been very intensively studied in the recent years, 
but to date have lacked a complete construction, which we establish here. 
Moreover, we show that the multiple $\SLE_\kappa$ pure partition functions are real-analytic in the parameter $\kappa \in (0,8)$, 
and that they decay to zero as a polynomial of $(8-\kappa)$ when $\kappa\to 8-$ (see Theorem~\ref{thm::PPF}). 
As a by-product, we also obtain a construction of global non-simple multiple chordal $\SLE_\kappa$ measures ($\kappa \in (4,8)$) uniquely determined by a re-sampling property. 
We show in Section~\ref{subsec::global_existence} that 
our construction agrees with the recent works~\cite{Zhan:Existence_and_uniqueness_of_nonsimple_multiple_SLE, AHSY:Conformal_welding_of_quantum_disks_and_multiple_SLE_the_non-simple_case} 
(see Theorem~\ref{thm::existenceglobalSLE}).
Third of all, we relate the Coulomb gas integrals and 
multiple $\SLE_\kappa$ pure partition functions together 
(in Propositions~\ref{prop::CGI_PPF}~\&~\ref{prop::CGI_PPF_renormalized}) 
in terms of the \emph{meander matrix} (the third key player, giving combinatorial coefficients, 
in Conjecture~\ref{conj::On}). 
This relation extends the earlier results in~\cite{FPW:Connection_probabilities_of_multiple_FK_Ising_interfaces}.

\subsection{Critical loop $O(n)$ models on planar domains}

Our first main objective is to give explicit conjectures concerning connection probabilities for loop $O(n)$ models at their critical regime. 
It is believed that in these models, the $O(n)$ loops converge to $\SLE_\kappa$ and $\CLE_\kappa$ processes, where $\kappa$ is related to $n$ as in~\eqref{eqn::On_kappa}. 
As part of the novel content in the present work, 
we derive key properties of the quantities appearing in the predicted formula, phrased in terms of Coulomb gas on the one hand, and 
$\SLE_\kappa$ partition function theory on the other hand.
Some of these properties are also needed in the second part of our work, where we construct the pure partition functions and global multiple chordal $\SLE_\kappa$ measures with the parameter range $\kappa\in (4,8)$.

\begin{figure}
\includegraphics[width=0.5\textwidth]{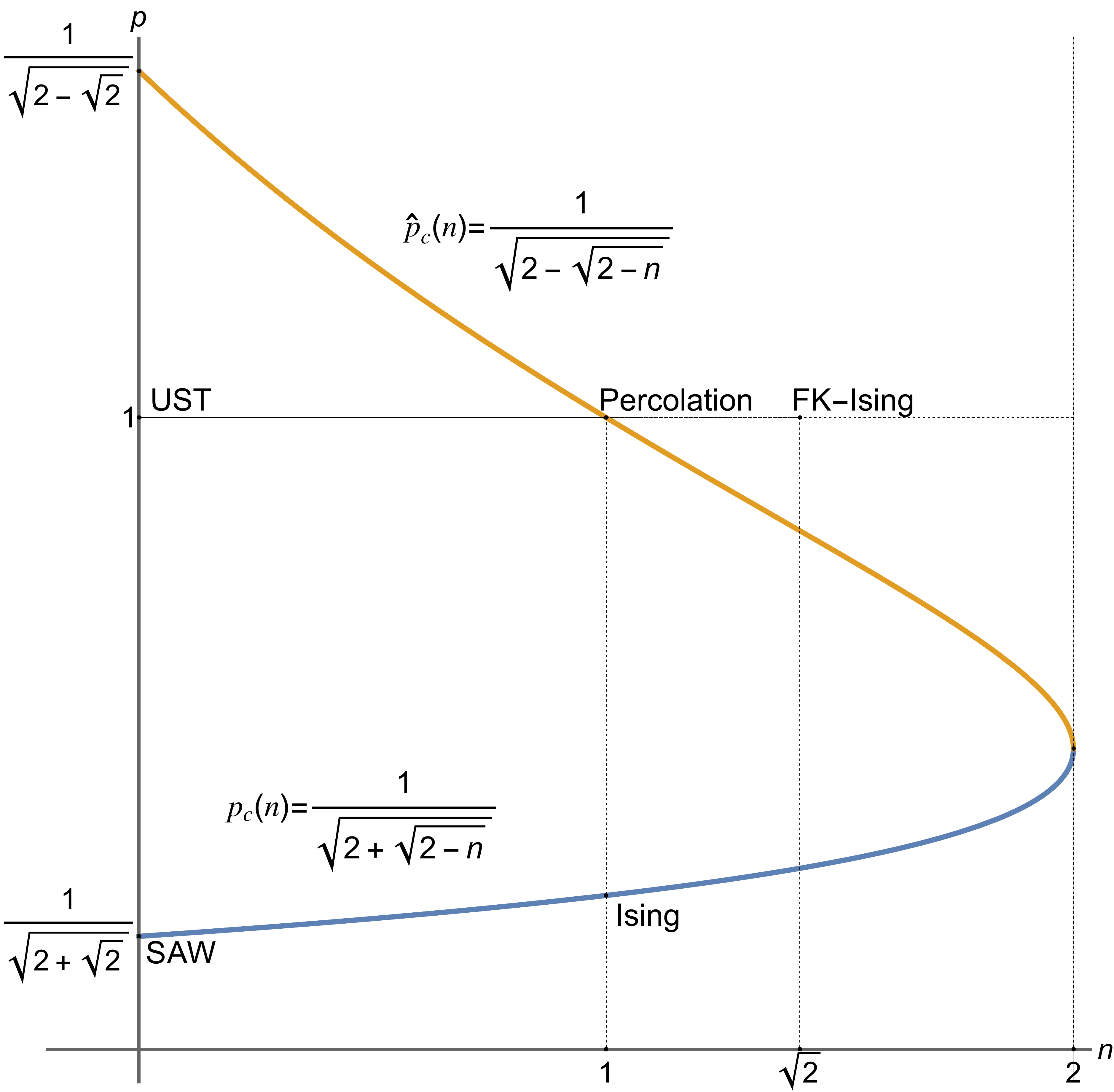}
\caption{\label{fig::Onphasediagram}
The $O(n)$ model has an interesting phase diagram. 
This figure plots the critical edge-weight $n\mapsto p_c(n)=(\sqrt{2+\sqrt{2-n}})^{-1}$ (in blue) 
and the critical edge-weight $n\mapsto \hat{p}_c(n)=(\sqrt{2-\sqrt{2-n}})^{-1}$ (in orange). 
Although the scaling limit when $p>p_c(n)$ is
conjecturally the same, at the line $\hat{p}_c(n)$ one should see higher order corrections to the critical behavior at the lattice level~\cite{Nienhuis:Exact_critical_point_and_exponents_of_the_On_model_in_two_dimensions}.
From the scaling limit, or SLE perspective, one can however only observe the critical phenomena at 
the critical line $p_c(n)$ (dilute phase) and above it (dense phase).}
\end{figure}

\subsubsection*{Definition of loop $O(n)$ models and boundary conditions}

Denote by $\hexa$ the hexagonal lattice. A loop is a finite subgraph of $\hexa$ which is isomorphic to a simple cycle. A loop configuration is a subgraph of $\hexa$ in which every vertex has even degree. For a loop configuration $\omega$, we denote by $\ell(\omega)$ the number of loops in $\omega$ and by $o(\omega)$ the number of edges of $\omega$.  Suppose $\Omega$ is a non-empty simply connected induced subgraph of $\hexa$. Let $n$ and $p$ be positive real numbers. The loop $O(n)$ measure on $\Omega$ with edge-weight $p$ is the probability measure on loop configurations defined by 
\begin{align*}
\PP[\omega]\propto p^{o(\omega)}n^{\ell(\omega)}.
\end{align*}
For recent progress on loop $O(n)$ models, see the survey~\cite{Peled-Spinka:Lectures_on_the_spin_and_loop_On_models} and references therein. 

It has been predicted by Nienhuis and Smirnov~\cite{Nienhuis:Exact_critical_point_and_exponents_of_the_On_model_in_two_dimensions, Kager-Nienhuis:Guide_to_SLE, Smirnov:Towards_conformal_invariance_of_2D_lattice_models} 
that the loop $O(n)$ model exhibits critical behavior when $n\in [0,2]$ with a critical edge-weight parameter $p_c(n)$ given by 
\begin{align*}
p_c(n)=\frac{1}{\sqrt{2+\sqrt{2-n}}}.
\end{align*}
Furthermore, there should be two critical regimes: $p=p_c$ and $p>p_c$. In both cases, the loops should scale in a suitable limit to $\SLE_\kappa$ processes with parameter $\kappa$ satisfying 
\begin{align}\label{eqn::On_kappa}
n= -2 \cos \big( 4\pi/\kappa\big). 
\end{align}
See Figure~\ref{fig::Onphasediagram} for the phase diagram.
The case $p=p_c(n)$ is called 
\emph{dilute phase} and the scaling limits of loops are simple,
while  the case $p>p_c(n)$ is called 
\emph{dense phase} 
and the scaling limits of loops are non-simple. 
Using these predictions, we will give precise Conjecture~\ref{conj::On} concerning connection probabilities of critical loop $O(n)$ interfaces, which also has a $\CLE_\kappa$ counterpart (Remark~\ref{rem:CLE}).
We begin with the setup.

\bigskip

For $N\ge 1$, by a \emph{(topological) polygon} we refer to a simply connected domain $\Omega \subsetneq \C$ with $2N$ distinct points $x_1, \ldots, x_{2N} \in \partial\Omega$ lying counterclockwise along the boundary, and we normally\footnote{At times, we might require more regularity, while on the other hand, 
one could also consider the boundary points in terms of prime ends, which however would not bring anything more to the present discussion.} assume that $\partial\Omega$ is locally connected. 
Throughout, we shall denote a generic polygon as $(\Omega; x_1, \ldots, x_{2N})$.
Let $(\Omega^{\delta}; x_1^{\delta}, \ldots, x_{2N}^{\delta})$ be a sequence of discrete polygons on $\delta\hexa$ approximating a given polygon $(\Omega; x_1, \ldots, x_{2N})$ as $\delta\to 0$.  
Consider the loop $O(n)$ model in the discrete polygon $(\Omega^{\delta}; x_1^{\delta}, \ldots, x_{2N}^{\delta})$ with the following boundary conditions, 
illustrated in Figure~\ref{fig::On_connection}(a): 
the $2N$ marked boundary points $x_1^{\delta}, \ldots, x_{2N}^{\delta}$ are connected pairwise outside of $\Omega^{\delta}$ according to a planar link pattern 
(non-crossing pair partition) 
\begin{align} \label{eq: link pattern ordering}
\begin{split}
& \beta = \{ \{a_1,b_1\},  \{a_2,b_2\},\ldots , \{a_N,b_N\}\} \; \in \; \LP_N \\
& \textnormal{with link endpoints ordered as } \; 
a_1 < a_2 < \cdots < a_N \textnormal{ and } a_r < b_r , \textnormal{ for all } 1 \leq r \leq N ,  \\
& \textnormal{and such that there are no indices } 1 \leq r , s \leq N \textnormal{ with } a_r < a_s < b_r < b_s , \end{split}
\end{align}
where $\{a_1, b_1,\ldots,  a_N, b_N  \} =  \{1,2,\ldots,2N\}$ and for convenience, we have chosen a particular ordering of the endpoints of the links $\{a_r,b_r\}$. 
We  also denote 
\begin{align*}
\LP := \bigsqcup_{N\geq 0} \LP_N .
\end{align*}

Let $\omega$ be a loop $O(n)$ configuration on $\Omega^{\delta}$ with boundary condition $\beta$, and denote its law by 
\begin{align*}
\PP_{\beta}^{\delta}=\PP_{\beta; \,(\Omega; x_1, \ldots, x_{2N})}^{\delta; \,(n, p)} . 
\end{align*}
In such a loop $O(n)$ configuration, there are $N$ curves (interfaces) 
on $\Omega^{\delta}$ connecting the marked points $x_1^{\delta}, \ldots, x_{2N}^{\delta}$ pairwise inside $\Omega^{\delta}$, as illustrated in Figure~\ref{fig::On_connection}(b). 
The connectivity of these $N$ interfaces induces a random link pattern $\conn^{\delta}$ in $\LP_N$. 
The connection probabilities $\PP_{\beta}^{\delta}[\conn^{\delta}=\alpha]$, for $\alpha\in\LP_N$, 
are expected to have an explicit conformally invariant formula 
in the scaling limit $\delta \to 0$ 
expressed in terms of 
combinatorial numbers encoded in the meander matrix (defined below) 
and ratios of specific $\SLE_\kappa$ partition functions 
(or certain CFT correlation functions~\cite{Peltola:Towards_CFT_for_SLEs}), 
which we will introduce and analyze thoroughly in this article. 
A similar but more restricted conjecture appeared in~\cite{FSKZ:A_formula_for_crossing_probabilities_of_critical_systems_inside_polygons, 
FPW:Connection_probabilities_of_multiple_FK_Ising_interfaces}.

\begin{figure}[ht!]
\begin{subfigure}[b]{\textwidth}
\begin{center}
\includegraphics[width=0.3\textwidth]{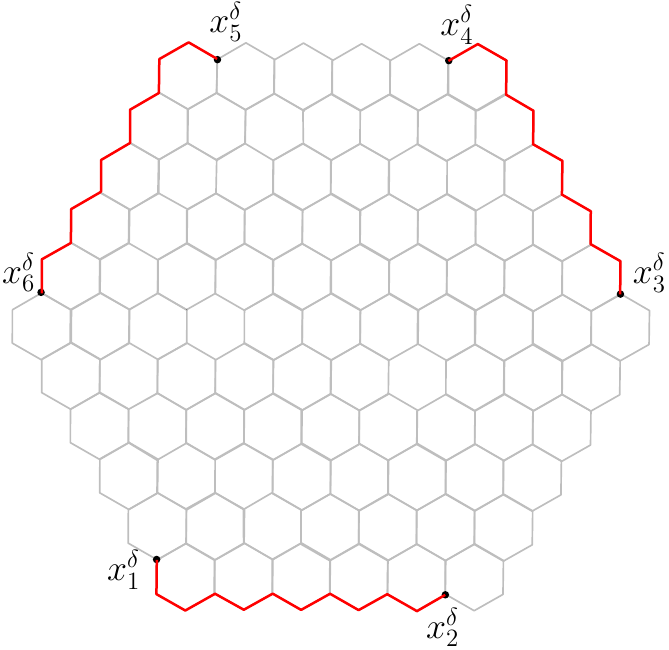}
\end{center}
\caption{Boundary condition $\beta=\{\{1,2\}, \{3,4\}, \{5,6\}\}$.}
\end{subfigure}\\
\vspace{0.5cm}
\begin{subfigure}[b]{0.19\textwidth}
\begin{center}
\includegraphics[width=\textwidth]{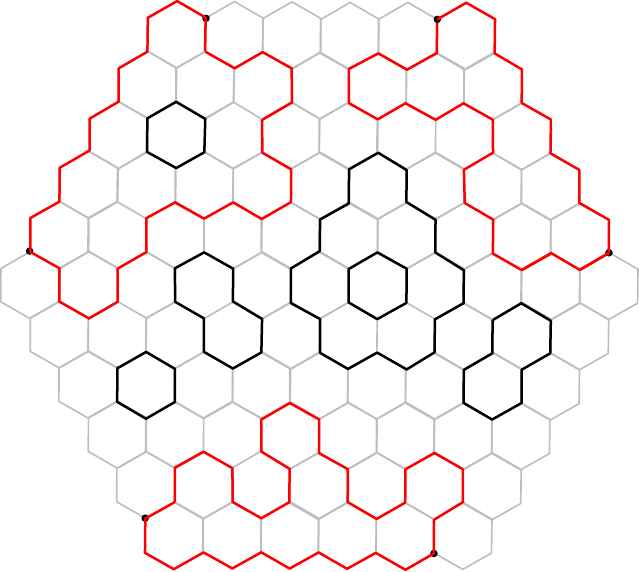}
\end{center}
\caption{}
\end{subfigure}
\begin{subfigure}[b]{0.19\textwidth}
\begin{center}
\includegraphics[width=\textwidth]{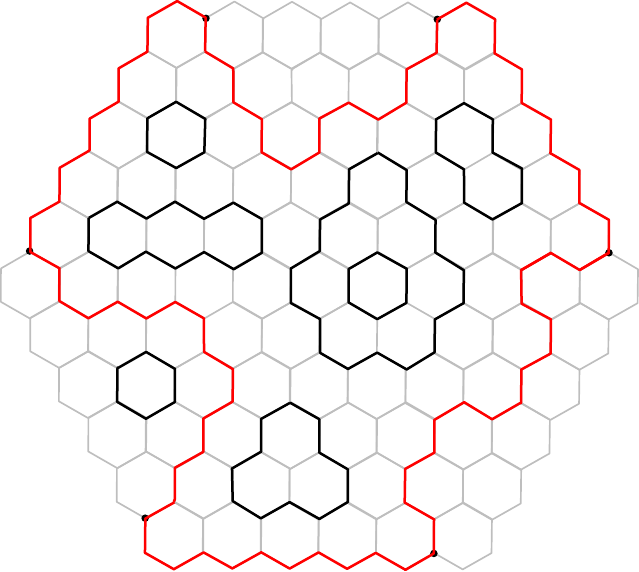}
\end{center}
\caption{}
\end{subfigure}
\begin{subfigure}[b]{0.19\textwidth}
\begin{center}
\includegraphics[width=\textwidth]{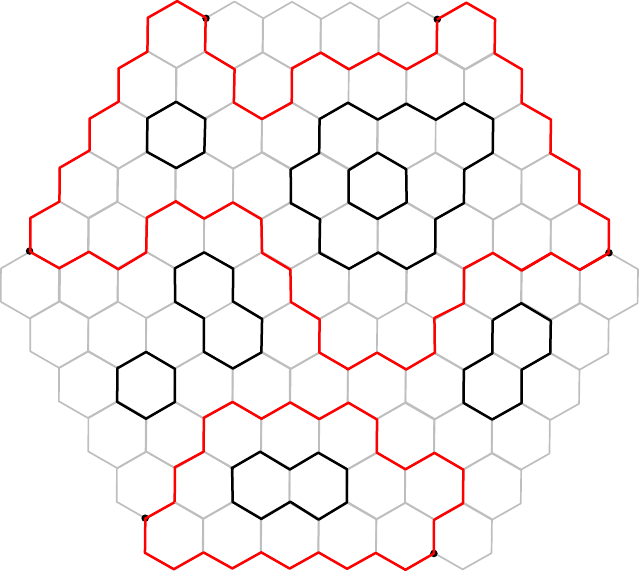}
\end{center}
\caption{}
\end{subfigure}
\begin{subfigure}[b]{0.19\textwidth}
\begin{center}
\includegraphics[width=\textwidth]{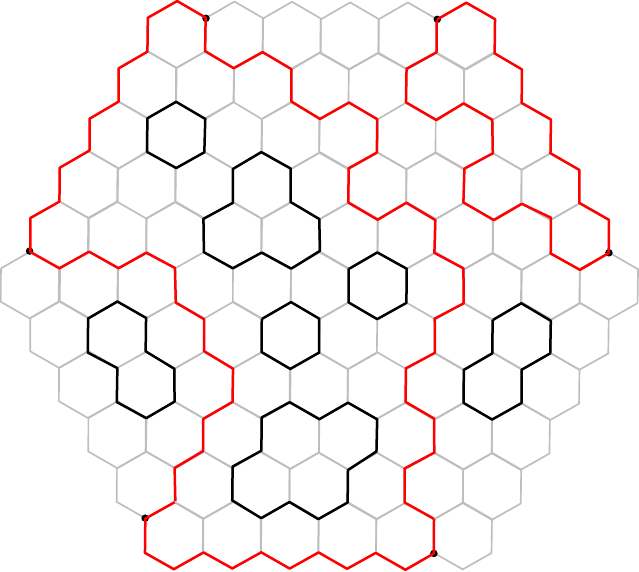}
\end{center}
\caption{}
\end{subfigure}
\begin{subfigure}[b]{0.19\textwidth}
\begin{center}
\includegraphics[width=\textwidth]{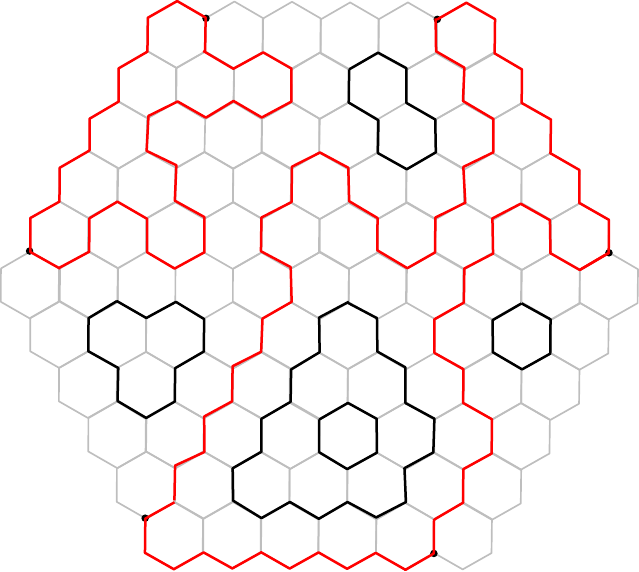}
\end{center}
\caption{}
\end{subfigure}

\caption{\label{fig::On_connection} In~(a), we illustrate a boundary condition $\beta=\{\{1,2\}, \{3,4\}, \{5,6\}\}$ for the $O(n)$ loop model. In~(b)-(f), we have five possible internal link patterns of the interfaces (in red): $\{\{1,2\}, \{3,4\}, \{5,6\}\}$, $\{\{1,6\}, \{2,3\}, \{4,5\}\}$, $\{\{1,2\}, \{3,6\}, \{4,5\}\}$, $\{\{1,6\}, \{2,5\}, \{3,4\}\}$, and $\{\{1,4\}, \{2,3\}, \{5,6\}\}$, respectively. 
The numbers of loops in the meander formed from the boundary condition and the internal link patterns are 3,1,2,2, and 2, respectively.}
\end{figure}

\subsubsection*{Conjecture for critical loop $O(n)$ models}

\begin{definition}\label{def::meandermatrix}
A \emph{meander} formed from two link patterns $\alpha,\beta\in\LP_N$ is the planar diagram obtained by placing $\alpha$ and the horizontal reflection $\beta$ on top of each other. We denote by $\ell(\alpha,\beta)$ the number of loops in the meander formed from $\alpha$ and $\beta$.  We define the \emph{meander matrix} $\{\meanderMatrix(\alpha, \beta) \colon \alpha,\beta\in\LP_N\}$ via 
\begin{align} \label{eqn::meandermatrix_def}
\meanderMatrix(\alpha,\beta) := \fugacity^{\ell(\alpha,\beta)} , \qquad \fugacity \in \C . 
\end{align}
\end{definition}

\begin{conjecture}\label{conj::On}
Consider the loop $O(n)$ model with $n\in (0,2]$ and with boundary condition $\beta\in\LP_N$. 
Denote by $\conn^{\delta}$ the random  connectivity in $\LP_N$ induced by the $N$ discrete interfaces in $\Omega^{\delta}$. 
\begin{enumerate}[leftmargin=1em]
\item 
\label{item::On_dilute}
For the critical dilute phase with $p=p_c$, let $\kappa$ be the solution to~\eqref{eqn::On_kappa} with $\kappa\in (8/3,4]$. 

\item 
\label{item::On_dense}
For the critical dense phase with $p>p_c$, let $\kappa$ be the solution to~\eqref{eqn::On_kappa} with $\kappa\in [4,8)$. 
\end{enumerate}
\noindent 
Then, the following conformally invariant connection probability formula holds\textnormal{:}
\begin{align} \label{eqn::On_connproba}
\begin{split}
\lim_{\delta\to 0} \PP_{\beta}^{\delta}[\conn^{\delta} = \alpha] 
= \; & \meanderMat_{n}(\alpha,\beta) \, \frac{\PartF_{\alpha}^{(\kappa)}(\realpt_1, \ldots, \realpt_{2N})}{\coulombGas_{\beta}^{(\kappa)}(\realpt_1, \ldots, \realpt_{2N})}
, \qquad \textnormal{for any } \alpha\in\LP_N,
\end{split}
\end{align}
where $\meanderMat_{n}(\alpha,\beta)$ is the meander matrix in Definition~\ref{def::meandermatrix} with $\fugacity=n$, 
and $\smash{\PartF_{\alpha}^{(\kappa)}(\realpt_1, \ldots, \realpt_{2N})}$ the pure partition function in Definition~\ref{def::PPF_general}, 
and $\smash{\coulombGas_{\beta}^{(\kappa)}(\realpt_1, \ldots, \realpt_{2N})}$ the Coulomb gas integral in Definition~\ref{def::CGI_def},
and $\realpt_j = \varphi(x_j)$ 
images of the marked points under 
any conformal map $\varphi \colon \Omega \to \HH = \{ z \in \C \colon \Im(z) > 0 \}$.
\end{conjecture}

The connection probabilities in Conjecture~\ref{conj::On} are also predicted in~\cite{FSKZ:A_formula_for_crossing_probabilities_of_critical_systems_inside_polygons}.
The aim of this article is to present them in a framework that is directly amenable to investigations using $\SLE_\kappa$ and $\CLE_\kappa$ processes.  
Some cases of Conjecture~\ref{conj::On} have been already proved rigorously in the literature:

\begin{itemize}[leftmargin=1em]
\item Critical loop $O(n)$ model with $n=1$ and $p=1$ corresponds to  critical \emph{Bernoulli percolation}, 
for which the conformal invariance is known~\cite{Smirnov:Critical_percolation_in_the_plane, 
Camia-Newman:2D_percolation_full_scaling_limit, 
Camia-Newman:Critical_percolation_exploration_path_and_SLE_proof_of_convergence}. 
In this case~\eqref{eqn::On_connproba} reads
\begin{align}\label{eqn::perco_connproba}
\lim_{\delta\to 0}\PP_{\beta}^{\delta}[\conn^{\delta}=\alpha] 
= \PartF_{\alpha}^{(6)}(\realpt_1, \ldots, \realpt_{2N}) ,
\end{align}
because $\meanderMat_{n}(\alpha,\beta) \equiv 1$ and $\coulombGas_{\beta}^{(6)} \equiv 1$ as in~\eqref{eqn::CGI_kappa6}. 
The case with $N=2$ is known as Cardy-Smirnov formula~\cite{Cardy:Critical_percolation_in_finite_geometries, Smirnov:Critical_percolation_in_the_plane}. 
The convergence in~\eqref{eqn::perco_connproba} with $N\ge 3$ was proven in~\cite{FPW:Crossing_probabilities_of_critical_percolation_interfaces}.

\item Critical loop $O(n)$ model with $n=\sqrt{2}$ and $p=1$ corresponds to the critical \emph{random-cluster model}  
with cluster-weight $2$, for which the conformal invariance  is  known on $\Z^2$~\cite{Chelkak-Smirnov:Universality_in_2D_Ising_and_conformal_invariance_of_fermionic_observables,  CDHKS:Convergence_of_Ising_interfaces_to_SLE}. 
Indeed, Chelkak~\&~Smirnov proved~\eqref{eqn::On_connproba} in the case $N=2$~\cite{Chelkak-Smirnov:Universality_in_2D_Ising_and_conformal_invariance_of_fermionic_observables}, 
and the case 
$N\ge 3$ was proven in~\cite[Theorem~1.8]{FPW:Connection_probabilities_of_multiple_FK_Ising_interfaces}.
In this case, the denominator $\smash{\coulombGas_{\beta}^{(16/3)}}$ has an explicit expression~\eqref{eqn::CGI_kappa16over3}.

\item Critical loop $O(n)$ model with $n\to 0$ and $p=1$ can be understood as a \emph{uniform spanning tree} (UST), 
for which the conformal invariance  is also  known on $\Z^2$~\cite{LSW:Conformal_invariance_of_planar_LERW_and_UST}. 
An analogue of the convergence in~\eqref{eqn::On_connproba} is proved in~\cite[Theorem~1.3]{LPW:UST_in_topological_polygons_partition_functions_for_SLE8_and_correlations_in_logCFT}, where one replaces $\meanderMat_{n}(\alpha,\beta)$ by the renormalized meander matrix $\meanderRenorm(\alpha,\beta)$ in~\eqref{eqn::meandermatrix_renormalized},
and one replaces $\smash{\PartF_{\alpha}^{(\kappa)}}$ and $\smash{\coulombGas_{\beta}^{(\kappa)}}$ by renormalized partition functions $\smash{\PartFRenorm_{\alpha}^{(8)}}$ and $\smash{\coulombGasRenorm_{\beta}^{(8)}}$ in~\eqref{eqn::PPF_continuity8} and~\eqref{eqn::CGI_PPF_8}, respectively.  
Related probabilities for UST models were also considered by Dub\'edat~\cite{Dubedat:Euler_integrals_for_commuting_SLEs} and Kenyon~\&~Wilson~\cite{Kenyon-Wilson:Boundary_partitions_in_trees_and_dimers}.

\item Critical loop $O(n)$ model with $n=1$ and $p=p_c(1)=1/\sqrt{3}$ corresponds to the critical (spin) \emph{Ising model}, 
for which the conformal invariance on $\Z^2$ is known as well~\cite{Chelkak-Smirnov:Universality_in_2D_Ising_and_conformal_invariance_of_fermionic_observables, CDHKS:Convergence_of_Ising_interfaces_to_SLE}. 
The convergence in~\eqref{eqn::On_connproba} with $N=2$ was proven by Chelkak~\&~Smirnov
in~\cite{Chelkak-Smirnov:Universality_in_2D_Ising_and_conformal_invariance_of_fermionic_observables},
and the case $N\ge 3$ was treated in~\cite{Peltola-Wu:Crossing_probabilities_of_multiple_Ising_interfaces}.
In this case, we have 
$\meanderMat_{n}(\alpha,\beta)=1$ and $\smash{\coulombGas_{\beta}^{(3)}}$ has an explicit expression~\eqref{eqn::CGI_kappa3}.  
\end{itemize}

\begin{remark} \label{rem:CLE}
Conformal loop ensemble, $\CLE_{\kappa}$, is argued~\cite{Nienhuis:Exact_critical_point_and_exponents_of_the_On_model_in_two_dimensions,
Kager-Nienhuis:Guide_to_SLE, Sheffield-Werner:CLEs} 
to be the only possible conformally invariant scaling limit for the joint law of all of the macroscopic loops for critical loop $O(n)$ model with ``uniform" free boundary condition, where $\kappa\in (8/3,8)$ and $n\in (0,2]$ are related through~\eqref{eqn::On_kappa}. 
In~\cite{Miller-Werner:Connection_probabilities_for_conformal_loop_ensembles}, 
Miller~\&~Werner introduced $\CLE_{\kappa}$ in a rectangle with alternating boundary conditions, and derived connection probabilities for $\CLE_{\kappa}$ loops which coincide with the right-hand side of~\eqref{eqn::On_connproba} with $N=2$.
In general, we can introduce the $\CLE_{\kappa}$ in arbitrary polygons with alternating boundary conditions,
by utilizing a similar exploration process as described in~\cite{Miller-Werner:Connection_probabilities_for_conformal_loop_ensembles}.
We can then ask what the general connection probabilities are. 
In light of the conjecture that critical $O(n)$ loops should converge to $\CLE_{\kappa}$ loops, 
it is natural to expect that the corresponding probabilities are given by the right-hand side of~\eqref{eqn::On_connproba} for all $N \geq 2$, as also conjectured in~\cite{AMY:MultipleSLEfromCLE}. 
\end{remark}

Finally, let us discuss the loop $O(n)$ model as $n\to 0$ and $p=p_c(0)=(\sqrt{2+\sqrt{2}})^{-1}$, which is closely related to self-avoiding walk. 
In the discrete polygon $(\Omega^{\delta}; x_1^{\delta}, \ldots, x_{2N}^{\delta})$, we consider configurations $\omega$ in $\Omega^{\delta}$ such that there are $N$ disjoint self-avoiding paths connecting pairwise among the marked points $(x_1^{\delta}, \ldots, x_{2N}^{\delta})$ and there are no loops. 
Denote by $\conn^{\delta}(\omega)$ the connectivity in $\LP_N$ induced by these $N$ discrete paths. 
Denote by $\LW_{\alpha}(\Omega^{\delta}; x_1^{\delta}, \ldots, x_{2N}^{\delta})$ the set of all configurations $\omega$ such that $\conn^{\delta}(\omega)=\alpha$. 
As a generalization of~\cite{DCS:Connective_constant_of_honeycomb_lattice} for $N=1$, 
the partition function of these walks 
(the number of configurations realizing the event $\{\conn^{\delta} = \alpha \}$)
should behave in the following manner\textnormal{:}
\begin{align}\label{eqn::SAW_conjecture}
\delta^{-\frac{5}{4}N} \, \sum_{\omega\in\LW_{\alpha}(\Omega^{\delta}; x_1^{\delta}, \ldots, x_{2N}^{\delta})}\Big(\frac{1}{\sqrt{2+\sqrt{2}}}\Big)^{o(\omega)} 
\quad \overset{\delta \to 0}{\approx} \quad 
\PartF_{\alpha}^{(8/3)}(\realpt_1, \ldots, \realpt_{2N}) 
, \qquad \alpha \in \LP_N ,
\end{align}
whose right-hand side is conformally covariant by Theorem~\ref{thm::PPF}. 
One possible way to understand the notation ``$\approx$'' is as 
\begin{align*}
\lim_{\delta\to 0} 
\frac{1}{\log\delta} 
\Big( \log(\textnormal{LHS of}~\eqref{eqn::SAW_conjecture}) 
- \log\PartF_{\alpha}^{(8/3)}(\realpt_1, \ldots, \realpt_{2N}) \Big)
= \frac{5}{4} N .
\end{align*}

\subsection{Coulomb gas integrals and their behavior in $\kappa$}
\label{subsec::CGI}

Representation theory of the Virasoro algebra provides a way to find solutions to PDEs appearing in conformal field theory and its applications to random geometry. 
In CFT, this approach is commonly known as ``Coulomb gas formalism''~\cite{DMS:CFT}, and it formally builds upon the free boson (Gaussian free field, GFF). 
While we will not use the GFF in the present work, we use the Coulomb gas formalism as a motivating framework to write down explicit formulas for $\SLE_\kappa$ partition functions, in the spirit of~\cite{Dubedat:Euler_integrals_for_commuting_SLEs, Kytola-Peltola:Pure_partition_functions_of_multiple_SLEs,
Flores-Kleban:Solution_space_for_system_of_null-state_PDE3}.
In this context, it is customary\footnote{Here, we opt for the notation that we find most commonly in CFT literature. Note that the charges $\alpha_0$ and $\alpha_i$ have nothing to do with link patterns, and we trust that the notation is clear from context.} 
to introduce parameters depending on $\kappa > 0$: 
\begin{itemize}[leftmargin=1em]
\item \emph{central charge} $c = c(\kappa)$ and \emph{background charge} $\alpha_0 = \alpha_0(\kappa)$:
\begin{align}\label{eqn::centralcharge}
c = \frac{(3 \kappa - 8)(6 - \kappa)}{2 \kappa} = 1 - 24 \alpha_0^2 
\qquad \textnormal{and} \qquad 
\alpha_0 = \frac{\kappa - 4}{4 \sqrt{\kappa}} ;
\end{align}

\item \emph{Coulomb charges} $\alpha_j = \alpha_j(\kappa) \in \R$ at the marked points $x_j \in \{x_1, x_2, \ldots, x_{2N}\}$;

\item \emph{screening charges} $\alpha_- = \alpha_-(\kappa)$ and $\alpha_+ = \alpha_+(\kappa)$: 
\begin{align*}
\alpha_- = - \frac{2}{\sqrt{\kappa}} 
\qquad \textnormal{and} \qquad 
\alpha_+ = \frac{\sqrt{\kappa}}{2}
\qquad \textnormal{such that} \qquad 
\alpha_- \alpha_+ = -1;
\end{align*}

\item \emph{conformal weights} $h_j = h_j(\kappa)$ at the marked points $x_j \in \{x_1, x_2, \ldots, x_{2N}\}$ are given by $\alpha_j(\kappa)$ as
\begin{align*}
h_j = \alpha_j^2 - 2 \alpha_j \alpha_0 .
\end{align*}
\end{itemize}
\noindent 
Choosing $\alpha_j = 1/\sqrt{\kappa}$ for all $j \in \{1, 2, \ldots, 2N\}$ 
(this is the charge associated to a degenerate CFT field at level two, known to be closely related to $\SLE_\kappa$ curves), the Coulomb gas integral 
\begin{align*}
\int \cdots \int 
\; \prod_{1\leq i<j\leq 2N}(x_{j}-x_{i})^{2 \alpha_i \alpha_j} 
\prod_{1\leq r<s\leq N}(u_{s}-u_{r})^{2 \alpha_-^2} 
\prod_{\substack{1\leq i\leq 2N \\ 1\leq r\leq N}}
(u_{r}-x_{i})^{2 \alpha_- \alpha_i}
\; \ud u_1 \cdots \ud u_N , 
\end{align*}
plays an important role in SLE theory. 
The conformal weight for a degenerate field at level two is 
\begin{align} \label{eq::degenerate_weight}
h_j = h(\kappa) := h_{1,2}(\kappa) := \frac{6 - \kappa}{2 \kappa} . 
\end{align}

\begin{figure}
\includegraphics[width=0.45\textwidth]{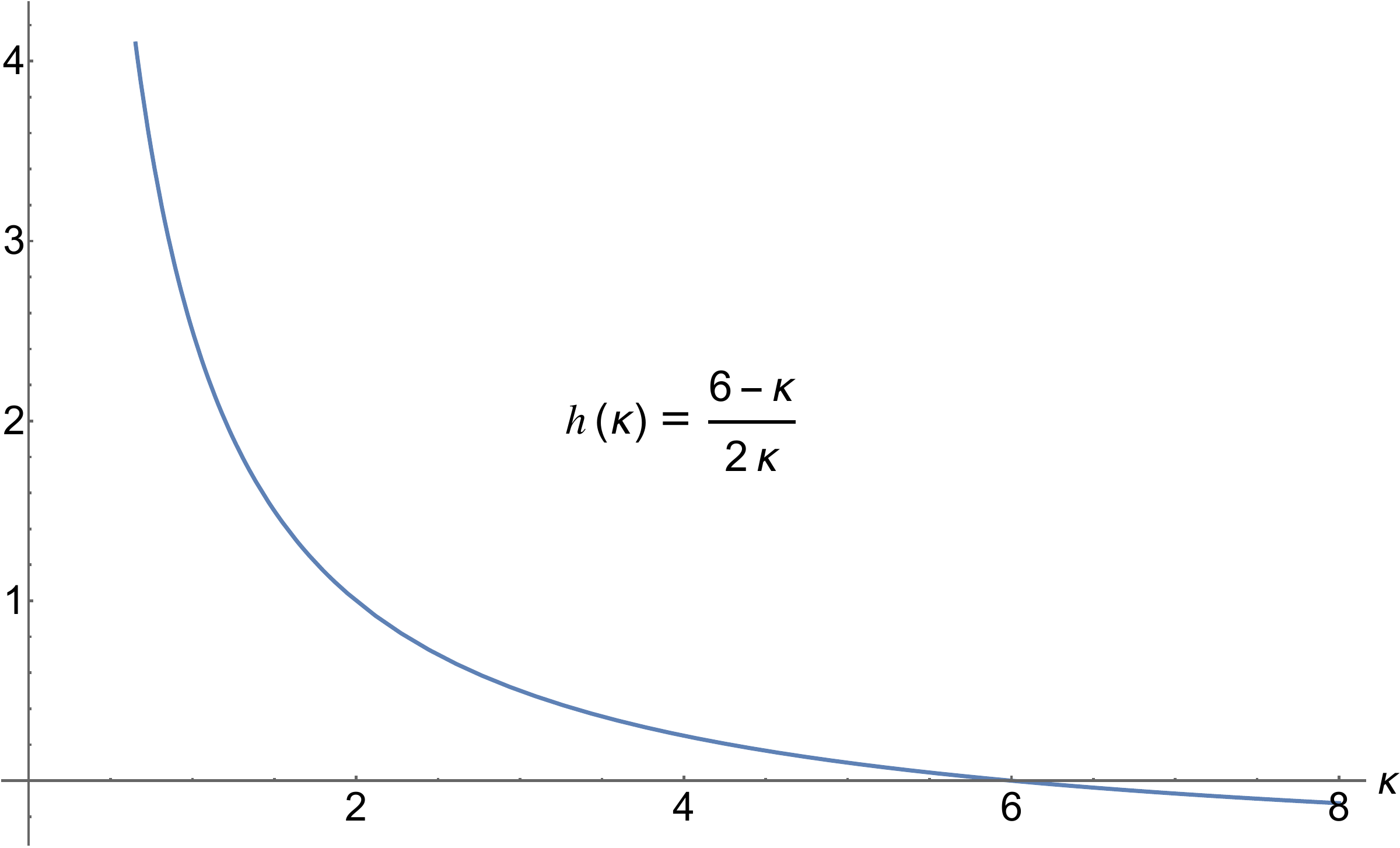}
\qquad
\includegraphics[width=0.45\textwidth]{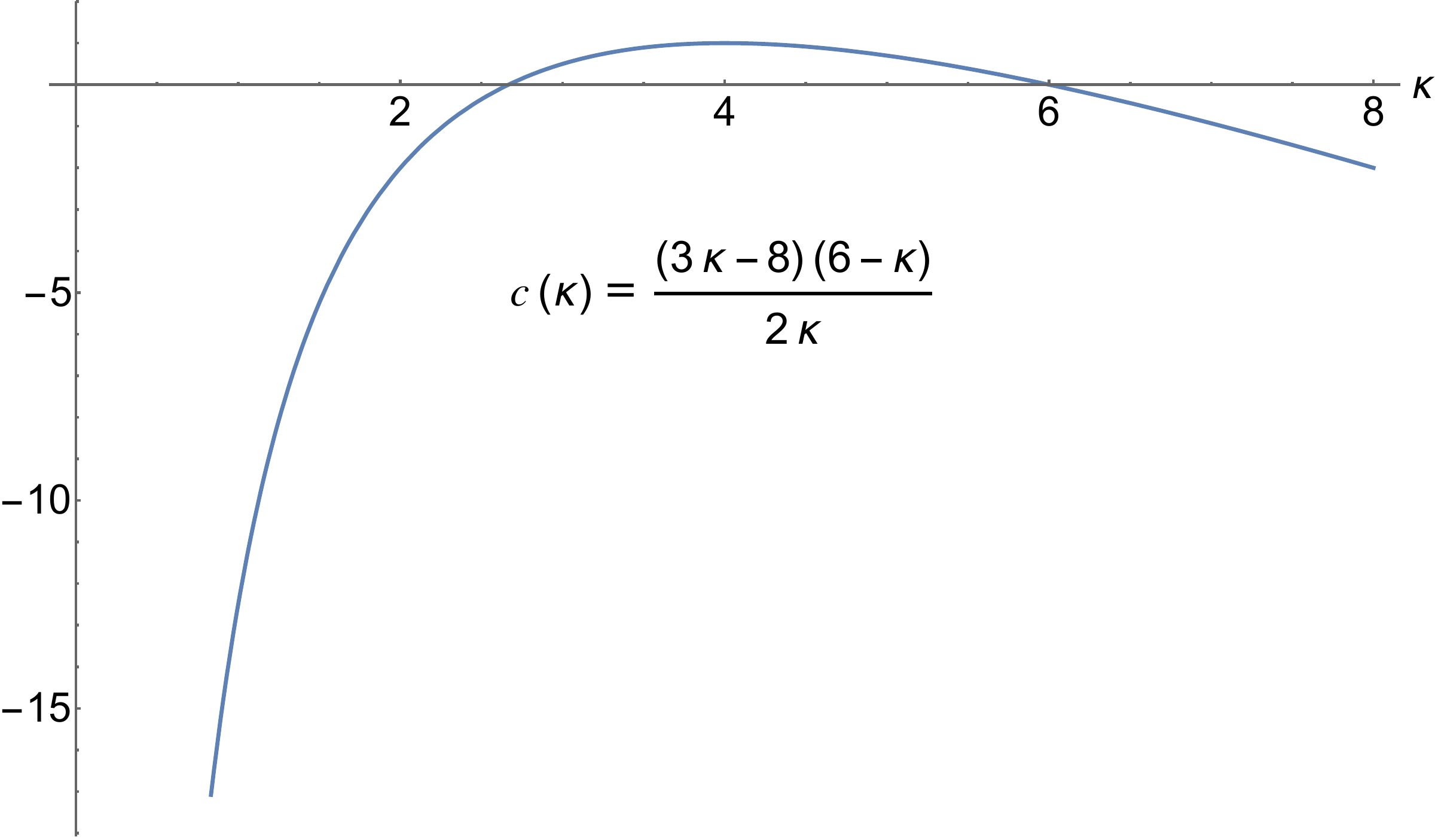}
\caption{
The left panel depicts a plot of the function $\kappa \mapsto h(\kappa)$. Note that $h(6) = 0$.
The right panel depicts a plot of the central charge $\kappa \mapsto c(\kappa)$. Note that $c(8/3)=c(6)=0$. }
\end{figure}

\begin{definition} \label{def::CGI_def}
For each $\beta \in \LP_N$ ordered as in~\eqref{eq: link pattern ordering}, we define the \emph{generic Coulomb gas integral} functions on the configuration space 
\begin{align}\label{eqn::chamber}
\chamber_{2N} := \big\{ \bs{x} := (x_{1},\ldots,x_{2N}) \in \R^{2N} \colon x_{1} < \cdots < x_{2N} \big\} 
\end{align}
as $\coulombGas_{\beta}^{(\kappa)} \colon \chamber_{2N} \to \C$, 
\begin{align}
\label{eqn::CGI_def}
\coulombGas_{\beta}^{(\kappa)} (\bs{x}) :=  \; &
\cst(\kappa)^N \, 
\ointclockwise_{\acycle^\beta_1}  \ud u_1 \ointclockwise_{\acycle^\beta_2}  \ud u_2 \cdots \ointclockwise_{\acycle^\beta_N} \ud u_N 
\; f_\beta^{(\kappa)}(\bs{x};\bs{u}) , \qquad 
\kappa\in (0,\infty) \setminus \big\{ \tfrac{8}{m} \colon m \in \bZpos \big\} ,
\end{align}
where the normalization factor is given by (see also Figure~\ref{fig:constant})
\begin{align}\label{eq: normalization_cst}
\cst(\kappa) := \frac{\fugacity(\kappa)}{4\sin^2(4\pi/\kappa)} 
\, \frac{\Gamma(2-8/\kappa)}{\Gamma(1-4/\kappa)^2} ,
\qquad \textnormal{with fugacity} \qquad 
\fugacity(\kappa) := -2\cos(4\pi/\kappa) ,
\end{align} 
the (multivalued) integrand is 
\begin{align} \label{eq: integrand}
f_\beta^{(\kappa)} (\bs{x};\bs{u}) := \; &
\prod_{1\leq i<j\leq 2N}(x_{j}-x_{i})^{2/\kappa} 
\prod_{1\leq r<s\leq N}(u_{s}-u_{r})^{8/\kappa} 
\prod_{\substack{1\leq i\leq 2N \\ 1\leq r\leq N}}
(u_{r}-x_{i})^{-4/\kappa} , 
\end{align} 
with $\bs{u} = (u_1,\ldots,u_N)$ disjoint  
and with branch chosen to be real and positive when
\begin{align} \label{eq: branch choice}
x_{a_r} < \Re(u_r) < x_{a_r+1} , \qquad \textnormal{ for all } 1 \leq r \leq N ,
\end{align}
and where each $\acycle^\beta_r$ in~\eqref{eqn::CGI_def} is a Pochhammer contour 
which surrounds each of the points $x_{a_r}, x_{b_r}$ once in the positive direction and once in the negative direction (using the notation~\eqref{eq: link pattern ordering}):
\begin{align}  \label{eq:Pochhammer}
\acycle^\beta_r
 \quad = \quad \acycle(x_{a_r}, x_{b_r})
 \quad = \quad \vcenter{\hbox{\includegraphics[scale=0.275]{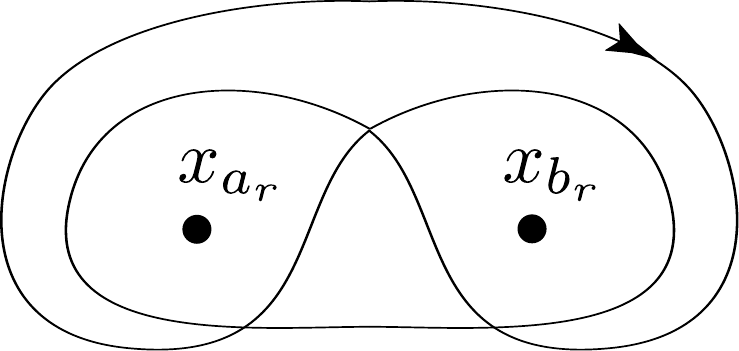}}} 
\end{align} 
and which does not surround any other marked point among $\{x_1, \ldots, x_{2N}\}$,
and lastly, all of the contours $\acycle^\beta_1, \acycle^\beta_2, \ldots, \acycle^\beta_N$
are assumed to be disjoint --- see Figure~\ref{fig:Pochhammer_examples}. 
\end{definition} 

\begin{figure}
\includegraphics[width=0.45\textwidth]{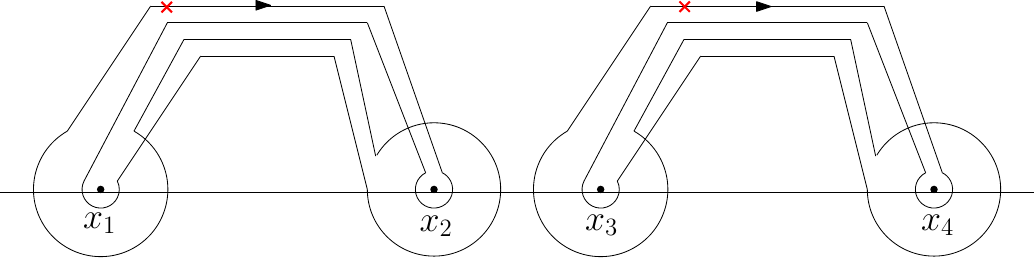}
$\qquad$
\includegraphics[width=0.45\textwidth]{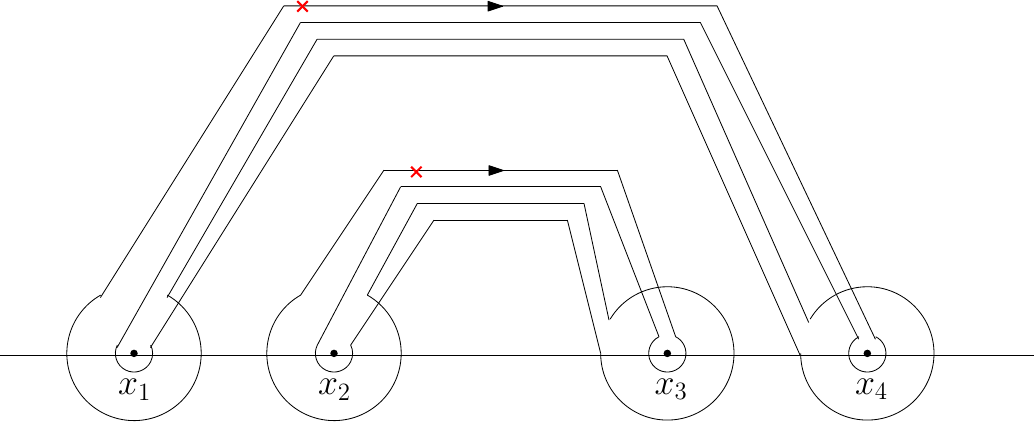}
\caption{\label{fig:Pochhammer_examples}
The two pairs of Pochhammer integration contours 
$\{\acycle(x_1,x_2), \acycle(x_3,x_4)\}$ and $\{\acycle(x_1,x_4), \acycle(x_2,x_3)\}$ used to define the functions $\smash{\coulombGas}^{(\kappa)}$ 
associated to the link patterns $\{\{1,2\},\{3,4\}\}$ and 
$\{\{1,4\},\{2,3\}\}$ in Definition~\ref{def::CGI_def}.
The branch choice~\eqref{eq: branch choice} is illustrated by a red cross.
}
\end{figure}

In other words, the integration contours for $\smash{\coulombGas_{\beta}^{(\kappa)}}$ 
in~\eqref{eqn::CGI_def} are determined by pairwise non-intersecting 
Pochhammer contours surrounding the marked points pairwise according to the connectivity $\beta$. 
Note that since the integration contours $\acycle^\beta_r$ avoid the marked points $x_{1},\ldots,x_{2N}$, the integral in~\eqref{eqn::CGI_def} 
is convergent for all $\kappa > 0$. 
We can also extend $\smash{\coulombGas_{\beta}^{(\kappa)}}$ to a multivalued function on the larger set
\begin{align} \label{eq: extended chamber}
\mathfrak{Y}_{2N} := \{ \bs{x} = (x_{1},\ldots,x_{2N}) \in \C^{2N} \colon x_{i} \neq x_{j} \textnormal{ for all } i \neq j \} .
\end{align}
In a few special cases, the Coulomb gas integrals $\smash{\coulombGas_{\beta}^{(\kappa)}}$ reduce to algebraic expressions --- see Remark~\ref{rem::CGI_examples}. 

\smallbreak

We have excluded particular values of $\kappa$ in~\eqref{eqn::CGI_def} because the constant $\cst(\kappa)$ is not well-defined at those exceptional points.
Nevertheless, it turns out that the function $\smash{\coulombGas_{\beta}^{(\kappa)}}$ is well-defined for all $\kappa \in (0,8)$. 
It~follows from~\cite[Propositions~2.2~\&~2.3]{FPW:Connection_probabilities_of_multiple_FK_Ising_interfaces}
that each $\smash{\coulombGas_{\beta}^{(\kappa)}}$ 
satisfies the following properties.

\begin{itemize}[leftmargin=3em]
\item[\textnormal{(PDE)}] 
\textnormal{\bf BPZ equations},  important in both CFT and SLE theory: 
\begin{align}\label{eqn::PDE}
\hspace*{-5mm}
\bigg[ 
\frac{\kappa}{2}\pdder{x_{j}}
+ \sum_{i \neq j} \Big( \frac{2}{x_{i}-x_{j}} \pder{x_{i}} 
- \frac{2h(\kappa)}{(x_{i}-x_{j})^{2}} \Big) \bigg]
F(x_1,\ldots,x_{2N}) =  0 , \qquad \textnormal{for all } j \in \{1,\ldots,2N\} .
\tag{\textnormal{PDE}}
\end{align}
\end{itemize}
\noindent 
These PDEs appear in SLE theory as a manifestation of martingales yielding SLE variants obtained from the chordal $\SLE_\kappa$~\cite{BBK:Multiple_SLEs_and_statistical_mechanics_martingales, 
Dubedat:Euler_integrals_for_commuting_SLEs, 
Dubedat:Commutation_relations_for_SLE, 
Graham:Multiple_SLEs, 
Kozdron-Lawler:Configurational_measure_on_mutually_avoiding_SLEs,
Lawler:Partition_functions_loop_measure_and_versions_of_SLE},
as well as in CFT literature~\cite{BPZ:Infinite_conformal_symmetry_in_2D_QFT, BPZ:Infinite_conformal_symmetry_of_critical_fluctuations_in_2D} 
emerging from the algebraic structure resulting from the infinitesimal conformal symmetry of planar critical models.
Notably, these PDEs are hypoelliptic, in terms of the H\"ormander bracket condition~\cite{Dubedat:SLE_and_Virasoro_representations_localization, Peltola-Wu:Global_and_local_multiple_SLEs_and_connection_probabilities_for_level_lines_of_GFF, AHSY:Conformal_welding_of_quantum_disks_and_multiple_SLE_the_non-simple_case, Karrila-Viitasaari:In_prep}.
In a recent work~\cite{ChenHuangWuCommutationChordal}, the authors investigate the solution space for the BPZ equations and show that its dimension is $2^{2N}$. Furthermore, they derive a hierarchy structure on the subspaces selected by conformal Ward identities (covariance properties~\eqref{eqn::COV}).
\begin{itemize}[leftmargin=3em]
\item[\textnormal{(COV)}] 
\textnormal{\bf M\"{o}bius covariance:} 
\begin{align}\label{eqn::COV}
F(x_{1},\ldots,x_{2N}) = 
\prod_{i=1}^{2N} \varphi'(x_{i})^{h(\kappa)} 
\times F(\varphi(x_{1}),\ldots,\varphi(x_{2N})) ,
\tag{\textnormal{COV}}
\end{align}
for all M\"{o}bius maps $\varphi$ of the upper half-plane $\HH$ 
such that $\varphi(x_{1}) < \cdots < \varphi(x_{2N})$.
\end{itemize}
\noindent 
Such a covariance property gives a version of global conformal symmetry, manifest in the conformal invariance of the $\SLE_{\kappa}$ process on the one hand, and for CFT primary fields on the other hand.

\begin{itemize}[leftmargin=3em]
\item[\textnormal{(PLB)}] 
\textnormal{\bf Power-law bound:} 
There exist constants $C>0$ and $p>0$ such that for all 
$N \geq 1$ and $\bs{x} \in \chamber_{2N}$, 
\begin{align}\label{eqn::PLB_weak_upper}
|F(\bs{x})| \le 
C \prod_{1 \leq i<j \leq 2N}(x_j-x_i)^{\mu_{ij}(p)}, 
\quad \textnormal{where } \quad
\mu_{ij}(p) :=
\begin{cases}
p, \quad &\textnormal{if } |x_j-x_i| > 1, \\
-p, \quad &\textnormal{if } |x_j-x_i| < 1.
\end{cases}
\tag{\textnormal{PLB}}
\end{align}
\end{itemize}
\noindent 
The bound~\eqref{eqn::PLB_weak_upper} appears in the works of Flores \&~Kleban~\cite{Flores-Kleban:Solution_space_for_system_of_null-state_PDE2, Flores-Kleban:Solution_space_for_system_of_null-state_PDE3},
where uniqueness results for solutions to the above PDEs were established using elliptic PDE techniques (cf.~Lemma~\ref{lem::PFuniqueness}).

\bigskip

In particular, in order to argue that each $\smash{\coulombGas_{\beta}^{(\kappa)}}$ is a multiple $\SLE_{\kappa}$ partition function, it only remains to check whether
$\smash{\coulombGas_{\beta}^{(\kappa)}}(\bs{x}) > 0$ for all $\bs{x} \in \chamber_{2N}$.
The positivity of the partition function is necessary for its usage in Loewner theory (indeed, it should arise from an exponential martingale in a Girsanov transform).
This positivity property has been, however, a substantial difficulty in the literature thus far. 
Indeed, from the iterated integral formulas~\eqref{eqn::CGI_def} of the complex, multivalued integrand~\eqref{eq: integrand} it is very hard to see in general nothing more than that the functions $\smash{\coulombGas_{\beta}^{(\kappa)}}$ are real-valued 
(by a simple reflection argument). 
We settle this problem in Theorem~\ref{thm::CGI}. 
Interestingly, the Coulomb gas integral functions are only positive when 
$\kappa$ avoids certain special intervals, and we also analyze the behavior of the zeros in a special case.
The positivity does however hold in the range $\kappa \in ( 8/3 ,8)$ relevant to critical loop $O(n)$ models.

\begin{figure}
\includegraphics[width=0.7\textwidth]{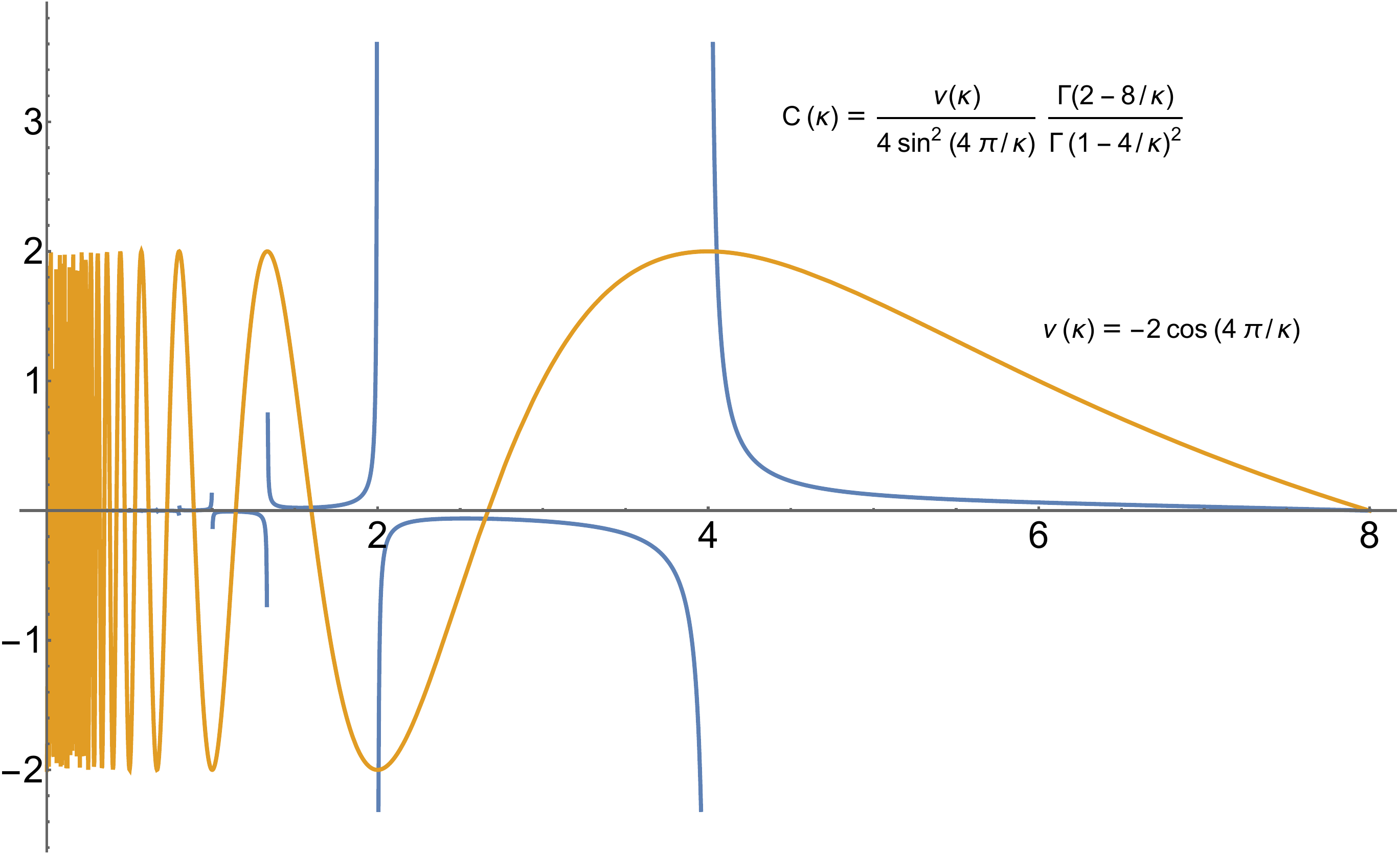}
\caption{\label{fig:constant}
Plot of the fugacity $\kappa\mapsto\fugacity(\kappa)$ (orange) and the constant $\kappa \mapsto \cst(\kappa)$ (blue). 
Note that $\fugacity(\kappa) = 0$ at points of the form $\kappa = \frac{8}{m}$ with $m \in 2\bZnn+1$ odd, 
and $\fugacity(\kappa) \in\{\pm 2\}$ at points of the form $\kappa = \frac{8}{m}$ with $m \in 2\bZpos$ even.
The constant $\cst(\kappa)$ diverges at points of the form $\kappa = \frac{8}{m}$ with $m \in 2\bZpos$ even (where $\fugacity(\kappa) \in\{\pm 2\}$).}
\end{figure}

\begin{theorem}\label{thm::CGI}
The Coulomb gas integral functions $\{\smash{\coulombGas_{\beta}^{(\kappa)}} \colon \beta \in \LP\}$ 
satisfy the following properties.
\begin{enumerate}[leftmargin=1em]
\item \label{item::CGI_contintuity}
They can be continuously extended to $(0,8) \ni \kappa$\textnormal{:} for each $\kappa \in (0,8)$, we have
\begin{align}\label{eq::CGI_contintuity}
\lim_{\kappa'\to\kappa}\coulombGas_{\beta}^{(\kappa')}=\coulombGas_{\beta}^{(\kappa)} , \qquad \beta \in \LP_N . 
\end{align}
Moreover, the extension is real-analytic in $\kappa\in (0,8)$.
\item \label{item::CGI_PDE_COV}
They satisfy \eqref{eqn::PDE},~\eqref{eqn::COV}, and~\eqref{eqn::PLB_weak_upper}.
Moreover, they satisfy
\begin{itemize}[leftmargin=3em]
\item[\textnormal{(POS)}] 
when $\kappa\in \smash{\underset{m\in \bZpos}{\bigcup} \big(\tfrac{8}{4m-1}, \tfrac{8}{4m-3}\big)}$, we have $\coulombGas_{\beta}^{(\kappa)}(\bs{x}) > 0 $ for all $\bs{x}\in\chamber_{2N}$ and $\beta\in\LP_N$\textnormal{;}

\medskip

\item[\textnormal{(NUL)}] 
when $\kappa\in \smash{\big\{ \tfrac{8}{2m+1} \colon m \in \bZpos \big\}}$, we have $\coulombGas_{\beta}^{(\kappa)}(\bs{x}) = 0 $ for all $\bs{x}\in\chamber_{2N}$ and $\beta\in\LP_N$\textnormal{;}

\medskip 
 
\item[\textnormal{(SGN)}] 
when $N\ge 2$ and $\kappa\in \smash{\underset{m\in \bZpos}{\bigcup}\big(\tfrac{8}{4m+1}, \tfrac{8}{4m-1}\big)}$, there exist points

\smallskip 

\begin{align*}
\bs{x}^+=\bs{x}^+(\kappa,\beta)\in\chamber_{2N}, \qquad
\bs{x}^0=\bs{x}^0(\kappa,\beta)\in\chamber_{2N}, \qquad
\bs{x}^-=\bs{x}^-(\kappa,\beta)\in\chamber_{2N},
\end{align*}
such that
\begin{align} \label{eqn::CGI_nonpos}
\coulombGas_{\beta}^{(\kappa)}(\bs{x}^+)>0,\qquad
\coulombGas_{\beta}^{(\kappa)}(\bs{x}^0)=0, \qquad
\coulombGas_{\beta}^{(\kappa)}(\bs{x}^-)<0. 
\end{align}
\end{itemize}
\end{enumerate}
\end{theorem}

We prove Theorem~\ref{thm::CGI} in Section~\ref{sec::consequences} by combining analysis of the Coulomb gas integrals and multiple SLE theory. 
In particular, we shall use the relation 
between the Coulomb gas integrals and 
the pure partition functions of multiple $\SLE_\kappa$, 
discussed in Section~\ref{subsec::intro_CGI_PPF}. 
The novel parts of Theorem~\ref{thm::CGI} are Item~\ref{item::CGI_contintuity} and the analysis of the zeroes in 
(POS), (NUL), (SGN) of Item~\ref{item::CGI_PDE_COV}, which we find quite curious.

\begin{remark}
In the case $N=2$, in Proposition~\ref{prop::F4points_pos} in Appendix~\ref{app::examples} 
we present a detailed analysis of the zeroes of the functions 
$\{\coulombGas_{\includegraphics[scale=0.2]{figures-arXiv/link-1}}^{(\kappa)}, \coulombGas_{\includegraphics[scale=0.2]{figures-arXiv/link-2}}^{(\kappa)}\}$ in terms of the cross-ratio of the marked points, $\chi = \CR(x_1, x_2, x_3, x_4) = \frac{(x_2-x_1)(x_4-x_3)}{(x_3-x_1)(x_4-x_2)}$. 
Interestingly, the cross-ratio $\chi$ oscillates as $\kappa\to 0$, see Eq.~\eqref{eqn::F4points_zero_semiclassical}. 
\end{remark}

Let us next elaborate on the exceptional values of $\kappa$ where the constant $\cst(\kappa)$ is not well-defined.
It is quite remarkable that the Coulomb gas integrals have different behavior for 
\begin{align} \label{eq:exceptoinal_kappas}
\kappa \; \in \; \big\{ \tfrac{4}{m} \colon m \in \bZpos \big\} 
\qquad \textnormal{or}\qquad 
\kappa \; \in \; \big\{ \tfrac{8}{2m+1} \colon m \in \bZpos \big\} 
\qquad \textnormal{or}\qquad \kappa=8.
\end{align}
\begin{itemize}[leftmargin=1em]
\item For 
$\kappa\in\big\{ \tfrac{4}{m} \colon m \in \bZpos \big\}$, 
the constant $\cst(\kappa)$ diverges (see Figure~\ref{fig:constant}) and the integral of $f_\beta^{(\kappa)}(\bs{x};\bs{u})$ is zero (see Lemma~\ref{lem:CGI_vanish}). 
Thus, it is not clear what the value of the product of the two is. 
We verify in Theorem~\ref{thm::CGI} that $\smash{\coulombGas_{\beta}^{(\kappa)}}$ is finite as a continuous extension in this case. 

\item For 
$\kappa\in \big\{ \tfrac{8}{2m+1} \colon m \in \bZpos \big\}$, the constant $\cst(\kappa)$ is finite (as a continuous extension, see Figure~\ref{fig:constant}) and the integral of $\smash{f_\beta^{(\kappa)}(\bs{x};\bs{u})}$ is zero (see Lemma~\ref{lem:CGI_vanish}), 
so one expects that the product of the two is zero. 
We verify in Theorem~\ref{thm::CGI} that  $\smash{\coulombGas_{\beta}^{(\kappa)}}$ is indeed zero (as a continuous extension) in this case. 
Furthermore, we will show that $\smash{\coulombGas_{\beta}^{(\kappa')}}$ tends to zero at rate $|\kappa'-\kappa|$ as $\kappa'\to \kappa$, see Proposition~\ref{prop::CGI_odd}. 

\item For $\kappa=8$, the constant $\cst(8)$ is zero and the integral of $\smash{f^{(8)}_{\beta}(\bs{x}; \bs{u})}$ is finite. 
Thus, $\smash{\coulombGas_{\beta}^{(8)}}$ is zero. Furthermore, we will show that $\smash{\coulombGas_{\beta}^{(\kappa')}}$ tends to zero at rate $(8-\kappa')^N$ as $\kappa'\to 8-$, see Proposition~\ref{prop::CGI_8}. 
\end{itemize}

\begin{proposition}\label{prop::CGI_odd}
Fix $\kappa \in \big\{ \tfrac{8}{2m+1} \colon m \in \bZpos \big\}$. 
Each Coulomb gas integral $\coulombGas_{\beta}^{(\kappa')}$ tends to zero at rate $(\kappa'- \kappa)$ as $\kappa' \to \kappa$\textnormal{:} 
\begin{align}\label{eqn::CGI_odd}
\coulombGasRenorm_{\beta}^{(\kappa)}(\bs{x}) 
:= \; & \lim_{\kappa'\to \kappa} 
\frac{ -\kappa^2 }{8\pi\sin(4\pi/\kappa)}  \, 
\frac{\coulombGas_{\beta}^{(\kappa')}(\bs{x})}{(\kappa'-\kappa)}
 \; \in \; (0,\infty)
, \qquad \bs{x}\in\chamber_{2N} , \; \beta \in \LP_N .
\end{align}
In particular, $\smash{\{\coulombGasRenorm_{\beta}^{(\kappa)} \colon \beta \in \LP_N\}}$ 
defined in~\eqref{eqn::CGI_odd} are multiple $\SLE_\kappa$ partition functions. 
\end{proposition}

\begin{proposition}\label{prop::CGI_8}
Fix $\kappa = 8$. 
Each Coulomb gas integral $\coulombGas_{\beta}^{(\kappa')}$ tends to zero at rate $(8-\kappa')^{N}$ as $\kappa'\to 8-$\textnormal{:}
\begin{align}\label{eqn::CGI_8}
\coulombGasRenorm_{\beta}^{(8)}(\bs{x}) 
:= \; & \lim_{\kappa'\to 8-} \frac{8}{\pi} \, 
\frac{\coulombGas_{\beta}^{(\kappa')}(\bs{x})}{(8-\kappa')^{N}} 
\; \in \; (0,\infty) ,
\qquad \bs{x} \in \chamber_{2N} , \; \beta \in \LP_N .
\end{align}
In particular, $\smash{\{\coulombGasRenorm_{\beta}^{(8)} \colon \beta \in \LP_N\}}$ defined in~\eqref{eqn::CGI_8} are multiple $\SLE_8$ partition functions. 
\end{proposition}

\subsection{Pure partition functions and their behavior in $\kappa$}
\label{subsec::PPF}

We now turn to SLE theory. 
Following Oded Schramm's seminal idea to describe critical interfaces in planar models dynamically as Loewner evolutions --- whose behavior is governed by uniformizing conformal maps and a one-dimensional driving process (a variant of Brownian motion) --- one is led to analyzing the set of the various natural Loewner driving processes.
For multiple chordal curves, these processes are classified in terms of multiple $\SLE_{\kappa}$ \emph{partition functions}: 
smooth positive functions $\PartF \colon \chamber_{2N} \to (0,\infty)$ 
satisfying the hypoelliptic partial differential equations~\eqref{eqn::PDE} (which essentially follow because the partition function describes a tilting of the Brownian motion by a local martingale),  
and the full M\"{o}bius covariance~\eqref{eqn::COV} (a direct consequence of the conformal invariance of the curves).

In the present work, we will give a complete classification and construction of multiple chordal $\SLE_\kappa$ curves and their partition functions for the full parameter range $\kappa \in (0,8)$ in general simply connected domains. 
The extremal points in the convex set of multiple chordal $\SLE_\kappa$ measures are provided by the ``pure'' partition functions $\smash{\PartF_{\alpha}^{(\kappa)}}$,
which also appear in the numerator in Conjecture~\ref{conj::On}.

\begin{definition} \label{def::PPF_general}
Fix $\kappa\in (0,8)$. 
The \emph{pure partition functions} of multiple $\SLE_{\kappa}$ are the recursive collection $\{\smash{\PartF_{\alpha}^{(\kappa)}} \colon \alpha \in \LP\}$ 
of positive functions $\smash{\PartF_{\alpha}^{(\kappa)}} \colon \chamber_{2N} \to (0,\infty)$, indexed by planar link patterns
$\alpha = \{ \{a_1,b_1\}, \ldots , \{a_N,b_N\}\}$ as in~\eqref{eq: link pattern ordering}, 
uniquely determined by the PDE system~\eqref{eqn::PDE}, 
M\"{o}bius covariance~\eqref{eqn::COV}, 
the power-law bound~\eqref{eqn::PLB_weak_upper}, 
as well as the following recursive asymptotics property~\eqref{eqn::PPF_ASY}: 
\begin{itemize}[leftmargin=3em]
\item[\textnormal{(ASY)}] 
\textnormal{\bf Asymptotics:} 
With $\smash{\PartF_{\emptyset}^{(\kappa)}} \equiv 1$ for the empty link pattern $\emptyset \in \LP_0$, the collection $\{\PartF_{\alpha}^{(\kappa)} \colon \alpha\in\LP_N\}$ satisfies the following recursive 
asymptotics property.
Fix $N \ge 1$ and $j \in \{1,2, \ldots, 2N-1 \}$. 
Then, 
\begin{align}
\label{eqn::PPF_ASY} 
\; & \lim_{x_j,x_{j+1}\to\xi} \frac{\PartF_{\alpha}^{(\kappa)}(\bs{x})}{ (x_{j+1}-x_j)^{-2h(\kappa)} }
= 
\begin{cases}
\PartF_{\alpha/\{j,j+1\}}^{(\kappa)}(\bs{\ddot{x}}_j), 
& \textnormal{if }\{j, j+1\}\in\alpha , \\
0 ,
& \textnormal{if }\{j, j+1\} \not\in \alpha ,
\end{cases}
\tag{\textnormal{ASY}}
\end{align}
where $\alpha/\{j,j+1\} \in \LP_{N-1}$ denotes the link pattern obtained from $\alpha$ by removing the link $\{j,j+1\}$ and relabeling the remaining indices by $1, 2, \ldots, 2N-2$, and 
$\xi \in (x_{j-1}, x_{j+2})$ and where 
\begin{align} \label{eqn::bs_notation}
\begin{split}
\bs{x} = \; & (x_1, \ldots, x_{2N}) \in \chamber_{2N} 
\qquad 
\textnormal{and}
\qquad 
\bs{\ddot{x}}_j = (x_1, \ldots, x_{j-1}, x_{j+2}, \ldots, x_{2N}) \in \chamber_{2N-2}. 
\end{split}
\end{align}

\end{itemize}
\end{definition}

Relaxing the positivity requirement, 
the collection $\{\PartF_{\alpha}^{(\kappa)} \colon \alpha\in\LP_N\}$ was first constructed in~\cite{Flores-Kleban:Solution_space_for_system_of_null-state_PDE3} indirectly by using Coulomb gas integrals for all $\kappa\in (0,8)$, and explicitly for all $\kappa \in (0,8) \setminus \QQ$ in~\cite{Kytola-Peltola:Pure_partition_functions_of_multiple_SLEs}, following the conjectures from~\cite{BBK:Multiple_SLEs_and_statistical_mechanics_martingales}. 
Explicit probabilistic constructions were given in~\cite[Theorem~1.1]{Peltola-Wu:Global_and_local_multiple_SLEs_and_connection_probabilities_for_level_lines_of_GFF} for $\kappa\le 4$  
and~\cite[Theorem~1.7]{Wu:Convergence_of_the_critical_planar_ising_interfaces_to_hypergeometric_SLE} for $\kappa\le 6$ 
--- see also the recent~\cite{Zhan:Existence_and_uniqueness_of_nonsimple_multiple_SLE}.
Also recently, Ang, Holden, Sun~\&~Yu gave a construction 
of some SLE partition functions
using LQG theory~\cite{AHSY:Conformal_welding_of_quantum_disks_and_multiple_SLE_the_non-simple_case} for $\kappa \in (6,8)$. 
In this construction, they do not derive the asymptotic properties of their partition functions. 
Nevertheless, we will indirectly show using global multiple $\SLE_\kappa$ measures as a bridge that these agree with the pure partition functions (up to a multiplicative constant)
--- see Remark~\ref{rem::NSLE_sanity} in Section~\ref{subsec::global_existence}.

\bigskip

Now, the second main purpose of the present article is to give a direct SLE theory construction for the pure partition functions with the parameter range $\kappa\in (4,8)$ (Theorem~\ref{thm::PPF}) and use it to construct global multiple SLEs  (Theorem~\ref{thm::existenceglobalSLE}). 
Let us emphasize a couple of crucial aspects of our results:
\begin{itemize}[leftmargin=1em]
\item The pure partition functions are well-defined and \emph{real-analytic} for all $\kappa\in (0,8)$. 
Moreover, the limit $\kappa \to 8-$ can be normalized to obtain a pure partition function associated to $\SLE_8$.

\item From the probabilistic construction it readily follows that the pure partition functions are \emph{positive}. 
This seems not obvious at all from the Coulomb gas perspective. 

\item 
The \emph{recursive asymptotics property}~\eqref{eqn::PPF_ASY} 
of the pure partition functions is their most important, characterizing property, providing the required boundary values for the unique determination of the hypoelliptic PDE boundary value problem~\eqref{eqn::PDE}--\eqref{eqn::COV}. 
It is important to note that the proof of the asymptotics properties~\eqref{eqn::PPF_ASY} for $\kappa \in (6,8)$ is quite hard (see Section~\ref{subsec:ssymptotics_PPF}), because $h(\kappa)<0$ and there are no bounds available to exchange the limit and expected value (cf.~\cite[Appendix~B]{Peltola:Towards_CFT_for_SLEs}).

\item As a consequence of their asymptotics~\eqref{eqn::PPF_ASY}, the pure partition functions are \emph{linearly independent}. 
Hence, they can be regarded as the ``extremal'' partition functions for multiple $\SLE_\kappa$ processes.

\end{itemize}

\begin{theorem}\label{thm::PPF}
The collection $\{\smash{\PartF_{\alpha}^{(\kappa)}} \colon \alpha \in \LP\}$ 
of pure partition functions exists and is unique for all $\kappa\in (0,8)$. 
They are linearly independent and 
have the following further properties.
\begin{enumerate}[leftmargin=1em]
\item \label{item:PPF_continuity}
They are continuous in $(0,8) \ni \kappa$\textnormal{:} 
for each $\kappa \in (0,8)$, we have
\begin{align}\label{eqn::PPF_continuity}
\lim_{\kappa' \to \kappa}\PartF_{\alpha}^{(\kappa')}(\bs{x}) 
= \PartF_{\alpha}^{(\kappa)}(\bs{x}) , \qquad \bs{x}\in\chamber_{2N} , \; \alpha \in \LP_N . 
\end{align}
Moreover, they are real-analytic in $\kappa\in (0,8)$.

\item \label{item:PPF_continuity8}
They tend to zero at rate $(8-\kappa')^{N-1}$ as $\kappa'\to 8-$, i.e., the following limit exists\textnormal{:} 
\begin{align}\label{eqn::PPF_continuity8}
\PartFRenorm_{\alpha}^{(8)}(\bs{x}) 
:= \lim_{\kappa'\to 8-} \frac{\PartF_{\alpha}^{(\kappa')}(\bs{x})}{(8-\kappa')^{N-1}} \; \in \; (0,\infty) , \qquad \bs{x}\in\chamber_{2N} , \; \alpha \in \LP_N .
\end{align}
\end{enumerate}
\end{theorem}

We prove Theorem~\ref{thm::PPF} in Section~\ref{subsec:finish_PPF_proof} using the relation between the pure partition functions and Coulomb gas integrals (see Section~\ref{subsec::intro_CGI_PPF}). 
Note that~\eqref{eqn::PPF_continuity8} tells that the pure partition functions decay as a polynomial in $\kappa$ as $\kappa\to 8-$ when $N\ge 2$. 
In contrast to this polynomial decay, the pure partition functions decay to zero exponentially fast as $\kappa\to 0+$: 
the quantity $\smash{-\kappa\log\PartF_{\alpha}^{(\kappa)}}$  
has a finite limit as $\kappa\to 0$, see~\cite[Eq.~(1.10)~\&~Theorem~1.5]{Peltola-Wang:LDP}. 
Interestingly, the limit functions also satisfy the semiclassical BPZ equations, by~\cite[Proposition~1.8]{Peltola-Wang:LDP}.  
See also~\cite{ABKM:Pole_dynamics_and_an_integral_of_motion_for_multiple_SLE0} for some recent investigations. 

\smallbreak

As a by-product, our results also give a construction of a global $N$-$\SLE_{\kappa}$ probability measure for all $\kappa \in (4,8)$; see Theorem~\ref{thm::existenceglobalSLE} and  Sections~\ref{sec::multipleSLEs}-\ref{sec::PPF}. Our proofs only rely on SLE techniques and some analysis and PDE theory (in particular, we do not use couplings with the Gaussian free field, GFF, except that 
we do use the reversibility of non-simple chordal $\SLE_\kappa$, proven in~\cite{Sheffield-Miller:Imaginary_geometry3} via SLE/GFF couplings).

\smallbreak

\subsection{Relation between pure partition functions and Coulomb gas integrals}\label{subsec::intro_CGI_PPF}

In fact, both collections $\{\smash{\coulombGas_{\alpha}^{(\kappa)}} \colon \alpha\in\LP_N\}$ and $\{\smash{\PartF_{\alpha}^{(\kappa)}} \colon \alpha\in\LP_N\}$ are solutions to the same hypoelliptic system~\eqref{eqn::PDE}, 
and they are related by a linear transformation encoded in the meander matrix (Definition~\ref{def::meandermatrix}) --- which is however not always invertible (see~\cite[Eq.~(5.6)]{DGG:Meanders_and_TL_algebra} and Lemma~\ref{lem::meandermatrix_invertible}).

\begin{proposition}\label{prop::CGI_PPF}
Fix $\kappa\in (0,8)$ and $\fugacity = \fugacity(\kappa)$ as in~\eqref{eq: normalization_cst}.
The 
Coulomb gas integrals in Definition~\ref{def::CGI_def} are linear combinations of pure partition functions: 
\begin{align}\label{eqn::CGI_PPF}
\coulombGas_{\beta}^{(\kappa)}(\bs{x}) = \sum_{\alpha\in\LP_N} \meanderMatrix(\alpha, \beta) \,  \PartF_{\alpha}^{(\kappa)}(\bs{x})
, \qquad  \bs{x} \in \chamber_{2N} , \; \beta \in \LP_N .
\end{align}
In particular, the functions $\{\coulombGas_{\beta}^{(\kappa)}\colon \beta\in\LP_N\}$ 
of Definition~\ref{eqn::CGI_def} 
are linearly independent for 
\begin{align}\label{eqn::meander_inv}
\kappa\in (0,8)\setminus \Big\{\tfrac{4p}{p'} \colon p,p' \in \bZpos \textnormal{ coprime and } 2 \le p \leq N + 1 \Big\}.
\end{align}
\end{proposition}

Proposition~\ref{prop::CGI_PPF} can be proven by similar arguments as~\cite[Proposition~1.10]{FPW:Connection_probabilities_of_multiple_FK_Ising_interfaces}. 
Indeed, the asserted equation~\eqref{eqn::CGI_PPF} follows 
from knowing that both sides of it 
satisfy the same PDE boundary value problem~(\ref{eqn::PDE},~\ref{eqn::COV},~\ref{eqn::CGI_ASY}),
where the asymptotic boundary conditions~\eqref{eqn::CGI_ASY} are specified in Corollary~\ref{cor::CGI_PDECOVASY} in Section~\ref{subsec:ASY}.
Invoking~\cite[Lemma~1]{Flores-Kleban:Solution_space_for_system_of_null-state_PDE2} 
shows that such a function is unique as long as the power-law bound~\eqref{eqn::PLB_weak_upper} satisfied.
The exceptional values of $\kappa$ require slightly more care.

\begin{example}
For instance, the meander matrix with $N=3$ reads
\begin{align*}
\left(
\begin{array}{c|ccccc}
&\includegraphics[scale=0.2]{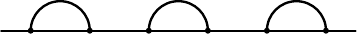}&\includegraphics[scale=0.2]{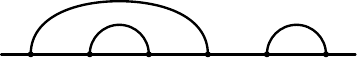}&\includegraphics[scale=0.2]{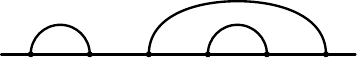}&\includegraphics[scale=0.2]{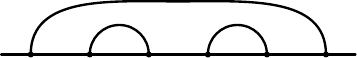}&\includegraphics[scale=0.2]{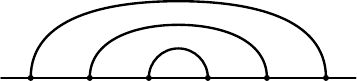}\\[.5em]
\hline
\includegraphics[scale=0.2]{figures-arXiv/link-3}& \fugacity^3 & \fugacity^2 & \fugacity^2 & \fugacity & \fugacity^2 \\
\includegraphics[scale=0.2]{figures-arXiv/link-4} &\fugacity^2 & \fugacity^3 & \fugacity & \fugacity^2 & \fugacity \\
 \includegraphics[scale=0.2]{figures-arXiv/link-5}&\fugacity^2 & \fugacity & \fugacity^3 & \fugacity^2 & \fugacity \\
 \includegraphics[scale=0.2]{figures-arXiv/link-6}&\fugacity & \fugacity^2 & \fugacity^2 & \fugacity^3 & \fugacity^2 \\
 \includegraphics[scale=0.2]{figures-arXiv/link-7}&\fugacity^2 & \fugacity & \fugacity & \fugacity^2 & \fugacity^3 \\
\end{array}
\right)
\end{align*}
Its determinant is 
$\fugacity^5 (\fugacity^2-2) (\fugacity^2-1)^4$,
which is zero when $\fugacity=0$ or $\fugacity=\pm\sqrt{2}$ or $\fugacity=\pm 1$.
\end{example}

\begin{remark}\label{rem::CGI_examples} 
Let us discuss three exceptional values in~\eqref{eqn::meander_inv}: $\kappa \in \{3,16/3,6\}$, for which $\nu(6)=\nu(3)=1$ and $\nu(16/3)=\sqrt{2}$, which are roots of the determinant of meander matrix with $N=3$. 
In these three cases, the Coulomb gas integrals are linearly dependent in general, but they have explicit expressions.
\begin{itemize}[leftmargin=1em]
\item
When $\kappa = 6$, we have $\nu(6)=1$. The meander matrix $\meanderMatrix$ is not invertible as long as $N\ge 2$. All Coulomb gas integrals are equal, given by
\begin{align}\label{eqn::CGI_kappa6}
\coulombGas_{\beta}^{(6)} (\bs{x}) = 1 
, \qquad  \bs{x} \in \chamber_{2N} , \; \beta \in \LP_N .
\end{align}
This is a feature of the target-independence of $\SLE_6$, or the fact that Bernoulli percolation is insensitive to boundary conditions (cf.~Conjecture~\ref{conj::On}).
Note, however, that not all Coulomb gas type integration formulas are constant in this case. 
For example, the Cardy-Smirnov formula~\cite{Cardy:Critical_percolation_in_finite_geometries, Smirnov:Critical_percolation_in_the_plane} for the percolation scaling limit crossing probability for a cluster from the interval $(x_1, x_2)$ to $(x_3, x_4)$ reads
\begin{align*}
\PartF_{\includegraphics[scale=0.2]{figures-arXiv/link-1}}^{(6)}(x_1, x_2, x_3, x_4)
= \frac{\Gamma(2/3)}{\Gamma(1/3)^2} \int_\chi^1 u^{-2/3} (1-u)^{-2/3} \, \ud u 
, \qquad \chi 
 = \frac{(x_2-x_1)(x_4-x_3)}{(x_3-x_1)(x_4-x_2)} .
\end{align*}
Let us compare this with the formula for $\smash{\coulombGas_{\includegraphics[scale=0.2]{figures-arXiv/link-2}}^{(6)}}$ obtained from the definition~\eqref{eqn::CGI_def}
and Proposition~\ref{prop::remove_integration}: 
\begin{align*}
\coulombGas_{\includegraphics[scale=0.2]{figures-arXiv/link-2}}^{(6)}(x_1, x_2, x_3, x_4)
= \; & \frac{\Gamma(2/3)}{\Gamma(1/3)^2} \, (1-\chi)^{1/3} \chi^{1/3} \, \int_\chi^1 (1-u)^{-2/3} (\chi-u)^{-2/3} u^{-2/3} \, \ud u
\\
= \; & \frac{\Gamma(2/3)}{\Gamma(1/3)^2} \, \chi^{-1/3} \, \int_0^1 (1- (\tfrac{\chi-1}{\chi}) \, v)^{-2/3} (1-v)^{-2/3} v^{-2/3} \, \ud v
\\
= \; & \chi^{-1/3} \, \hF\big(\tfrac{2}{3}, \tfrac{1}{3} , \tfrac{2}{3}; \tfrac{\chi-1}{\chi} \big)
\; = \; 1 .
\end{align*}
While the integration contours in both formulas coincide, the integrands are different. 
 
\item
When $\kappa = 3$, we have $\fugacity(3)=1$. The meander matrix $\meanderMatrix$ is not invertible as long as $N\ge 2$. All Coulomb gas integrals are equal, given by a well-known explicit Pfaffian formula (related to fermion, or boundary energy density correlations in the Ising model~\cite{Hongler-Smirnov:Energy_density_in_planar_Ising_model, Izyurov:Smirnovs_observable_for_free_boundary_conditions_interfaces_and_crossing_probabilities}, see also~\cite{Peltola-Wu:Crossing_probabilities_of_multiple_Ising_interfaces}),
\begin{align}\label{eqn::CGI_kappa3}
\coulombGas_{\beta}^{(3)} (\bs{x})
= \mathrm{pf} \bigg(  \frac{1}{x_{j}-x_{i}} \bigg)_{i,j=1}^{2N} 
, \qquad  \bs{x} \in \chamber_{2N} , \; \beta \in \LP_N .
\end{align}

\item
When $\kappa = 16/3$, we have $\fugacity(16/3)=\sqrt{2}$. The meander matrix $\meanderMatrix$ is not invertible as long as $N\ge 3$. 
In this case, the Coulomb gas integrals are linearly dependent~\cite[Theorem~1.5]{FPW:Connection_probabilities_of_multiple_FK_Ising_interfaces}, 
\begin{align}\label{eqn::CGI_kappa16over3}
\coulombGas_{\beta}^{(16/3)} (\bs{x}) = &
\prod_{s=1}^N |x_{b_s}-x_{a_s}|^{-1/8}
\bigg(\sum_{\bs{\sigma} \in \{\pm 1\}^N}\prod_{1\le s < t \le N} \Big(
\frac{|x_{a_t}-x_{a_s}| \, |x_{b_s}-x_{b_t}|}{|x_{b_t}-x_{a_s}| \, |x_{b_s}-x_{a_t}|}
\Big)^{\sigma_s \sigma_t / 4} \bigg)^{1/2} 
, \\&  \bs{x} \in \chamber_{2N} , \; \beta \in \LP_N .\notag
\end{align}
where $\bs{\sigma} = (\sigma_1, \ldots, \sigma_N)\in \{\pm 1\}^N$, reminiscent to spin correlation functions in the Ising model~\cite{Izyurov:On_multiple_SLE_for_the_FK_Ising_model}. 
\end{itemize}
\end{remark}

\begin{remark}\label{rem::duality}
The (Duplantier) duality of $\SLE_{\kappa}$ is a symmetry in the parameter $\kappa$ preserving the central charge $c(\kappa)=(3\kappa-8)(6-\kappa)/2\kappa$:
if $\kappa\le 4$ and $\kappa'\ge 4$, then Eq.~\eqref{eqn::centralcharge} shows that  
\begin{align*}
\kappa\kappa'=16\qquad\Longleftrightarrow\qquad c(\kappa)=c(\kappa'). 
\end{align*}
That this relation has a deeper meaning for $\SLE_{\kappa}$ curves 
was predicted by Bertrand Duplantier using ideas from quantum gravity (see, e.g.,~\cite{Duplantier:Conformal_random_geometry}).  
This duality is now well-understood in terms of Hausdorff dimensions of SLE curves: the frontier of $\SLE_{\kappa'}$ curve is locally the same as $\SLE_{\kappa}$ curve, and in particular, the frontier of $\SLE_{\kappa'}$ has the same Hausdorff dimension as $\SLE_{\kappa}$ 
(see, e.g.,~\cite{Dubedat:Duality_of_SLE, Zhan:Duality_of_chordal_SLE}).

In contrast, there is another symmetry that is much more mysterious.
Namely, for the fugacity $\fugacity(\kappa)=-2\cos(4\pi/\kappa)$,
dilute and dense models with the same fugacity share similar algebraic structure. Indeed, Eq.~\eqref{eq: normalization_cst} shows that
if $\kappa\in (8/3,4]$ and $\hat{\kappa}\in [4,8)$, then
\begin{align*}
\frac{1}{\kappa}+\frac{1}{\hat{\kappa}}=\frac{1}{2}\qquad\Longleftrightarrow\qquad 
\fugacity(\kappa)=\fugacity(\hat{\kappa}). 
\end{align*}
The fugacity parameter $\fugacity(\kappa)$ appears as the parameter of the Temperley-Lieb algebra 
(see, e.g.,~\cite{Flores-Kleban:Solution_space_for_system_of_null-state_PDE3,
Flores-Peltola:Standard_modules_radicals_and_the_valenced_TL_algebra, 
Flores-Peltola:Higher_spin_QSW, 
GSNJRS:Global_symmetry_and_conformal_bootstrap_in_the_two-dimensional_On_model}), 
and $q(\kappa) = \exp(4 \pi \ii / \kappa)$ as the parameter of the associated quantum group $U_q(\mathfrak{sl}(2,\C))$
(e.g.,~\cite{Kytola-Peltola:Conformally_covariant_boundary_correlation_functions_with_quantum_group}),
where the symmetry $\kappa \; \leftrightarrow \; \hat{\kappa}$ corresponds with $q(\hat{\kappa}) = 1 / q(\kappa) = \exp(-4 \pi \ii / \kappa)$. 
A special case of this symmetry is the connection between 
the Ising model and percolation: 
\begin{align*}
\frac{1}{3}+\frac{1}{6}=\frac{1}{2}. 
\end{align*}
Considering the critical-temperature Ising model as a special case of Conjecture~\ref{conj::On} with $n=1$ and $p=p_c(1)$, 
the high-temperature Ising model corresponds to the special case with $n=1$ and $p>p_c(1)$. 
This is related to the conjecture that the universality class of 
high-temperature Ising model is the same as that of 
critical Bernoulli percolation (see, e.g.~\cite{Kager-Nienhuis:Guide_to_SLE, Smirnov:Towards_conformal_invariance_of_2D_lattice_models, BCM:High_temperature_Ising_model_on_the_triangular_lattice_is_critical_Bernoulli_percolation_model}). 
In~(\ref{eqn::CGI_kappa6},~\ref{eqn::CGI_kappa3}), we see that Coulomb gas integrals are equal for both $\kappa=3$ and $\kappa=6$,
which consistent with the fact that $\fugacity(3)=\fugacity(6)=1$. 
\end{remark}

In the cases where the  Coulomb gas integrals identically vanish, we may renormalize them as in Propositions~\ref{prop::CGI_odd}~\&~\ref{prop::CGI_8}. 
We also define the \emph{renormalized meander matrix} $\{\meanderRenorm(\alpha, \beta) \colon \alpha, \beta\in\LP_N\}$ via
\begin{align} \label{eqn::meandermatrix_renormalized}
\meanderRenorm(\alpha, \beta) 
:= \lim_{\fugacity \to 0} \frac{1}{\fugacity} \, 
\meanderMatrix(\alpha,\beta) =
\begin{cases}
1, &\textnormal{if the meander formed from $\alpha$ and $\beta$ has one loop,}\\
0,&\textnormal{otherwise.}
\end{cases}
\end{align}

Using the renormalized meander matrix, we can extend the relations in Proposition~\ref{prop::CGI_PPF} to the exceptional values of $\kappa$ 
(see Section~\ref{subsec::renormalized_pf}).
In this way, we obtain partition functions for $\SLE_\kappa$ processes with exceptional values of $\kappa$.
It would be interesting to find probabilistic or CFT interpretations for these partition functions. 
We shall discuss some possible CFT relations in Section~\ref{subsec:Frob}.

\begin{proposition}
\label{prop::CGI_PPF_renormalized}\
\begin{enumerate}[leftmargin=1em]
\item \label{item::CGI_PPF_renormalized1}
Fix $\kappa\in \big\{ \tfrac{8}{2m+1} \colon m\in \bZpos \big\}$. 
Each renormalized Coulomb gas integral~\eqref{eqn::CGI_odd} is a linear combination of pure partition functions:
\begin{align}\label{eqn::CGI_PPF_odd}
\coulombGasRenorm_{\beta}^{(\kappa)}(\bs{x}) 
\; = \; \sum_{\alpha\in\LP_N} \meanderRenorm(\alpha, \beta) \, \PartF_{\alpha}^{(\kappa)} (\bs{x})\; \in \; (0,\infty)
, \qquad \bs{x}\in\chamber_{2N} , \; \beta \in \LP_N .
\end{align}
In particular, the functions $\{\coulombGasRenorm_{\beta}^{(\kappa)}: \beta\in\LP_N\}$ defined in~\eqref{eqn::CGI_odd} are linearly independent.

\item \label{item::CGI_PPF_renormalized2}
Fix $\kappa=8$. Each renormalized Coulomb gas integral~\eqref{eqn::CGI_8} is a linear combination of the renormalized pure partition functions~\eqref{eqn::PPF_continuity8}\textnormal{:} 
\begin{align}\label{eqn::CGI_PPF_8}
\coulombGasRenorm_{\beta}^{(8)}(\bs{x}) 
\; = \; \sum_{\alpha\in\LP_N} \meanderRenorm(\alpha, \beta) \, \widehat{\PartF}_{\alpha}^{(8)} (\bs{x})\; \in \; (0,\infty)
, \qquad \bs{x}\in\chamber_{2N} , \; \beta \in \LP_N .
\end{align}
In particular, both $\{\coulombGasRenorm_{\beta}^{(8)}: \beta\in\LP_N\}$ and $\{\widehat{\PartF}_{\alpha}^{(8)}: \alpha\in\LP_N\}$ are linearly independent. 
\end{enumerate}
\end{proposition}

\subsection{Asymptotics, Frobenius series, and CFT interpretations}
\label{subsec:Frob}

The pure partition functions (Section~\ref{subsec::PPF}) are singled out as unique solutions to the PDE boundary value problem~\eqref{eqn::PDE}--\eqref{eqn::COV} 
via their recursive asymptotics property~\eqref{eqn::PPF_ASY}, as in Theorem~\ref{thm::PPF}.
In CFT parlance, the two possible asymptotics in~\eqref{eqn::PPF_ASY} correspond to two possible fusion channels of the degenerate field $\Phi_{1,2}$ (of conformal weight $h_{1,2}$~\eqref{eq::degenerate_weight}) into the identity field $\Phi_{1,1}$ and the field $\Phi_{1,3}$ which is degenerate at level three. Indeed, the corresponding exponents in~\eqref{eqn::PPF_ASY} have the form
\begin{align} \label{eq::indicial_exponents}
-2h(\kappa) = h_{1,1}(\kappa) - 2 h_{1,2}(\kappa) 
\qquad \textnormal{and} \qquad 
\tfrac{2}{\kappa} = h_{1,3}(\kappa) - 2 h_{1,2}(\kappa) ,
\end{align}
where $h_{1,1}(\kappa)= 0$ and $h_{1,3}(\kappa) = (8-\kappa)/\kappa$
(see~\cite{Peltola:Towards_CFT_for_SLEs} and references therein for more background).
The Coulomb gas integrals are generically (i.e., when $\kappa \notin \smash{\big\{ \tfrac{8}{m} \colon m \in \bZpos \big\}}$) 
also singled out as unique solutions to the PDE boundary value problem~\eqref{eqn::PDE}--\eqref{eqn::COV} but with different asymptotic boundary conditions (see Corollary~\ref{cor::CGI_PDECOVASY}). 
In the present work, we will also investigate the subleading terms in the expansion of the Coulomb gas integrals when two of the variables collide. 
For simplicity of notation, and without loss of generality, we only consider the case where $|x_2 - x_1| \to 0$.
Analogously to the solutions of the hypergeometric equation (see also Appendix~\ref{app::examples}), 
solutions to~\eqref{eqn::PDE} admit a Frobenius series 
with the two indicial exponents~\eqref{eq::indicial_exponents}. 
We investigate the explicit terms in this expansion in Proposition~\ref{prop::Frobenius_F} 
(and some exceptional cases in Propositions~\ref{prop::CGI8over3_Frobenius}
and~\ref{prop::CGI8_Frobenius} ). 
Notably, these terms agree with the findings 
in~\cite{Flores-Peltola:Standard_modules_radicals_and_the_valenced_TL_algebra,
Flores-Peltola:Higher_spin_QSW,
Kytola-Peltola:Conformally_covariant_boundary_correlation_functions_with_quantum_group,
Flores-Peltola:Solution_space_of_BSA_PDEs},
where a quantum group structure on the solution space of~\eqref{eqn::PDE}--\eqref{eqn::COV} was found and investigated,
as well as with the fusion properties expected from (possibly logarithmic) extended\footnote{Observe that it is evident from the Coulomb gas integrals in general 
(see~\cite{Peltola:Basis_for_solutions_of_BSA_PDEs_with_particular_asymptotic_properties} for the generic case) 
that they describe correlation functions outside of CFT minimal models, because their arbitrary fusions are non-zero.}
minimal models in the physics literature~\cite{DMS:CFT, PRZ:Logarithmic_minimal_models, Read-Saleur:Associative-algebraic_approach_to_log_CFT}.

\subsubsection{Loop $O(n)$ model --- generic case}

In the terms appearing in the Frobenius series, 
we will use the Coulomb gas type integrands with general numbers 
of $d \geq 2$ variables and $\ell \geq 0$ integration variables (with any branch choice), 
\begin{align} 
\label{eq: integrand_gen}
f^{(\kappa)} (\bs{x};\bs{u}) := \; &
\prod_{1\leq i<j\leq d}(x_{j}-x_{i})^{2/\kappa} 
\prod_{1\leq r<s\leq \ell}(u_{s}-u_{r})^{8/\kappa} 
\prod_{\substack{1\leq i\leq d \\ 1\leq r\leq \ell}}
(u_{r}-x_{i})^{-4/\kappa} 
\end{align}
having negative screening charges $\alpha_- = - 2/\sqrt{\kappa}$ at the integration variables $\bs{u} = (u_1, u_2, \ldots, u_\ell)$, 
and symmetric charges 
$\alpha_j = 1/\sqrt{\kappa}$ for all $j \in \{1, 2, \ldots, d\}$ 
at all of the variables $\bs{x} = (x_1, x_2, \ldots, x_d)$. 

\begin{proposition}\label{prop::Frobenius_F}
Fix $\kappa\in (8/3,8)$ and $N\ge 2$. 
Each Coulomb gas integral $\smash{\coulombGas_{\beta}^{(\kappa)}}$, 
with $\beta \in \LP_N$, has the following expansion 
\textnormal{(}using the notation~\eqref{eqn::bs_notation}\textnormal{)}. 
\begin{itemize}[leftmargin=1em]
\item If $\{1,2\}\in\beta$, 
then for all $\xi\in (-\infty, x_3)$, 
\begin{align}\label{eqn::CGI_Frobenius_paired}
\coulombGas_{\beta}^{(\kappa)}(\bs{x}) 
= \; & \fugacity(\kappa) \, 
\coulombGas_{\beta/\{1,2\}}^{(\kappa)}(\bs{\ddot{x}}_1)
\, (x_{2}-x_1)^{-2h(\kappa)}
\; + \; 
o\big( (x_{2}-x_1)^{2/\kappa} \big) , \qquad \textnormal{as } x_1, x_{2}\to \xi, 
\end{align}
where $\beta/\{1,2\} \in \LP_{N-1}$ denotes the link pattern obtained from $\beta$ by removing the link $\{1,2\}$ and relabeling the remaining indices by $1, 2, \ldots, 2N-2$.

\item If $\{1,2\}\not\in\beta$, 
then for all $\xi\in (-\infty, x_{3})$, 
\begin{align}\label{eqn::CGI_Frobenius_tying}
\begin{split}
\coulombGas_{\beta}^{(\kappa)}(\bs{x}) 
= \; & \coulombGas_{\wp_1(\beta)/\{1,2\}}^{(\kappa)}(\bs{\ddot{x}}_1)
\, (x_{2}-x_1)^{-2h(\kappa)} 
\; + \; \coulombGas_{\beta/\vee_1}^{(\kappa)}(\xi, \bs{\ddot{x}}_1) 
\,  (x_{2}-x_1)^{2/\kappa} \\
\; & + \; o\big((x_{2}-x_1)^{2/\kappa}\big) , \qquad \textnormal{as } x_1, x_{2}\to \xi,
\end{split}
\end{align}
where $\wp_1$ 
is the ``tying operation'' defined by 
\begin{align*}
\wp_1 \colon \LP_N\to \LP_N , \qquad
\wp_1(\beta) = 
\big(\beta\setminus(\{1,k_1\}, \{2, k_2\})\big)\cup \{1,2\}\cup \{k_1, k_2\} , 
\end{align*}  
where $k_1$ \textnormal{(}resp.~$k_2$\textnormal{)} 
is the pair of $1$ \textnormal{(}resp.~$2$\textnormal{)} in $\beta$ \textnormal{(}and $\{1,k_1\}, \{2, k_2\}, \{k_1, k_2\}$ are unordered\textnormal{)},
\begin{align*}
\vcenter{\hbox{\includegraphics[scale=1]{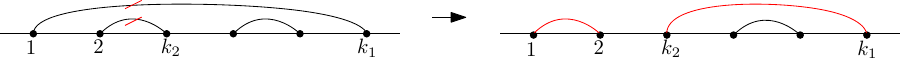}}} 
\end{align*}
and where $\smash{\coulombGas_{\beta/\vee_1}^{(\kappa)}(\xi, \bs{\ddot{x}}_1)}$ is the ``fused'' Coulomb gas integral function defined as follows. 
\begin{itemize}[leftmargin=1em]
\item[$\triangleright$] If $\kappa\in(8/3,8)\setminus\{4\}$, then 
\begin{align} \label{eqn::CGI_fused_generic}
\coulombGas_{\beta/\vee_1}^{(\kappa)}(\xi, \bs{\ddot{x}}_1) 
:=  \; &  
\frac{\cst(\kappa)^{N-2} \,(\fugacity(\kappa)^2-1)}{\hF\big(\tfrac{4}{\kappa}, 1-\tfrac{4}{\kappa}, \tfrac{8}{\kappa}; 1\big)} 
\; \ointclockwise_{\acycle(x_{a_3},x_{b_3})}\ud u_3 \cdots \ointclockwise_{\acycle(x_{a_N},x_{b_N})} \ud u_N 
\; \hat{f}_{\beta/\vee_1}^{(\kappa)} (\xi, \bs{\ddot{x}}_1; \bs{\ddot{u}}_1),
\end{align}
where $\bs{\ddot{u}}_1 := (u_3,u_4,\ldots,u_N)$ and 
\begin{align} \label{eqn::integrand_CGI_fused_generic}
\hat{f}_{\beta/\vee_1}^{(\kappa)} (\xi, \bs{\ddot{x}}_1; \bs{\ddot{u}}_1) 
:= \; & \prod_{r=3}^N (u_r-\xi)^{16/\kappa-2} 
\prod_{j=3}^{2N} (x_j-\xi)^{1-8/\kappa} 
\; f^{(\kappa)} (\bs{\ddot{x}}_1; \bs{\ddot{u}}_1), 
\end{align}
whose branch is chosen to be real and positive when $x_{a_r}<\Re (u_r)< x_{a_r+1}$ for $r\in \{3,4,\ldots, N\}$.

\item[$\triangleright$]
If $\kappa=4$, then $\smash{\coulombGas_{\beta/\vee_1}^{(4)}}$ is a linear combination of rational functions with explicit formula given by Eq.~\eqref{eqn::CGI4_fused} in Proposition~\ref{prop::CGI4_Frobenius} in Appendix~\ref{app::Frob}. 
\end{itemize}
\end{itemize}
\end{proposition}

\begin{proposition} \label{prop::third_order_PDE}
Fix $\kappa\in (8/3,8)$ and $\beta\in \LP_N$ with $\{1,2\}\notin \beta$. 
Each function $\smash{\coulombGas_{\beta/\vee_1}^{(\kappa)}}$ defined in~\eqref{eqn::CGI_fused_generic} satisfies the
third order BPZ PDE
\begin{align}\label{eqn::third_order_PDE}
\Big[	\frac{\partial^3}{\partial \xi^3} -\frac{16}{\kappa}\LL_{-2}\frac{\partial}{\partial \xi} +\frac{8(8-\kappa)}{\kappa^2} \LL_{-3}\Big] \coulombGas_{\beta/\vee_1}^{(\kappa)}(\xi, \bs{\ddot{x}}_1)=0,
\end{align}
where 
\begin{align}\label{eqn::L_operators}
\LL_{-2}:= \; &\sum_{i=3}^{2N} \Big(\frac{h(\kappa)}{(x_i-\xi)^2}-\frac{1}{x_i-\xi}\frac{\partial}{\partial x_i}\Big) 
\qquad \textnormal{and} \qquad
\LL_{-3} := \: \sum_{i=3}^{2N} \Big( \frac{2h(\kappa)}{(x_i-\xi)^3}-\frac{1}{(x_i-\xi)^2}\frac{\partial}{\partial x_i}\Big).
\end{align}
\end{proposition}

We prove Proposition~\ref{prop::Frobenius_F} in Section~\ref{subsec: Proof of Frob generic} and Proposition~\ref{prop::third_order_PDE} in Section~\ref{subsec::3rdorderBPZ}.

\medskip

The Frobenius series~(\ref{eqn::CGI_Frobenius_paired},~\ref{eqn::CGI_Frobenius_tying}) are consistent with the fusion rules in a CFT model with two fields $\Phi_{1,2}(x_1)$ and $\Phi_{1,2}(x_2)$ 
merging to $\Phi_{1,1}(\xi) = 1$ in the former case, or a combination of $\Phi_{1,1}(\xi) = 1$ and $\Phi_{1,3}(\xi)$ in the latter case. 
The factor $\smash{\coulombGas_{\beta/\vee_1}^{(\kappa)}(\xi, \bs{\ddot{x}}_1)}$ 
represents a correlation function of a product 
$\Phi_{1,3}(\xi) \Phi_{1,2}(x_3) \cdots \Phi_{1,2}(x_{2N})$, 
while the Coulomb gas integral function 
$\smash{\coulombGas_{\beta}^{(\kappa)}(\bs{{x}})}$ 
represents a correlation function of a product 
$\Phi_{1,2}(x_1) \Phi_{1,2}(x_2) \Phi_{1,2}(x_3) \cdots \Phi_{1,2}(x_{2N})$. 
The third order BPZ PDE in Proposition~\ref{prop::third_order_PDE} is a feature of the field $\Phi_{1,3}$, degenerate at level three, while the field $\Phi_{1,2}$ is degenerate at level two --- and its correlation functions satisfy the very thoroughly studied second order equations~\eqref{eqn::PDE}.

Higher order PDEs do not seem to have a probabilistic interpretation, 
but Proposition~\ref{prop::third_order_PDE} (which essentially follows from a computation and the earlier work~\cite{Kytola-Peltola:Conformally_covariant_boundary_correlation_functions_with_quantum_group}),
demonstrates that such higher order PDEs predicted by CFT are indeed present for the Coulomb gas integrals (viz.~partition functions).
See also~\cite{Dubedat:SLE_and_Virasoro_representations_fusion,KKP:Boundary_correlations_in_planar_LERW_and_UST,Peltola:Basis_for_solutions_of_BSA_PDEs_with_particular_asymptotic_properties}, 
and the recent~\cite{KLPR:Planar_UST_branches_and_degenerate_boundary_correlations,LPR:Fused_Specht_polynomials_and_c_equals_1_degenerate_conformal_blocks}.
We expect the same results to also hold for all parameters 
$\kappa\in (0,8) \setminus \big\{ \tfrac{8}{m} \colon m \in \bZpos \big\}$. However, the detailed analysis of the Coulomb gas integral becomes increasingly complicated in the intervals $\kappa\in \big(\tfrac{8}{4m-1}, \tfrac{8}{4m-3}\big)$ with $m = 2,3,4,\ldots$.

The CFT describing the geometry of critical $O(n)$ loop models is expected to be logarithmic for all $n \in (0,2)$~\cite{Gorbenko-Zan:Two-dimensional_On_models_and_logarithmic_CFTs,
GSNJRS:Global_symmetry_and_conformal_bootstrap_in_the_two-dimensional_On_model}. 
At the first level of fusion, however, 
we only see anomalous behavior in the borderline cases $\kappa\in \{8/3,8\}$, 
that is, $m \in \{1,3\}$, which correspond respectively to the logarithmic minimal models 
$M(2,1)$ and $M(2,3)$ with central charges
$c = -2$ (dense $O(0)$ model, or UST, Proposition~\ref{prop::CGI8_Frobenius}), 
and $c = 0$ (dilute $O(0)$ model, or self-avoiding walk, Proposition~\ref{prop::CGI8over3_Frobenius}). 
More generally, the exceptional values of $\kappa \in \smash{\big\{ \tfrac{8}{m} \colon m \in \bZpos \big\}}$ correspond to central charges\footnote{Note also that with $m = \tfrac{2p'}{p}$, we have $\kappa = \tfrac{4 p}{p'}$.} 
\begin{align*}
c \big( 8/m \big) = 13 - 6 \big( 2/m + m/2 \big) 
= 1 - 6 \, \frac{(p-p')^2}{p p'}  ,
\qquad m = \frac{2p'}{p} ,
\end{align*}
with $p,p' \geq 1$ coprime or $(p,p') = (1,1)$. 
The case $m=2$ (loop $O(2)$ model) corresponds to $(p,p') = (1,1)$ and $c=1$ (with $\kappa=4$). 
In the three cases~\eqref{eq:exceptoinal_kappas}, we have the following identifications\footnote{Traditionally, the cases where $p,p' \geq 2$ are coprime label the CFT \emph{minimal models} $M(p,p')$
--- theories comprising a finite number of primary fields~\cite{DMS:CFT} --- 
but Coulomb gas integrals describe correlation functions outside these models~\cite{Peltola:Basis_for_solutions_of_BSA_PDEs_with_particular_asymptotic_properties}.}.

\begin{itemize}[leftmargin=1em]
\item For 
$\kappa\in\big\{ \tfrac{4}{m} \colon m \in \bZpos \big\}$, 
in which case $\smash{\coulombGas_{\beta}^{(\kappa)}}$ is finite, we see that $p=1$ and $p'=m/2$.
The logarithmic minimal models $M(1,p')$ were discussed, e.g., in~\cite{PRZ:Logarithmic_minimal_models} and references therein.

\item For 
$\kappa\in \big\{ \tfrac{8}{2m+1} \colon m \in \bZpos \big\}$, in which case  
$\smash{\coulombGas_{\beta}^{(\kappa)}}$ is zero but the renormalized function 
$\smash{\coulombGasRenorm_{\beta}^{(\kappa)}}$ is well-defined  via~\eqref{eqn::CGI_odd}, we see that $p=2$ and $p'=m$.
The logarithmic minimal models $M(2,p')$ were discussed, e.g., in~\cite{PRZ:Logarithmic_minimal_models, Read-Saleur:Associative-algebraic_approach_to_log_CFT} and references therein.
In Proposition~\ref{prop::CGI8over3_Frobenius} we consider the special case where $m=3$,
corresponding to the dilute $O(0)$ model, or self-avoiding walk ($c = 0$ and $\kappa = 8/3$).

\item For $\kappa=8$, in which case $\smash{\coulombGas_{\beta}^{(8)}}$ is zero but the renormalized function 
$\smash{\coulombGasRenorm_{\beta}^{(8)}}$ is well-defined  via~\eqref{eqn::CGI_8}, we see that $p=2$ and $p'=1$.
The logarithmic minimal model $M(2,1)$ was discussed, e.g.,
in~\cite{PRZ:Logarithmic_minimal_models,
Pearce-Rasmussen:Solvable_critical_dense_polymers,
Read-Saleur:Associative-algebraic_approach_to_log_CFT} 
and references therein.
We consider this case, corresponding to the dense $O(0)$ model, or UST ($c = -2$ and $\kappa = 8$),
in Proposition~\ref{prop::CGI8_Frobenius}.
\end{itemize}

A minimal model $M(p,p')$ comprises fields 
$\{ \Phi_{r,s} \colon r,s \in \bZpos , \;  1 \leq r \leq p'-1 , \; 1 \leq s \leq p-1 , \; p r > p' s \}$. 
In all of the above cases, the field $\Phi_{1,2}$ is absent from the corresponding minimal model.
This could heuristically explain why $\smash{\coulombGas_{\beta}^{(\kappa)}}$ are identically zero in the cases where $\kappa = 8/m$ with $m$ odd.
(The cases where $\kappa = 8/m$ with $m$ even are outside of the minimal models anyway, since $p=1$ in these cases.)

\subsubsection{Logarithmic minimal model $M(2,3)$ --- dilute $O(0)$ loops with $\kappa=8/3$ and $c(\kappa) = 0$}

When $\kappa = 8/3$, the Coulomb gas integrals vanish (see Lemma~\ref{lem:CGI_vanish}).
The renormalized integrals, however, make sense as in Proposition~\ref{prop::CGI_odd}.
We can derive the explicit Frobenius series also in this case. 
From the physics perspective, models which are believed to be related to dilute $O(0)$ loop models
and logarithmic minimal models have been considered, e.g., in~\cite{Read-Saleur:Associative-algebraic_approach_to_log_CFT}.
Gurarie~\cite[Eq.~(45)]{Gurarie:Logarithmic_operators_and_logarithmic_conformal_field_theories}  
computed the OPE of the anticipated field $\Phi_{1,2}$ in such a model, finding the asymptotic expansion 
(see also~\cite[Appendix~B]{FSK:Multiple_SLE_connectivity_weights_for_rectangles_hexagons_and_octagons})
\begin{align*}
\langle \Phi_{1,2}(x_1) \Phi_{1,2}(x_2) \cdots \rangle  
= \; & c_1 \, (x_2-x_1)^{-5/4} \, \langle \Phi_{1,1}(\xi)  \cdots\rangle  
\; + \; c_2 \, (x_2-x_1)^{3/4} \, \langle \Psi(\xi) \cdots \rangle  
\\
& \; + \; (c_3 + c_2 \log(x_2-x_1)) \, (x_2-x_1)^{3/4} \, \langle T(\xi) \cdots \rangle  
\\
& \; + \; o\big((x_2-x_1)^{3/4}\big) , \qquad \textnormal{as } x_1, x_2 \to \xi .
\end{align*}
where $\Psi(\xi)$ denotes the so-called logarithmic partner of the stress-energy tensor $T(\xi)$, 
and $c_1$, $c_2$, and $c_3$ are some constants 
(which are explicit in Proposition~\ref{prop::CGI8over3_Frobenius}, 
if the functions $\smash{\coulombGasRenorm_{\beta}^{(8/3)}}$ and $\smash{\coulombGasRenorm_{\beta/\vee_1}^{(8/3)}}$ are interpreted suitably).
The relationship of $\SLE_{8/3}$ with the stress-energy tensor was observed already by 
Friedrich~\&~Werner~\cite{Friedrich-Werner:Conformal_restriction_highest_weight_representations_and_SLE}, but we have not found a mathematically complete treatment. 

\begin{proposition}\label{prop::CGI8over3_Frobenius}
Fix $\kappa=8/3$ and $N\ge 2$. 
Each renormalized Coulomb gas integral $\smash{\coulombGasRenorm_{\beta}^{(8/3)}(\bs{x})}$ 
in Proposition~\ref{prop::CGI_odd}, with $\beta \in \LP_N$, has the following expansion. 
\begin{itemize}[leftmargin=1em]
\item If $\{1,2\}\in\beta$, 
then for all $\xi\in (-\infty, x_3)$, 
\begin{align}\label{eqn::CGIRe_Frobenius_paired}
\coulombGasRenorm_{\beta}^{(8/3)}(\bs{x})
= \; & \coulombGasRenorm_{\beta/\vee_1}^{(8/3)}(\xi,\bs{\ddot{x}}_1) \, (x_2-x_1)^{3/4}
\; + \; o\big((x_2-x_1)^{3/4}\big),\qquad \textnormal{as }x_1, x_{2}\to \xi ,
\end{align}
where 
\begin{align}\label{eqn::CGIRe_fused_paired}
\coulombGasRenorm_{\beta/\vee_1}^{(8/3)}(\xi, \bs{\ddot{x}}_1)
:= \; & \frac{\cst(8/3)^{N-2}}{ \hF\big(\tfrac{3}{2},-\tfrac{1}{2},3;1\big)}
\; \ointclockwise_{\acycle(x_{a_3},x_{b_3})}\ud u_3 \cdots \ointclockwise_{\acycle(x_{a_N},x_{b_N})} \ud u_N 
\; \hat{f}_{\beta/\vee_1}^{(8/3)} (\xi, \bs{\ddot{x}}_1; \bs{\ddot{u}}_1) ,
\end{align}
and $\smash{\hat{f}_{\beta/\vee_1}^{(8/3)}}$ is defined in~\eqref{eqn::integrand_CGI_fused_generic} with $\kappa=8/3$ and
\begin{align*}
\cst(8/3) = \lim_{\kappa\to 8/3} \cst(\kappa) = -\frac{1}{16}
\qquad \textnormal{and} \qquad 
\hF\big(\tfrac{3}{2},-\tfrac{1}{2},3;1\big) = \frac{32}{15\pi} .
\end{align*}

\item If $\{1,2\}\not\in\beta$, 
then for all $\xi \in (-\infty, x_3)$, 
\begin{align}\label{eqn::CGIRe_Frobenius_tying}
\coulombGasRenorm_{\beta}^{(8/3)}(\bs{x})
= & \;
\coulombGasRenorm_{\wp_1(\beta)/\{1,2\}}^{(8/3)}(\bs{\ddot{x}}_1) \, (x_2-x_1)^{-5/4}
\; - \; \frac{1}{\pi} \, (-1)^{(b_2-1)/2} \, \coulombGasRenorm_{\wp_1(\beta)/\vee_1}^{(8/3)}(\xi, \bs{\ddot{x}}_1) \, (x_2-x_1)^{3/4} \, |\log(x_2-x_1)| \notag\\
& \; + \; O\big((x_2-x_1)^{3/4} \big) , \qquad \textnormal{as } x_1, x_2 \to \xi .
\end{align}
\end{itemize}
\end{proposition}

We prove Proposition~\ref{prop::CGI8over3_Frobenius} in Section~\ref{subsec: Proof of Frob 83}.

\subsubsection{Logarithmic minimal model $M(2,1)$ --- dense $O(0)$ loops with $\kappa=8$ and $c(\kappa) = -2$}

Conjecturally, the geometry of dense $O(0)$ loops with $c = -2$ should be closely related to $\SLE_8$ processes and interfaces in the UST model. 
The scaling limit of the latter has been intensively studied (e.g.,~\cite{LSW:Conformal_invariance_of_planar_LERW_and_UST,
Dubedat:Euler_integrals_for_commuting_SLEs, 
Kenyon-Wilson:Boundary_partitions_in_trees_and_dimers,
LPW:UST_in_topological_polygons_partition_functions_for_SLE8_and_correlations_in_logCFT}). 
In the physics literature, possible conformal field theories relevant to this model 
are among the most investigated logarithmic CFTs --- see, e.g.,~\cite{Gurarie-Ludwig:Conformal_algebras_of_two-dimensional_disordered_systems,
Read-Saleur:Associative-algebraic_approach_to_log_CFT}  
and references therein. 
In particular, Gurarie~\cite[Eq.~(11,12)]{Gurarie:Logarithmic_operators_in_conformal_field_theory} 
computed the OPE of the anticipated field $\Phi_{1,2}$, finding the asymptotic expansion 
(see also~\cite[Appendix~B]{FSK:Multiple_SLE_connectivity_weights_for_rectangles_hexagons_and_octagons})
\begin{align*}
\langle \Phi_{1,2}(x_1) \Phi_{1,2}(x_2) \cdots \rangle  
= \; & (c_1 + c_2 \log(x_2-x_1)) \, (x_2-x_1)^{1/4} \, \langle \Phi_{1,1}(\xi) \cdots \rangle  
\; + \; c_2 \, (x_2-x_1)^{1/4} \, \langle \Psi(\xi) \cdots \rangle  \\
& \; + \; o\big((x_2-x_1)^{1/4} \big) , \qquad \textnormal{as } x_1, x_2 \to \xi .
\end{align*}
where $\Psi(\xi)$ denotes the so-called logarithmic partner of the identity field $\Phi_{1,1}(\xi) = 1$, and $c_1$ and $c_2$ are some constants (which are explicit in Proposition~\ref{prop::CGI8_Frobenius}, if the functions $\smash{\coulombGasRenorm_{\beta}^{(8)}}$ 
are interpreted suitably).

\begin{proposition}\label{prop::CGI8_Frobenius}
Fix $\kappa=8$ and $N\ge 2$. Each renormalized Coulomb gas integral $\coulombGasRenorm_{\beta}^{(8)}(\bs{x})$ in Proposition~\ref{prop::CGI_8}, with $\beta\in\LP_N$, has the following expansion. 
\begin{itemize}[leftmargin=1em]
\item If $\{1,2\}\in\beta$, 
then for all $\xi\in (-\infty, x_3)$, 
\begin{align}\label{eqn::UST_paired}
\coulombGasRenorm_{\beta}^{(8)}(\bs{x})
= \; & \frac{\pi}{8} \, \coulombGasRenorm_{\beta/\{1,2\}}^{(8)}(\bs{\ddot{x}}_1)(x_{2}-x_1)^{1/4}
\; + \; o\big((x_2-x_1)^{1/4}\big) , \qquad \textnormal{as } x_1, x_2\to \xi .
\end{align}
\item If $\{1,2\}\not\in\beta$, 
then for all $\xi\in (-\infty, x_3)$, 
\begin{align} \label{eqn::UST_log}
\begin{split}
\coulombGasRenorm_{\beta}^{(8)}(\bs{x})
= & \; \frac{1}{8} \, \coulombGasRenorm_{\wp_1(\beta)/\{1,2\}}^{(8)}(\bs{\ddot{x}}_1) 
\, (x_2-x_1)^{1/4} 
\, |\log(x_{2}-x_1)| \; + \; O\big((x_2-x_1)^{1/4}\big) , \quad \textnormal{as } x_1, x_2 \to \xi .
\end{split}
\end{align}
\end{itemize}
\end{proposition}
The expansion~\eqref{eqn::UST_paired} and a weaker version of~\eqref{eqn::UST_log} were proved in~\cite[Proposition~2.9]{LPW:UST_in_topological_polygons_partition_functions_for_SLE8_and_correlations_in_logCFT}. 
We prove Proposition~\ref{prop::CGI8_Frobenius} in Section~\ref{subsec::USTlog}. 

\medskip

\begin{remark}
Illustration for the terms appearing in Frobenius series in Propositions~\ref{prop::Frobenius_F}, \ref{prop::CGI8over3_Frobenius},~\&~\ref{prop::CGI8_Frobenius}. 
\begin{equation*}
\begin{array}{|c|c|c|c|c|}
\hline
&&&&\\
&&(x_2-x_1)^{1-6/\kappa}&(x_2-x_1)^{2/\kappa}|\log(x_2-x_1)|&(x_2-x_1)^{2/\kappa}\\
\hline
&&&&\\
\kappa\in(8/3,8)&\{1,2\}\in\beta&\checkmark&\times & \times\\
\hline
&&&&\\
\kappa\in(8/3,8)&\{1,2\}\not\in\beta&\checkmark&\times & \checkmark\\
\hline
&&&&\\
\kappa=8/3&\{1,2\}\in\beta&\times&\times & \checkmark\\
\hline
&&&&\\
\kappa=8/3&\{1,2\}\not\in\beta&\checkmark&\checkmark & \checkmark\\
\hline
&&&&\\
\kappa=8&\{1,2\}\in\beta&\checkmark&\times & (\checkmark)\\
\hline
&&&&\\
\kappa=8&\{1,2\}\not\in\beta&\checkmark&\checkmark & (\checkmark)\\
\hline
\end{array}
\end{equation*}

\end{remark}

\paragraph*{Outline of this article.} 
In Section~\ref{sec::CGI}, we derive properties of the Coulomb gas integrals, including a general M\"{o}bius covariance property, a cyclic permutation relation, and recursive asymptotics properties. 
In Section~\ref{sec::multipleSLEs}, we construct a conformally invariant measure on multiple $\SLE_\kappa$ curves which will be used to construct pure partition functions and global multiple SLEs for $\kappa\in (4,8)$ in Section~\ref{sec::PPF}. 
In Section~\ref{sec::consequences}, using the tools from preceding sections, 
we complete the proofs of Theorems~\ref{thm::CGI}~\&~\ref{thm::PPF},  Propositions~\ref{prop::CGI_odd}~\&~\ref{prop::CGI_8}, and Propositions~\ref{prop::CGI_PPF}~\&~\ref{prop::CGI_PPF_renormalized}. In Section~\ref{sec:Frobenius_series_proofs}, prove  Propositions~\ref{prop::Frobenius_F}--\ref{prop::CGI8_Frobenius} (concerning Frobenius series expansions and BPZ PDEs).  

The bulk of this article is supplemented by several appendices that contain technical material or illuminating examples.
The most important one is the last 
Appendix~\ref{app::Frob}, deriving explicit formulas for partition functions in the case $\kappa=4$.
The first Appendix~\ref{app::identities} gathers some useful simple integral identities,
and the second Appendix~\ref{app::additional_formulas} gathers useful but technical Coulomb gas integral formulas for usage in Sections~\ref{sec::multipleSLEs}--\ref{sec::PPF}. 
The third Appendix~\ref{app::examples} contains illuminating examples of the partition functions ($N=1,2$).

\medbreak
\paragraph*{Acknowledgments.}
\begin{itemize}[leftmargin=1em]
\item We thank Bernard Nienhuis for several interesting discussions related to $O(n)$ models and Coulomb gas. 
We thank Nicolai Makarov and Jiaxin Zhang for inspiring discussions about Coulomb gas, and Xin Sun for stimulating discussions about multiple SLEs and sharing their work.
We thank Dapeng Zhan for his interest and very useful correspondences regarding an earlier version of this article. 

\item Part of this research was performed while Y.F., E.P., and H.W. visited the Institute for Pure and Applied Mathematics (IPAM), supported by the National Science Foundation (Grant No.~DMS-1925919).

\item During this work, M.L. was supported by the Knut and Alice Wallenberg Foundation.

\item During this work, Y.F. and H.W. were supported by Beijing Natural Science Foundation (JQ20001) and New Cornerstone Investigator Program 100001127. H.W. is partly affiliated at Yanqi Lake Beijing Institute of Mathematical Sciences and Applications, Beijing, China.

\item This material is part of a project that has received funding from the  European Research Council (ERC) under the European Union's Horizon 2020 research and innovation programme (101042460): 
ERC Starting grant ``Interplay of structures in conformal and universal random geometry'' (ISCoURaGe) 
and from the Academy of Finland grant number 340461 ``Conformal invariance in planar random geometry.''
E.P.~is also supported by 
the Academy of Finland Centre of Excellence Programme grant number 346315 ``Finnish centre of excellence in Randomness and STructures (FiRST)'' 
and by the Deutsche Forschungsgemeinschaft (DFG, German Research Foundation) under Germany's Excellence Strategy EXC-2047/1-390685813, 
as well as the DFG collaborative research centre ``The mathematics of emerging effects'' CRC-1060/211504053.

\end{itemize}


\bigskip{}
\section{Coulomb gas integrals}
\label{sec::CGI}
Throughout, 
we consider link patterns 
$\beta \in \LP_N$ with link endpoints ordered as in~\eqref{eq: link pattern ordering}.
Recall that
\begin{align} \label{eq::F_as_C_times_H}
\coulombGas_{\beta}
= \coulombGas_{\beta}^{(\kappa)} 
:= \cst(\kappa)^N \, \coulomb_\beta^{(\kappa)} ,\qquad \textnormal{where}\qquad
\cst(\kappa) := 
\frac{\fugacity(\kappa)}{4\sin^2(4\pi/\kappa)} 
\, \frac{\Gamma(2-8/\kappa)}{\Gamma(1-4/\kappa)^2} ,
\end{align}
and where, using the branch choice~\eqref{eq: branch choice} of the integrand $f_\beta = \smash{f_\beta^{(\kappa)}}$ defined in~\eqref{eq: integrand}, we have
\begin{align}
\label{eqn::coulombgasintegral_Poch}
\coulomb_\beta (\bs{x}) 
\; = \;
\coulomb_\beta^{(\kappa)} (\bs{x}) 
\; := \; \ointclockwise_{\acycle^\beta_1}  \ud u_1 \ointclockwise_{\acycle^\beta_2}  \ud u_2 \cdots \ointclockwise_{\acycle^\beta_N} \ud u_N 
\; f_\beta^{(\kappa)}(\bs{x};\bs{u}) , 
\qquad \bs{x} \in \chamber_{2N} ,
\end{align}
where we choose the Pochhammer contours $\acycle^\beta_1 ,\ldots, \acycle^\beta_N$ as in Figure~\ref{fig:Pochhammer_examples}.
Choosing $\alpha_j = 1/\sqrt{\kappa}$ for all $1 \leq j \leq 2N$ 
(this is the charge associated to a degenerate CFT field at level two), we see that the integral~\eqref{eqn::coulombgasintegral_Poch} can be written in the Coulomb gas form with $N$ screenings $\alpha_- = - 2/\sqrt{\kappa}$, 
\begin{align*}
\coulomb_\beta (\bs{x}) 
= \; & \underset{\acycle^\beta_1 \times \cdots \times \acycle^\beta_N}{\ointclockwise}
\; \prod_{1\leq i<j\leq 2N}(x_{j}-x_{i})^{2 \alpha_i \alpha_j} 
\prod_{1\leq r<s\leq N}(u_{s}-u_{r})^{2 \alpha_-^2} 
\prod_{\substack{1\leq i\leq 2N \\ 1\leq r\leq N}}
(u_{r}-x_{i})^{2 \alpha_- \alpha_i}
\; \ud u_1 \cdots \ud u_N , 
\end{align*}
where the branch choice is implicit. 
The conformal weight for a degenerate field at level two is given by $h_j = h(\kappa) = \frac{6 - \kappa}{2 \kappa}$.  
It gives the M\"obius covariance factor as in Proposition~\ref{prop: full Mobius covariance}.

\bigskip

Now, the main purpose of this section is to collect properties of $\smash{\coulomb_\beta}$.
\begin{itemize}[leftmargin=2em]
\item We prove a general M\"{o}bius covariance property in Proposition~\ref{prop: full Mobius covariance} in Section~\ref{subsec::MobiusCov}, involving also M\"{o}bius transformations that may move the point $\infty$.

\item We prove a cyclic permutation relation in Proposition~\ref{prop: rotation symmetry} in Section~\ref{subsec::rotation_sym}. 
This agrees with the deep relationship of the Coulomb gas integrals with representation theory of the quantum group $U_q(\mathfrak{sl}(2,\C))$, 
where $q = \exp(4 \pi \ii / \kappa)$. 
Indeed, at least when $\kappa \notin \QQ$ so that $q$ is not a root of unity, 
the monodromy of the Coulomb gas integrals is encoded into the (universal) R-matrix  of $U_q(\mathfrak{sl}(2,\C))$~\cite{Kassel:Quantum_groups} 
acting on a tensor product of its type-one two-dimensional simple modules~\cite{Felder-Wieczerkowski:Topological_representations_of_quantum_group, 
Kytola-Peltola:Conformally_covariant_boundary_correlation_functions_with_quantum_group}.
Our findings in Proposition~\ref{prop: rotation symmetry} perfectly match with this structure.
Note also that $\fugacity(\kappa) = - q - \tfrac{1}{q}$.

\item We give an alternative formula for $\smash{\coulomb_\beta}$ in Proposition~\ref{prop::remove_integration} in Section~\ref{subsec:Integral removal}, 
in terms of $N-1$ integrals satisfying charge neutrality~\cite{DMS:CFT,
Flores-Kleban:Solution_space_for_system_of_null-state_PDE3}, 
including the case where the integrand~\eqref{eq: integrand with one less screening variable b} equals 
\begin{align*} 
\hat{f} (\bs{x}; u_1, \ldots, u_{N-1})
:= \; &
\prod_{\substack{1\leq i<j\leq 2N \\i,j \neq 2N}}(x_{j}-x_{i})^{2 \alpha_i \alpha_j} 
\prod_{1\leq i<2N}(x_{2N}-x_{i})^{2 \alpha_i \alpha_{2N}} \\
\; & \times \; 
\prod_{\substack{1\leq t<s\leq N-1}}(u_{s}-u_{t})^{2 \alpha_-^2} 
\prod_{\substack{1\leq i\leq 2N-1 \\ 1\leq s\leq N-1}}
(u_{s}-x_{i})^{2 \alpha_- \alpha_i} 
\prod_{\substack{1\leq s\leq N-1}}
(u_{s}-x_{2N})^{2 \alpha_- \alpha_{2N}} ,
\end{align*}
where
$\alpha_j = 1/\sqrt{\kappa}$ for all $1 \leq j \leq 2N-1$ as before, 
and 
$\alpha_{2N} = \frac{\kappa - 6}{2\sqrt{\kappa}}$ is the so-called \emph{conjugate charge}\footnote{The charges associated to degenerate CFT fields are commonly labeled as $\alpha_{r,s} = \tfrac{1}{2} \big( (1-r) \alpha_+ + (1-s) \alpha_- \big)$. 
The conjugate charge is 
$\alpha_{r,s}^* = \tfrac{1}{2} \big( (1+s) \alpha_+ + (1+r) \alpha_- \big)$. 
In this notation, we have $\alpha_{1,2} = 1/\sqrt{\kappa}$ and $\alpha_{1,2}^* = \tfrac{1}{2} \big( (1+s) \alpha_+ + (1+r) \alpha_- \big)$.} 
of $1/\sqrt{\kappa}$,  
so that the sum of the charges equals the background charge:
\begin{align*}
\Big( \sum_{j=1}^{2N-1} \alpha_j \Big) + \Big( \frac{\kappa - 6}{2\sqrt{\kappa}} \Big) + (N-1) \alpha_- = 2 \alpha_0 .
\end{align*}
This shows in particular that the Coulomb gas integral~\eqref{eqn::CGI_def} is consistent with the earlier works of 
Dub\'edat~\cite{Dubedat:Euler_integrals_for_commuting_SLEs},
and Flores \&~Kleban~\cite{Flores-Kleban:Solution_space_for_system_of_null-state_PDE3}. 
Hence, our formulas, which are symmetric in the marked points, indeed agree with those considered in 
the physics literature, and demonstrates that the charge neutrality assumption, commonly imposed in physics, 
is not actually needed here.

\item 
We prove a recursive asymptotics property 
for $\coulomb_\beta$ 
in Proposition~\ref{prop::H_ASY} in Section~\ref{subsec:ASY},
from which we obtain the asymptotics for $\coulombGas_{\beta}$ in the case where
$\kappa\in (0,8) \setminus \big\{ \tfrac{8}{m} \colon m \in \bZpos \big\}$, see Corollary~\ref{cor::CGI_PDECOVASY}. 
\end{itemize}

\subsection{General M\"obius covariance}
\label{subsec::MobiusCov}

In this section, we prove a very general M\"obius covariance property for general Coulomb gas integrals. 

\begin{lemma} \label{lem::reflection}
Fix $\beta \in \LP_N$ and let $-\beta \in \LP_N$ be the link pattern obtained from $\beta$ by reflecting it across the imaginary axis. 
Then, we have
\begin{align} \label{eq::reflection}
\coulomb_{-\beta} (-x_{2N},\ldots,-x_{1}) 
= \coulomb_{\beta} (x_{1},\ldots,x_{2N}) , 
\qquad x_{1} < \cdots < x_{2N} .
\end{align}
\end{lemma}

\begin{proof}
Write the reflected link pattern as 
\begin{align*}
-\beta=\{\{a'_1,b'_1\},\ldots,\{a'_N,b'_N\}\}, \qquad a_1' < a_2' < \cdots < a_N' \textnormal{ and } a_r' < b_r' , \textnormal{ for all } 1 \leq r \leq N ,
\end{align*}
and let $\pi$ be the permutation of $\{1,2,\ldots,2N\}$ determined by
\begin{align*}
a'_{\pi(r)}=2N+1-b_r, \quad b'_{\pi(r)}=2N+1-a_r, \quad 1\leq r\leq N.
\end{align*}
Making in $\coulomb_{-\beta} (-x_{2N},\ldots,-x_{1})$ the change of variables $v_{\pi(r)}=-u_r$ for $1\leq r\leq N$ and collecting the phase factors relative to our branch choice~\eqref{eq: branch choice} of the integrand $\smash{f_\beta}$ defined in~\eqref{eq: integrand} --- to observe that they all cancel out --- we readily obtain~\eqref{eq::reflection}. 
For instance, for $N=1$ and $v=-u$, we simply have
\begin{align*}
\coulomb_{\includegraphics[scale=0.2]{figures-arXiv/link-0}} (-y,-x) 
= \; & \underset{\acycle^{\includegraphics[scale=0.2]{figures-arXiv/link-0}}(-y,-x)}{\ointclockwise} \; f_{\includegraphics[scale=0.2]{figures-arXiv/link-0}}^{(\kappa)}(-y,-x;u)  \; \ud u 
= - \underset{\acycle^{\includegraphics[scale=0.2]{figures-arXiv/link-0}}(x,y)}{ \ointctrclockwise} \; f_{\includegraphics[scale=0.2]{figures-arXiv/link-0}}^{(\kappa)}(-y,-x;-v)  \; \ud v 
\\
= \; &  \underset{\acycle^{\includegraphics[scale=0.2]{figures-arXiv/link-0}}(x,y)}{ \ointclockwise} \; f_{\includegraphics[scale=0.2]{figures-arXiv/link-0}}^{(\kappa)}(x,y;v)  \; \ud v \,
= \coulomb_{\includegraphics[scale=0.2]{figures-arXiv/link-0}} (x,y) ,
\end{align*}
where $\acycle^{\includegraphics[scale=0.2]{figures-arXiv/link-0}}$ is the Pochhammer contour~\eqref{eq:Pochhammer}. 
For $N=2$ and $\beta = \vcenter{\hbox{\includegraphics[scale=0.4]{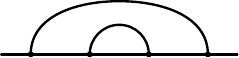}}}$, we have
\begin{align*}
\coulomb_{\includegraphics[scale=0.2]{figures-arXiv/link-2}} (-x_4,-x_3,-x_2,-x_1) 
= \; & \underset{\acycle^{\includegraphics[scale=0.2]{figures-arXiv/link-0}}(-x_4,-x_1) \times \acycle^{\includegraphics[scale=0.2]{figures-arXiv/link-0}}(-x_3,-x_2)
}{\ointclockwise} \; f_{\includegraphics[scale=0.2]{figures-arXiv/link-0}}^{(\kappa)}(-x_4,-x_3,-x_2,-x_1;u_1,u_2)  \; \ud u_1 \ud u_2
\\
= \; & \underset{\acycle^{\includegraphics[scale=0.2]{figures-arXiv/link-0}}(x_1,x_4) \times \acycle^{\includegraphics[scale=0.2]{figures-arXiv/link-0}}(x_2,x_3)
}{\ointclockwise} \; f_{\includegraphics[scale=0.2]{figures-arXiv/link-0}}^{(\kappa)}(x_1,x_2,x_3,x_4;v_1,v_2)  \; \ud v_1 \ud v_2,
\end{align*}
where $(v_1,v_2)=(-u_1,-u_2)$ (see also Figure~\ref{fig::reflection}) and the phase factor in the transformation equals
$ \exp\big( \pi \ii \big( 2 \alpha_- \alpha_2 + 2 \alpha_- \alpha_3 + 2 \alpha_- \alpha_- \big)\big)
= \exp\big( \! -4 \pi \ii \frac{1}{\sqrt{\kappa}} \big(\frac{1}{\sqrt{\kappa}} + \frac{1}{\sqrt{\kappa}} - \frac{2}{\sqrt{\kappa}} \big)\big) = 1$.
\end{proof}

\begin{figure}[ht!]
\begin{center}
\includegraphics[width=0.9\textwidth]{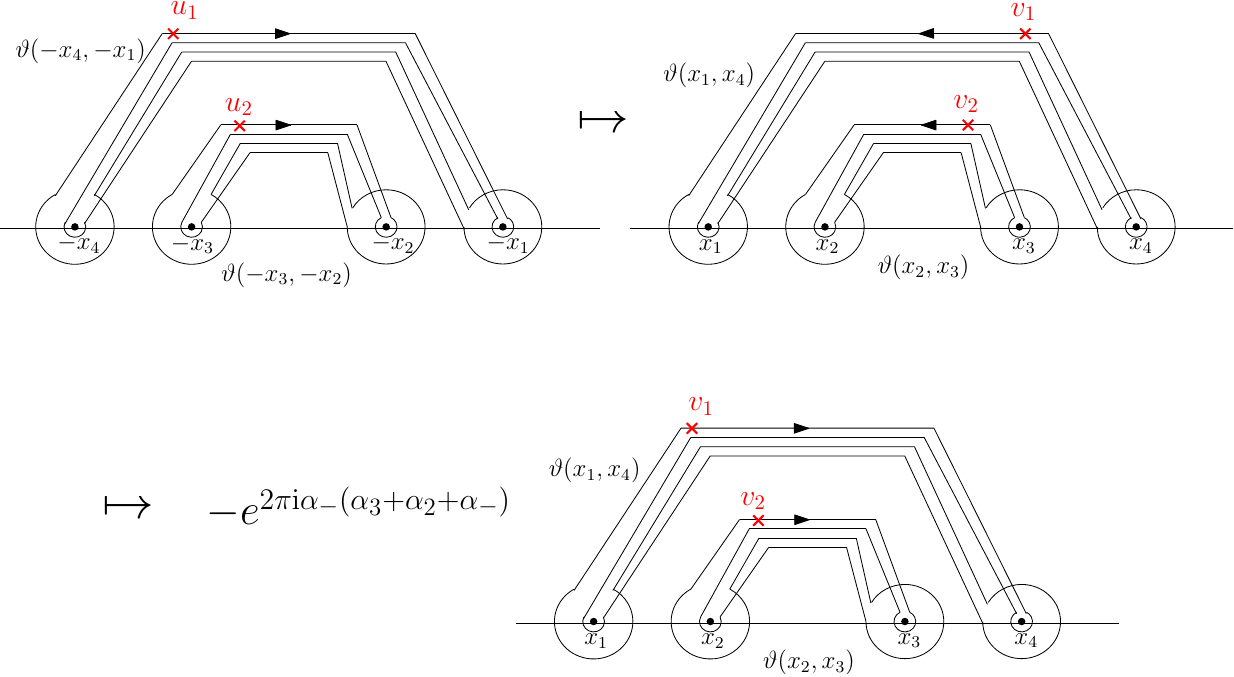}
\end{center}
\caption{\label{fig::reflection} 
Illustration of the branch choices in the proof of Lemma~\ref{lem::reflection} 
for $N=2$ 
}
\end{figure}

\begin{proposition} \label{prop: full Mobius covariance}
Fix $\beta \in \LP_N$. The function $\smash{\coulomb_{\beta}}$ satisfies the general M\"obius covariance 
\begin{align}\label{eqn::general_COV_for_H}
\coulomb_{\beta} (\bs{x}) = 
\prod_{i=1}^{2N} | \varphi'(x_{i}) |^{h(\kappa)} \times
\, \coulomb_{\varphi(\beta)} (\varphi(\bs{x})) , 
\qquad \bs{x} \in \chamber_{2N} ,
\end{align}
for all M\"obius maps $\varphi \in \mathrm{SL}(2,\R)$, 
where 
$\varphi(\bs{x}) := (\varphi(x_{j+1}) \ldots , \varphi(x_{2N}) , \varphi(x_{1}) , \ldots , \varphi(x_{j}))$, 
so that 
\begin{align*}
\varphi(x_{j+1}) < \varphi(x_{j+2}) < \cdots < \varphi(x_{2N}) < \varphi(x_{1}) < \varphi(x_{2}) < \cdots  < \varphi(x_{j})
\end{align*} 
with some $j \in \{0,1,2,\ldots,2N-1\}$,
and where $\varphi(\beta) \in \LP_N$ is the link pattern obtained from $\beta$ via permuting the indices according to the permutation of the boundary points induced by $\varphi$ and then ordering the link endpoints appropriately\footnote{Note that since $\varphi$ is a diffeomorphism on $\R$, for any $\beta \in \LP_N$, also $\varphi(\beta) \in \LP_N$ is a planar link pattern.}. 
\end{proposition}

\begin{proof}
The proof is a generalization of the arguments appearing in~\cite[Proposition~4.15]{Kytola-Peltola:Conformally_covariant_boundary_correlation_functions_with_quantum_group} 
(for $\kappa \notin \QQ$).  
Recall that the M\"obius group $\mathrm{SL}(2,\R)$ is generated by translations $z \mapsto z + b$, for $b \in \R$, 
dilations $z \mapsto a z$, for $a \in \R \setminus \{0\}$, 
and special conformal transformations $\varphi_{c} \colon z \mapsto \frac{z}{1 + c z}$, for $c \in \R$.
The covariance~\eqref{eqn::general_COV_for_H} under translations and scalings $z \mapsto a z$, for $a > 0$, is straightforward to check.  
The covariance~\eqref{eqn::general_COV_for_H}
for dilations $\delta_{a} \colon z \mapsto a z$, for $a < 0$, can be checked by considering the $a$-variation of the right-hand side of~\eqref{eqn::general_COV_for_H} with $\varphi = \delta_{a}$ in the following manner. 
Fix $a_0 \in \R \setminus \{0\}$ and a small $\varepsilon>0$, and consider $a \in (a_0-\varepsilon,a_0+\varepsilon)$. Note that the integration contours $\smash{\acycle_1^{\delta_{a}(\beta)} \times \cdots \times \acycle_N^{\delta_{a}(\beta)}}$ stay in the same homotopy class under a variation by $\delta_{a}$ when $\varepsilon$ is small enough.  Hence, the $a$-variation of the right-hand side of~\eqref{eqn::general_COV_for_H} with $\varphi = \delta_{a}$ equals
\begin{align}
\nonumber
\; & \der a \bigg(
\prod_{i=1}^{2N} | \delta_{a}'(x_{i}) |^{h(\kappa)} \times 
\, \underset{\acycle_1^{\delta_{a}(\beta)} \times \cdots \times \acycle_N^{\delta_{a}(\beta)}}{\ointclockwise} f_{\delta_{a}(\beta)} (\delta_{a}(\bs{x});\bs{u}) \; \ud u_1 \cdots \ud u_N 
\bigg) \\
\label{eq: dilation}
=\; & \frac{1}{a} \prod_{i=1}^{2N} | \delta_{a}'(x_{i}) |^{h(\kappa)} \times \, 
\underset{\acycle_1^{\delta_{a}(\beta)} \times \cdots \times \acycle_N^{\delta_{a}(\beta)}}{\ointclockwise}
\sum_{j=1}^{2N} \big( a x_{j} \partial_j + h(\kappa) \big) f_{\delta_{a}(\beta)} (\delta_{a}(\bs{x});\bs{u}) \; \ud u_1 \cdots \ud u_N .
\end{align}
A short computation shows that
\begin{align*}
\sum_{j=1}^{2N} \Big( y_{j} \pder{y_{j}} + h(\kappa) \Big) f_{\delta_{a}(\beta)} (\bs{y};\bs{u})
= - \sum_{r=1}^{N} \pder{u_{r}} \big( u_{r} f_{\delta_{a}(\beta)} \big) (\bs{y};\bs{u}) .
\end{align*}
Applying this to $\bs{y} = \delta_{a}(\bs{x})$, we see that~\eqref{eq: dilation} equals  
\begin{align*}
\textnormal{\eqref{eq: dilation}}
=\; & - \frac{1}{a} \prod_{i=1}^{2N} | \delta_{a}'(x_{i}) |^{h(\kappa)} \times \, 
\underset{\acycle_1^{\delta_{a}(\beta)} \times \cdots \times \acycle_N^{\delta_{a}(\beta)}}{\ointclockwise}
\sum_{r=1}^{N} \pder{u_{r}} \big( u_{r} f_{\delta_{a}(\beta)} \big)
(\delta_{a}(\bs{x});\bs{u}) \; \ud u_1 \cdots \ud u_N ,
\end{align*}
which equals zero because each term in the sum vanishes by integration by parts, as the Pochhammer contours are homologically trivial. 
Therefore, the right-hand side of the asserted formula~\eqref{eqn::general_COV_for_H} with $\varphi = \delta_{a}$ is constant in $a \in (a_0-\varepsilon,a_0+\varepsilon)$ for any $a_0 \in \R \setminus \{0\}$ and small enough $\varepsilon>0$. 
On the one hand, since at $a_0 = 1$ we have $\delta_{1} = \mathrm{id}$, this constant equals $\smash{\coulomb_{\beta}} (\bs{x})$ for all $a > 0$ (as we already knew). 
On the other hand, since at $a_0 = -1$ we have $\delta_{-1} = -\mathrm{id}$, and $\delta_{-1}'(z) = -1$, 
to prove the asserted formula~\eqref{eqn::general_COV_for_H} for dilations it suffices to invoke Lemma~\ref{lem::reflection}: 
the right-hand side equals $(-1)^{2N} \smash{\coulomb_{-\beta}} (-x_{2N},\ldots,-x_{1}) = \smash{\coulomb_{\beta}} (x_{1},\ldots,x_{2N})$.

Lastly, to prove the covariance property~\eqref{eqn::general_COV_for_H}
for special conformal transformations $\varphi_{c} \colon z \mapsto \frac{z}{1 + c z}$, for $c \in \R$, 
we generalize and refine the same type of argument (see also~\cite[Proof of Proposition 2.2]{FPW:Connection_probabilities_of_multiple_FK_Ising_interfaces}).
Let us consider the $c$-variation of the right-hand side of~\eqref{eqn::general_COV_for_H} with $\varphi = \varphi_{c}$ 
for fixed $\bs{x}\in\chamber_{2N}$.  
In the case of present interest, the main difference and difficulty is caused by the fact that the permutation of the marked points is allowed to be nontrivial.
Since (for fixed $\bs{x}\in \chamber_{2N}$ and $\varphi_c$) both sides of~\eqref{eqn::general_COV_for_H} are analytic in $\kappa$, 
we may assume that $\kappa\in (4,8)$ and extend the conclusion analytic continuation. 
Note that $\varphi'_c(z)=(1+cz)^{-2} > 0$. 
By translation invariance, we may assume that $x_1>0$, so that
\begin{align*}
-\tfrac{1}{x_1} < -\tfrac{1}{x_2} < \cdots < -\tfrac{1}{x_{2N}} < 0 \qquad \textnormal{and} \qquad	\varphi_{-\frac{1}{x_j}}(x_j) = \infty , \quad 1\leq j\leq 2N .
\end{align*}
Fix $j\in \{0,1,2,\ldots,2N\}$, a constant\footnote{We use the convention that $\tfrac{-1}{x_0} := -\infty$ and $\tfrac{-1}{x_{2N+1}} := +\infty$.} 
$c_0 \in (-\tfrac{1}{x_j} , -\tfrac{1}{x_{j+1}})$ and a small $\varepsilon>0$. 
Consider $c \in (c_0 - \varepsilon,c_0 + \varepsilon)$. 
Provided that $\varepsilon$ is small enough, the integration contours $\smash{\acycle_1^{\varphi_{c}(\beta)} \times \cdots \times \acycle_N^{\varphi_{c}(\beta)}}$ stay in the same homotopy class under a variation by $\varphi_{c}$, 
so the $c$-variation of the right-hand side of~\eqref{eqn::general_COV_for_H} with $\varphi = \varphi_{c}$ equals
\begin{align} 
\nonumber
\; & \der c \bigg(
\prod_{i=1}^{2N} \varphi_{c}'(x_{i})^{h(\kappa)} \times \,
\underset{\acycle_1^{\varphi_{c}(\beta)} \times \cdots \times \acycle_N^{\varphi_{c}(\beta)}}{\ointclockwise} f_{\varphi_{c}(\beta)} (\varphi_{c}(\bs{x});\bs{u}) \; \ud u_1 \cdots \ud u_N 
\bigg) \\
\label{eq: spec conf transf}
=\; & - \prod_{i=1}^{2N} \varphi_{c}'(x_{i})^{h(\kappa)} \times 
\, \underset{\acycle_1^{\varphi_{c}(\beta)} \times \cdots \times \acycle_N^{\varphi_{c}(\beta)}}{\ointclockwise}
\sum_{j=1}^{2N} \Big( \varphi_{c}(x_{j})^{2} \partial_j + 2 h(\kappa) \, \varphi_{c}(x_{j}) \Big) f_{\varphi_{c}(\beta)} (\varphi_{c}(\bs{x});\bs{u}) \; \ud u_1 \cdots \ud u_N .
\end{align}
This can be evaluated by observing (via a long calculation combined with Liouville theorem, 
as in~\cite[Lemma~4.14]{Kytola-Peltola:Conformally_covariant_boundary_correlation_functions_with_quantum_group})
that the integrand function $\smash{f_{\varphi_{c}(\beta)}}$ defined in~\eqref{eq: integrand} satisfies the partial differential equation
\begin{align*}
\sum_{j=1}^{2N} \Big( y_{j}^{2}\pder{y_{j}} + 2 h(\kappa) \, y_{j}\Big) f(\bs{y};\bs{u})
=\; & \sum_{r=1}^{N} \pder{u_{r}} 
\big(g(u_{r};\bs{y};\bs{\dot{u}}_r) \; f_{\varphi_{c}(\beta)} (\bs{y};\bs{u}) \big),
\end{align*}
where 
$\bs{\dot{u}}_r = (u_{1},\ldots,u_{r-1},u_{r+1},\ldots,u_{N})$
and 
$g$ is a rational function symmetric in its last $N-1$
variables, and whose only poles are where some of its arguments coincide. 
Applying this to $\bs{y} = \varphi_{c}(\bs{x})$, we obtain
\begin{align*}
\textnormal{\eqref{eq: spec conf transf}} 
= \; & - \prod_{i=1}^{2N} \varphi_{c}'(x_{i})^{h(\kappa)} \times  
\, \underset{\acycle_1^{\varphi_{c}(\beta)} \times \cdots \times \acycle_N^{\varphi_{c}(\beta)}}{\ointclockwise}
\sum_{r=1}^{N} \pder{u_{r}} 
\big(g(u_{r};\varphi_{c}(\bs{x});\bs{\dot{u}}_r) \; 
f_{\varphi_{c}(\beta)} (\varphi_{c}(\bs{x});\bs{u}) \big) 
\; \ud u_1 \cdots \ud u_N ,
\end{align*}
which equals zero because each term in the sum vanishes by integration by parts, as the Pochhammer contours are homologically trivial. 
Therefore, the right-hand side of the asserted formula~\eqref{eqn::general_COV_for_H} with $\varphi = \varphi_{c}$ is constant in $c \in (c_0-\varepsilon,c_0+\varepsilon)$ for any $c_0 \in (-\tfrac{1}{x_j} , -\tfrac{1}{x_{j+1}})$ and small enough $\varepsilon>0$. 
It remains to prove that these constants are all equal; for then by choosing $c_0 = 0$ 
so that $\varphi_{0} = \mathrm{id}$, we see that this constant equals $\smash{\coulomb_{\beta}} (\bs{x})$, 
which is the desired left-hand side of the asserted formula~\eqref{eqn::general_COV_for_H}.
Thus, let us check that 
\begin{align} \label{eqn::general_COV_for_H_aux2}
\lim_{c\nearrow -\frac{1}{x_j} } \; \prod_{i=1}^{2N} | \varphi'_c(x_{i}) |^{h(\kappa)} \times
\, \coulomb_{\varphi_c(\beta)} (\varphi_c(\bs{x}))  
= \lim_{c\searrow -\frac{1}{x_j}} \; \prod_{i=1}^{2N} | \varphi'_c(x_{i}) |^{h(\kappa)} \times
\, \coulomb_{\varphi_c(\beta)} (\varphi_c(\bs{x})) , 
\end{align}
for any $j\in \{1,\ldots,2N\}$. 
Without loss of generality, we consider $j=2N$ and choose $r\in \{1,2\ldots,N\}$ such that $b_r=2N$. 
Note that $\varphi'_c(z)=(1+cz)^{-2}$. 
Write $\realpt_i=\varphi_{-\frac{1}{x_{2N}}}(x_i)$, for $1\leq i\leq 2N-1$. 
Now, we have
\begin{align} \label{eqn::cov_leftlimit_aux1}
&\textnormal{left-hand side of~\eqref{eqn::general_COV_for_H_aux2}}\notag\\
=\; & x_{2N}^{-2h(\kappa)} 
\prod_{i=1}^{2N-1} \Big| 1- \frac{x_i}{x_{2N}} \Big|^{-2h(\kappa)} 
\times 
\lim_{c\nearrow -\frac{1}{x_{2N}}} \; \Big|\frac{x_{2N}}{1+c\,x_{2N}} \Big|^{2h(\kappa)}  \ointclockwise_{\realacycle_1^\beta}\ud u_1\cdots\ointclockwise_{\realacycle_{r-1}^{\beta}}\ud u_{r-1} \notag
\\
\; & \qquad\qquad\qquad\qquad\qquad\qquad\qquad\times
\ointclockwise_{\realacycle_{r+1}^{\beta}}\ud u_{r+1}\cdots \ointclockwise_{\realacycle_{N}^{\beta}}\ud u_{N}\ointclockwise_{\acycle(\varphi_{c}(x_{2N}), \varphi_{c}(x_{a_r}))}\ud u_r \; f_{\varphi_c(\beta)} (\varphi_c(\bs{x});\bs{u})\notag
\\
=\; & x_{2N}^{-2h(\kappa)} \prod_{i=1}^{2N-1} \Big| 1- \frac{x_i}{x_{2N}} \Big|^{-2h(\kappa)}
	\; \ointclockwise_{\realacycle_1^{\beta}}\ud u_1\cdots\ointclockwise_{\realacycle_{r-1}^{\beta}}\ud u_{r-1}\ointclockwise_{\realacycle_{r+1}^{\beta}}\ud u_{r+1}\cdots \ointclockwise_{\realacycle_{N}^{\beta}}\ud u_{N}   \; \mathring{f}_\beta (\realpt_1, \ldots , \realpt_{2N-1};\bs{\dot{u}}_r) \notag\\
\;& \times \lim_{c\nearrow -\frac{1}{x_{2N}}} \; \Big|\frac{x_{2N}}{1+c\,x_{2N}} \Big|^{2h(\kappa)} \ointclockwise_{\acycle(\varphi_{c}(x_{2N}), \varphi_{c}(x_{a_r}))}\ud u_r \;
	\frac{(\varphi_c(x_{a_r})-\varphi_c(x_{2N}))^{2/\kappa}}{(u_r-\varphi_c(x_{2N}))^{4/\kappa}(u_r-\varphi_c(x_{a_r}))^{4/\kappa}} \notag\\
\;& \qquad\qquad\qquad\qquad\qquad 
\times\prod_{1\leq s\leq N\atop s\neq r}\frac{(\varphi_c(x_{a_s})-\varphi_c(x_{2N}))^{2/\kappa}(\varphi_c(x_{b_s})-\varphi_c(x_{2N}))^{2/\kappa}(u_r-u_s)^{8/\kappa}}{(u_s-\varphi_c(x_{2N}))^{4/\kappa}(u_r-\varphi_c(x_{a_s}))^{4/\kappa}(u_r-\varphi_c(x_{b_s}))^{4/\kappa}},
\end{align}
where $\smash{\mathring{f}_\beta}$ 
is the integrand~\eqref{eq: integrand_gen} 
with $d=2N-1$ and $\ell=N-1$, and branch chosen to be real and positive when $\realpt_{a_s} < \Re (u_s) < \realpt_{a_s+1}$ for all $s \neq r$, and where the Pochhammer contours 
$\smash{\realacycle_1^{\beta},\ldots,\realacycle_{r-1}^{\beta},\realacycle_{r+1}^{\beta},\ldots,\realacycle_N^{\beta}}$ are all disjoint, and such that each $\smash{\realacycle_s^{\beta}}$ 
surrounds the points $\realpt_{a_s},\realpt_{b_s}$ and does not surround any other points among $\{\realpt_1,\ldots,\realpt_{r-1},\realpt_{r+1},\ldots,\realpt_{2N-1}\}$. 
Next, we claim that for each fixed $(x_1,\ldots,x_{2N-1}) \in \chamber_{2N-1}$ and $\bs{\dot{u}}_r\in \smash{\realacycle_1^{\beta}}\times\cdots \times  \smash{\realacycle_{r-1}^{\beta}}\times \smash{\realacycle_{r+1}^{\beta}}\times \cdots \times  \smash{\realacycle_N^{\beta}}$, we have  
	\begin{align} \label{eqn::cov_leftlimit_aux2}
		&\lim_{c\nearrow -\frac{1}{x_{2N}}} \; \Big|\frac{x_{2N}}{1+c\,x_{2N}} \Big|^{2h(\kappa)} \ointclockwise_{\acycle(\varphi_{c}(x_{2N}), \varphi_{c}(x_{a_r}))}\ud u_r \;
		\frac{(\varphi_c(x_{a_r})-\varphi_c(x_{2N}))^{2/\kappa}}{(u_r-\varphi_c(x_{2N}))^{4/\kappa}(u_r-\varphi_c(x_{a_r}))^{4/\kappa}} \notag\\
		\;& \qquad\qquad\times\prod_{1\leq s\leq N\atop s\neq r}\frac{(\varphi_c(x_{a_s})-\varphi_c(x_{2N}))^{2/\kappa}(\varphi_c(x_{b_s})-\varphi_c(x_{2N}))^{2/\kappa}(u_r-u_s)^{8/\kappa}}{(u_s-\varphi_c(x_{2N}))^{4/\kappa}(u_r-\varphi_c(x_{a_s}))^{4/\kappa}(u_r-\varphi_c(x_{b_s}))^{4/\kappa}}\times \mathring{f}_\beta (\realpt_1, \ldots , \realpt_{2N-1};\bs{\dot{u}}_r)\notag\\
		&=\bigg(\ointclockwise_{\acycle(0,1)} v^{-4/\kappa}(v-1)^{-4/\kappa} \, \ud v \bigg) \times  \mathring{f}_\beta (\realpt_1, \ldots , \realpt_{2N-1};\bs{\dot{u}}_r),
	\end{align}
	where the branch of $v\mapsto v^{-4/\kappa}(v-1)^{-4/\kappa}$ is chosen to be real and positive when $0<\Re (v)<1$.
Given this identity~\eqref{eqn::cov_leftlimit_aux2}, plugging it into~\eqref{eqn::cov_leftlimit_aux1} shows that the left-hand side of~\eqref{eqn::general_COV_for_H_aux2} equals
\begin{align}\label{eqn::cov_leftlimit_aux3}
	\; & x_{2N}^{-2h(\kappa)} 
	\prod_{i=1}^{2N-1} \Big| 1- \frac{x_i}{x_{2N}} \Big|^{-2h(\kappa)} 
	\; \ointclockwise_{\acycle(0,1)} v^{-4/\kappa}(v-1)^{-4/\kappa} \, \ud v 
	\notag\\
	\; & \qquad\qquad\times 
	\ointclockwise_{\realacycle_1^{\beta}}\ud u_1\cdots\ointclockwise_{\realacycle_{r-1}^{\beta}}\ud u_{r-1}\ointclockwise_{\realacycle_{r+1}^{\beta}}\ud u_{r+1}\cdots \ointclockwise_{\realacycle_{N}^{\beta}}\ud u_{N}   \; \mathring{f}_\beta (\realpt_1, \ldots , \realpt_{2N-1};\bs{\dot{u}}_r).
\end{align}
One can similarly check that the right-hand side of~\eqref{eqn::general_COV_for_H_aux2} equals~\eqref{eqn::cov_leftlimit_aux3}
as well, which proves~\eqref{eqn::general_COV_for_H_aux2}.

\smallbreak

It now remains to prove the claimed identity~\eqref{eqn::cov_leftlimit_aux2}. 
To this end, fix $(x_1,\ldots,x_{2N-1}) \in \chamber_{2N-1}$ and $\bs{\dot{u}}_r\in \smash{\realacycle_1^{\beta}}\times\cdots \times  \smash{\realacycle_{r-1}^{\beta}}\times \smash{\realacycle_{r+1}^{\beta}}\times \cdots \times  \smash{\realacycle_N^{\beta}}$. 
Consider the contour $\acycle^c(0,1):=\psi_c(\acycle(\varphi_c(x_{2N}),\varphi_c(x_{a_r})))$, where $\psi_c(z):=\tfrac{z-\varphi_c(x_{2N})}{\varphi_c(x_{a_r})-\varphi_c(x_{2N})}$.
A change of variables $v:=\psi_c(u_r)$ shows that the left-hand side of~\eqref{eqn::cov_leftlimit_aux2} equals 
\begin{align}
\nonumber
	&\lim_{c\nearrow -\frac{1}{x_{2N}}} \; \bigg(\ointclockwise_{\acycle^c(0,1)} v^{-4/\kappa}(v-1)^{-4/\kappa} \, \ud v \bigg)
\times \; \mathring{f}_\beta (\realpt_1, \ldots , \realpt_{2N-1};\bs{\dot{u}}_r)\\
\nonumber
	&\qquad\quad \times \underbrace{\prod_{1\leq s\leq N\atop s\neq r}\tfrac{(\varphi_c(x_{a_s})-\varphi_c(x_{2N}))^{2/\kappa}(\varphi_c(x_{b_s})-\varphi_c(x_{2N}))^{2/\kappa}((\varphi_c(x_{a_r})-\varphi_c(x_{2N}))v+\varphi_c(x_{2N})-u_s)^{8/\kappa}}{(u_s-\varphi_c(x_{2N}))^{4/\kappa}((\varphi_c(x_{a_r})-\varphi_c(x_{2N}))v+\varphi_c(x_{2N})-\varphi_c(x_{a_s}))^{4/\kappa}((\varphi_c(x_{a_r})-\varphi_c(x_{2N}))v+\varphi_c(x_{2N})-\varphi_c(x_{b_s}))^{4/\kappa}}}_{\; =: \; T_c(\bs{x};\bs{\dot{u}}_r;v)}\\
	\label{eqn::cov_leftlimit_aux5}
	&=4\pi \sin^2(4\pi/\kappa)\; \mathring{f}_\beta (\realpt_1, \ldots , \realpt_{2N-1};\bs{\dot{u}}_r)\;\times\lim_{c\nearrow -\frac{1}{x_{2N}}} \;\int_{0}^1 \ud v \, |v(1-v) |^{-4/\kappa} \times T_c(\bs{x};\bs{\dot{u}}_r;v),
\end{align}
where the branch of the term $T_c(\bs{x};\bs{\dot{u}}_r;v)$ is chosen to be real and positive when 
$0<\Re (v)< \smash{\tfrac{\realpt_{1}-\varphi_c(x_{2N})}{\varphi_c(x_{a_r})-\varphi_c(x_{2N})}}$ 
and $\realpt_{a_s}<\Re(u_s)<\realpt_{b_s}$ for all $s\neq r$. 
We claim that one can exchange the limit and integral in~\eqref{eqn::cov_leftlimit_aux5}:
\begin{align} \label{eqn::cov_leftlimit_aux6}
\lim_{c\nearrow -\frac{1}{x_{2N}}} \;\int_{0}^1 \ud v \, | v(1-v) |^{-4/\kappa} \times T_c(\bs{x};\bs{\dot{u}}_r;v)=\;\int_{0}^1 \ud v \, | v(1-v)|^{-4/\kappa} \times \lim_{c\nearrow -\frac{1}{x_{2N}}} \; T_c(\bs{x};\bs{\dot{u}}_r;v). 
\end{align} 
Assume that~\eqref{eqn::cov_leftlimit_aux6} holds. Since $\varphi_c(x_{2N})\to -\infty$ as $c\nearrow -\frac{1}{x_{2N}}$, we obtain 
\begin{align*}
\textnormal{\eqref{eqn::cov_leftlimit_aux5}} \; = \; &
4\pi \sin^2(4\pi/\kappa) \;\mathring{f}_\beta (\realpt_1, \ldots , \realpt_{2N-1};\bs{\dot{u}}_r)\times \bigg(\int_{0}^1 \ud v \, | v(1-v) |^{-4/\kappa}\bigg)\\
= \; &\mathring{f}_\beta (\realpt_1, \ldots , \realpt_{2N-1};\bs{\dot{u}}_r)\;\times \bigg(\ointclockwise_{\acycle(0,1)} v^{-4/\kappa}(v-1)^{-4/\kappa} \, \ud v \bigg),
\end{align*}
which gives the claimed~\eqref{eqn::cov_leftlimit_aux2}. By analytic continuation, this concludes the proof of~\eqref{eqn::cov_leftlimit_aux2}.

To finalize the proof, we prove the claim~\eqref{eqn::cov_leftlimit_aux6}. 
Consider a small $\epsilon>0$. 
On the one hand, note that there exists a constant $C_{\epsilon}\in (0,\infty)$ such that
\begin{align*}
|T_c(\bs{x};\bs{\dot{u}}_r;v)|\leq C_{\epsilon},\quad \textnormal{for all }v\in [0,1-\epsilon] \textnormal{ and }c<-\tfrac{1}{x_{2N}} \textnormal{ sufficiently close to }-\tfrac{1}{x_{2N}}.
\end{align*}
Therefore, using the dominated convergence theorem, we have
\begin{align} \label{eqn::cov_leftlimit_aux6_1}
\lim_{c\nearrow -\frac{1}{x_{2N}}} \;\int_{0}^{1-\epsilon} \ud v | v(1-v) |^{-4/\kappa} \times T_c(\bs{x};\bs{\dot{u}}_r;v)=\;\int_{0}^{1-\epsilon} \ud v \, | v(1-v) |^{-4/\kappa} \times \lim_{c\nearrow -\frac{1}{x_{2N}}} \; T_c(\bs{x};\bs{\dot{u}}_r;v). 
\end{align}
On the other hand, using the change of variables $t=(1-v)(\varphi_c(x_{a_r})-\varphi_c(x_{2N}))$, we have 
\begin{align*}
\Big| v^{-4/\kappa}(1-v)^{-4/\kappa}\times \;T_c(\bs{x};\bs{\dot{u}}_r;v)\Big| 
\leq &\;\big|\varphi_c(x_{a_r})-\varphi_c(x_{2N})\big|^{4/\kappa}t^{-4/\kappa}\\
&\;\quad\times\Big|\prod_{1\leq s\leq N\atop s\neq r}\tfrac{(\varphi_c(x_{a_s})-\varphi_c(x_{2N}))^{2/\kappa}(\varphi_c(x_{b_s})-\varphi_c(x_{2N}))^{2/\kappa}(-t+\varphi_c(x_{a_r})-u_s)^{8/\kappa}}{(u_s-\varphi_c(x_{2N}))^{4/\kappa}(-t+\varphi_c(x_{a_r})-\varphi_c(x_{a_s}))^{4/\kappa}(-t+\varphi_c(x_{a_r})-\varphi_c(x_{b_s}))^{4/\kappa}}\Big|\\
\leq &\; C \big|\varphi_c(x_{a_r})-\varphi_c(x_{2N})\big|^{4/\kappa}t^{-4/\kappa}\\
&\;\quad\times \Big|\prod_{1\leq s\leq N\atop s\neq r}\tfrac{(-t+\varphi_c(x_{a_r})-u_s)^{8/\kappa}}{(-t+\varphi_c(x_{a_r})-\varphi_c(x_{a_s}))^{4/\kappa}(-t+\varphi_c(x_{a_r})-\varphi_c(x_{b_s}))^{4/\kappa}}\Big|,
\end{align*}
for all $v\in [1-\epsilon,1]$ and corresponding $t\in [0,\epsilon(\varphi_{c}(x_{a_r})-\varphi_{c}(x_{2N}))]$, 
where $C\in (0,\infty)$ is some constant that is independent of $c<-\tfrac{1}{x_{2N}}$ and $t>0$. Therefore (recall that we have required that $\kappa \in (4,8)$ here),
\begin{align*}
\Big|\int_{1-\epsilon}^1 \ud v \, | v(1-v) |^{-4/\kappa} \times T_c(\bs{x};\bs{\dot{u}}_r;v)\Big| 
\leq &\;\hat{C} \, \big|\varphi_c(x_{a_r})-\varphi_c(x_{2N})\big|^{4/\kappa-1}
\Big(1+ \epsilon^{1-4/\kappa}(\varphi_c(x_{a_r})-\varphi_c(x_{2N}))^{1-4/\kappa}\Big)\\
\leq &\; \hat{C} \, \Big(\big|\varphi_c(x_{a_r})-\varphi_c(x_{2N})\big|^{4/\kappa-1}+\epsilon^{1-4/\kappa}\Big),
\end{align*}
where $\hat{C}\in (0,\infty)$ is a constant that is independent of $c<-\tfrac{1}{x_{2N}}$. 
This implies that 
\begin{align}\label{eqn::cov_leftlimit_aux6_2}
\limsup_{c\nearrow -\frac{1}{x_{2N}}}\; \Big|\int_{1-\epsilon}^1 \ud v \, | v(1-v) |^{-4/\kappa} \times T_c(\bs{x};\bs{\dot{u}}_r;v)\Big|\leq \hat{C} \epsilon^{1-4/\kappa}.
\end{align}
Combining~\eqref{eqn::cov_leftlimit_aux6_1} with~\eqref{eqn::cov_leftlimit_aux6_2} gives~\eqref{eqn::cov_leftlimit_aux6}, as we set out to prove.
\end{proof}

\begin{corollary} \label{cor: full Mobius covariance F}
Fix $\beta \in \LP_N$. The function $\smash{\coulombGas_{\beta}}$ satisfies the  M\"obius covariance 
\begin{align*}
\coulombGas_{\beta} (\HH;\bs{x}) 
\; = \; \prod_{i=1}^{2N} \varphi'(x_{i})^{h(\kappa)} \times
\, \coulombGas_{\varphi(\beta)} (\HH;\varphi(\bs{x})) 
\; = \; \prod_{i=1}^{2N} |\varphi'(x_{i})|^{h(\kappa)} \times
\, \coulombGas_{\varphi(\beta)} (\HH;\varphi(\bs{x})) , 
\qquad \bs{x} \in \chamber_{2N} ,
\end{align*}
for all M\"obius maps $\varphi \in \mathrm{PSL}(2,\R)$ preserving the upper half-plane, 
where $\varphi'(x_{i}) > 0$ for all $1 \leq i \leq 2N$. 
\end{corollary}

\begin{proof}
For $\kappa\in (0,8) \setminus \big\{ \tfrac{8}{m} \colon m \in \bZpos \big\}$, 
this is immediate from Proposition~\ref{prop: full Mobius covariance}, since $\smash{\coulombGas_{\beta}}$ is a constant multiple of 
$\smash{\coulomb_\beta}$.
For exceptional $\kappa\in \big\{ \tfrac{8}{m} \colon m \in \bZpos \big\}$, 
the claim then follows by continuity in $\kappa \in \C \setminus \{0\}$.
\end{proof}

\begin{remark} \label{rem: definition F polygon}
We extend the definition of $\coulombGas_{\beta}(\Omega; \bs{x})$ to more general (nice\footnote{By a \emph{nice polygon}, 
we refer to a polygon such that the marked boundary points $x_1, \ldots, x_{2N}$ lie on sufficiently regular boundary segments
for the derivatives of $\varphi$ to be defined there, e.g. $C^{1+\eps}$ for some $\eps>0$.}) polygons $(\Omega; \bs{x})$ 
by sending $\Omega$ onto $\HH$ via 
any conformal map $\varphi$ with $\varphi(x_1)<\cdots<\varphi(x_{2N})$,
using the conformal covariance
\begin{align}\label{eqn::def_polygon_CGI}
\coulombGas_{\beta} (\Omega; \bs{x}) := 
\prod_{j=1}^{2N}|\varphi'(x_j)|^{h(\kappa)} \times \, \coulombGas_{\beta} (\Omega; \varphi(\bs{x})) 
\end{align}
Thanks to Corollary~\ref{cor: full Mobius covariance F}, this definition is independent of the choice of $\varphi$. 
Proposition~\ref{prop: full Mobius covariance} is also consistent: 
\begin{align*}
\coulombGas_{\beta} (\HH;\bs{x}) = 
\prod_{i=1}^{2N} |\varphi'(x_{i})|^{h(\kappa)} \times
\, \coulombGas_{\beta} (\HH^*;\tilde{\varphi}(\bs{x})) , 
\qquad \bs{x} \in \chamber_{2N} ,
\end{align*}
for all M\"obius maps $\varphi \colon \HH \to \HH^* := \{z \in \C \colon \Im(z) < 0\}$, 
where 
$\tilde{\varphi}(\bs{x}) := (\varphi(x_{1}) , \ldots , \varphi(x_{2N}))$. 
\end{remark}

\subsection{Integral removals}
\label{subsec:Integral removal}
We next verify that our Coulomb gas integral formula~\eqref{eqn::CGI_def} is consistent with the earlier works of 
Julien Dub\'edat~\cite{Dubedat:Euler_integrals_for_commuting_SLEs},
and Steven Flores \&~Peter Kleban~\cite{Flores-Kleban:Solution_space_for_system_of_null-state_PDE3}, 
involving one less integration, 
in the spirit of the physics literature where the charge neutrality assumption is imperative. 
Our original formulas have the virtue that they do not (need to) satisfy any charge neutrality, and they are completely symmetric in the marked points.
However, the alternative formulas in Proposition~\ref{prop::remove_integration} are very useful in computations.

\begin{proposition} \label{prop::remove_integration}
Fix $\beta \in \LP_N$ and $r \in \{1,2,\ldots,N\}$.
For $\kappa\in (0,8) \setminus \big\{ \tfrac{8}{m} \colon m \in \bZpos \big\}$, we have
\begin{align} 
\label{eq: Dub consistency general b}
\coulomb_\beta(\bs{x} )
= \; & \frac{\fugacity(\kappa) }{\cst(\kappa)} \,
\ointclockwise_{\acycle^\beta_1} \ud u_1 
\cdots \ointclockwise_{\acycle^\beta_{r-1}} \ud u_{r-1} 
\ointclockwise_{\acycle^\beta_{r+1}} \ud u_{r+1} 
\cdots \ointclockwise_{\acycle^\beta_{N}} \ud u_{N} 
\; \hat{f}_\beta^b (\bs{x};\bs{\dot{u}}_r) \\
\label{eq: Dub consistency general a}
= \; & \frac{\fugacity(\kappa) }{\cst(\kappa)} \,
\ointclockwise_{\acycle^\beta_1} \ud u_1 
\cdots \ointclockwise_{\acycle^\beta_{r-1}} \ud u_{r-1} 
\ointclockwise_{\acycle^\beta_{r+1}} \ud u_{r+1} 
\cdots \ointclockwise_{\acycle^\beta_{N}} \ud u_{N} 
\; \hat{f}_\beta^a (\bs{x};\bs{\dot{u}}_r) , \qquad \bs{x} \in \chamber_{2N} ,
\end{align}
where $\bs{\dot{u}}_r = (u_1, \ldots, u_{r-1},u_{r+1},\ldots,u_{N})$ and 
the branches of the multivalued integrands 
\begin{align} \label{eq: integrand with one less screening variable b}
\begin{split}
\hat{f}_\beta^b (\bs{x};\bs{\dot{u}}_r)
:= \; &
\prod_{\substack{1\leq i<j\leq 2N \\i,j \neq b_r}}(x_{j}-x_{i})^{2/\kappa} 
\prod_{1\leq i<b_r}(x_{b_r}-x_{i})^{-2h(\kappa)} 
\prod_{b_r<j\leq 2N}(x_{j}-x_{b_r})^{-2h(\kappa)} \\
\; & \times \; 
\prod_{\substack{1\leq t<s\leq N \\ t,s \neq r}}(u_{s}-u_{t})^{8/\kappa} 
\prod_{\substack{1\leq i\leq 2N  \\i \neq b_r\\ 1\leq s\leq N  \\ s \neq r}}
(u_{s}-x_{i})^{-4/\kappa} 
\prod_{\substack{1\leq s\leq N  \\ s \neq r}}
(u_{s}-x_{b_r})^{4h(\kappa)} ,
\end{split}
\end{align}  
which carries the conjugate charge at $x_{b_r}$, 
and
\begin{align} \label{eq: integrand with one less screening variable a}
\begin{split}
\hat{f}_\beta^a (\bs{x};\bs{\dot{u}}_r)
:= \; &
\prod_{\substack{1\leq i<j\leq 2N \\i,j \neq a_r}}(x_{j}-x_{i})^{2/\kappa} 
\prod_{1\leq i<a_r}(x_{a_r}-x_{i})^{-2h(\kappa)} 
\prod_{a_r<j\leq 2N}(x_{j}-x_{a_r})^{-2h(\kappa)} \\
\; & \times \; 
\prod_{\substack{1\leq t<s\leq N \\ t,s \neq r}}(u_{s}-u_{t})^{8/\kappa} 
\prod_{\substack{1\leq i\leq 2N  \\i \neq a_r\\ 1\leq s\leq N  \\ s \neq r}}
(u_{s}-x_{i})^{-4/\kappa} 
\prod_{\substack{1\leq s\leq N  \\ s \neq r}}
(u_{s}-x_{a_r})^{4h(\kappa)} ,
\end{split}
\end{align}  
which carries the conjugate charge at $x_{a_r}$, 
are chosen to be real and positive when
$x_{a_s} < \Re(u_s) < x_{a_s+1}$ for all $s \in \{1,2,\ldots,N\} \setminus \{r\}$.
\textnormal{(}Here, $h(\kappa) = \frac{6 - \kappa}{2 \kappa}$.\textnormal{)}
\end{proposition}

Note that the omitted integral is associated to a variable carrying the conjugate charge.
This is important in order to preserve the desired asymptotics properties (Proposition~\ref{prop::H_ASY}).

\begin{proof}
It suffices to show~\eqref{eq: Dub consistency general b} when $b_r=2N$ and~\eqref{eq: Dub consistency general a} when $a_r=1$, because for the other cases the claim then follows after conjugating by a suitable M\"obius map and using Proposition~\ref{prop: full Mobius covariance}.

Fix $\bs{x}=(x_1,\ldots,x_{2N})\in\chamber_{2N}$. 
Suppose $b_r=2N$ (in this case, $a_r = 2r-1$) and 
denote  by $\smash{\auxcoulomb_\beta}$ the right-hand side of~\eqref{eq: Dub consistency general b}. 
Note that if $\smash{\coulomb_{\beta}(\bs{x}) = \auxcoulomb_\beta(\bs{x} ) = 0}$, 
then~\eqref{eq: Dub consistency general b} is clear, 
and otherwise we may assume without loss of generality that $\smash{\coulomb_{\beta}(\bs{x})\neq 0}$, 
and consider the ratio $
\smash{\auxcoulomb_\beta(\bs{x} )/\coulomb_{\beta}(\bs{x})}$
of the right and left-hand sides. 
For large $R >0$, let $\varphi_R$ be the M\"obius map of the upper half-plane $\HH$ such that
\begin{align*}
\begin{cases}
	-\infty<\varphi_R(x_1)<\varphi_R(x_2)<\cdots<\varphi_R(x_{2N})<\infty,\\
	\varphi_R(x_1)=x_1,\quad \varphi_R(x_{2N-1})=x_{2N-1},\quad \varphi_R(x_{2N})=R.
\end{cases}
\end{align*} 
Denote $\varphi_R(\bs{x}):=(\varphi_R(x_1),\ldots,\varphi_R(x_{2N}))$ and 
$\realpt_j := \underset{R \to \infty}{\lim} \varphi_R(x_j)$ for $1\le j\le 2N-1$. Note that 
\begin{align*}
-\infty<x_1=\realpt_1<\realpt_2<\cdots<\realpt_{2N-1}=x_{2N-1}<R<\infty. 
\end{align*}  
Because both $\smash{\coulomb_{\beta}}$ and $\smash{\auxcoulomb_\beta}$
are M\"obius covariant (respectively by Proposition~\ref{prop: full Mobius covariance} and~\cite[Lemma~2]{Dubedat:Euler_integrals_for_commuting_SLEs}), we may compute their ratio in the form
\begin{align} \label{eqn::reduction_aux0}
\; & \frac{\auxcoulomb_{\beta}(\bs{x})}{\coulomb_{\beta}(\bs{x})} 
\; = \; \lim_{R\to\infty} \frac{\auxcoulomb_{\beta}(\varphi_R(\bs{x}))}{\coulomb_{\beta}\big(\varphi_R(\bs{x})\big)} 
\; \\
= \; & \frac{\fugacity(\kappa)}{\cst(\kappa)} 
\lim_{R\to \infty} \frac{R^{2h(\kappa)}  \; 
\ointclockwise_{\realacycle_1^\beta}\ud u_1\cdots\ointclockwise_{\realacycle_{r-1}^{\beta}}\ud u_{r-1}\ointclockwise_{\realacycle_{r+1}^{\beta}}\ud u_{r+1}\cdots \ointclockwise_{\realacycle_{N}^{\beta}}\ud u_{N} 
\; \hat{f}_\beta^b (\varphi_R(\bs{x});\bs{\dot{u}}_r)}{R^{2h(\kappa)} \; \ointclockwise_{\realacycle_1^\beta}\ud u_1\cdots\ointclockwise_{\realacycle_{r-1}^{\beta}}\ud u_{r-1}\ointclockwise_{\realacycle_{r+1}^{\beta}}\ud u_{r+1}\cdots \ointclockwise_{\realacycle_{N}^{\beta}}\ud u_{N}\ointclockwise_{\acycle(\varphi_R(x_{a_r}),R)}\ud u_r \; f_\beta (\varphi_R(\bs{x});\bs{u})},  \notag
\end{align} 
where the Pochhammer contours 
$\smash{\realacycle_1^{\beta},\ldots,\realacycle_{r-1}^{\beta},\realacycle_{r+1}^{\beta},\ldots,\realacycle_N^{\beta}}$ are all disjoint, and each $\smash{\realacycle_s^{\beta}}$ 
surrounds the points $\realpt_{a_s},\realpt_{b_s}$ and does not surround any other points among $\{\realpt_1,\ldots,\realpt_{r-1},\realpt_{r+1},\ldots,\realpt_{2N-1}\}$.
Now, after collecting the powers
$0 = 2h(\kappa)-2h(\kappa)+4(N-1)h(\kappa)-4(N-1)h(\kappa)$ of the blow-up parameter $R$, 
the limit of the numerator equals
\begin{align} \label{eqn::reduction_aux1}
\ointclockwise_{\realacycle_1^{\beta}}\ud u_1\cdots\ointclockwise_{\realacycle_{r-1}^{\beta}}\ud u_{r-1}\ointclockwise_{\realacycle_{r+1}^{\beta}}\ud u_{r+1}\cdots \ointclockwise_{\realacycle_{N}^{\beta}}\ud u_{N}   \; \mathring{f}_\beta (\realpt_1, \ldots , \realpt_{2N-1};\bs{\dot{u}}_r) ,
\end{align}
where $\smash{\mathring{f}_\beta}$ 
is the integrand function~\eqref{eq: integrand_gen} 
with $d=2N-1$ and $\ell=N-1$,
and branch chosen to be real and positive when $\realpt_{a_s} < \Re (u_s) < x_{a_s+1}$ for all $s \neq r$.
For the denominator, using the change of variables $v := \frac{u_r - \varphi_R(x_{a_r})}{R - \varphi_R(x_{a_r})}$ and Eq.~\eqref{eqn::beta_1} from Lemma~\ref{lem::beta_acycle} to evaluate the $r$:th integral as
\begin{align} \label{eqn::Poch_beta_again}
\ointclockwise_{\acycle(0,1)} v^{-4/\kappa}(v-1)^{-4/\kappa} \, \ud v = \frac{\fugacity(\kappa)}{\cst(\kappa)} , 
\qquad \textnormal{for } 
\kappa\in (0,8) \setminus \big\{ \tfrac{8}{m} \colon m \in \bZpos \big\},
\end{align} 
we obtain 
\begin{align} \label{eqn::reduction_aux2}
\; & \lim_{R \to \infty} R^{2h(\kappa)} 
\; \ointclockwise_{\realacycle_1^\beta}\ud u_1\cdots\ointclockwise_{\realacycle_{r-1}^{\beta}}\ud u_{r-1}\ointclockwise_{\realacycle_{r+1}^{\beta}}\ud u_{r+1}\cdots \ointclockwise_{\realacycle_{N}^{\beta}}\ud u_{N}\ointclockwise_{\acycle(\varphi_R(x_{a_r}),R)}\ud u_r 
\; f_\beta (\varphi_R(\bs{x});\bs{u}) 
\notag 
\\
= \; & \bigg( \ointclockwise_{\acycle(0,1)} v^{-4/\kappa}(v-1)^{-4/\kappa} \ud v \bigg)
\; \ointclockwise_{\realacycle_1^{\beta}}\ud u_1\cdots\ointclockwise_{\realacycle_{r-1}^{\beta}}\ud u_{r-1}\ointclockwise_{\realacycle_{r+1}^{\beta}}\ud u_{r+1}\cdots \ointclockwise_{\realacycle_N^{\beta}}\ud u_{N} \; \mathring{f}_\beta (\realpt_1, \ldots , \realpt_{2N-1};\bs{\dot{u}}_r) 
\notag 
\\
=\; & \frac{\fugacity(\kappa) }{\cst(\kappa)} \, 
\ointclockwise_{\realacycle_1^{\beta}}\ud u_1\cdots\ointclockwise_{\realacycle_{r-1}^{\beta}}\ud u_{r-1}\ointclockwise_{\realacycle_{r+1}^{\beta}}\ud u_{r+1}\cdots \ointclockwise_{\realacycle_N^{\beta}}\ud u_{N}  \; f_\beta (\realpt_1, \ldots , \realpt_{2N-1};\bs{\dot{u}}_r) .
\end{align}
Plugging~(\ref{eqn::reduction_aux1},~\ref{eqn::reduction_aux2}) into~\eqref{eqn::reduction_aux0}, we obtain \eqref{eq: Dub consistency general b}:
\begin{align*}
\frac{\auxcoulomb_{\beta}(\bs{x})}{\coulomb_{\beta}(\bs{x})} 
\; = \;  \lim_{R \to \infty} \frac{\auxcoulomb_{\beta}(\varphi_R(\bs{x}))}{\coulomb_{\beta}\big(\varphi_R(\bs{x})\big)} \; = \;  1 .
\end{align*}
Similarly, the second asserted identity~\eqref{eq: Dub consistency general a} when $a_r=1$ can be proven 
by sending $x_1$ to $-\infty$.
\end{proof}

\subsection{Cyclic permutation symmetry}
\label{subsec::rotation_sym}

We next record a simple property of $\smash{\coulomb_\beta}$ when its variables are cyclically permuted within $\mathfrak{Y}_{2N}$.
To state it, we denote by 
$\sigma = \bigl(\begin{smallmatrix}
  1 & 2 & 3 & \cdots & 2N-1 & 2N \\
  2 & 3 & 4 & \cdots &  2N  & 1
\end{smallmatrix}\bigr)$
the cyclic counterclockwise permutation of the indices, and 
for each $\beta \in \LP_N$, we denote by $\sigma(\beta) \in \LP_N$ the link pattern obtained from $\beta$ via permuting the indices by $\sigma$ and then ordering the link endpoints appropriately. For example, with $N=2$ we have
\begin{align*}
\beta = \; &
\vcenter{\hbox{\includegraphics[scale=0.3]{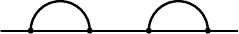}}} = \{\{1, 2\}, \{3, 4\} \} 
\qquad \longmapsto \qquad
\sigma(\beta) = \vcenter{\hbox{\includegraphics[scale=0.3]{figures-arXiv/link-2.pdf}}} = \{\{1, 4\}, \{2, 3\} \}
\\
\beta = \; &
\vcenter{\hbox{\includegraphics[scale=0.3]{figures-arXiv/link-2.pdf}}} = \{\{1, 4\}, \{2, 3\} \}
\qquad \longmapsto \qquad
\sigma(\beta) = \vcenter{\hbox{\includegraphics[scale=0.3]{figures-arXiv/link-1.pdf}}} = \{\{1, 2\}, \{3, 4\} \} .
\end{align*}
Note that this cannot be realized as an action by a M\"obius transformation, unless there are only two (or three) points.
Indeed, it is not possible to fix the images of more than three points under a M\"obius transformation.
Rather, the permutation should be thought of as an action of a \emph{mapping class group} (of the upper half-plane with marked points on the boundary, say) or of a suitable subgroup\footnote{Cyclic permutations form a subgroup of the symmetric group and the latter is a quotient of the braid group.} 
of the \emph{braid group} (that is, the fundamental group of the punctured Riemann sphere). 
Also, this cannot be realized as an action by the group of cyclic permutations, unless $q=\exp(4 \pi \ii / \kappa)$ is a third root of unity,
because according to Proposition~\ref{prop: rotation symmetry}, 
permuting $2N$ times results in a multiplicative
factor $(- \exp(-6 \pi \ii / \kappa))^{2N}$.

\begin{proposition} \label{prop: rotation symmetry}
Fix $\beta \in \LP_N$. 
Let $\bs{\rho} \colon [0,1] \to \mathfrak{Y}_{2N}$ 
be a path from 
$\bs{\rho}(0) = \boldsymbol{x} = (x_{1},x_{2},\ldots,x_{2N-1},x_{2N}) \in \chamber_{2N}$
to $\bs{\rho}(1) := (x_{2},x_{3},\ldots,x_{2N-1},x_{2N},x_{1}) \in \mathfrak{Y}_{2N}$ 
such that $\bs{\rho} = (\rho_1, \rho_2, \ldots, \rho_{2N})$ satisfy
\begin{align*}
\begin{cases}
\Im(\rho_j(t)) = 0 , & \textnormal{for all } j \in \{1,2,\ldots,2N-1\} , \\
\Im(\rho_j(t)) \geq 0 , & \textnormal{for } j = 2N ,
\end{cases}
\qquad \textnormal{for all } t \in [0,1] 
\end{align*}
see also Figure~\ref{fig::braid}.
Then, we have
\begin{align} \label{eq: rotation symmetry}
\coulomb_{\sigma(\beta)} (\bs{\rho}(1))
\; = \; - e^{-\frac{6 \pi \ii}{\kappa}} \,
\coulomb_{\beta} (\bs{\rho}(0)) 
\; = \; - q^{3/2} \,
\coulomb_{\beta} (\bs{\rho}(0)) , \qquad q=\exp(4 \pi \ii / \kappa) .
\end{align}
\end{proposition}

\begin{figure}[ht!]
\includegraphics[width=0.5\textwidth]{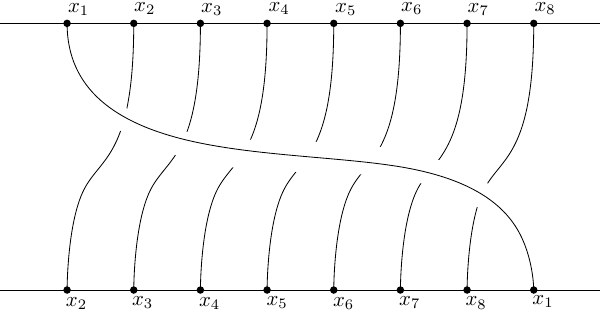}
\caption{\label{fig::braid}
Illustration of the braid transformation $\bs{\rho}$ in Proposition~\ref{prop: rotation symmetry}. }
\end{figure}

\begin{proof}
Note that the path $\bs{\rho}$ transforms the integration contours associated to $\beta$ in $\smash{\coulomb_{\beta}}$ into the integration contours associated to $\sigma(\beta)$ in $\smash{\coulomb_{\beta}}$.
Hence, it is clear that~\eqref{eq: rotation symmetry} holds up to some multiplicative phase factor. To find it, we simplify $\smash{\coulomb_{\beta}}$ using Proposition~\ref{prop::remove_integration}: 
choosing $r$ in~\eqref{eq: Dub consistency general b} such that $b_r=2N$ (in this case, $a_r = 2r-1$), 
we see that $\bs{\rho}$ gives rise to the following phase factors:
\begin{itemize}[leftmargin=2em]
\item the factor $\underset{1\leq i\leq 2N-1}{\prod} (x_{2N}-x_{i})^{1-6/\kappa} $ in~\eqref{eq: integrand with one less screening variable b} gets a phase 
$(-1)^{2N-1} \, \exp( -\frac{6 \pi \ii }{\kappa} (2N-1))$; 

\smallbreak

\item and the factor $\underset{\substack{1\leq s\leq N  \\ s \neq r}}{\prod}
(u_{s}-x_{2N})^{12/\kappa-2} $ in~\eqref{eq: integrand with one less screening variable b} gets a phase $\exp( \frac{12 \pi \ii }{\kappa} (N-1))$. 
\end{itemize}
The asserted equality~\eqref{eq: rotation symmetry} follows by collecting the overall phase factor.
\end{proof}

\begin{remark} \label{rem:quantum_group}
Let $\mathsf{M}_d$ denote the $d$-dimensional type-one simple module of the Hopf algebra $U_q(\mathfrak{sl}(2,\C))$
(see, e.g.,~\cite{Kassel:Quantum_groups} for definitions).   
A space of $\SLE_\kappa$ partition functions and Coulomb gas integral solutions to the second order BPZ PDEs~\eqref{eqn::PDE} was shown in~\cite{Felder-Wieczerkowski:Topological_representations_of_quantum_group, 
Kytola-Peltola:Pure_partition_functions_of_multiple_SLEs, Kytola-Peltola:Conformally_covariant_boundary_correlation_functions_with_quantum_group} 
in the generic case to coincide with the set of highest-weight vectors in 
the tensor product module $\smash{\mathsf{M}_2^{\otimes 2N}}$. 
The solutions satisfying the full M\"obius covariance belong to the trivial 
submodule $( \tfrac{1}{N+1} \binom{2N}{N}) \; \mathsf{M}_1$ of this tensor product module, 
whose multiplicity is the $N$:th Catalan number.
The (inverse) R-matrix of $U_q(\mathfrak{sl}(2,\C))$ represents monodromy transformations that braid the variables $\bs x = (x_1, \ldots, x_{2N})$ in the Riemann sphere.
It acts on the tensor product 
$\mathsf{M}_2 \otimes \mathsf{M}_2 \cong \mathsf{M}_1 \oplus \mathsf{M}_3$ in the block diagonal form 
\begin{displaymath}
= \left( 
\begin{matrix}
  - q^{3/2} & 0 & 0 & 0 \\
  0 & q^{-1/2} & 0 & 0  \\
  0 & 0 & q^{-1/2} & 0  \\
  0 & 0 & 0 & q^{-1/2}
\end{matrix}
\right) ,
\end{displaymath}
and its eigenvalue on $\mathsf{M}_1$ equals $- q^{3/2}$.
Note that~\eqref{eq: rotation symmetry} in Proposition~\ref{prop: rotation symmetry} is consistent with this.
\end{remark}

\subsection{Contour deformations}
\label{subsec:Contour_deformations}

We now consider Coulomb gas integrals with general numbers 
of $d \geq 2$ variables and $\ell \geq 0$ screening variables (with any branch choice), 
\begin{align} \label{eq: integral_gen_app}
\underset{\Gamma}{\int \cdots \int }{}
\; \underbrace{\prod_{1\leq i<j\leq d}(x_{j}-x_{i})^{2 \alpha_i \alpha_j} 
\prod_{1\leq r<s\leq \ell}(u_{s}-u_{r})^{2 \alpha_-^2} 
\prod_{\substack{1\leq i\leq d \\ 1\leq r\leq \ell}}
(u_{r}-x_{i})^{2 \alpha_- \alpha_i}}_{=: \; f(\bs{x};\bs{u})}
\; \ud u_1 \cdots \ud u_\ell , 
\end{align}
with Coulomb charges $\alpha_j = \alpha_j(\kappa) \in \R$ 
at the variables $\bs{x} = (x_1, x_2, \ldots, x_d)$, and negative screening charges $\alpha_- = - 2/\sqrt{\kappa}$ at the integration variables $\bs{u} = (u_1, u_2, \ldots, u_\ell)$.
We assume that the integration contours in~\eqref{eq: integral_gen_app} have the form 
$\Gamma = \Gamma_1 \times \Gamma_2 \times \cdots \times \Gamma_\ell$  
such that $\Gamma_r \cap \Gamma_s = \emptyset$ for all $r \neq s$.
For example, they could be Pochhammer contours\footnote{Throughout, we always assume that $\acycle(x_{a}, x_{b})$ does not surround any other marked point among $\{x_1, \ldots, x_d\}$.} 
$\acycle(x_{a}, x_{b})$
as in~\eqref{eq:Pochhammer}, 
line integrals, 
or simple loops surrounding some of the marked points.
Specifically, we let $\intloop(x_{a}, x_{b})$ denote a clockwise oriented simple loop on $\C \setminus \{x_1, \ldots, x_d\}$ surrounding two points $x_{a}$ and $x_{b}$ and none of the other points among $\{x_1, \ldots, x_d\}$. 
(A limiting case of such a loop as $x_{b} \to +\infty$ is illustrated in Eq.~\eqref{eq: rho contour} in Lemma~\ref{lem::beta_rholoop}.)
We also denote by $\partial B(x,\epsilon)$ the clockwise oriented circle of radius $\epsilon$ around a point $x$.

\medskip
 
The next simple but useful lemma shows how to transform an integration around a Pochhammer contour into an integration involving line integrals between the marked points together with integrations around two small loops surrounding these points. 
To avoid any ambiguity, by the integration symbol $\smash{\landupint_{x+\epsilon}^{y-\epsilon}} \ud u$, 
we indicate that the integration is performed from $(x+\epsilon)$ to $(y-\epsilon)$ in the upper half-plane.

\begin{lemma}\label{lem::Poch_line_relation}
Fix $r \in \{1,2,\ldots,\ell\}$ such that $\Gamma_r = \acycle(x_{a}, x_{b})$ for some $a,b \in \{1,2,\ldots,d\}$, $a \neq b$. 
Then, 
\begin{align*} 
\; & 
\int_{\Gamma_1}  \ud u_1 \cdots \int_{\Gamma_\ell} \ud u_\ell 
\; f(\bs{x};\bs{u}) \\
= \; & 
\int_{\Gamma_1} \ud u_1 
\cdots \int_{\Gamma_{r-1}} \ud u_{r-1} 
\bigg( 
(1 - e^{-4 \pi\ii \alpha_- \alpha_{b}}) (1 - e^{4 \pi\ii \alpha_-\alpha_{a}})
\landupint_{x_{a}+\epsilon}^{x_{b}-\epsilon} \ud u_r
\\
\; & 
\; + \; 
(1 - e^{-4 \pi\ii \alpha_- \alpha_{b}}) 
\ointclockwise_{\partial B(x_{a},\epsilon)}\ud u_r
\; + \; 
 (1 - e^{4 \pi\ii \alpha_-\alpha_{a}})
\ointclockwise_{\partial B(x_{b},\epsilon)}\ud u_r
\bigg)
\int_{\Gamma_{r+1}} \ud u_{r+1} 
\cdots \int_{\Gamma_\ell} \ud u_\ell 
\; f(\bs{x};\bs{u}) , 
\end{align*}
for all $\epsilon > 0$ small enough.
\end{lemma}
\begin{proof}
The claim follows in a straightforward manner by contour deformation, 
collecting the phase factors of the integrand in~\eqref{eq: integral_gen_app} around the points $x_{a}$ and $x_{b}$.
\end{proof}

\begin{figure}[ht!]
\begin{center}
\includegraphics[width=0.45\textwidth]{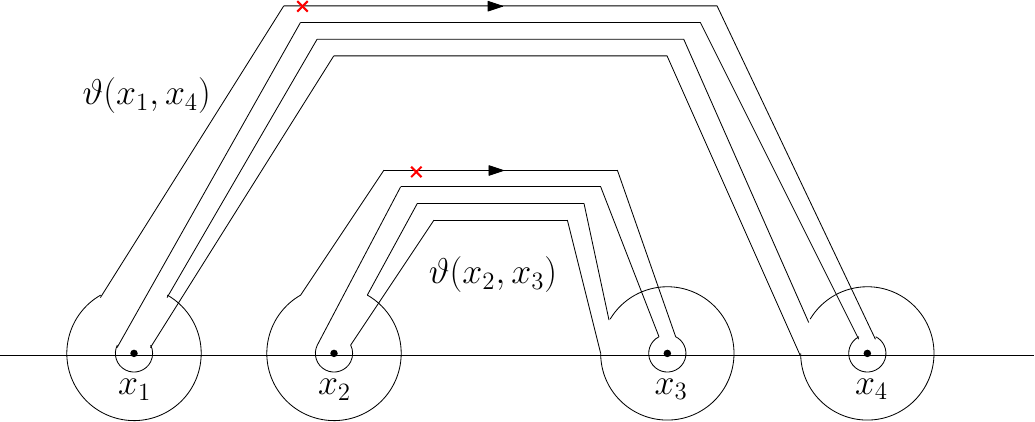}
$\quad$
\includegraphics[width=0.45\textwidth]{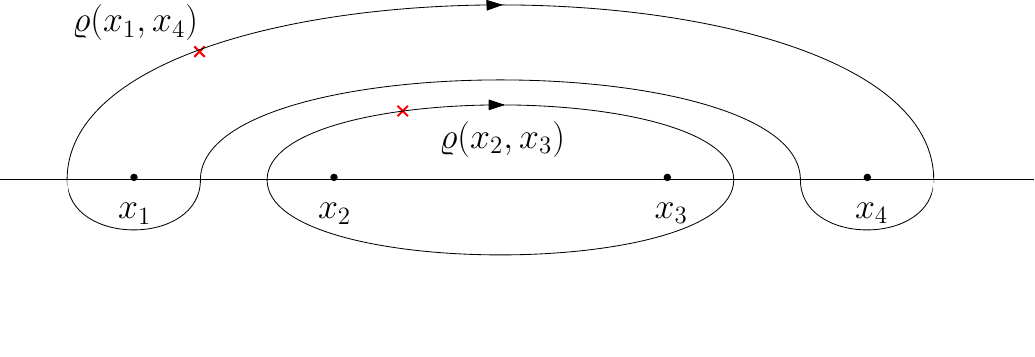}
\end{center}
\caption{\label{fig::loop_deformation_loop} 
Illustration of Corollary~\ref{cor::pochhammer_to_loop_8overodd}: 
when $\kappa=\frac{8}{2m-1}$ for $m\in\bZpos$, 
the integral of $f(\bs{x}; \bs{u})$ along the Pochhammer contour $\acycle(x_{a_r}, x_{b_r})$ is the same as $2$ times its integral along the simple clockwise loop $\intloop(x_{a_r}, x_{b_r})$. 
}
\end{figure}

\begin{corollary}\label{cor::pochhammer_to_loop_8overodd}
Fix $\kappa\in \big\{\tfrac{8}{2m-1} \colon m\in \bZpos\big\}$, $\beta\in\LP_N$, and $\bs{x}\in \chamber_{2N}$. 
Suppose that the charges in~\eqref{eq: integral_gen_app} are $\alpha_j = 1/\sqrt{\kappa}$ for all $j \in \{1, 2, \ldots, d\}$. 
Fix $r \in \{1,2,\ldots,\ell\}$ such that $\Gamma_r = \acycle(x_{a_r}, x_{b_r})$. Then, we have 
\begin{align*} 
\; & 
\int_{\Gamma_1} \ud u_1 \cdots 
\int_{\Gamma_{r-1}} \ud u_{r-1}
\ointclockwise_{\acycle(x_{a_r},x_{b_r})} \ud u_r
\int_{\Gamma_{r+1}} \ud u_{r+1} \cdots \int_{\Gamma_\ell} \ud u_\ell 
\; f(\bs{x};\bs{u}) \\
= \; & 2 \; 
\int_{\Gamma_1} \ud u_1 \cdots 
\int_{\Gamma_{r-1}} \ud u_{r-1}
\ointclockwise_{\intloop(x_{a_r},x_{b_r})} \ud u_r
\int_{\Gamma_{r+1}} \ud u_{r+1} \cdots \int_{\Gamma_\ell} \ud u_\ell 
\; f(\bs{x};\bs{u}) . 
\end{align*}
\end{corollary}
\begin{proof}
Similarly as when proving Lemma~\ref{lem::Poch_line_relation},
a straightforward contour deformation shows that
\begin{align*}
\ointclockwise_{\intloop(x_{a_r},x_{b_r})} \ud u_r \; f(\bs{x};\bs{u}) 
= \; & (1 - e^{-4 \pi\ii \alpha_- \alpha_{b_r}})
\landupint_{x_{a_r}+\epsilon}^{x_{b_r}-\epsilon} \ud u_r \; f(\bs{x};\bs{u}) 
\\
\; & 
\; + \; 
\ointclockwise_{\partial B(x_{a_r},\epsilon)}\ud u_r
\; f(\bs{x};\bs{u}) 
\; + \; 
\ointclockwise_{\partial B(x_{b_r},\epsilon)}\ud u_r
\; f(\bs{x};\bs{u}) .
\end{align*}
The claim now follows from Lemma~\ref{lem::Poch_line_relation}, after noting that
where $1 - e^{-4 \pi\ii \alpha_- \alpha_{b_r}} = 2 = 1 - e^{4 \pi\ii \alpha_- \alpha_{a_r}}$ 
due to the choices $\alpha_{a_r} = \alpha_{b_r} = 1/\sqrt{\kappa}$ and $\alpha_- = - 2/\sqrt{\kappa}$, since $\kappa = 8/m$ with $m \geq 1$ odd.
\end{proof}

We next record a very useful general property of Coulomb gas integrals of type~\eqref{eq: integral_gen_app} with $\ell \in \{2,3,\ldots, d\}$ and $\kappa\in\big\{ \tfrac{8}{2m-1} \colon m \in\bZpos\big\}$ 
and integration contours $\Gamma = \intloop_1 \times \intloop_2 \times \cdots \times \intloop_\ell$, 
where $\intloop_1, \intloop_2, \ldots, \intloop_\ell$ are clockwise\footnote{One could also orient them counterclockwise --- what is important is that all loops have the same orientation.} oriented simple mutually non-intersecting loops on $\C \setminus \mathcal{C}$, where
$\mathcal{C} = \gamma_1 \times \gamma_2 \times \cdots \times \gamma_{d'}$ are a choice of branch cuts for the multivalued integrand in~\eqref{eq: integral_gen_app} 
(with $d' = d/2$ if $d$ is even, and $d' = (d+1)/2$ if $d$ is odd). 
Namely, for any fixed $s \in \{1,2,\ldots,\ell\}$, we can replace the integration 
along $\intloop_s$ by an integration along another simple loop $\Gloop_s$ obtained from $\intloop_s$ by pulling it over some of the other loops in $\Gamma$ (as specified in Lemma~\ref{lem::circle_integral_general} below).
This property hold regardless of the branch choice for the integrand $f$ defined in~\eqref{eq: integral_gen_app}, as long as $\kappa\in\big\{ \tfrac{8}{2m-1} \colon m \in \bZpos\big\}$. 
The setup is illustrated in Figure~\ref{fig::loop_deformation_pulling}. 
This is a generalization of~\cite[Lemma~2.12]{LPW:UST_in_topological_polygons_partition_functions_for_SLE8_and_correlations_in_logCFT} which deals with the case $\kappa=8$. 

\begin{figure}[ht!]
\begin{center}
\includegraphics[width=0.45\textwidth]{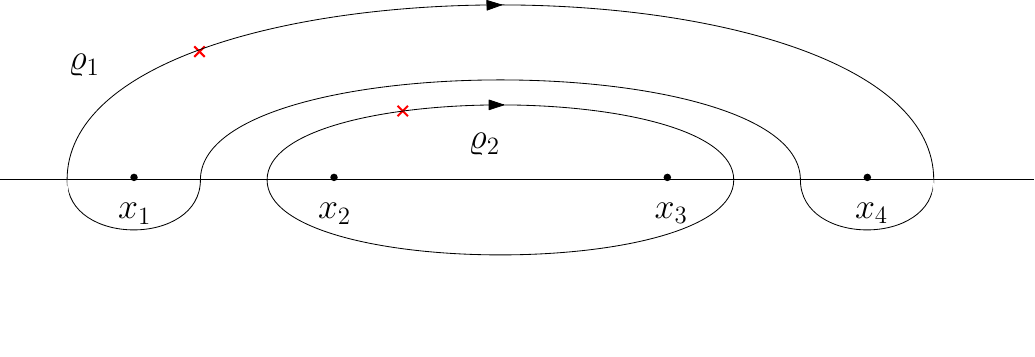}
$\quad$
\includegraphics[width=0.45\textwidth]{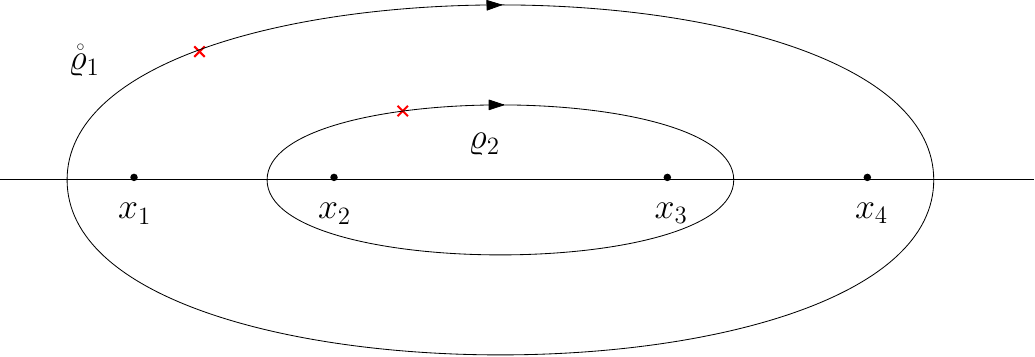}
\end{center}
\caption{\label{fig::loop_deformation_pulling} 
Illustration of Lemma~\ref{lem::circle_integral_general}:
when $\kappa=\frac{8}{2m-1}$ for $m\in\bZpos$, 
the integral of $f(\bs{x}; \bs{u})$ along the simple clockwise loop $\intloop_r$ is the same as its integral along a loop $\Gloop_r$ possibly surrounding more marked points. 
}
\end{figure}

\begin{lemma} \label{lem::circle_integral_general}
Fix $\kappa\in\big\{ \tfrac{8}{2m-1} \colon m \in \bZpos\big\}$. 
Fix $\ell \in \{2,3,\ldots,d\}$ and $s \in \{1,2,\ldots,\ell\}$. 
Let $\Gloop_s$ be a clockwise oriented simple loop on $\C \setminus \gamma_1 \times \gamma_2 \times \cdots \times \gamma_{d'}$ obtained from $\intloop_s$ by pulling $\intloop_s$ over some of the other loops in $\{\intloop_1, \intloop_2, \ldots, \intloop_\ell\} \setminus \{\intloop_s\}$\textnormal{:}
\begin{align} \label{eq::Gloop_in_homology}
\Gloop_s = \intloop_s + \sum_{r \in I_s} \intloop_r ,
\end{align}
where $I_s \subset \{1,2,\ldots, \ell\} \setminus \{s\}$. 
Then, writing $\bs{u} = (u_1, \ldots, u_\ell)$, we have
\begin{align} \label{eq::integra_deformation}
\begin{split}
\; & 
\ointclockwise_{\intloop_1}  \ud u_1 \ointclockwise_{\intloop_2}  \ud u_2 
\cdots \ointclockwise_{\intloop_\ell} \ud u_\ell 
\; f (\bs{x};\bs{u}) \\
= \; & \ointclockwise_{\intloop_1} \ud u_1 
\cdots \ointclockwise_{\intloop_{s-1}} \ud u_{s-1} 
\ointclockwise_{\Gloop_s} \ud u_s
\ointclockwise_{\intloop_{s+1}} \ud u_{s+1} 
\cdots \ointclockwise_{\intloop_{\ell}} \ud u_{\ell} 
\; f (\bs{x};\bs{u}) ,
\end{split}
\end{align}
where $f$ is the integrand defined in~\eqref{eq: integral_gen_app}.
\end{lemma}

Note that the replacement in Lemma~\ref{lem::circle_integral_general} is only valid when we integrate $f (\bs{x}; u_1, \ldots, u_\ell)$
along all $\ell$ loops $\intloop_1, \intloop_2, \ldots, \intloop_\ell$, 
in which case we can use antisymmetry to get cancellations.

\begin{proof}
Note that since $\kappa\in\big\{ \tfrac{8}{2m-1} \colon m \in\bZpos\big\}$, the integrand~\eqref{eq: integral_gen_app} reads
	\begin{align}\label{eqn::8overoddintegrand} 
		f (\bs{x};\bs{u}) = \; &
		\prod_{1\leq i<j\leq d}(x_{j}-x_{i})^{m/4} 
		\prod_{1\leq r<s\leq \ell}(u_{s}-u_{r})^{m} 
		\prod_{\substack{1\leq i\leq d \\ 1\leq r\leq \ell}}
		(u_{r}-x_{i})^{-m/2} .
	\end{align}
After deforming and decomposing the loop $\Gloop_s$ into the linear combination~\eqref{eq::Gloop_in_homology},  
we see by antisymmetry of the integrand~\eqref{eqn::8overoddintegrand} (since $m$ is odd) 
and Fubini's theorem that only $\intloop_s$ can give a non-zero contribution on the right-hand side of~\eqref{eq::integra_deformation}: 
indeed, for any $r \neq s$, the double-integral of~\eqref{eqn::8overoddintegrand} along $\intloop_r$ vanishes (here, we use the assumption that $\ell \geq 2$), 
\begin{align*}
\ointclockwise_{\intloop_r} \ud u_r \ointclockwise_{\intloop_r}  \ud u_s
\; f (\bs{x};u_1,\ldots,u_\ell)
= 0 ,
\end{align*}
by antisymmetry of the integrand~\eqref{eqn::8overoddintegrand} with respect to the exchange 
$u_r \leftrightarrow u_s$.
\end{proof}

\subsection{Asymptotics properties}
\label{subsec:ASY}
It was proved in the series~\cite{Flores-Kleban:Solution_space_for_system_of_null-state_PDE1, Flores-Kleban:Solution_space_for_system_of_null-state_PDE2, Flores-Kleban:Solution_space_for_system_of_null-state_PDE3, Flores-Kleban:Solution_space_for_system_of_null-state_PDE4} 
of articles that when $\kappa \in (0,8)$, the space 
\begin{align}\label{eq: solution space}
\mathcal{S}_N := \{ F \colon \chamber_{2N} \to \C \colon F 
\textnormal{ satisfies \eqref{eqn::PDE}, \eqref{eqn::COV}, and~\eqref{eqn::PLB_weak_upper}} \} 
\end{align}
has dimension $\mathrm{dim} \, \mathcal{S}_N = \tfrac{1}{N+1} \binom{2N}{N}$, the $N$:th Catalan number.
One can uniquely specify functions in this space in terms of their asymptotic behavior as 
two of the variables in $\chamber_{2N}$ tend together. 

\begin{lemma}\label{lem::PFuniqueness}
{\textnormal{\cite[Lemma~1]{Flores-Kleban:Solution_space_for_system_of_null-state_PDE2}}}
Fix $\kappa\in (0,8)$. Let $F \colon \chamber_{2N}\to \C$ be a function satisfying~\eqref{eqn::PDE}, \eqref{eqn::COV}, and~\eqref{eqn::PLB_weak_upper}. If $F$ also has the asymptotics property
\begin{align*}
\lim_{x_j, x_{j+1}\to\xi} \frac{F(x_1, \ldots, x_{2N})}{(x_{j+1}-x_j)^{-2h(\kappa)}}=0, \qquad\textnormal{for all }j\in \{1,2,3,\ldots, 2N-1\} \textnormal{ and }\xi\in (x_{j-1}, x_{j+2})
\end{align*}
\textnormal{(}with the convention that $x_0=-\infty$ and $x_{2N+1}=\infty$\textnormal{)}, then $F\equiv 0$. 
\end{lemma}

This is a version of an operator product expansion (OPE) for the corresponding CFT correlation functions, important for identifying solutions of interest to problems in random geometry, such as the SLE pure partition functions and interface crossing probabilities~\cite{BBK:Multiple_SLEs_and_statistical_mechanics_martingales, Kytola-Peltola:Pure_partition_functions_of_multiple_SLEs, 
Peltola:Towards_CFT_for_SLEs, 
Peltola-Wu:Crossing_probabilities_of_multiple_Ising_interfaces, 
FPW:Connection_probabilities_of_multiple_FK_Ising_interfaces}.
In this section, we consider asymptotics of the Coulomb gas integral solutions in $\mathcal{S}_N$.

\begin{proposition}\label{prop::H_ASY} 
Fix $\kappa\in (0,8) \setminus \big\{ \tfrac{8}{m} \colon m \in \bZpos \big\}$. 
With $\coulomb_\emptyset \equiv 1$ for the empty link pattern $\emptyset \in \LP_0$, 
the collection $\{\coulomb_\beta \colon \beta \in \LP_N\}$ satisfies 
the following recursive asymptotics property.
Fix $N \ge 1$ and $j \in \{1,2, \ldots, 2N-1 \}$. 
Then, for all $\xi \in (x_{j-1}, x_{j+2})$, using the notation~\eqref{eqn::bs_notation}, 
\begin{align}
\label{eqn::H_ASY} 
\; & \lim_{x_j,x_{j+1}\to\xi} \frac{\coulomb_\beta(\bs{x})}{ (x_{j+1}-x_j)^{-2h(\kappa)} }
= \frac{1}{\cst(\kappa)} \, 
\begin{cases}
\fugacity(\kappa) \, \coulomb_{\beta/\{j,j+1\}}(\bs{\ddot{x}}_j),
& \textnormal{if }\{j, j+1\}\in\beta , \\
\coulomb_{\wp_j(\beta)/\{j,j+1\}}(\bs{\ddot{x}}_j),
& \textnormal{if }\{j, j+1\} \not\in \beta , 
\end{cases}
\end{align}
where $\beta/\{j,j+1\} \in \LP_{N-1}$ denotes the link pattern obtained from $\beta$ by removing the link $\{j,j+1\}$ and relabeling the remaining indices by $1, 2, \ldots, 2N-2$, 
and $\wp_j$ 
is the ``tying operation'' defined by 
\begin{align} \label{eq: tying operation}
\wp_j \colon \LP_N\to \LP_N , \qquad
\wp_j(\beta) = 
\big(\beta\setminus(\{j,k_1\}, \{j+1, k_2\})\big)\cup \{j,j+1\}\cup \{k_1, k_2\} , 
\end{align}  
where $k_1$ \textnormal{(}resp.~$k_2$\textnormal{)} 
is the pair of $j$ \textnormal{(}resp.~$j+1$\textnormal{)} in $\beta$ \textnormal{(}and $\{j,k_1\}, \{j+1, k_2\}, \{k_1, k_2\}$ are unordered\textnormal{)}.
\begin{align*}
\vcenter{\hbox{\includegraphics[scale=1]{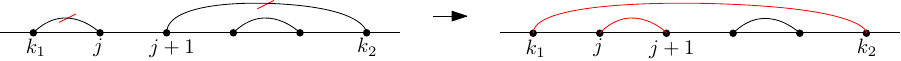}}} 
\end{align*}

\smallbreak

\noindent
Moreover, we have 
\textnormal{(}using the cyclic indexing convention $2N+1 \mapsto 1$ to define the tying operation $\wp_{2N}$\textnormal{)}
\begin{align}
\label{eqn::H_ASY_at_infinity} 
\; & \lim_{\substack{x_1 \to -\infty \\ x_{2N} \to +\infty}} \frac{\coulomb_\beta(\bs{x})}{ |x_{2N} - x_1|^{-2h(\kappa)} }
= \frac{1}{\cst(\kappa)} \, 
\begin{cases}
\fugacity(\kappa) \, \coulomb_{\beta/\{1, 2N\}}(x_2, \ldots, x_{2N-1}),
& \textnormal{if }\{1, 2N\}\in\beta , \\
\coulomb_{\wp_{2N}(\beta)/\{1, 2N\}}(x_2, \ldots, x_{2N-1}),
& \textnormal{if }\{1, 2N\} \not\in \beta .
\end{cases}
\end{align}
\end{proposition}

\cite[Appendix~C]{FPW:Connection_probabilities_of_multiple_FK_Ising_interfaces}
provides a proof for~\eqref{eqn::H_ASY} for the case $\{j,j+1\}\notin \beta$ with $\kappa\in (4,8)$, which is rather long and technical. 
The proof presented here uses Proposition~\ref{prop::remove_integration} 
and works for all $\kappa\in (0,8)\setminus \big\{\tfrac{8}{m} \colon m\in\bZpos \big\}$ and is much shorter. 
Readers may also have a look at proof of Proposition~\ref{prop::H4points_asy}, where we prove the asymptotics for $N=2$ and the notations there are simpler.

\begin{proof}
The normalization property $\coulomb_\emptyset \equiv 1$ is understood as an empty product of integrals in the definition~\eqref{eqn::coulombgasintegral_Poch}. 
The asymptotics in~\eqref{eqn::H_ASY_at_infinity} is a consequence of~\eqref{eqn::H_ASY} and Proposition~\ref{prop: full Mobius covariance}.

In the case $\{j,j+1\}\in\beta$, 
let $\acycle_r^\beta \ni u_r$ be the Pochhammer contour~\eqref{eq:Pochhammer} 
which surrounds the points $x_j$ and $x_{j+1}$. 
Because the integration contours remain bounded away from each other, 
we see that the integral relevant for evaluating the limit is
\begin{align*} 
\lim_{x_j, x_{j+1} \to \xi} 
\ointclockwise_{\acycle_r^\beta} \ud u_r \; \frac{f_{\beta}(\bs{x};\bs{u})}{(x_{j+1} - x_j)^{-2h(\kappa)}}  .
\end{align*}
After making the change of variables $v = \frac{u_r - x_j}{x_{j+1} - x_j}$ in this integral and collecting all the factors,
carefully noting that no branch cuts are crossed, 
and after taking into account cancellations, we obtain~\eqref{eqn::H_ASY}: 
\begin{align*}
\; & \lim_{x_j,x_{j+1}\to\xi} \frac{\coulomb_\beta(\bs{x})}{ (x_{j+1}-x_j)^{-2h(\kappa)} } 
\\
\; & =
\bigg( \ointclockwise_{\acycle(0,1)} \ud u  \; (u-1)^{-4/\kappa} u^{-4/\kappa} \bigg)
\; \ointclockwise_{\acycle_1^\beta}  \ud u_1 \cdots \ointclockwise_{\acycle_{r-1}^\beta}  \ud u_{r-1} 
\ointclockwise_{\acycle_{r+1}^\beta} \ud u_{r+1}   \cdots \ointclockwise_{\acycle_N^\beta} \ud u_N \; 
f_{\beta / \{j,j+1\}}(\bs{\ddot{x}}_j;\bs{\dot{u}}_r) \\
\; & = \frac{\fugacity(\kappa)}{\cst(\kappa)} 
\; \coulomb_{\beta/\{1,2\}}(\bs{\ddot{x}}_{j}) 
\end{align*}
where $\bs{\dot{u}}_r := (u_1,u_2,\ldots,u_{r-1},u_{r+1},\ldots,u_{N})$, 
and the prefactor evaluates to~\eqref{eqn::Poch_beta_again}.

In the case $\{j,j+1\}\notin \beta$, thanks to Proposition~\ref{prop: full Mobius covariance}, 
by conjugating by a suitable M\"obius map 
it suffices without loss of generality to check~\eqref{eqn::H_ASY} in the case $j=1$. 
Write $\hat{\beta} := \beta\setminus (\{\{1,b_1\},\{2,b_2\}\})$ (we do not relabel the indices here), and denote by $\Gamma_{\hat{\beta}}$ the integration contours in $\coulombGasH_{\beta}$ other than the two Pochhammer contours surrounding 
the pairs $\{x_{1},x_{b_1}\}$ and $\{x_{2},x_{b_2}\}$, respectively. 
Proposition~\ref{prop::remove_integration}, gives
\begin{align*}
\coulombGasH_{\beta}(\bs{x}) 
= \; & \frac{\fugacity(\kappa) }{\cst(\kappa)} \,
\ointclockwise_{\acycle^\beta_2} \ud u_2
\cdots \ointclockwise_{\acycle^\beta_{N}} \ud u_{N} 
\; \hat{f}_\beta^b (\bs{x};\bs{\dot{u}}_1) 
\\
= & \; \frac{\fugacity(\kappa) }{\cst(\kappa)} 
(x_{2}-x_1)^{2/\kappa} \, (x_{b_1}-x_{1})^{-2h(\kappa)} \, (x_{b_1}-x_2)^{-2h(\kappa)} \,  (x_{b_2}-x_{1})^{2/\kappa}(x_{b_2}-x_{2})^{2/\kappa} 
\\
\; & \times \prod_{\substack{3\leq i\leq 2N \\ i \notin \{b_1,b_2\}}} \big( (x_i - x_1)^{2/\kappa} \, (x_i - x_{2})^{2/\kappa} \big)
\\
\; & \times \int_{\Gamma_{\hat{\beta}}} 
\ud \bs{\ddot{u}}_1 \; \ointclockwise_{\acycle(x_{2},x_{b_2})}\ud u_2 
\; \tilde{f}(\bs{\ddot{x}}_1;\bs{\ddot{u}}_1) \, 
(u_2-x_{b_1})^{4h(\kappa)} \, (u_2-x_{b_2})^{-4/\kappa} \, (u_2-x_1)^{-4/\kappa} \, (u_2-x_{2})^{-4/\kappa} \\
\; &  \times 
\prod_{\substack{3\leq i\leq 2N \\ i\notin \{b_1,b_2\}}} (u_2-x_i)^{-4/\kappa} 
\prod_{3\leq s\leq N}(u_s-u_2)^{8/\kappa}  \prod_{3\leq s\leq N} \big( (u_s-x_1)^{-4/\kappa} \, (u_s-x_{2})^{-4/\kappa} \big)
\end{align*}
where $\bs{\ddot{u}}_1 = (u_3, u_4, \ldots, u_N)$ 
and $\smash{\hat{f}_{\beta}^b}$ is defined in~\eqref{eq: integrand with one less screening variable b} with $r=1$ and the conjugate charge at $x_{b_1}$, 
and 
\begin{align*} 
\tilde{f} (\bs{\ddot{x}}_1;\bs{\ddot{u}}_1)
:= \; &
\prod_{\substack{3\leq i<l\leq 2N \\ i,l \neq b_1}}(x_{l}-x_{i})^{2/\kappa} 
\prod_{\substack{3\leq i\leq 2N \\ i\neq b_1}} |x_{b_1}-x_{i}|^{-2h(\kappa)} 
 \\
\; & \times \; 
\prod_{3\leq t<s\leq N}(u_{s}-u_{t})^{8/\kappa} 
\prod_{\substack{3\leq i\leq 2N  \\i \neq b_1\\ 3\leq s\leq N  }}
(u_{s}-x_{i})^{-4/\kappa} 
\prod_{\substack{3\leq s\leq N  \\ }}
(u_{s}-x_{b_1})^{4h(\kappa)} ,
\end{align*}  
and the branch of $\tilde{f} (\bs{\ddot{x}}_1;\cdot)$ 
and the rest of the products in the integrand is chosen to be real and positive when $x_{a_s} < \Re (u_s) < x_{a_s+1}$ for all $s \in \{3,4,\ldots,N\}$.
Next, we use the change of variables
\begin{align*}
v = \frac{\CR(x_{b_2},x_{b_1},x_{2},u_2)}{\CR(x_1,x_2,x_{b_2},x_{b_1})} ,
\end{align*}
where $\CR$ is the cross-ratio 
\begin{align}\label{eqn::crossratio}
\CR(y_1,y_2,y_3,y_4) := \frac{(y_2-y_1)(y_4-y_3)}{(y_3-y_1)(y_4-y_2)} ,
\end{align} 
to obtain
\begin{align}
\notag 
\coulombGasH_{\beta}(\bs{x}) 
= \; & \frac{\fugacity(\kappa) }{\cst(\kappa)} \,
(x_{2}-x_1)^{2/\kappa} \, (x_{b_1}-x_{b_2})^{8/\kappa-1} \, (x_{b_1}-x_{1})^{-2h(\kappa)} \, (x_{b_2}-x_{2})^{-2h(\kappa)} \, (x_{b_1}-x_2)^{-2/\kappa} \, (x_{b_2}-x_{1})^{-2/\kappa} 
\\
 \; & \times \chi^{1-8/\kappa} 
\; \int_{\Gamma_{\hat{\beta}}} \ud \bs{\ddot{u}}_1 
 \ointclockwise_{\acycle(0,\frac{1}{\chi})}\ud v
 \; \tilde{f}(\bs{\ddot{x}}_1;\bs{\ddot{u}}_1) \; 
 \tilde{g}(x_1,x_2,v; \bs{\ddot{x}}_1;\bs{\ddot{u}}_1) 
 \, k(\chi,v)  ,
\label{eqn::H_ASY_change_of_variables}
\end{align}
denoting $\chi := \CR(x_1,x_2,x_{b_2},x_{b_1})$ and
\begin{align*}
\tilde{g}(x_1,x_2,v; \bs{\ddot{x}}_1;\bs{\ddot{u}}_1) 
:= \; & \prod_{\substack{3\leq i\leq 2N \\ i \notin \{b_1,b_2\}}}
\Big| x_{b_1}-(x_{b_1}-x_2) 
\Big( 1 + v \, \chi \, \frac{x_{b_2}-x_2}{x_{b_1}-x_{b_2}} \Big)^{-1} - x_i \Big|^{-4/\kappa}
\\
& \; \times \prod_{3\leq s\leq N} \Big(u_s-x_{b_1}+(x_{b_1}-x_2) 
\Big( 1 + v \, \chi \, \frac{x_{b_2}-x_2}{x_{b_1}-x_{b_2}}\Big)^{-1}\Big)^{8/\kappa}
\\ 
& \;  \times \prod_{3\leq s\leq N } \big( (u_s-x_1)^{-4/\kappa}(u_s-x_{2})^{-4/\kappa}
\big) 
\prod_{\substack{3\leq i\leq 2N \\ i\notin \{b_1,b_2\}}} \big( (x_i-x_1)^{2/\kappa} \, (x_i-x_{2})^{2/\kappa} \big) ,
\end{align*}
whose branch is chosen to be real and positive when 
\begin{align*}
0 < \Re (v) < \frac{\CR(x_{b_2},x_{b_1},x_{2},x_3)}{\CR(x_1,x_2,x_{b_2},x_{b_1})} 
, \qquad \textnormal{and} \qquad 
x_{a_s}<\Re (u_s)<x_{a_s+1} , \qquad \textnormal{for all } s\in \{3,4,\ldots,N\} ,
\end{align*}
and where the branch of 
\begin{align*}
k(\chi,v) 
:= \; & 
(1+(1-\chi)v)^{-4/\kappa} v^{-4/\kappa} (1-\chi v)^{-4/\kappa},
\end{align*}
is chosen to be real and positive when $0 < \Re (v) < 1/\chi$.
Note that $k(\chi,v) = f(0, \tfrac{1}{\chi-1}, \tfrac{1}{\chi}; v)$ has the Coulomb gas integrand form~\eqref{eq: integral_gen_app} with  
symmetric charges $\alpha = 1/\sqrt{\kappa}$ 
at the variables $(0, \tfrac{1}{\chi-1}, \tfrac{1}{\chi})$
and the negative screening charge $\alpha_- = - 2/\sqrt{\kappa}$ at $v$.
Hence, using Lemma~\ref{lem::Poch_line_relation}, we obtain
\begin{align*}
\ointclockwise_{\acycle(0,\frac{1}{\chi})} k(\chi,v) \, \ud v 
= \; & 
(1 - e^{8 \pi\ii/\kappa }) (1 - e^{-8 \pi\ii /\kappa})
\int_{\epsilon}^{\frac{1}{\chi}-\epsilon} k(\chi,v) \,\ud v
\\
\; & 
\; + \; 
(1 - e^{8 \pi\ii/\kappa}) 
\ointclockwise_{\partial B(0,\epsilon)} k(\chi,v) \,\ud v
\; + \; 
 (1 - e^{-8 \pi\ii/\kappa})
\ointclockwise_{\partial B(\frac{1}{\chi},\epsilon)} k(\chi,v) \, \ud v ,
\end{align*}
for all $\epsilon > 0$ small enough. 
Now, we take $\chi \to 0+$ (while keeping $\epsilon$ fixed).
On the one hand, using the fact that $8/\kappa>1$ and the dominated convergence theorem, 
we see that the first two terms have a limit,
\begin{align*}
\; & (1 - e^{8 \pi\ii/\kappa }) (1 - e^{-8 \pi\ii /\kappa})
\lim_{\chi \to 0+} 
\landupint_{\epsilon}^{\frac{1}{\chi}-\epsilon} k(\chi,v) \, \ud v
\; + \; 
(1 - e^{8 \pi\ii/\kappa}) 
\lim_{\chi \to 0+} \ointclockwise_{\partial B(0,\epsilon)} k(\chi,v) \,\ud v
\\
= \; & (1-e^{8\pi\ii/\kappa}) \ointclockwise_{\intloop(0,\infty)} (1+v)^{-4/\kappa} v^{-4/\kappa} \, \ud v ,
\end{align*}
where $\intloop(0,\infty)$ is a contour which starts from $\infty$, 
then winds around $0$ once in 
clockwise direction and finally goes back to $\infty$,
illustrated in Eq.~\eqref{eq: rho contour} in Lemma~\ref{lem::beta_rholoop}.
On the other hand, the third term is
\begin{align*}
\lim_{\chi \to 0+} 
\ointclockwise_{\partial B(\frac{1}{\chi},\epsilon)} k(\chi,v) \, \ud v = 0 ,
\end{align*}
since $| k(\chi,v) | \sim v^{-8/\kappa}$ as $v \to \infty$.
Consequently, we have 
\begin{align*}
\lim_{x_1,x_2 \to \xi} \frac{\coulombGasH_{\beta}(\bs{x})}{(x_{2}-x_1)^{-2h(\kappa)}}
= \; & \frac{\fugacity(\kappa)}{\cst(\kappa)}
\; \int_{\Gamma_{\hat{\beta}}}\tilde{f}(\bs{\ddot{x}}_1;\bs{\ddot{u}}_1) \, \ud \bs{\ddot{u}}_1 \; 
\bigg( \lim_{\chi \to 0+}\ointclockwise_{\acycle(0,\frac{1}{\chi})} k(\chi,v) \, \ud v \bigg)
\\
= \; & \frac{\fugacity(\kappa)}{\cst(\kappa)}
\; \int_{\Gamma_{\hat{\beta}}}\tilde{f}(\bs{\ddot{x}}_1;\bs{\ddot{u}}_1) \, \ud \bs{\ddot{u}}_1 \;  \bigg( (1-e^{8\pi\ii/\kappa}) \; \int_{\intloop(0,\infty)} (1+v)^{-4/\kappa} v^{-4/\kappa} \, \ud v \bigg) .
\end{align*}
Using Eq.~\eqref{eqn::beta_2} from Lemma~\ref{lem::beta_rholoop} to evaluate the integral
\begin{align*}
\int_{\intloop(0,\infty)} (1+v)^{-4/\kappa} v^{-4/\kappa} \, \ud v  
= \frac{(1-e^{-8\pi\ii/\kappa})}{\fugacity(\kappa)}  \, \frac{\Gamma(1-4/\kappa)^2}{\Gamma(2-8/\kappa)} 
\end{align*}
and Proposition~\ref{prop::remove_integration}, we conclude that 
\begin{align*}
\lim_{x_1,x_2 \to \xi} \frac{\coulombGasH_{\beta}(\bs{x})}{(x_{2}-x_1)^{-2h(\kappa)}}
= \; & \frac{\fugacity(\kappa)}{\cst(\kappa)}
\; (1-e^{8\pi\ii/\kappa}) \, \frac{(1-e^{-8\pi\ii/\kappa})}{\fugacity(\kappa)}
\, \frac{\Gamma(1-4/\kappa)^2}{\Gamma(2-8/\kappa)}
\; \int_{\Gamma_{\hat{\beta}}}\tilde{f}(\bs{\ddot{x}}_1;\bs{\ddot{u}}_1) \, \ud \bs{\ddot{u}}_1
\\
= \; & \frac{1}{\cst(\kappa)} \, \bigg( 
\frac{\fugacity(\kappa)}{\cst(\kappa)}
\; \int_{\Gamma_{\hat{\beta}}}\tilde{f}(\bs{\ddot{x}}_1;\bs{\ddot{u}}_1) \, \ud \bs{\ddot{u}}_1 \bigg) 
\; = \; 
\frac{1}{\cst(\kappa)} 
\; \coulomb_{\wp_1(\beta)/\{1,2\}}(\bs{\ddot{x}}_1) ,
\end{align*}
which shows~\eqref{eqn::H_ASY} for $j=1$ and $\{1,2\} \notin \beta$ and finishes the proof.  
\end{proof}

\begin{corollary}\label{cor::CGI_PDECOVASY}
Fix $\kappa\in (0,8) \setminus \big\{ \tfrac{8}{m} \colon m \in \bZpos \big\}$. 
The collection $\{\coulombGas_{\beta} \colon \beta \in \LP_N\}$  
satisfies the properties~\eqref{eqn::PDE},~\eqref{eqn::COV},
and~\eqref{eqn::PLB_weak_upper}. Furthermore, with $\smash{\coulombGas_{\emptyset}} \equiv 1$ for the empty link pattern $\emptyset \in \LP_0$, 
the following recursive 
asymptotics property holds.
Fix $N \ge 1$ and $j \in \{1,2, \ldots, 2N-1 \}$. 
Then, for all $\xi \in (x_{j-1}, x_{j+2})$, using the notation~\eqref{eqn::bs_notation}, 
\begin{align}
\label{eqn::CGI_ASY} 
\; & \lim_{x_j,x_{j+1}\to\xi} \frac{\coulombGas_{\beta}(\bs{x})}{ (x_{j+1}-x_j)^{-2h(\kappa)} }
= 
\begin{cases}
\fugacity(\kappa) \, \coulombGas_{\beta/\{j,j+1\}}(\bs{\ddot{x}}_j),
& \textnormal{if }\{j, j+1\}\in\beta , \\
\coulombGas_{\wp_j(\beta)/\{j,j+1\}}(\bs{\ddot{x}}_j),
& \textnormal{if }\{j, j+1\} \not\in \beta .
\end{cases}
\tag{$\coulombGas$\textnormal{-ASY}}
\end{align}
Moreover, we have 
\begin{align}
\label{eqn::CGI_ASY_at_infinity} 
\; & \lim_{\substack{x_1 \to -\infty \\ x_{2N} \to +\infty}} \frac{\coulombGas_{\beta}(\bs{x})}{ |x_{2N} - x_1|^{-2h(\kappa)} }
= 
\begin{cases}
\fugacity(\kappa) \, \coulombGas_{\beta/\{1, 2N\}}(x_2, \ldots, x_{2N-1}),
& \textnormal{if }\{1, 2N\}\in\beta , \\
\coulombGas_{\wp_{2N}(\beta)/\{1, 2N\}}(x_2, \ldots, x_{2N-1}),
& \textnormal{if }\{1, 2N\} \not\in \beta .
\end{cases}
\tag{$\coulombGas\textnormal{-ASY}_\infty$}
\end{align}
\end{corollary}

\begin{proof}
\eqref{eqn::PDE}~\&~\eqref{eqn::COV} are well-known, see, e.g.,~\cite[Propositions~2.2~\&~2.3]{FPW:Connection_probabilities_of_multiple_FK_Ising_interfaces}.
(See also Proposition~\ref{prop: full Mobius covariance} for a stronger covariance under $\varphi \in \mathrm{SL}(2,\R)$.) 
The bound~\eqref{eqn::PLB_weak_upper} is clear from the definition. 
The asymptotics properties~(\ref{eqn::CGI_ASY},~\ref{eqn::CGI_ASY_at_infinity}) follow from Proposition~\ref{prop::H_ASY} and the relation~\eqref{eq::F_as_C_times_H}.
\end{proof}

\subsection{Coulomb gas integrals for exceptional values of $\kappa$}

Let us record here the observation that
the Coulomb gas type integrals vanish for the special values of the parameter
of the form $\kappa = 8/m$ for integers $m\ge 2$, including $\kappa = 4$, $\kappa = 8/3$,
and $\kappa = 2$. 

\begin{lemma} \label{lem:CGI_vanish}
	For $\kappa \in \big\{ \tfrac{8}{m} \colon m \in \bZpos+1 \big\}$, we have
	\begin{align*}
		\coulomb_\beta (\bs{x}) \equiv 0 , \qquad \textnormal{for all } \beta \in \LP_N \textnormal{ and } \bs{x} \in \chamber_{2N} .
	\end{align*}
\end{lemma}

\begin{proof} 
We first check the case $\kappa=\tfrac{4}{m}$ for $m\in \bZpos$.  
Note that the integrand~\eqref{eq: integrand} reads
\begin{align*}
	f_\beta (\bs{x};\bs{u}) = \; &
	\prod_{1\leq i<j\leq 2N}(x_{j}-x_{i})^{m/2} 
	\prod_{1\leq r<s\leq N}(u_{s}-u_{r})^{2m} 
	\prod_{\substack{1\leq i\leq 2N \\ 1\leq r\leq N}}
	(u_{r}-x_{i})^{-m} , \qquad \kappa=\tfrac{4}{m} .
\end{align*}
For any fixed $s\in \{1,2,\ldots, N\}$, and $\bs{x}\in \chamber_{2N}$, 
and $\bs{\dot{u}}_s\in \acycle_1^{\beta}\times \acycle_{2}^{\beta}\times \cdots \times \acycle_{s-1}^{\beta}\times \acycle_{s+1}^{\beta}\times \cdots\times \acycle_{N}^{\beta}$,
the function 
$u_s\mapsto f_{\beta}(\bs{x};\bs{u})$ is holomorphic on $\C\setminus \{x_1,\ldots,x_{2N},u_1,\ldots,u_{s-1},u_{s+1},\ldots,u_N\}$.
Because the Pochhammer contour $\smash{\acycle^\beta_s}$ surrounds both of its poles $\{x_{a_s},x_{b_s}\}$ 
once in the positive direction and once in the negative direction,
we see by the residue theorem that the integral of $f_{\beta}$ around $\smash{\acycle^\beta_s}$ vanishes:
\begin{align*}
\ointclockwise_{\acycle^\beta_s}\ud u_s f_{\beta}(\bs{x};\bs{u})
= \; & 2\pi \ii \, \Big( \underset{u_s=x_{a_s}}{\mathrm{Res}}f_\beta(\bs{x};\bs{u}) \Big)  (1-1)
\;  + \; 2\pi \ii \, \Big(\underset{u_s=x_{b_s}}{\mathrm{Res}} f_\beta(\bs{x};\bs{u})\Big) (1-1) 
= 0.
\end{align*}
Consequently, using Fubini's theorem, we conclude that $\coulomb_\beta (\bs{x}) \equiv 0$ for $\kappa\in \big\{\tfrac{8}{2m} \colon m\in \bZpos\big\}$.

Next, we check the case $\kappa=\frac{8}{2m+1}$ for $m\in \bZpos$. 
Note that the integrand~\eqref{eq: integrand} reads
	\begin{align*}
		f_\beta (\bs{x};\bs{u}) = \; &
		\prod_{1\leq i<j\leq 2N}(x_{j}-x_{i})^{(2m+1)/4} 
		\prod_{1\leq r<s\leq N}(u_{s}-u_{r})^{2m+1} 
		\prod_{\substack{1\leq i\leq 2N \\ 1\leq r\leq N}}
		(u_{r}-x_{i})^{-(2m+1)/2}, \qquad \kappa=\tfrac{8}{2m+1} .
	\end{align*}
Using Corollary~\ref{cor::pochhammer_to_loop_8overodd}, we can replace the Pochhammer contours $\acycle^\beta_r = \acycle(x_{a_r}, x_{b_r})$~\eqref{eq:Pochhammer} 
in the integral $\coulomb_\beta$~\eqref{eqn::coulombgasintegral_Poch} by simple clockwise loops $\smash{\intloop_r^\beta}$, 
each surrounding the two points $\{x_{a_r},x_{b_r}\}$ and no other points (see Figure~\ref{fig::loop_deformation_loop}).
Moreover, a repeated application of Lemma~\ref{lem::circle_integral_general} 
shows that we can replace these loops $\smash{\intloop_r^\beta}$ 
by nested loops $\smash{\Gloop_r^\beta}$ that are symmetric with respect to the real axis --- 
see Figure~\ref{fig::loop_deformation_pulling}. 
Moreover, we can then deform the outermost such loop into a loop surrounding the point $\infty$ at infinity.

Now, thanks to Proposition~\ref{prop: full Mobius covariance}, we may assume that $\{1,2N\}\in \beta$ without loss of generality. 
Then, for any fixed $\bs{x}\in \chamber_{2N}$ and $\bs{\dot{u}}_1\in \acycle_{2}^{\beta}\times \acycle_{3}^{\beta}\times \cdots \times \acycle_{N}^{\beta}$, 
the function 
$u_1 \mapsto f_{\beta}(\bs{x};\bs{u})$ is holomorphic on $\C\setminus [x_1,x_{2N}]$.	
In this case, the aforementioned outermost loop is $\smash{\Gloop_1^\beta}$, 
and using the residue theorem, we see that
	\begin{align*}
		\coulomb_\beta (\bs{x}) 
		= \; & \ointclockwise_{\acycle^\beta_1}  \ud u_1 \ointclockwise_{\acycle^\beta_2}  \ud u_2 \cdots \ointclockwise_{\acycle^\beta_N} \ud u_N 
		\; f_\beta(\bs{x};\bs{u}) 
		&& \textnormal{[by~\eqref{eqn::coulombgasintegral_Poch}]}		
		\\
				= \; & 2^N \, \ointclockwise_{\intloop^\beta_2}  \ud u_2\cdots  \ointclockwise_{\intloop^\beta_N}\ud u_N \ointclockwise_{\intloop^\beta_1} \ud u_1 
				\; f_{\beta}(\bs{x};\bs{u})
		&& \textnormal{[by Corollary~\ref{cor::pochhammer_to_loop_8overodd}]}
				\\
				= \; & 2^N \, \ointclockwise_{\Gloop^\beta_2}  \ud u_2\cdots  \ointclockwise_{\Gloop^\beta_N}\ud u_N \ointclockwise_{\Gloop^\beta_1} \ud u_1 
				\; f_{\beta}(\bs{x};\bs{u})		
		&& \textnormal{[by Lemma~\ref{lem::circle_integral_general}]}
						\\
				= \; & - 2^{N+1} \, \pi \ii \, \ointclockwise_{\Gloop^\beta_2}  \ud u_2\cdots  \ointclockwise_{\Gloop^\beta_N}\ud u_N 
			 \, \Big( \underset{u_1=\infty}{\mathrm{Res}}f_\beta(\bs{x};\bs{u}) \Big)
		&& \textnormal{[by residue theorem]}
					\\
				= \; & 0 ,	
	\end{align*}
since $\underset{u_1=\infty}{\mathrm{Res}}f_\beta(\bs{x};\bs{u}) = - \underset{|u_1|\to \infty}{\lim} \, u_1 f_{\beta}(\bs{x};\bs{u}) = 0$. This concludes the proof.
\end{proof}

\begin{remark}
Note that Lemma~\ref{lem:CGI_vanish} does not apply to the case where $\kappa = 8$, 
because the residue of the integrand $u_1 f_{\beta}(\bs{x};\bs{u})$ is not zero in that case.
\end{remark}


\bigskip{}
\section{Constructing variants of multiple $\SLE_\kappa$ measures for $\kappa\in (4,8)$}
\label{sec::multipleSLEs}
In this section, we consider measures on multiple interacting SLE curves. Let us first gather some terminology about multiple SLEs. 
To begin, recall that the chordal $\SLE_\kappa$ is a random curve $\eta$, parameterized by speed $\kappa > 0$ and a simply connected domain $\Omega$ together with two distinct endpoints $x_1, x_2 \in \partial \Omega$ of the curve $\eta$ on the boundary 
(so $\SLE_\kappa$ is a random chord in $(\Omega; x_1, x_2)$). 
It is a conformally invariant process (Definition~\ref{def: conformally invariant}), 
and taking as 
the reference domain the upper half-plane $\HH$,  
i.e., $(\Omega; x_1, x_2) = (\HH; 0, \infty)$, the chordal $\SLE_\kappa$ curve $\eta$ 
can be constructed as a random Loewner chain via solving the ODE
\begin{align} \label{eq:LE}
\partial_t g_t(z) = \frac{2}{g_t(z) - W_t} , 
\qquad \textnormal{with initial condition} \qquad 
g_0(z) = z , \qquad z \in \overline{\HH} ,
\end{align}
where the driving function $W_t = \sqrt{\kappa} B_t$ is given by the standard one-dimensional Brownian motion. 
For each $z \in \overline{\HH}$, the evolution $(g_t(z), t\ge 0)$ 
is well-defined only up to the blow-up time of the ODE~\eqref{eq:LE}, 
$\tau(z) := \sup\{t\ge 0: \min_{s\in [0,t]}|g_s(z)-W_s|>0\}$,
when the denominator becomes zero, and the domain of the conformal map $g_t$ is thus given by the complement of the $\SLE_\kappa$ hull $K_t := \{z\in\HH \colon \tau(z)\leq t\}$.
The conformal maps $g_t \colon \HH\setminus K_t \to \HH$ are uniquely determined and normalized as $\lim_{z\to\infty}|g_t(z)-z|=0$.
When $\kappa \in (0,4]$, we have $K_t = \eta[0,t]$ while when $\kappa > 4$, the curve $\eta$ has self-touchings and $K_t$ is the filling of $\eta[0,t]$ in the sense that the domain 
$\HH \setminus K_t$ is the unbounded connected component of $\HH \setminus \eta[0,t]$. 
Notably, when $\kappa \geq 8$, the chordal $\SLE_\kappa$ curve is space-filling --- 
thus, in the present article we shall be concerned with the range $\kappa \in (0,8)$.
Detailed introductions can be found in textbooks, e.g.~\cite{Werner:Random_planar_curves_and_SLE, Lawler:Conformally_invariant_processes_in_the_plane, Kemppainen:SLE_book}. 

\paragraph*{Curve space.}
We denote by $X_{\simple}(\Omega; x_1, x_2)$ the set of continuous simple unparameterized curves in $\Omega$ connecting $x_1$ and $x_2$ such that they only touch the boundary $\partial\Omega$ in $\{x_1, x_2\}$. 
We denote by $X_0(\Omega; x_1, x_2)$ its closure 
in the metric topology of 
the set $X$ of all planar oriented curves 
(i.e., continuous mappings from $[0,1]$ to $\C$ modulo reparameterization) 
endowed with the metric 
\begin{align}\label{eqn::curve_metric}
\dist(\eta, \tilde{\eta}):=\inf_{\psi,\tilde{\psi}}\sup_{t\in [0,1]}|\eta(\psi(t))-\tilde{\eta}(\tilde{\psi}(t))| , \qquad \eta, \tilde{\eta} \in X ,
\end{align} 
where the infimum is taken over all increasing homeomorphisms $\psi, \tilde{\psi} \colon [0,1]\to [0,1]$. 

Note that curves in $X_0(\Omega; x_1, x_2)$ may have multiple points.
In particular, the chordal $\SLE_{\kappa}$ curves belong to this space almost surely when $\kappa>4$~\cite{Rohde-Schramm:Basic_properties_of_SLE}. 
(See Section~\ref{subsec:SLE_preli} for a brief survey of the definition of SLE curves.)
For $N\ge 2$ and $\alpha = \{ \{a_1,b_1\},  \{a_2,b_2\},\ldots , \{a_N,b_N\}\} \in\LP_N$, we denote by\footnote{The metric space $(X, \dist)$ is complete and separable, while the space $X_0^{\alpha}(\Omega; x_1, \ldots, x_{2N})$ is not complete.} 
\begin{align*}
X_0^{\alpha}(\Omega; x_1, \ldots, x_{2N}) 
\end{align*} 
the collection of curves $(\eta_1, \ldots, \eta_N)$ such that, for each $s\in\{1, \ldots, N\}$, we have $\eta_s\in X_0(\Omega; x_{a_s}, x_{b_s})$ and $\eta_s$ does not disconnect any two points $x_a, x_b$ such that $\{a,b\}\in\alpha$ from each other. 
We also set 
\begin{align}\label{eq:curve_space}
X_0(\Omega; x_1, \ldots, x_{2N})=\bigsqcup_{\alpha\in\LP_N}X_0^{\alpha}(\Omega; x_1, \ldots, x_{2N}).
\end{align}

\begin{definition}\label{def::NSLE}
Fix $\kappa\in (0,8)$, and $N\ge 2$, and $\alpha=\{\{a_1, b_1\}, \ldots, \{a_N, b_N\}\}\in\LP_N$. 
We call a probability measure on the families $(\eta_1, \ldots, \eta_N)\in X_0^{\alpha}(\Omega; x_1, \ldots, x_{2N})$ a \emph{global $N$-$\SLE_{\kappa}$ associated to $\alpha$}
if it satisfies the following resampling property: for each $s\in\{1, \ldots, N\}$, the conditional law  of the curve $\eta_s$ given $\{\eta_1, \ldots, \eta_N\}\setminus\{\eta_s\}$ is the chordal $\SLE_{\kappa}$ connecting $x_{a_s}$ and $x_{b_s}$ in the connected component of the domain $\Omega\setminus\bigcup_{r\neq s}\eta_k$ having the endpoints $x_{a_s}$ and $x_{b_s}$ of $\eta_s$ on its boundary. 
\end{definition}

A Loewner chain \emph{associated to the partition function} $\PartF$ 
starting from $x_i$ (and with launching points $(x_1, \ldots, x_{2N})$)
is the Loewner chain whose driving function $W$ solves the following system of SDEs: 
\begin{align} \label{eq:Loewner_SDE}
\begin{cases}
\ud W_t = \sqrt{\kappa} \, \ud B_t + \kappa(\partial_i\log \PartF)(V_t^{1}, \ldots, V_t^{i-1}, W_t, V_t^{i+1}, \ldots, V_t^{2N}) \, \ud t, \qquad W_0=\varphi(x_i),\\
\ud V_t^j =\frac{2 \, \ud t}{V_t^j-W_t},\qquad 
V_0^j=\varphi(x_j),\quad j\in\{1,\ldots, i-1, i+1, \ldots, 2N\}, 
\end{cases}
\end{align} 
Note that the chordal $\SLE_{\kappa}$ in $\HH$ from $x$ to $y$ is the Loewner chain associated to the partition function $\PartF(x, y) = |x-y|^{-2h(\kappa)}$.
Note also that in~\eqref{eq:Loewner_SDE} it suffices to consider the partition function $\PartF$ only up to a multiplicative constant --- 
which however is canonically fixed by requiring the asymptotics~\eqref{eqn::PPF_ASY}.

\begin{theorem}\label{thm::existenceglobalSLE}
Fix $\kappa\in (0,8)$ and $\alpha\in\LP_N$. 
\begin{enumerate}
\item There exists a unique probability measure on global $N$-$\SLE_{\kappa}$ associated to $\alpha$.

\item Suppose $(\eta_1, \ldots, \eta_N)$ has the law of global $N$-$\SLE_{\kappa}$ associated to $\alpha$. 
Then, for each $s \in \{1,2,\ldots,N\}$, 
the marginal law of $\eta_s$ is a Loewner chain associated to the pure partition function $\PartF_{\alpha}$. 
\end{enumerate}
\end{theorem}

When $\kappa\le 4$, the existence of global $N$-$\SLE_{\kappa}$ is given in~\cite{Kozdron-Lawler:Configurational_measure_on_mutually_avoiding_SLEs} and~\cite[Theorem~1.3]{Peltola-Wu:Global_and_local_multiple_SLEs_and_connection_probabilities_for_level_lines_of_GFF}, 
and the uniqueness was proven in~\cite[Theorem~1.2]{BPW:On_the_uniqueness_of_global_multiple_SLEs} 
(see also~\cite{Sheffield-Miller:Imaginary_geometry2} and~\cite[Appendix~A]{MSW:Non-simple_SLE_curves_are_not_determined_by_their_range} for the case $N=2$). 
The marginal law was also considered in~\cite[Proposition~3.5]{Peltola-Wu:Global_and_local_multiple_SLEs_and_connection_probabilities_for_level_lines_of_GFF}
for $\kappa\in (0,4]$. 
The recent~\cite{AHSY:Conformal_welding_of_quantum_disks_and_multiple_SLE_the_non-simple_case} gives a construction of global $N$-$\SLE_{\kappa}$ for $\kappa\in (4,8)$ by welding of quantum disks, 
\cite{AMY:MultipleSLEfromCLE} gives a construction for the same range using conformal loop ensemble,
and~\cite{Zhan:Existence_and_uniqueness_of_nonsimple_multiple_SLE} concerns the uniqueness for $\kappa\in (4,8)$. 
In Sections~\ref{sec::multipleSLEs}--\ref{sec::PPF}, we obtain a quite different construction for $\kappa\in (4,8)$ and derive the marginal law in Section~\ref{subsec::global_existence}. 
Our results moreover imply the asymptotic fusion rules~\eqref{eqn::PPF_ASY} for $\PartF_{\alpha}$ not established by the other constructions, which are crucial in applications and classification results.

\begin{remark}
In the chordal case, the resampling property uniquely determines the global multiple $\SLE_{\kappa}$ measure with $\kappa\in (0,8)$. 
Such a characterization fails in the radial setting: because the curves can spiral around their common endpoint, any linear combination of multiple radial $\SLE_{\kappa}$ measures with different spiraling rates satisfies the resampling property. 
The spiraling rates yield an additional real parameter. 
See~\cite[Section~1.4]{Wang-Wu:Commutation_relations_for_two-sided_radial_SLE} for two curves, and~\cite{HPW:Multiradial_SLE_with_spiral} and references therein for more general case.
\end{remark}

In the rest of this section, we fix $\kappa\in (4,8)$ and omit the dependence on $\kappa$ from the notation,
while we shall highlight $N \geq 1$. 
In this and the next Section~\ref{sec::PPF}, we make use of two particular Coulomb gas integrals:
$\smash{\coulombGas_{\!\! \vcenter{\hbox{\includegraphics[scale=0.8]{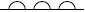}}}_N}}$ and 
$\smash{\coulombGas_{\!\! \vcenter{\hbox{\includegraphics[scale=0.8]{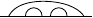}}}_N}}$, 
where 
\begin{align*}
\vcenter{\hbox{\includegraphics[scale=1.2]{figures-arXiv/link61.pdf}}}_N
\quad = \qquad & \{\{1,2\}, \{3,4\}, \ldots, \{2N-1, 2N\}\} \; \in \; \LP_N, \\
\vcenter{\hbox{\includegraphics[scale=1.2]{figures-arXiv/link62.pdf}}}_N
\quad = \qquad & \{\{1,2N\}, \{2,3\}, \{4,5\},\ldots, \{2N-2, 2N-1\}\} \; \in \; \LP_N. 
\end{align*}

\subsection{Conformally invariant measures on curves}
\label{subsec:SLE_preli}

The goal of this section is to construct a certain probability measure on the curve space $X_0(\Omega; x_1, \ldots, x_{2N})$ defined in~\eqref{eq:curve_space}. 
It serves as a tool to construct the global $N$-$\SLE_{\kappa}$ measures of  Theorem~\ref{thm::existenceglobalSLE}. 
To this end, we will first have to introduce some notation and terminology\footnote{As usual, we write $\bs{x} = (x_1, \ldots, x_{2N})$ and $\varphi(\bs{x}) = (\varphi(x_1), \ldots, \varphi(x_{2N}))$.}.

\paragraph*{Conformal invariance.}
The global multiple SLE measures of Theorem~\ref{thm::existenceglobalSLE} are conformally invariant. 

\begin{definition} \label{def: conformally invariant}
A collection $\{\PP(\HH; \bs{x}) \colon \bs{x} \in\chamber_{2N}\}$ of probability measures on $X_0(\HH; \bs{x})$ 
is \emph{M\"obius invariant} 
if the following property holds. 
Suppose $(\eta_1, \ldots, \eta_N)\sim \PP(\HH; \bs{x})$ and $\varphi \colon \HH \to \HH$ is a M\"obius map such that $\varphi(x_1)<\cdots<\varphi(x_{2N})$. 
Then, the law of the collection $(\varphi(\eta_1), \ldots, \varphi(\eta_N))$ is $\PP(\HH; \varphi(\bs{x}))$. 

Whenever the collection $\{\PP(\HH; \bs{x} ) \colon \bs{x} \in\chamber_{2N}\}$ is M\"obius invariant, we may extend its definition for general polygons $(\Omega; \bs{y}) = (\Omega; y_1,\ldots,y_{2N})$ via conformal invariance: 
$\PP(\Omega; \bs{y})$ is defined to be the pushforward law of $(\varphi(\eta_1), \ldots, \varphi(\eta_N))$, where $\varphi \colon \HH \to \Omega$ is any conformal bijection such that $\varphi^{-1}(y_1) < \cdots < \varphi^{-1}(y_{2N})$ and $(\eta_1, \ldots, \eta_N) \sim \PP(\HH; \varphi^{-1}(\bs{y}))$. 
We then say that the collection $\{\PP(\Omega; \bs{y}) \colon (\Omega; \bs{y}) \textnormal{ is a polygon} \}$ of probability measures is \emph{conformally invariant}. 
\end{definition}

\paragraph*{SLE with force points.}
The \emph{$\SLE_{\kappa}(\underline{\rho})$ process} is 
a variant of the chordal $\SLE_{\kappa}$ where one keeps track of additional points on the boundary. We will encounter the following version of it. 
Fix $y_1 < \cdots < y_{2N}$ and $j \in \{1,\ldots,2N\}$. 
Let $\underline{y}^L = (y_1, \ldots, y_{j-1})$ and 
$\underline{y}^R = (y_{j+1}, \ldots, y_{2N})$ and 
$\underline{\rho}^L = (\rho_1,\ldots, \rho_{j-1}) \in \R^{j-1}$ and 
$\underline{\rho}^R = (\rho_{j+1},\ldots, \rho_{2N}) \in \R^{2N-j}$. 
An $\SLE_{\kappa}(\underline{\rho}^L;\underline{\rho}^R)$ process in $(\HH;y_j,\infty)$ 
started at $y_j$ with force points $(\underline{y}^L;\underline{y}^R)$ is the Loewner evolution driven by $W_t$ that solves
the following system of integrated SDEs: 
\begin{align}
\label{eq: SLE_kappa_rho}
\begin{split}
W_t = y_j+\; & \sqrt{\kappa} B_t \, + \, 
\sum_{i=1}^{j-1} 
\int_0^t\frac{\rho_i \, \ud s}{W_s - g_s(y_i)} \, + \, 
\sum_{i=j+1}^{2N}
\int_0^t\frac{\rho_i \, \ud s}{W_s - g_s(y_i)} , \\
g_t(y_i) = \; & y_i \, + \,  \int_0^t\frac{2 \, \ud s}{g_s(y_i) - W_s} ,  \qquad i \in \{1,2,\ldots,2N\} \setminus \{j\} ,
\end{split}
\end{align}
where $B_t$ is the one-dimensional Brownian motion started at $B_0 = 0$. 
Note that the process $g_t(y_i)$ is the time-evolution of the point $y_i$ for each $i \neq j$. 
We define the \emph{continuation threshold} 
of the $\SLE_{\kappa}(\underline{\rho}^L;\underline{\rho}^R)$ 
to be the infimum of the time $t$ for which 
\begin{align*}
\textnormal{either} \quad \sum_{\substack{i \in \{1,\ldots,j-1\} \\ g_t(y_i) = W_t}} \rho_i \le -2 , 
\qquad \textnormal{or} \quad \sum_{\substack{i \in \{j+1,\ldots,2N\} \\ g_t(y_i) = W_t}} \rho_i \le -2 . 
\end{align*}
By~\cite{Sheffield-Miller:Imaginary_geometry1}, 
the $\SLE_{\kappa}(\underline{\rho}^L;\underline{\rho}^R)$
process is well-defined up to the continuation threshold, and it is almost surely
generated by a continuous curve up to and including the continuation threshold.
Concretely, it is the Loewner chain associated to the partition function 
\begin{align*}
\PartF(y_1, \ldots, y_{2N}) 
= \prod_{\substack{ 1 \leq i \leq 2N \\ i \neq j}} (y_i - y_j)^{\rho_i/\kappa} 
\; \prod_{\substack{ 1 \leq i < l \leq 2N \\ i,l \neq j}} (y_l - y_i)^{\rho_l \rho_i / 2\kappa} .
\end{align*}

\paragraph{Poisson kernel.}~(Boundary) Poisson kernel $H(\Omega; x, y)$ is defined for nice Dobrushin domains $(\Omega; x, y)$. When $\Omega=\HH$, we have
\begin{align*}
H(\HH; x,y)=|y-x|^{-2},\quad x,y\in\R. 
\end{align*}
For a nice Dobrushin domain $(\Omega; x, y)$, we define the Poisson kernel via conformal covariance: 
\begin{align*}
H(\Omega; x, y)=|\varphi'(x)\varphi'(y)|H(\HH; \varphi(x), \varphi(y)), 
\end{align*}
where $\varphi$ is any conformal map from $\Omega\to \HH$.
Note that the Poisson kernel satisfies the following monotonicity.
Let $(\Omega; x, y)$ be a nice Dobrushin domain and let $U\subset\Omega$ be a subdomain that agrees with $\Omega$ in a neighborhood of the boundary arc $(yx)$. Then, we have
\begin{align}\label{eqn::Poisson_mono}
H(U; x, y)\le H(\Omega; x, y).
\end{align}

\subsection{Construction of a probability measure $\SLEmeasure_N$}
\label{subsec::QN}

In this section, we will construct probability measures $\SLEmeasure_N = \SLEmeasure_N(\Omega;\bs{x})$ on $N$-tuples of chordal curves $\bs{\eta} := (\eta_1, \ldots, \eta_N) \in X_0(\Omega; \bs{x})$ for every $N\ge 1$ by induction. 
When $N=1$, the measure $\SLEmeasure_1=\SLEmeasure_1(\Omega;x_1,x_2)$ on the Dobrushin domain equals the law of $\SLE_\kappa$ on $\Omega$ from $x_1$ to $x_2$. 
Note that by the reversibility of $\SLE_\kappa$~\cite{Sheffield-Miller:Imaginary_geometry3}, 
we have $\SLEmeasure_1(\Omega;x_1,x_2) = \SLEmeasure_1(\Omega;x_2,x_1)$. 
We use this fact in the proof of Proposition~\ref{prop::Qn_2b_marginal}.
When $N=2$, a similar construction was used in~\cite{FLW:Decomposition_of_global_2SLE_and_application_for_critical_FK_Ising_model} to derive a decomposition of global 2-SLE. As $N$ becomes larger, there are more possibilities for the curves to hit different segments of the boundary. To address all of them, we introduce the following construction which involves heavy notation. Unfortunately, we find it unavoidable. 
We shall highlight the orientation of the curves by denoting 
\begin{align*}
(\vec\eta_1, \ldots, \vec\eta_N) \sim  \vec{\SLEmeasure}_N = \vec{\SLEmeasure}_N(\Omega;\bs{x}) ,
\end{align*}
to emphasize that $\vec\eta_s$ starts from $x_{2s-1}$ for each $1\le s\le N$. 
The construction is given in Definition~\ref{def::QN}.

Given the construction of $\smash{\vec{\SLEmeasure}_N}$, 
we will also define another probability measure $\cev{\SLEmeasure}_N = \cev{\SLEmeasure}_N(\Omega;\bs{x})$ as follows. 
Choose an anti-conformal map $\anticonf$ from $\Omega$ onto $\Omega$ and sample $(\vec\eta_1,\ldots,\vec\eta_N) \sim \vec{\SLEmeasure}_N(\Omega;\anticonf(x_{2N}),\ldots,\anticonf(x_1))$. Then, we define $\cev{\SLEmeasure}_N$ to be the law of $(\cev\eta_1, \ldots,\cev\eta_N) := (\anticonf^{-1}(\vec\eta_N), \ldots,\anticonf^{-1}(\vec\eta_1))$. 
Note that under $\cev{\SLEmeasure}_N$, each curve $\cev\eta_s$ starts from $x_{2s}$ for $1\le s\le N$.
(Corollary~\ref{cor::reversibility} implies that these measures agree.)

Now, fix $N \geq 2$ and suppose we have constructed $\vec{\SLEmeasure}_n$ for all $n \le N-1$. 
We will then construct the measure $\vec{\SLEmeasure}_N$ in this section --- see Definition~\ref{def::QN}. 
Throughout, let $\phi$ be the M\"obius transformation
sending\footnote{Here, it is useful to require that $\phi(x_1) = -\infty$ and $\phi(x_2) = 0$, while the normalization $\phi(x_3) = 1$ could be arbitrary.} 
$x_1 \mapsto -\infty$, and $x_2 \mapsto 0$, and $x_3 \mapsto 1$, 
\begin{align}\label{eqn::Mobius_def}
\phi(z) := \bigg( \frac{z-x_2}{z-x_1}\bigg) \bigg( \frac{x_3-x_1}{x_3-x_2} \bigg) , \qquad z\in\HH .
\end{align}
We first construct $N$ conformally invariant probability measures $\smash{\SLEmeasure^{x_1\to x_{2\fixedindex}}_N}$ on $X_0(\HH; \bs{x})$ with $\bs{x} \in\chamber_{2N}$, indexed by $\fixedindex \in \{1,2,\ldots, N\}$, such that the random curve $\vec\eta_1$ under $(\vec\eta_1, \ldots, \vec\eta_N) \sim \smash{\SLEmeasure^{x_1\to x_{2\fixedindex}}_N}$ ends at $x_{2\fixedindex}$ almost surely (Definition~\ref{def::Qn2b}). 
The probability measure $\vec{\SLEmeasure}_N$ will then be a convex combination of $\{\smash{\SLEmeasure^{x_1\to x_{2\fixedindex}}_N} \colon \fixedindex \in \{1\ldots, N\} \}$ weighted by the probabilities for the curve $\vec\eta_1$ to end at $x_{2},x_{4},\ldots,x_{2N}$.

\begin{definition}\label{def::Qn2b}
Fix $\fixedindex \in \{1,2,\ldots, N\}$. For each polygon $(\Omega; \bs{y})$, 
the probability measure $\smash{\SLEmeasure_N^{y_1\to y_{2\fixedindex}}}(\Omega;\bs{y})$ 
on $X_0(\Omega;\bs{y})$ is constructed uniquely in the following manner. 
Under $(\vec\eta_1, \ldots, \vec\eta_N) \sim \smash{\SLEmeasure_N^{y_1\to y_{2\fixedindex}}}(\Omega;\bs{y})$, the curve $\vec\eta_1$ traverses from $y_1$ to $y_{2\fixedindex}$, and each $\vec\eta_s$ starts from $y_{2s-1}$, for $2\le s\le N$. 
The measure $\smash{\SLEmeasure_N^{y_1\to y_{2\fixedindex}}}(\Omega;\bs{y})$ is the pushforward of 
$\smash{\SLEmeasure^{x_1\to x_{2\fixedindex}}_N}:=\smash{\SLEmeasure^{x_1\to x_{2\fixedindex}}_N}(\HH;\bs{x})$, which is constructed uniquely in the following manner. 

\begin{enumerate}[leftmargin=2em]
\item \label{step: sample random point}
Sample a random point $\bs{\dot U}_\fixedindex = (U_1,\ldots,U_{\fixedindex-1},U_{\fixedindex+1},\ldots, U_N)$ 
such that\footnote{Throughout, we use the convention that if $\fixedindex=1$, we do not sample any point on the left side of $\phi(x_2)$; and if $\fixedindex=N$, we do not sample any point on the right side of $\phi(x_{2N})$.}
\begin{align}\label{eqn::posu}
U_s \in 
\begin{cases}
(\phi(x_{2s-1}),\phi(x_{2s})) , 
& \textnormal{for } 1\le s\le \fixedindex-1 , \\
(\phi(x_{2s}),\phi(x_{2s+1})) , 
& \textnormal{for } \fixedindex+1\le s\le N-1 , \\
(\phi(x_{2N}),+\infty) ,
& \textnormal{for } s = N ,
\end{cases}
\end{align}
whose joint law has the following density with respect to the Lebesgue measure 
(see Figure~\ref{fig::Definition32}):
\begin{align}\label{eqn::densityruj}
\LU(\bs{\dot U}_\fixedindex) 
:= \frac{|f(\phi(\bs{\dot x}_1);\bs{\dot U}_\fixedindex)|}
{\int_{-\infty}^{\phi(x_2)} \ud u_1 \cdots \int_{\phi(x_{2\fixedindex-3})}^{\phi(x_{2\fixedindex-2})} \ud u_{\fixedindex-1}\int_{\phi(x_{2\fixedindex+2})}^{\phi(x_{2\fixedindex+3})} \ud u_{\fixedindex+1} \cdots \int_{\phi(x_{2N})}^{+\infty} \ud u_N \; |f(\phi(\bs{\dot x}_1);\bs{\dot u}_\fixedindex)|} 
\end{align}
where $\bs{\dot{x}}_1 = (x_2, x_3, \ldots, x_{2N}) \in\chamber_{2N-1}$ and 
$\phi(\bs{\dot{x}}_1) = (\phi(x_2),\ldots,\phi(x_{2N})) \in\chamber_{2N-1}$, and
\begin{align} \label{eq: integrand_N-1}
f (\bs{y}; \bs{v}) := \; &
\prod_{1\leq i<j\leq 2N-1}(y_{j}-y_{i})^{2/\kappa} 
\prod_{1\leq r<s\leq N-1}(v_{s}-v_{r})^{8/\kappa} 
\prod_{\substack{1\leq i\leq 2N-1 \\ 1\leq r\leq N-1}}
(v_{r}-y_{i})^{-4/\kappa} .
\end{align} 

\item \label{step: sample SLE kappa rho}
Given $\bs{\dot U}_\fixedindex$, sample a curve $\gamma_\fixedindex$ as an
$\SLE_{\kappa}(\underline{\rho}^L;\underline{\rho}^R)$ process in $(\HH;\phi(x_{2\fixedindex}),\infty)$ with force points 
\begin{align*}
\underline{x}^L = \; & (U_1,\phi(x_2),\phi(x_3),\ldots, U_{\fixedindex-1},\phi(x_{2\fixedindex-2}),\phi(x_{2\fixedindex-1})) , \\
\underline{x}^R = \; & (\phi(x_{2\fixedindex+1}),\phi(x_{2\fixedindex+2}),U_{\fixedindex+1},\ldots, \phi(x_{2N-1}),\phi(x_{2N}),U_N) ,
\end{align*}
and weights
$\underline{\rho}^L = (-4,2,2,\ldots,-4,2,2)$ on its left and 
$\underline{\rho}^R = (2,2,-4,\ldots,2,2,-4)$ on its right 
(so the weights are $-4$ for each $U_s$ and $2$ for each $\phi(x_i)$).
It follows from~\cite[Lemma~15]{Dubedat:Duality_of_SLE} that 
\begin{itemize}
\item the curve $\gamma_\fixedindex$ does not hit the intervals
$\overset{\fixedindex-1}{\underset{s=1}{\bigcup}}[U_s,\phi(x_{2s+1})]$
nor the intervals
$\overset{N}{\underset{s=\fixedindex+1}{\bigcup}}[\phi(x_{2s-1}),U_s]$;

\item the curve $\gamma_\fixedindex$ hits both intervals $(-\infty,U_1)$ and $(U_N,+\infty)$, almost surely.
\end{itemize}

\item \label{step: sample right curves}
Given $\gamma_\fixedindex$, we denote by $\Omega_1^R,\ldots,\Omega_r^R$ the (random) connected components of $\HH\setminus\gamma_\fixedindex$ whose boundaries contain some of the points $\phi(x_i)$ for $2\le i\le 2\fixedindex-1$. 
For each $1\le m\le r$, let $\smash{\phi(x_{2s^R_m}),\ldots,\phi(x_{2s^R_{m+1}-1})}$ 
be the boundary points contained in $\partial \Omega_m^R$.
Using the induction hypothesis that we have already constructed $\cev{\SLEmeasure}_n$ for all $n \le N-1$,
sample independently for each $1\le m\le l$ the curves 
\begin{align*}
\bs\gamma_m^R := 
\big( \gamma_{s^R_m},\ldots,\gamma_{s^R_{m+1}-1} \big) \sim \cev{\SLEmeasure}_{s^R_{m+1}-s^R_{m}} \big( \Omega_m^R; \phi(x_{2s^R_m}),\ldots, \phi(x_{2s^R_{m+1}-1}) \big) ,
\end{align*}
where each $\gamma_s$ is a curve starting from $\phi(x_{2s-1})$ for $2\le s\le \fixedindex$. 
See also Figure~\ref{fig::Definition32}.

\item \label{step: sample left curves}
Given $\gamma_\fixedindex$, we denote by $\Omega_1^L,\ldots,\Omega_l^L$ the (random) connected components of $\HH\setminus\gamma_\fixedindex$ whose boundaries contain some of the points $\phi(x_i)$ for $2\fixedindex+1\le i\le 2N$. 
For each $1\le m\le l$, let $\smash{\phi(x_{2s^L_m+1}),\ldots, \phi(x_{2s^L_{m+1}})}$ 
be the boundary points contained in $\partial \Omega_m^L$.
Using the induction hypothesis that we have already constructed $\vec{\SLEmeasure}_n$ for all $n \le N-1$,
sample independently for each $1\le m\le l$ the curves 
\begin{align*}
\bs\gamma_m^L := 
\big( \gamma_{s^L_m+1},\ldots,\gamma_{s^L_{m+1}} \big) \sim \vec{\SLEmeasure}_{s^L_{m+1}-s^L_{m}} \big( \Omega_m^L; \phi(x_{2s^L_m+1}),\ldots, \phi(x_{2s^L_{m+1}}) \big) ,
\end{align*}
where each $\gamma_s$ is a curve starting from $\phi(x_{2s-1})$ for $\fixedindex+1\le s\le N$. 
See also Figure~\ref{fig::Definition32}.

\item \label{step: finish}
Denote by $\cev\gamma_\fixedindex$ the time-reversal of $\gamma_\fixedindex$. 
Set $\vec{\SLEmeasure}_N^{x_1\to x_{2\fixedindex}}$ to be the law of $(\vec\eta_1, \ldots, \vec\eta_N)$,
where\footnote{Importantly, only $\gamma_\fixedindex$ is time-reversed in this construction.} 
\begin{align*}
\vec\eta_1 := \phi^{-1}(\cev\gamma_\fixedindex) ,
\quad 
\vec\eta_{s+1} := \phi^{-1}(\gamma_{s}) , \quad \textnormal{for } 1\le s\le \fixedindex-1 ,
\quad \textnormal{and} \quad 
\vec\eta_{s} := \phi^{-1}(\gamma_{s}) , \quad \textnormal{for } \fixedindex+1\le s\le N ,
\end{align*}
Note that under the measure $\SLEmeasure_N^{x_1\to x_{2\fixedindex}}$, 
the curve $\vec\eta_1$ traverses from $x_1$ to $x_{2\fixedindex}$, 
and each $\vec\eta_s$ starts from $x_{2s-1}$, for $2\le s\le N$.

\item 
Lastly, by conformal invariance of $\SLE_\kappa$, the collection $\{\smash{\SLEmeasure^{x_1\to x_{2\fixedindex}}_N}(\HH; \bs{x}) \colon \bs{x}\in\chamber_{2n}\}$ is M\"obius invariant.
Hence, it uniquely determines the probability measures $\smash{\SLEmeasure_N^{y_1\to y_{2\fixedindex}}}(\Omega;\bs{y})$ for all polygons $(\Omega; \bs{y})$.
\end{enumerate}
\end{definition}

\begin{figure}[ht!]
\begin{subfigure}[b]{\textwidth}
\begin{center}
\includegraphics[width=0.9\textwidth]{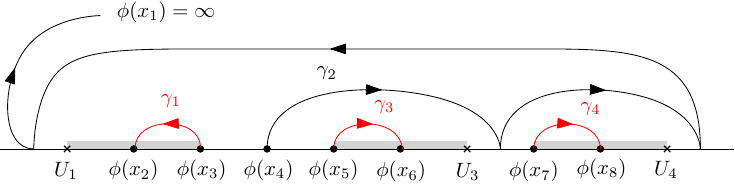}
\end{center}
\caption{
We first sample points $(U_1, U_3, U_4)$ according to~\eqref{eqn::densityruj}. 
Second, given $(U_1, U_3, U_4)$, we sample $\gamma_2$ from $\phi(x_4)$ to $\infty$ as $\SLE_{\kappa}(-4, 2, 2; 2, 2, -4, 2, 2, -4)$ with force points $(U_1, \phi(x_2), \phi(x_3); \phi(x_5), \phi(x_6), U_3, \phi(x_7), \phi(x_8), U_4)$. 
Note that the curve $\gamma_2$ does not hit the intervals marked in grey, because the accumulated force point weight on these intervals exceeds $2\ge\kappa/2-2$ (which guarantees that $\gamma_2$ does not hit these intervals). 
Third, we sample the curves $\gamma_1, \gamma_3, \gamma_4$ in connected components of $\HH\setminus\gamma_2$ using the induction hypothesis. }
\end{subfigure}\\
\vspace{0.5cm}
\begin{subfigure}[b]{\textwidth}
\begin{center}
\includegraphics[width=0.9\textwidth]{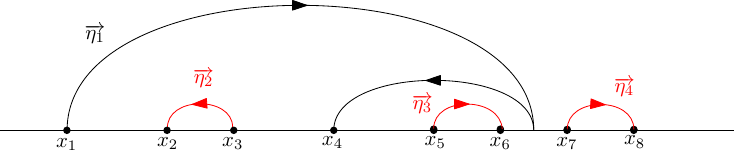}
\end{center}
\caption{
Denote by $\cev\gamma_2$ the time-reversal of $\gamma_2$. 
We set $\vec\eta_1 := \phi^{-1}(\cev\gamma_2)$, which traverses from $x_1$ to $x_4$. We set $\vec\eta_2 := \phi^{-1}(\gamma_1)$, which starts from $x_3$, and $\vec\eta_3 := \phi^{-1}(\gamma_3)$, which starts from $x_5$, and $\vec\eta_4 := \phi^{-1}(\gamma_4)$, which starts from $x_7$.}
\end{subfigure}

\caption{\label{fig::Definition32} Example with $N=4$ and $\fixedindex=2$ in Definition~\ref{def::Qn2b}.}
\end{figure}

We will next gather properties of the measure $\smash{\SLEmeasure_N^{y_1\to y_{2\fixedindex}}}(\Omega;\bs{y})$, before we continue the construction of $\vec{\SLEmeasure}_N$.
At this point, we know that the construction of the probability measures 
$\SLEmeasure_n^{y_1\to y_{2\fixedindex}}(\Omega;\bs{y})$ is complete for all $n \leq N$,
and the construction of the probability measures $\vec{\SLEmeasure}_n(\Omega;\bs{y})$ 
is complete for all $n \leq N-1$. 

\begin{proposition}\label{prop::Qn_2b_marginal}
Fix $\fixedindex \in \{1,2,\ldots, N\}$. 
The marginal law of $\vec\eta_1$ under $\smash{\SLEmeasure^{x_1\to x_{2\fixedindex}}_N}$ can be described as follows. 
Let $\eta$ be the chordal $\SLE_\kappa$ on $(\HH;x_1,x_{2\fixedindex})$, and define the event
\begin{align*}
A_\fixedindex(\eta) := \big\{\eta \cap (x_{2s},x_{2s+1}) = \emptyset \textnormal{ for all } 1\le s \le \fixedindex-1 \big\} \cap 
\big\{\eta \cap (x_{2s-1},x_{2s}) = \emptyset \textnormal{ for all } \fixedindex+1\le s \le N \big\}.
\end{align*}
On the event $A_\fixedindex(\eta)$, denote by $H_1,\ldots,H_\np$ the connected components of $\HH\setminus\eta$ whose boundaries contain points in $\{x_1,\ldots,x_{2N}\}\setminus\{x_1,x_{2\fixedindex}\}$. 
For each $1\le m\le \np$, denote by $x_{j_m},\ldots,x_{j_{m+1}-1}$ the points contained in $\partial H_m$. 
Then, using the notation~\eqref{eqn::bs_notation}, 
the marginal law of $\vec\eta_1$ under $\smash{\SLEmeasure^{x_1\to x_{2\fixedindex}}_N}$  is the same as the law of $\eta$ weighted by
\begin{align}\label{eqn::RNderigeneral}
\frac{(x_{2\fixedindex}-x_1)^{-2h(\kappa)}}{\hat{\LR}_{\omega_{\fixedindex}}(\bs{x})} \; \one\{ A_\fixedindex(\eta) \} \; 
\prod_{m=1}^{\np} \frac{1}{\fugacity(\kappa)} \, \coulombGas_{\!\! \vcenter{\hbox{\includegraphics[scale=0.8]{figures-arXiv/link62.pdf}}}_{\frac{j_{m+1}-j_m}{2}}}(H_m;x_{j_m},\ldots, x_{j_{m+1}-1})
\quad > \quad 0 ,
\end{align}
where 
$\hat{\LR}_{\omega_\fixedindex}$ are the functions appearing in Lemma~\ref{lem::another_CGI_with_conjugate_charge}.
\end{proposition}

We will prove Proposition~\ref{prop::Qn_2b_marginal} in Section~\ref{subsec: proof of marginal prop}.
Note that $\coulombGas_{\!\! \vcenter{\hbox{\includegraphics[scale=0.8]{figures-arXiv/link62.pdf}}}_{m}} > 0$ for all $m \in \bZpos$ by Lemma~\ref{lem::RN_equivalent_def2} 
and $\hat{\LR}_{\omega_\fixedindex} > 0$ by Lemma~\ref{lem::another_CGI_with_conjugate_charge}.

\begin{lemma}\label{lem::Qn_2b_positive}
Fix $\fixedindex \in \{1,2,\ldots, N\}$. 
Denote by $\conn(\vec{\bs\eta})$ the random link pattern in $\LP_N$ formed by the curves $\vec{\bs\eta} := (\vec\eta_1,\ldots,\vec\eta_N) \sim \smash{\SLEmeasure^{x_1\to x_{2\fixedindex}}_N}$. 
For any $\alpha\in \LP_N$ such that $\{1,2\fixedindex\}\in\alpha$, we have
\begin{align}\label{eq:Qn_2b_positive}
\SLEmeasure_N^{x_1\to x_{2\fixedindex}} [\conn(\vec{\bs\eta}) = \alpha] > 0. 
\end{align}
\end{lemma}

\begin{proof}
We prove the claim~\eqref{eq:Qn_2b_positive} by induction on the number of curves. 
The base case is obvious\footnote{In this case, $\SLEmeasure_1^{x_1\to x_2}$ equals the law of $\SLE_\kappa$ from $x_1$ to $x_2$.}.
We shall perform the induction step at the same time as we prove Lemma~\ref{lem::QN_positive},  
using the following assumptions.
\begin{itemize}[leftmargin=2em]
\item The construction of the measures 
$\SLEmeasure_n^{y_1\to y_{2\fixedindex}}(\Omega;\bs{y})$ is complete and 
Proposition~\ref{prop::Qn_2b_marginal} holds for all $n \leq N$.

\item (Induction hypothesis): 
For all $n \leq N-1$, we have 
$\SLEmeasure_n^{x_1\to x_{2\fixedindex}} [\conn(\vec{\bs\eta}) = \alpha] > 0$ for all $\fixedindex \in \{1,\ldots, n\}$ 
and for all $\alpha\in \LP_n$ such that $\{1,2\fixedindex\}\in\alpha$. 

\item The construction of the measures 
$\vec{\SLEmeasure}_n(\Omega;\bs{y})$ is complete and Lemma~\ref{lem::QN_positive} holds for all $n \leq N-1$.
\end{itemize}
We will complete the proof in the course of the proof of Lemma~\ref{lem::QN_positive}.
\end{proof}

Before finishing the proofs of Proposition~\ref{prop::Qn_2b_marginal} and Lemma~\ref{lem::Qn_2b_positive}, we continue the construction of $\vec{\SLEmeasure}_N$. 
For $\alpha\in\LP_N$ such that $\{1,2\fixedindex\}\in\alpha$, we define $\alpha^R \in \LP_{\fixedindex-1}$ to be the link pattern induced by the links of $\alpha$ on $\{2,3,\ldots,2\fixedindex-1\}$ nested by the link $\{1,2\fixedindex\}$.  
Using Lemma~\ref{lem::Qn_2b_positive}, we set
\begin{align}\label{eqn::palpha}
p_\alpha(\bs x) 
:= \frac{ \hat{\LR}_{\omega_\fixedindex}(\bs x) \, w_\fixedindex(\alpha) \, \SLEmeasure_N^{x_1\to x_{2\fixedindex}} [\conn(\vec{\bs\eta}) = \alpha] }{
\underset{1 \leq s \leq N}{\sum} \hat{\LR}_{\omega_s}(\bs x)
\underset{\substack{\beta\in\LP_N \\ \{1,2s\}\in\beta}}{\sum} w_s(\beta) \, \SLEmeasure_N^{x_1\to x_{2s}} [\conn(\vec{\bs\eta}) = \beta ] } , 
\qquad \bs{x}\in\chamber_{2N} ,
\end{align}
where $\hat{\LR}_{\omega_s}$ are the functions appearing in Lemma~\ref{lem::another_CGI_with_conjugate_charge}, and 
\begin{align}\label{eqn::defqj}
w_s(\alpha):= 
\begin{cases}
1 , & s=1 , \\
\fugacity(\kappa) \dfrac{\meanderMatrix(\vcenter{\hbox{\includegraphics[scale=0.8]{figures-arXiv/link61.pdf}}}_{s-1},\alpha^R)}{\meanderMatrix(\vcenter{\hbox{\includegraphics[scale=0.8]{figures-arXiv/link62.pdf}}}_{s-1},\alpha^R)} , 
& 2\le s\le N. \\
\end{cases}
\end{align}
The reason for the choice of $w_s(\alpha)$ will become clear in Lemma~\ref{lem::meander}. 
(Note that the proof of Lemma~\ref{lem::Qn_2b_positive} only relies on the induction hypothesis which implies that $p_\alpha$ is well-defined for $\alpha\in\LP_n$ with $n \leq N-1$.
The induction then shows that $p_\alpha$ is well-defined for $\alpha\in\LP_N$.)

\begin{lemma} \label{lem:p_Mobius}
The function $p_\alpha$ defined in~\eqref{eqn::palpha} is M\"obius invariant, that is, 
$p_\alpha \circ \varphi = p_\alpha$ 
for all M\"obius maps $\varphi$ of the upper half-plane $\HH$ 
such that $\varphi(x_{1}) < \cdots < \varphi(x_{2N})$.
\end{lemma}

\begin{proof}
The functions $\LR_{\omega_\fixedindex}$ are M\"obius covariant by Lemma~\ref{lem:PDECOVASY_for_LR} 
and their covariance factors cancel out in the ratio~\eqref{eqn::palpha}.
The other terms in~\eqref{eqn::palpha} are M\"obius invariant.
\end{proof}

\begin{definition}\label{def::QN}
Denote by $\conn(\vec{\bs\eta})$ the random link pattern in $\LP_N$ formed by the curves $\vec{\bs\eta}$. 
Set 
\begin{align}\label{eqn::defQN}
\vec{\SLEmeasure}_N[\, \cdot \,] 
\; = \; \vec{\SLEmeasure}_N(\HH;\bs x)[\, \cdot \,] 
\; := \; \sum^N_{\fixedindex=1} \sum_{\substack{\alpha\in\LP_N \\ \{1,2\fixedindex\}\in\alpha}} 
p_{\alpha}(\bs x) \, \SLEmeasure_N^{x_1\to x_{2\fixedindex}} [ \; \cdot \, \cond \conn(\vec{\bs\eta}) = \alpha],
\end{align}
where the weight $p_{\alpha}(\bs{x})$ is defined in~\eqref{eqn::palpha}, and 
$\smash{\SLEmeasure^{x_1\to x_{2\fixedindex}}_N}$ is the probability measure from Definition~\ref{def::Qn2b}.
Since $p_\alpha$ is M\"obius invariant by Lemma~\ref{lem::Qn_2b_positive},
the collection $\{\vec{\SLEmeasure}_N(\HH; \bs{x}) \colon \bs{x}\in\chamber_{2N}\}$ is M\"obius invariant. 
Hence, it determines a conformally invariant collection $\{\vec{\SLEmeasure}_N \colon (\Omega; \bs{y}) \textnormal{ is a polygon} \}$ of probability measures.
\end{definition}

\begin{lemma}\label{lem::QN_positive}
Denote by $\conn(\vec{\bs\eta})$ the random link pattern in $\LP_N$ formed by the curves $\vec{\bs\eta} \sim \vec{\SLEmeasure}_N$. 
Then, 
\begin{align}\label{eq:QN_positive}
\vec{\SLEmeasure}_N [\conn(\vec{\bs\eta}) = \alpha] > 0 \qquad \textnormal{for any } \alpha\in \LP_N .
\end{align}
\end{lemma}

During the proof of Lemma~\ref{lem::QN_positive} we also complete the proof of Lemma~\ref{lem::Qn_2b_positive}.
\begin{proof}
We prove the claim~\eqref{eq:QN_positive} by induction on the number of curves. 
The base case is obvious\footnote{In this case, $\vec{\SLEmeasure}_1^{x_1\to x_2}$ equals the law of $\SLE_\kappa$ from $x_1$ to $x_2$.}.
For the induction step, we use the following assumptions.
\begin{itemize}[leftmargin=2em]
\item The construction of the measures 
$\SLEmeasure_n^{y_1\to y_{2\fixedindex}}(\Omega;\bs{y})$ is complete and 
Proposition~\ref{prop::Qn_2b_marginal} holds for all $n \leq N$.

\item The construction of the measures 
$\vec{\SLEmeasure}_n(\Omega;\bs{y})$ is complete for all $n \leq N-1$.

\item (Induction hypothesis): For all $n \leq N-1$, the following hold:
\begin{itemize}[leftmargin=5em]
\item[Lemma~\ref{lem::Qn_2b_positive}:]
$\SLEmeasure_n^{x_1\to x_{2\fixedindex}} [\conn(\vec{\bs\eta}) = \alpha] > 0$ for all $\fixedindex \in \{1,\ldots, n\}$ 
and for all $\alpha\in \LP_n$ such that $\{1,2\fixedindex\}\in\alpha$;

\item[Lemma~\ref{lem::QN_positive}:]
$\vec{\SLEmeasure}_n [\conn(\vec{\bs\eta}) = \alpha] > 0$ for all $\alpha\in \LP_n$.
\end{itemize}
\end{itemize}
Consider the event 
$E := \{\vec\eta_1 \cap (x_2,x_{2\fixedindex-1}) = \emptyset\} \cap \{\vec\eta_1 \cap (x_{2\fixedindex+1},x_{2n}) = \emptyset\}$. 
Note that $\SLEmeasure_n^{x_1\to x_{2\fixedindex}}[E] > 0$, for otherwise, 
Proposition~\ref{prop::Qn_2b_marginal} with $N=3$ and the marked boundary points $(x_1,x_2,x_{2\fixedindex-1},x_{2\fixedindex},x_{2\fixedindex+1},x_{2n})$ would imply that the expectation of~\eqref{eqn::RNderigeneral} is zero. 
For $\alpha\in\LP_n$ such that $\{1,2\fixedindex\}\in\alpha$, we define $\alpha^L \in \LP_{n-\fixedindex}$ to be the link pattern induced by the links of $\alpha$ on $\{2\fixedindex+1,\ldots,2n\}$ outside of the link $\{1,2\fixedindex\}$. 
Write $\vec{\bs\eta}^R := (\vec\eta_2,\ldots,\vec\eta_{\fixedindex-1})$ 
and $\vec{\bs\eta}^L := (\vec\eta_{\fixedindex+1},\ldots,\vec\eta_{n})$. 
By the construction in Definitions~\ref{def::Qn2b}~\&~\ref{def::QN}, 
conditional on $\vec\eta_1$ (and on the event $E$), under 
$\SLEmeasure_n^{x_1\to x_{2\fixedindex}}$ 
the law of $\vec{\bs\eta}^R$ 
equals $\cev{\SLEmeasure}_n(\Omega^R;x_{2},\ldots,x_{2\fixedindex-1})$, 
and the law of $\vec{\bs\eta}^L$ 
equals $\vec{\SLEmeasure}_n(\Omega^L;x_{2\fixedindex+1},\ldots,x_{2n})$,
where $\Omega^R$ (resp.~$\Omega^L$) is the connected component of $\HH \setminus \vec\eta_1$ containing the relevant curves $\vec{\bs\eta}^R$ 
(resp.~$\vec{\bs\eta}^L$).
Hence, we have
\begin{align*}
\; & \SLEmeasure_n^{x_1\to x_{2\fixedindex}}[\conn(\vec{\bs\eta}) = \alpha] \\
\ge \; & \SLEmeasure_n^{x_1\to x_{2\fixedindex}} \Big[ \SLEmeasure_n^{x_1\to x_{2\fixedindex}}
 \big[ \conn(\vec{\bs\eta}^R) = \alpha^R 
 \textnormal{ and } \conn(\vec{\bs\eta}^L) = \alpha^L
 \cond \vec\eta_1
 \big] \, \one\{E\}
\Big]
\\ 
= \; & \SLEmeasure_n^{x_1\to x_{2\fixedindex}} \Big[ \SLEmeasure_{\fixedindex-1}^{x_1\to x_{2\fixedindex}}
 \big[ \conn(\vec{\bs\eta}^R) = \alpha^R
 \cond \vec\eta_1
 \big] 
\; \SLEmeasure_{n-\fixedindex}^{x_1\to x_{2\fixedindex}}
 \big[ \conn(\vec{\bs\eta}^L) = \alpha^L
 \cond \vec\eta_1
 \big] \, \one\{E\}
\Big]
\\ 
= \; & \SLEmeasure_n^{x_1\to x_{2\fixedindex}} \Big[ \cev{\SLEmeasure}_{\fixedindex-1}(\Omega^R;x_{2},\ldots,x_{2\fixedindex-1})
 \big[ \conn(\bs\gamma^R) = \alpha^R
 \big] 
\; \vec{\SLEmeasure}_{n-\fixedindex}(\Omega^L;x_{2\fixedindex+1},\ldots,x_{2n})
 \big[ \conn(\bs\gamma^L) = \alpha^L
 \big] \, \one\{E\}
\Big]
\\
> \; & 0 ,
\end{align*}
where
$\bs\gamma^R := (\gamma_2,\ldots,\gamma_{\fixedindex-1})$ 
and $\bs\gamma^L := (\gamma_{\fixedindex+1},\ldots,\gamma_{n})$. 
This proves Lemma~\ref{lem::Qn_2b_positive}. 
To complete the induction step for Lemma~\ref{lem::QN_positive}, it just remains to note that
$\vec{\SLEmeasure}_n[\conn(\vec{\bs\eta}) = \alpha] = p_{\alpha}(\bs{x}) > 0$,
by Lemma~\ref{lem::Qn_2b_positive} for $N=n$. 
\end{proof}

\subsection{Marginal law: proof of Proposition~\ref{prop::Qn_2b_marginal}}
\label{subsec: proof of marginal prop}

In this section, we derive the marginal law of $\vec\eta_1$. 
For simplicity of notation, let us assume that $\fixedindex=1$; the general case can be proved similarly.
Recall the definition of $\phi$ in~\eqref{eqn::Mobius_def}, with 
$\phi(x_1) = \infty$, and $\phi(x_2) = 0$, and $\phi(x_3) = 1$. 
Recall that in Item~\ref{step: finish} in the construction of $\smash{\SLEmeasure^{x_1\to x_2}_N}$, the curve $\vec\eta_1 := \phi^{-1}(\cev\gamma_1)$ is defined via the time-reversal of the 
$\SLE_{\kappa}(\underline{\rho}^R)$ process $\gamma_1$ in $(\HH;0,\infty)$. 
Hence, it suffices to derive the marginal law of $\cev\gamma_1$. 
To this end, we will compare the law of $\gamma_1$ to the law of the chordal $\SLE_{\kappa}$ curve $\eta$ in $(\HH;0,\infty)$.

\begin{lemma}\label{lem::mart1general}
Let $\eta$ be the chordal $\SLE_\kappa$ in $(\HH;0,\infty)$, with associated conformal maps 
$g_t \colon \HH\setminus K_t \to \HH$ and driving function $W_t = \sqrt{\kappa} B_t$ given by the standard one-dimensional Brownian motion. 
Fix points $\bs{\dot u}_1 = (u_2,\ldots,u_N)$ satisfying~\eqref{eqn::posu}. 
Let $\tau$ be the hitting time of $\eta$ at 
\begin{align*}
I^{\bs{\dot u}_1} := \bigcup^N_{s=2} (\phi(x_{2s-1}), u_s) ,
\end{align*}
and let $T$ be the hitting time of $\eta$ at $(\phi(x_3),+\infty) \setminus \smash{I^{\bs{\dot u}_1}}$. 
Next, define
\begin{align}\label{eqn::mart1}
\begin{split}
M_t^{\bs{\dot u}_1} = M_t := \; &
\one \{t<T\wedge\tau\} \, \prod_{j=3}^{2N} g_t'(\phi(x_j))^{h} \; 
\prod_{s=2}^{N} g_t'(u_s)  \\
\; & \times |f(W_t,g_t(\phi(x_3)),g_t(\phi(x_4)),\ldots,g_t(\phi(x_{2N}));g_t(u_2),\ldots,g_t(u_N))| ,
\end{split}
\end{align}
where $f$ is the function~\eqref{eq: integrand_N-1}. 
Then, the following two properties hold.
\begin{itemize}[leftmargin=2em]
\item
The process $(M_t)_{t\ge 0}$ is a uniformly integrable martingale for $\eta$. 
Moreover:
\begin{enumerate}
\item 
On the event $\{T<\tau\}$, 
denote by $\Omega_\eta$ the bounded connected component of $\HH\setminus \eta[0,T]$ 
whose boundary contains $\phi(x_3)$.  
Suppose $\eta(T) \in (u_p,\phi(x_{2p+1}))$ for some \textnormal{(}random\textnormal{)} $p \in \{2,3,\ldots,N\}$.
Choose any conformal map $\varphi_\eta$ from $\Omega_\eta$ onto $\HH$ such that $\varphi_\eta(\eta(T))=\infty$. 
Define 
\begin{align*}
S_T := \; & \prod_{j=3}^{2p} |\varphi_\eta'(\phi(x_j))|^{h} \; 
\prod_{s=2}^{p} |\varphi_\eta'(u_s)|
 \\
\; & \times \big| f(\varphi_\eta(\phi(x_3)),\varphi_\eta(\phi(x_4)),\ldots,\varphi_\eta(\phi(x_{2p}));\varphi_\eta(u_2),\ldots,\varphi_\eta(u_p)) \big| , \\
V_T := \; & \prod_{j=2p+1}^{2N} g_T'(\phi(x_j))^{h} \; 
\prod_{s=p+1}^{N} g_T'(u_s) \\
\; & \times \big| f(W_T,g_T(\phi(x_{2p+1})),\ldots,g_T(\phi(x_{2N}));g_T(u_{p+1}),\ldots,g_T(u_N)) \big| .
\end{align*}
Then, we have
\begin{align}\label{eqn::mart1limitsec4}
M_{T} := \lim_{t\to T} M_t = \one \{T<\tau\} \, S_T \, V_T.
\end{align}
\item
On the event $\{\tau< T\}$, 
we have
\begin{align}\label{eqn::mart1limitsec5}
\lim_{t\to \tau}M_t=0.
\end{align}
\end{enumerate}
\item
Let $\gamma^{\bs{\dot u}_1}$ be an
$\SLE_{\kappa}(\underline{\rho}^R)$ process in $(\HH;0,\infty)$ with all force points on its right side, 
\begin{align*}
\underline{x}^R = \; & (\phi(x_{3}),\phi(x_{4}),u_{2},\ldots, \phi(x_{2N-1}),\phi(x_{2N}),u_N)
\end{align*}
and with weights $\underline{\rho}^R = (2,2,-4,\ldots,2,2,-4)$
\textnormal{(}so the weights are $-4$ for each $u_s$ and $2$ for each $\phi(x_i)$\textnormal{)}.
Then, the law of $\eta$ tilted by $(M_{t\wedge T\wedge\tau})_{t\ge 0}$ is the same as the law of $\gamma^{\bs{\dot u}_1}$.
\end{itemize}
\end{lemma}

\begin{proof}
We know from~\cite[Theorem~6]{Schramm-Wilson:SLE_coordinate_changes}, $M_t$ in~\eqref{eqn::mart1} is a local martingale for $\eta$ and the law of $\eta$ tilted by $(M_{t})_{t\ge 0}$ 
is the same as the law of $\gamma^{\bs{\dot u}_1}$ up to time $T\wedge\tau$. 
It remains to show that $(M_{t\wedge T\wedge\tau})_{t\ge 0}$ is uniformly integrable 
and to derive the terminal values~\eqref{eqn::mart1limitsec4} and~\eqref{eqn::mart1limitsec5}. 

\smallbreak

{\bf Terminal value~\eqref{eqn::mart1limitsec4}.}
Let us first assume that $T<\tau$ and it hits the set $(\phi(x_3),+\infty) \setminus \smash{I^{\bs{\dot u}_1}}$ at the point 
$\eta(T) \in (u_p,\phi(x_{2p+1}))$ for $p \in \{2,3,\ldots,N\}$.
Note that on the event $\{T<\tau\}$, the curve $\eta$ has a positive distance from $\smash{I^{\bs{\dot u}_1}}$ almost surely.  
Write $M_t=M_{1,t}\, M_{2,t}\, M_{3,t}\, M_{4,t}$, where
\begin{align*}
M_{1,t} := \; & 
\prod_{j=3}^{2p} g_t'(\phi(x_j))^{h} \; 
\prod_{s=2}^{p} g_t'(u_s) 
\times |f(g_t(\phi(x_3)),g_t(\phi(x_4)),\ldots,g_t(\phi(x_{2p}));g_t(u_2),\ldots,g_t(u_p))| , \\
M_{2,t} := \; & 
\prod_{j=2p+1}^{2N} g_t'(\phi(x_j))^{h} \; 
\prod_{s=p+1}^{N} g_t'(u_s) 
\times |f(W_t,g_t(\phi(x_{2p+1})),\ldots,g_t(\phi(x_{2N}));g_t(u_{p+1}),\ldots,g_t(u_N))| , \\
M_{3,t} := \; & 
\prod_{s=2}^p \prod_{j=2p+1}^{2N}
\left(\frac{\left(g_t(\phi(x_j))-g_t(u_s)\right)^2}{(g_t(\phi(x_{j}))-g_t(\phi(x_{2s-1})))(g_t(\phi(x_{j}))-g_t(\phi(x_{2s})))}\right)^{-2/\kappa} \\
& \; \times 
\prod_{s=2}^p \prod_{q=p+1}^N
\left(\frac{\left(g_t(u_q)-g_t(u_s)\right)^2}{(g_t(u_q)-g_t(\phi(x_{2s-1})))(g_t(u_q)-g_t(\phi(x_{2s})))}\right)^{4/\kappa},\\
M_{4,t} := \; & 
\prod_{s=2}^p 
\left(\frac{\left(g_t(\phi(x_{2s-1}))-W_t\right)\left(g_t(\phi(x_{2s}))-W_t\right)}{\left(g_t(u_s)-W_t\right)^2}\right)^{2/\kappa}.
\end{align*}
We derive the terminal value~\eqref{eqn::mart1limitsec4} by considering the limits of $M_{1,t},\ldots, M_{4,t}$ one by one. 
\begin{enumerate}[leftmargin=2em]
\item \label{item::M1t}
We write $M_{1,t}$ in the following form:
\begin{align}\label{eqn::mart_decomposition1}
M_{1,t} = \; &
\prod_{3\le i < j\le 2p} \left(\frac{g_t'(\phi(x_j)) \, g_t'(\phi(x_i))}{(g_t(\phi(x_j))-g_t(\phi(x_i)))^2}\right)^{-1/\kappa}
\times \prod_{\substack{3 \le i\le 2p \\ 2 \le s \le p}}\left(\frac{g_t'(\phi(x_i)) \, g_t'(u_s)}{(g_t(u_s)-g_t(\phi(x_i)))^2}\right)^{2/\kappa}
\\
\nonumber
& \; \times 
\prod_{2\le s < q \le p} \left(\frac{g_t'(\phi(u_q)) \, g_t'(\phi(u_s))}{(g_t(\phi(u_q))-g_t(\phi(u_s)))^2}\right)^{-4/\kappa}
\times 
\prod_{2\le s\le p}\left(\frac{g_t'(u_s)^2}{g_t'(\phi(x_{2s-1})) \, g_t'(\phi(x_{2s}))}\right)^{(\kappa-4)/2\kappa}.
\end{align}
Note that the quantity 
\begin{align*}
\frac{g_t'(\phi(x_j)) \, g_t'(\phi(x_i))}{(g_t(\phi(x_j))-g_t(\phi(x_i)))^2}
\end{align*} 
equals the Poisson kernel for the connected component of $\HH\setminus\eta[0,t]$ whose boundary contains $\phi(x_j)$ and $\phi(x_i)$. 
As $t\to T$, it converges to the Poisson kernel $H(\Omega_\eta;\phi(x_j),\phi(x_i))$, so
\begin{align*}
\prod_{3\le i < j\le 2p} \left(\frac{g_t'(\phi(x_j)) \, g_t'(\phi(x_i))}{(g_t(\phi(x_j))-g_t(\phi(x_i)))^2}\right)^{-1/\kappa}
\; & \quad \overset{t \to T}{\longrightarrow} \quad 
\prod_{3\le i < j\le 2p} \left| \frac{\varphi_\eta'(\phi(x_j)) \, \varphi_\eta'(\phi(x_i))}{(\varphi_\eta(\phi(x_j))-\varphi_\eta(\phi(x_i)))^2} \right|^{-1/\kappa} ,
\end{align*}
and similarly, 
\begin{align*}
\prod_{\substack{3 \le i\le 2p \\ 2 \le s \le p}} \left(\frac{g_t'(\phi(x_i)) \, g_t'(u_s)}{(g_t(u_s)-g_t(\phi(x_i)))^2}\right)^{2/\kappa}
\; & \quad \overset{t \to T}{\longrightarrow} \quad 
\prod_{\substack{3 \le i\le 2p \\ 2 \le s \le p}} \left| \frac{\varphi_\eta'(\phi(x_i)) \, \varphi_\eta'(u_s)}{(\varphi_\eta(u_s)-\varphi_\eta(\phi(x_i)))^2} \right|^{2/\kappa} , \\
\prod_{2\le s < q \le p} \left(\frac{g_t'(\phi(u_q)) \, g_t'(\phi(u_s))}{(g_t(\phi(u_q))-g_t(\phi(u_s)))^2}\right)^{-4/\kappa}
\; & \quad \overset{t \to T}{\longrightarrow} \quad 
\prod_{2\le s < q \le p} \left| \frac{\varphi_\eta'(\phi(u_q)) \, \varphi_\eta'(\phi(u_s))}{(\varphi_\eta(\phi(u_q))-\varphi_\eta(\phi(u_s)))^2} \right|^{-4/\kappa}
\end{align*}
Next, recall that $g_t$ is the conformal map from the unbounded component of $\HH\setminus\eta[0,t]$ onto $\HH$ normalized at $\infty$. 
Let $\xi_t$ to be the conformal map from 
the unbounded component of $\HH\setminus\eta[0,t]$ onto $\HH$ such that $\xi_t(\phi(x_3)) = 0$, 
and $\xi_t(\phi(x_4)) = 1$, and $\xi_t(\eta(T))=\infty$,
\begin{align*}
\xi_t(z) := \Big( \frac{g_t(z) - g_t(\phi(x_3))}{g_t(\eta(T))-g_t(z)} \Big)
\Big( \frac{g_t(\eta(T))-g_t(\phi(x_4))}{g_t(\phi(x_4))-g_t(\phi(x_3))} \Big).
\end{align*}
Then, we have
\begin{align*}
\frac{g_t'(u_s)^2}{g_t'(\phi(x_{2s-1})) \, g_t'(\phi(x_{2s}))}
= \frac{\xi'_t(u_s)^2}{\xi'_t(\phi(x_{2s-1}))\xi'_t(\phi(x_{2s}))} 
\; \frac{(g_t(\eta(T))-g_t(u_s))^2}{(g_t(\eta(T))-g_t(\phi(x_{2s-1})))(g_t(\eta(T))-g_t(\phi(x_{2s})))} .
\end{align*}
A standard argument (for instance, see~\cite[Eq.~(4.6)]{Liu-Wu:Scaling_limits_of_crossing_probabilities_in_metric_graph_GFF}) then shows that
\begin{align}\label{eqn::asy=1}
\prod_{2\le s\le p}\left(\frac{g_t'(u_s)^2}{g_t'(\phi(x_{2s-1})) \, g_t'(\phi(x_{2s}))}\right)^{(\kappa-4)/2\kappa}
\quad \overset{t \to T}{\longrightarrow} \quad 
\prod_{2\le s\le p}\left| \frac{\varphi_\eta'(u_s)^2}{\varphi_\eta'(\phi(x_{2s-1})) \, \varphi_\eta'(\phi(x_{2s}))} \right|^{(\kappa-4)/2\kappa} .
\end{align}
Note that the right side of~\eqref{eqn::asy=1} is independent of the choice of $\varphi_\eta$ 
as long as $\varphi_\eta(\eta(T))=\infty$.
Combining these observations, we find that, almost surely on the event $\{T<\tau\}$, 
\begin{align*}
\lim_{t\to T}M_{1,t}=S_T.
\end{align*}

\item \label{item::M2t}
For $M_{2,t}$, by continuity, we have $\underset{t\to T}{\lim} \, M_{2,t}=V_T$.

\item \label{item::M3t}
For $M_{3,t}$, by continuity, we have $\underset{t\to T}{\lim} \, M_{3,t}=1$.

\item \label{item::M4t}
By a standard argument (see for instance~\cite[Eq.~(B.2)]{Peltola-Wu:Global_and_local_multiple_SLEs_and_connection_probabilities_for_level_lines_of_GFF}), we find that $\underset{t\to T}{\lim} \, M_{4,t}=1$.
\end{enumerate}
Combining Items~\ref{item::M1t}--\ref{item::M4t}, we obtain~\eqref{eqn::mart1limitsec4} on the event $\{T<\tau\}$. 

\smallbreak

{\bf Uniform integrability.}
We argue next that $\left\{M_{t\wedge T\wedge\tau}\right\}_{t\ge 0}$ is uniformly integrable.
For $\eps>0$, denote by $\tau_\eps$ the hitting time of $\eta$ at the union of the $\eps$-neighborhood of $\smash{I^{\bs{\dot u}_1}}$ and $\partial B(0,1/\eps):=\{z\in\C \colon |z|=1/\eps\}$. 
\begin{itemize}[leftmargin=2em]
\item 
For $M_{1,t}$, we consider the decomposition in~\eqref{eqn::mart_decomposition1}.
Because the Poisson kernel is monotone by~\eqref{eqn::Poisson_mono}, 
we see that there exists a constant $C=C(\eps)$ such that for all $t\le \tau_\eps$, we have
\begin{align*}
\frac{1}{C} \; \le \;  \bigg\lbrace
\; & \prod_{3\le i < j\le 2p}
\frac{g_t'(\phi(x_j)) \, g_t'(\phi(x_i))}{(g_t(\phi(x_j))-g_t(\phi(x_i)))^2}, \\
&
\prod_{\substack{3 \le i\le 2p \\ 2 \le s \le p}} \frac{g_t'(\phi(x_i)) \, g_t'(u_s)}{(g_t(u_s)-g_t(\phi(x_i)))^2}
, \;
\prod_{2\le s < q \le p} \frac{g_t'(\phi(u_q)) \, g_t'(\phi(u_s))}{(g_t(\phi(u_q))-g_t(\phi(u_s)))^2}
\bigg\rbrace
\; \le \; C .
\end{align*}
It thus remains to bound the fourth product in~\eqref{eqn::mart_decomposition1}. 
To this end, note that
\begin{align*}
\partial_t \log \left( \frac{g_t'(u_s)^2}{g_t'(\phi(x_{2s-1})) \, g_t'(\phi(x_{2s}))}\right)
= \; & \frac{-4}{(g_t(u_s)-W_t)^2} 
+ \frac{2}{(g_t(\phi(x_{2s-1}))-W_t)^2} 
+ \frac{2}{(g_t(\phi(x_{2s}))-W_t)^2} \\ 
\ge \; & 0 .
\end{align*}
Therefore, there exists a constant $C'=C'(\eps)$ such that for all $t\le \tau_\eps$, we have
\begin{align*}
\frac{g_t'(u_s)^2}{g_t'(\phi(x_{2s-1})) \, g_t'(\phi(x_{2s}))}
\leq 
\frac{|\varphi_\eta'(u_s)|^2}{|\varphi_\eta'(\phi(x_{2s-1})) \, \varphi_\eta'(\phi(x_{2s}))|}
\leq C' .
\end{align*}
We thus see that $|M_{1,t}|$ is bounded on the event $\{t\le\tau_\eps\}$ by a constant depending on $\eps$.

\item 
Similarly, we find that $|M_{2,t}|$ is bounded on the event $\{t\le\tau_\eps\}$ by a constant depending on $\eps$.

\item Let us then consider $M_{3,t}$. For all $2\le s\le p$ and $p+1\le q\le N$, we have
\begin{align*}
\bigg| \frac{\left(g_t(u_q)-g_t(u_s)\right)^2}{(g_t(u_q)-g_t(\phi(x_{2s-1})))(g_t(u_q)-g_t(\phi(x_{2s})))} \bigg| \leq 1 .
\end{align*}
For all $2\le s\le p$ and $2p+1\le j\le 2N$, note that the process
\begin{align*}
Y_{s,j}(t) 
:= \frac{\left(g_t(\phi(x_j))-g_t(u_s)\right)^2}{(g_t(\phi(x_{j}))-g_t(\phi(x_{2s-1})))(g_t(\phi(x_{j}))-g_t(\phi(x_{2s})))} 
\end{align*}
satisfies
\begin{align*}
\partial_t \log Y_{s,j}(t) 
= \; &  
\frac{2}{W_t-g_t(\phi(x_j))} \bigg( 
\frac{-2}{W_t-g_t(u_s)}+\frac{1}{W_t-g_t(\phi(x_{2s-1}))}+\frac{1}{W_t-g_t(\phi(x_{2s}))} \bigg) \\ 
\ge \; & 0 ,
\end{align*}
which implies that
\begin{align*}
Y_{s,j}(t) \ge \; &
\frac{\left(\phi(x_j)-u_s\right)^2}{(\phi(x_j)-\phi(x_{2s-1}))(\phi(x_j)-\phi(x_{2s}))} ,
\end{align*}
and we may conclude that $|M_{3,t}|$ is bounded on the event $\{t\le\tau_\eps\}$ by a constant depending on $\eps$.

\item Lastly, note that $|M_{4,t}| \le 1$ is bounded.
\end{itemize}
Combining these observations, we see that $|M_t|$ is bounded on the event $\{t\le\tau_\eps\}$ by a constant depending on $\eps$. We will next take $\eps \to 0$. 
Denote by $P_\eps$ the law of $\eta$ weighted by $M_{\tau_\eps}/M_0$. 
Then, by construction, $P_\eps$ is the same as the law of $\gamma^{\bs{\dot u}_1}$ up to the hitting time at the union of $\eps$-neighborhoods of $\smash{I^{\bs{\dot u}_1}}$ and $\partial B(0,1/\eps)$. 
On the one hand, as the collection 
$\{P_\eps\}_{\eps>0}$ is consistent in $\eps$, by Kolmogorov's extension theorem 
there exists a probability measure $P_*$ such that, for each $\eps$, 
the curve $\eta$ has the law of $\gamma^{\bs{\dot u}_1}$ up to the hitting 
time at the union of $\eps$-neighborhoods of $\smash{I^{\bs{\dot u}_1}}$ and $\partial B(0,1/\eps)$. 
On the other hand, 
we know that $\gamma^{\bs{\dot u}_1}$ has a positive distance from the intervals in $\smash{I^{\bs{\dot u}_1}}$ almost surely. 
Hence, the measure $P_*$ is exactly the law of $\gamma^{\bs{\dot u}_1}$ stopped at the hitting time at $(\phi(x_3),+\infty)\setminus \smash{I^{\bs{\dot u}_1}}$. 
This implies that $(M_{t\wedge T\wedge\tau})_{t\ge 0}$ is uniformly integrable and the law of $\eta$ tilted by $(M_{t\wedge T\wedge\tau})_{t\ge 0}$ 
is the same as the law of $\gamma^{\bs{\dot u}_1}$ stopped at the hitting time at $(\phi(x_3),+\infty)\setminus \smash{I^{\bs{\dot u}_1}}$. 

\smallbreak

{\bf Terminal value~\eqref{eqn::mart1limitsec5}.}
To finish the proof, we derive the limit of $M_t$ on the event $\{\tau< T\}$:
since $\gamma^{\bs{\dot u}_1}$ has a positive distance from the intervals in $\smash{I^{\bs{\dot u}_1}}$ almost surely, we have
\begin{align*}
\E \Big[ \one\{\tau<T\} \frac{M_{\tau}}{M_0} \Big] = \PP^{\bs{\dot u}_1}[\gamma^{\bs{\dot u}_1}\text{ hits } I^{\bs{\dot u}_1}] = 0 ,
\end{align*}
where $\E$ is the expected value under the chordal $\SLE_\kappa$ in $(\HH;0,\infty)$, 
and $\smash{\PP^{\bs{\dot u}_1}}$ is the $\SLE_{\kappa}(\underline{\rho}^R)$ measure in $(\HH;0,\infty)$.
This proves~\eqref{eqn::mart1limitsec5} and completes the proof. 
\end{proof}

\begin{proof}[Proof of Proposition~\ref{prop::Qn_2b_marginal}]
Assume that $\fixedindex=1$. 
Let $\smash{\SLEmeasureEx_N^{x_1\to x_2}}$ be the expected value under the measure $\smash{\SLEmeasure_N^{x_1\to x_2}}$ from Definition~\ref{def::Qn2b}.
Recall that $\vec\eta_1 = \phi^{-1}(\cev\gamma_1)$, 
where the curve $\cev\gamma_1$ is the time-reversal of an 
$\SLE_{\kappa}(\underline{\rho}^R)$ process $\gamma_1$ in $(\HH;0,\infty)$ with force points 
\begin{align*}
\underline{x}^R = \; & (\phi(x_3),\phi(x_4),U_2,\ldots, \phi(x_{2N-1}),\phi(x_{2N}),U_N) ,
\end{align*}
and weights 
$\underline{\rho}^R = (2,2,-4,\ldots,2,2,-4)$ on its right; 
the points $(U_2, U_2, \ldots, U_N)$ being random, with density~\eqref{eqn::densityruj}.
Hence, for any bounded continuous function $F$, we have 
\begin{align}\label{eqn::intexpression}
\SLEmeasureEx_N^{x_1\to x_2} [ F(\vec\eta_1) ] 
= & \; \int_{\phi(x_{4})}^{\phi(x_{5})} \ud u_{2}\cdots \int_{\phi(x_{2N})}^{+\infty}\ud u_{N} 
\; \frac{|f(\phi(\bs{\dot x}_1);\bs{\dot u}_1)| \; \E^{\bs{\dot u}_1} \big[ F(\phi^{-1}(\cev\gamma^{\bs{\dot u}_1})) \big] }
{\int_{\phi(x_{4})}^{\phi(x_{5})} \ud v_{2} \cdots \int_{\phi(x_{2N})}^{+\infty} \ud v_N \; |f(\phi(\bs{\dot x}_1);\bs{\dot v}_1)|}
\end{align}
where $\smash{\E^{\bs{\dot u}_1}}$ is the expected value under the $\SLE_{\kappa}(\underline{\rho}^R)$ measure in $(\HH;0,\infty)$.
We will compare this to the law of the chordal $\SLE_{\kappa}$ curve $\eta$ in $(\HH;0,\infty)$. 
We denote by $\Omega_1,\ldots,\Omega_\np$ the connected components of $\HH\setminus\eta$ whose boundaries contain points in $\{\phi(x_3), \phi(x_4),\ldots,\phi(x_{2N})\}$, and
we denote by $\phi(x_{j_m}),\ldots,\phi(x_{j_{m+1}-1})$ the points contained in $\partial \Omega_m$, for $1\le m\le \np$. 
We also define $T_m$ to be the hitting time of $\eta$ at $(\phi(x_{j_{m+1}-1}),+\infty)$ and 
choose conformal maps $\varphi_m$ from $\Omega_m$ onto $\HH$ such that $\varphi_m(\eta(T_m))=\infty$. 
Then, by applying Lemma~\ref{lem::mart1general} finitely many times iteratively, we obtain 
\begin{align}\label{eqn::intexpressionRN}
\E^{\bs{\dot u}_1} \big[ F(\phi^{-1}(\cev\gamma^{\bs{\dot u}_1})) \big]
= \frac{1}{M^{\bs{\dot u}_1}_0} \, 
\E \Big[F(\phi^{-1}(\cev\eta)) \, \one\{A^{\bs{\dot u}_1}\} \, \prod_{m=1}^{\np} R^{\bs{\dot u}_1}_{m}\Big] ,
\end{align}
where $\E$ is the expected value for the chordal $\SLE_\kappa$ curve $\eta$ in $(\HH;0,\infty)$,
and $\cev\eta$ is its time-reversal, 
and $A^{\bs{\dot u}_1}(\eta) := \{\eta \cap I^{\bs{\dot u}_1} = \emptyset \}$, and
\begin{align*}
R^{\bs{\dot u}_1}_{m}
= & \; \prod_{j = j_{m}}^{j_{m+1}-1} |\varphi_m'(\phi(x_j))|^{h} 
\prod_{s = (j_{m}+1)/2}^{(j_{m+1}-1)/2} |\varphi_m'(u_s)| \\
&\; \times \big| f(\varphi_m(\phi(x_{j_m})),\ldots,\varphi_m(\phi(x_{j_{m+1}-1}));\varphi_m(u_{(j_{m}+1)/2}),\ldots,\varphi_m(u_{(j_{m+1}-1)/2})) \big| .
\end{align*}
Note that $\smash{M_0^{\bs{\dot u}_1} = |f(\phi(\bs{\dot x}_1);\bs{\dot u}_1)|}$ by~\eqref{eqn::mart1}. 
Combining~(\ref{eqn::intexpression},~\ref{eqn::intexpressionRN}), we obtain
\begin{align}\label{eqn::intexpressionRNratio}
\SLEmeasureEx_N^{x_1\to x_2} [ F(\vec\eta_1) ] 
= & \; \frac{\int_{\phi(x_{4})}^{\phi(x_{5})} \ud u_{2}\cdots \int_{\phi(x_{2N})}^{+\infty}\ud u_{N} 
\; \E \Big[F(\phi^{-1}(\cev\eta)) \, \one\{A^{\bs{\dot u}_1}\} \, \prod_{m=1}^{\np} R^{\bs{\dot u}_1}_{m}\Big]}
{\int_{\phi(x_{4})}^{\phi(x_{5})} \ud v_{2} \cdots \int_{\phi(x_{2N})}^{+\infty} \ud v_N \; |f(\phi(\bs{\dot x}_1);\bs{\dot v}_1)|} \\
\nonumber
= & \; \frac{\E \Big[ \one\{A_1(\eta)\} \, F(\phi^{-1}(\cev\eta)) \, 
\prod_{m=1}^{\np} 
\int_{\phi(x_{j_m+1})}^{\phi(x_{j_m+2})} \ud u_{(j_m+1)/2} \cdots \int_{\phi(x_{j_{m+1}-1})}^{\eta(T_m)} \ud u_{(j_{m+1}-1)/2} \; R^{\bs{\dot u}_1}_{m} \Big]}
{\int_{\phi(x_{4})}^{\phi(x_{5})} \ud v_{2} \cdots \int_{\phi(x_{2N})}^{+\infty} \ud v_N \; |f(\phi(\bs{\dot x}_1);\bs{\dot v}_1)|} ,
\end{align}
where 
\begin{align*}
A_1(\eta) := \{\eta \cap (\phi(x_{2s-1}), \phi(x_{2s})) = \emptyset \textnormal{ for all } 2\le s \le N \}.
\end{align*}
We will simplify the numerator and denominator of~\eqref{eqn::intexpressionRNratio} separately.
\begin{itemize}[leftmargin=2em]
\item For the numerator, using the change of variables $v_s = \varphi_m(u_s)$ for all $s$, 
and invoking Lemma~\ref{lem::RN_equivalent_def2} and the conformal covariance of  
$\coulombGas_{\!\! \vcenter{\hbox{\includegraphics[scale=0.8]{figures-arXiv/link62.pdf}}}_{\frac{j_{m+1}-j_m}{2}}}$, 
we see that
\begin{align*}
\; & \fugacity(\kappa) \, \hat{\cst}(\kappa)^{(j_{m+1}-j_m)/2} 
\; \int_{\phi(x_{j_m+1})}^{\phi(x_{j_m+2})} \ud u_{(j_m+1)/2} \cdots \int_{\phi(x_{j_{m+1}-1})}^{\eta(T_m)} \ud u_{(j_{m+1}-1)/2} \; R^{\bs{\dot u}_1}_{m} 
\\
= \; & \prod_{j = j_{m}}^{j_{m+1}-1} |\varphi_m'(\phi(x_j))|^{h} \;
\; \times \; \coulombGas_{\!\! \vcenter{\hbox{\includegraphics[scale=0.8]{figures-arXiv/link62.pdf}}}_{\frac{j_{m+1}-j_m}{2}}}(\varphi_m(\phi(x_{j_m})),\ldots,\varphi_m(\phi(x_{j_{m+1}-1}))) 
&& \textnormal{[by Lemma~\ref{lem::RN_equivalent_def2}]}
\\
= \; & \coulombGas_{\!\! \vcenter{\hbox{\includegraphics[scale=0.8]{figures-arXiv/link62.pdf}}}_{\frac{j_{m+1}-j_m}{2}}}(\Omega_m;\phi(x_{j_m}),\ldots,\phi(x_{j_{m+1}-1})) ,
&& \textnormal{[by~\eqref{eqn::def_polygon_CGI}]}
\\
= \; & \prod_{j=j_m}^{j_{m+1}-1} | \phi'(x_j) |^{-h} \; \times 
\; \coulombGas_{\!\! \vcenter{\hbox{\includegraphics[scale=0.8]{figures-arXiv/link62.pdf}}}_{\frac{j_{m+1}-j_m}{2}}}(H_m;x_{j_m},\ldots,x_{j_{m+1}-1}) ,
\end{align*}
where $H_m = \phi^{-1}(\Omega_m)$ and $\hat{\cst}(\kappa) = 4\sin^2(4\pi/\kappa) \, \cst(\kappa)$.

\item For the denominator, we make the change of variables $u_1 = \phi^{-1}(v_N)$ and $u_r = \phi^{-1}(v_{r-1})$ for all $3 \leq r \leq N$. 
Then, using the identity $\frac{\phi(z)-\phi(w)}{z-w} = \sqrt{\phi'(z)} \sqrt{\phi'(w)}$ (see, e.g.,~\cite[Lemma~4.7]{Kytola-Peltola:Pure_partition_functions_of_multiple_SLEs}) for the M\"obius transformation~\eqref{eqn::Mobius_def} satisfying $\phi^{-1}(\infty) = x_1$ and  
$\phi'(z) = \tfrac{x_2-x_1}{(z-x_1)^2} \, \tfrac{x_3-x_1}{x_3-x_2}$, 
we obtain 
\begin{align*}
\; & \int_{\phi(x_{4})}^{\phi(x_{5})} \ud v_{2} \cdots \int_{\phi(x_{2N})}^{+\infty} \ud v_N \; |f(\phi(\bs{\dot x}_1);\bs{\dot v}_1)| \\
= \; & \int_{x_{2N}}^{x_1} \ud u_1 
\int_{x_{4}}^{x_{5}} \ud u_3 \cdots \int_{x_{2N-2}}^{x_{2N-1}} \ud u_N
\; |f(\bs{\dot x}_1;\bs{\dot u}_2)| 
\; \prod_{\substack{1 \leq r \leq N \\ r \neq 2}} | \phi'(u_r) |^{-2h(\kappa)} 
\\
= \; & 
\bigg( \frac{(x_2-x_1)(x_3-x_1)}{x_3-x_2} \bigg)^{-2 (N-1) h(\kappa)}
\prod_{j=2}^{2N} (x_j - x_1)^{2h(\kappa)}  \;
\int_{x_{2N}}^{x_1} \ud u_1 
\int_{x_{4}}^{x_{5}} \ud u_3 \cdots \int_{x_{2N-2}}^{x_{2N-1}} \ud u_N
\; |\hat{f}^a(\bs{x};\bs{\dot u}_2)| \\
= \; & (x_2-x_1)^{2h(\kappa)} \,
\prod_{j=3}^{2N} | \phi'(x_j) |^{-h} \; \times 
\int_{x_{2N}}^{x_1} \ud u_1 
\int_{x_{4}}^{x_{5}} \ud u_3 \cdots \int_{x_{2N-2}}^{x_{2N-1}} \ud u_N
\; |\hat{f}^a(\bs{x};\bs{\dot u}_2)| \\
= \; & \frac{(x_2-x_1)^{2h(\kappa)}}{\hat{\cst}(\kappa)^{N-1}} \,
\prod_{j=3}^{2N} | \phi'(x_j) |^{-h} \; \times 
\; \hat{\LR}_{\omega_1}(\bs{x}) ,
\end{align*}
using Lemma~\ref{lem::another_CGI_with_conjugate_charge}, 
where $\hat{f}^a$ defined in~\eqref{eq: integrand_gen_conjugate_a} has the conjugate charge at $x_1$.
\end{itemize}
We conclude that~\eqref{eqn::intexpressionRNratio} equals
\begin{align} \label{eqn::RNderichordalSLE} 
\SLEmeasureEx_N^{x_1\to x_2} [ F(\vec\eta_1) ] 
= & \; 
\frac{\E \Big[ \one\{A_1(\eta)\} \, F(\phi^{-1}(\cev\eta)) \, 
\prod_{m=1}^{\np} 
\frac{1}{\fugacity(\kappa)} \, \coulombGas_{\!\! \vcenter{\hbox{\includegraphics[scale=0.8]{figures-arXiv/link62.pdf}}}_{\frac{j_{m+1}-j_m}{2}}}(H_m;x_{j_m},\ldots,x_{j_{m+1}-1}) \Big]}
{(x_2 - x_1)^{2h(\kappa)} \; 
\hat{\LR}_{\omega_1}(\bs{x})} .
\end{align}
The asserted formula~\eqref{eqn::RNderigeneral} follows by conformal invariance and reversibility~\cite{Sheffield-Miller:Imaginary_geometry3} of 
the chordal $\SLE_\kappa$ curve on $(\HH;0,\infty)$: 
since $\phi \colon \HH \to \HH$ sends $x_1 \mapsto \infty$ and $x_2 \mapsto 0$,
the curve $\phi^{-1}(\eta)$ has the law of 
chordal $\SLE_\kappa$  on $(\HH;x_2,x_1)$,
and the time-reversal $\phi^{-1}(\cev\eta)$ has the law of chordal $\SLE_\kappa$ on $(\HH;x_1,x_2)$.
\end{proof}


\bigskip{}
\section{Multiple $\SLE_\kappa$ and pure partition functions for $\kappa\in (4,8)$}
\label{sec::PPF}
We now use the measures constructed in Section~\ref{subsec::QN} 
to construct global multiple $\SLE_{\kappa}$ processes and their pure partition functions for $\kappa\in (4,8)$. Together with prior results~\cite{Peltola-Wu:Global_and_local_multiple_SLEs_and_connection_probabilities_for_level_lines_of_GFF, 
Wu:Convergence_of_the_critical_planar_ising_interfaces_to_hypergeometric_SLE}, 
this completes the construction of the pure partition functions for the full parameter range $\kappa\in (0,8)$.

Recall that $\SLEmeasure_N^{x_1\to x_{2\fixedindex}}$ (Definition~\ref{def::Qn2b}) 
and $\vec{\SLEmeasure}_N$ (Definition~\ref{def::QN}) 
are conformally invariant probability measures on (oriented) curves $\vec{\bs\eta} := (\vec\eta_1,\ldots,\vec\eta_N) \in X_0(\Omega; \bs{x})$ on the curve space~\eqref{eq:curve_space} on polygons $(\Omega; \bs{x})$. 
We will mainly work with them in the reference domain $\Omega = \HH$, writing 
$\SLEmeasure_N^{x_1\to x_{2\fixedindex}}=\SLEmeasure_N^{x_1\to x_{2\fixedindex}}(\HH; \bs{x})$ and $\SLEmeasure_N=\SLEmeasure_N(\HH; \bs{x})$. 
Using the normalization factor of~\eqref{eqn::palpha}, set 
\begin{align}\label{eqn::def_KN}
\LK_{\vcenter{\hbox{\includegraphics[scale=0.8]{figures-arXiv/link62.pdf}}}_N}(\bs x) := \fugacity(\kappa) \, 
\sum^N_{\fixedindex=1} \hat{\LR}_{\omega_\fixedindex}(\bs x) \, \sum_{\substack{\beta \in\LP_N  \\ \{1,2\fixedindex\}\in\beta}} 
w_\fixedindex(\beta) \, \SLEmeasure^{x_1\to x_{2\fixedindex}}_N [\conn(\vec{\bs\eta}) = \beta ]  , \qquad \bs{x}\in\chamber_{2N} ,
\end{align}
where $\conn(\vec{\bs\eta})$ denotes the random link pattern in $\LP_N$ formed by the curves $\vec{\bs\eta} \sim \smash{\SLEmeasure^{x_1\to x_{2\fixedindex}}_N}$,
and where $\hat{\LR}_{\omega_\fixedindex}$ are the functions appearing in Lemma~\ref{lem::another_CGI_with_conjugate_charge}  
and the weights $w_\fixedindex(\beta)$ are defined in~\eqref{eqn::defqj}. 
In fact, we have $\LK_{\vcenter{\hbox{\includegraphics[scale=0.8]{figures-arXiv/link62.pdf}}}_N} = \smash{\coulombGas_{\vcenter{\hbox{\includegraphics[scale=0.8]{figures-arXiv/link62.pdf}}}_N}}$, see Lemma~\ref{lem::PPF_using_CGI}.
To prove this, we need key properties of the Coulomb gas integral functions 
$\smash{\coulombGas_{\vcenter{\hbox{\includegraphics[scale=0.8]{figures-arXiv/link62.pdf}}}_N}}$: 
they are uniquely characterized as solutions to the PDE boundary value problem in the space $\mathcal{S}_N$~\eqref{eq: solution space} 
by their asymptotic properties~\eqref{eqn::CGI_ASY} in Corollary~\ref{cor::CGI_PDECOVASY}.

\begin{definition}\label{def::Zalpha}
Fix $\kappa \in (4,8)$. 
We define the the \emph{pure partition function} of multiple $\SLE_\kappa$ by
\begin{align}\label{eqn::PPF_def_QN}
\PartF_{\alpha}(\bs x) = 
\vec{\PartF}_{\alpha}(\bs x) := \LK_{\vcenter{\hbox{\includegraphics[scale=0.8]{figures-arXiv/link62.pdf}}}_N} (\bs x) \, 
\frac{\vec{\SLEmeasure}_N[\conn(\vec{\bs\eta}) = \alpha ]}{\meanderMatrix(\vcenter{\hbox{\includegraphics[scale=0.8]{figures-arXiv/link62.pdf}}}_N,\alpha)}
 , \qquad \bs{x}\in\chamber_{2N} , \; \alpha \in \LP_N . 
\end{align}
We define analogously
\begin{align}\label{eqn::defLFalpha}
\cev{\PartF}_{\alpha}(\bs x) := \LK_{\vcenter{\hbox{\includegraphics[scale=0.8]{figures-arXiv/link62.pdf}}}_N} (\bs x) \, 
\frac{\cev{\SLEmeasure}_N[\conn(\cev{\bs\eta}) = \alpha ]}{\meanderMatrix(\vcenter{\hbox{\includegraphics[scale=0.8]{figures-arXiv/link62.pdf}}}_N,\alpha)}
 , \qquad \bs{x}\in\chamber_{2N} , \; \alpha \in \LP_N . 
\end{align}
\end{definition}

Recall that, under the measure $\cev{\SLEmeasure}_N$, each curve $\cev\eta_s$ starts from $x_{2s}$ for $1\le s\le N$ (while under the measure $\vec{\SLEmeasure}_N$, each curve $\vec\eta_s$ starts from $x_{2s-1}$ for $1\le s\le N$). We will prove that the two functions~\eqref{eqn::PPF_def_QN} and~\eqref{eqn::defLFalpha} 
are in fact the same (Lemma~\ref{lem::difpartition}).
This is a feature of reversibility for multiple $\SLE_\kappa$ curves.

\begin{remark}
When $N=1$, we have $\PartF_{\vcenter{\hbox{\includegraphics[scale=0.2]{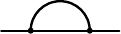}}}}(x_1,x_2) = (x_2-x_1)^{-2h}$. Further examples are discussed in Appendix~\ref{app::examples}.
\end{remark}

The function $\LK_{\vcenter{\hbox{\includegraphics[scale=0.8]{figures-arXiv/link62.pdf}}}_N}(\bs x)$ is defined in~\eqref{eqn::def_KN} in terms of the probability measure $\SLEmeasure_N^{x_1\to x_{2\fixedindex}}$. By definition of $\vec{\SLEmeasure}_N$ in~\eqref{eqn::defQN}, 
and similarly for $\cev{\PartF}$, we obtain the factorization formula 
\begin{align}\label{eqn::usefulexpression}
\vec{\PartF}_{\alpha}(\bs x) =  
\fugacity(\kappa) \, 
\hat{\LR}_{\omega_\fixedindex}(\bs x)
\frac{w_\fixedindex(\alpha) \, \SLEmeasure^{x_1\to x_{2\fixedindex}}_N [\conn(\vec{\bs\eta}) = \alpha ]}{\meanderMatrix(\vcenter{\hbox{\includegraphics[scale=0.8]{figures-arXiv/link62.pdf}}}_N,\alpha)} , 
\qquad \{1,2\fixedindex\}\in\alpha .
\end{align}

\bigskip

The goal of this section is to prove that $\{\vec{\PartF}_{\alpha} \colon \alpha\in\LP_N\}$ indeed 
are the \emph{pure} partition functions of multiple $\SLE_{\kappa}$, that is, 
they are positive (Lemma~\ref{lem::ppf_continuity_and_positivity}) 
and satisfy the PDE system~\eqref{eqn::PDE} (see Proposition~\ref{prop::PPF_PDE}), 
M\"{o}bius covariance~\eqref{eqn::COV} (see Lemma~\ref{lem::PPF_COV}),
the power-law bound in~\eqref{eqn::PLB_weak_upper} (see Lemma~\ref{lem::PPF_PLB_weak}),
as well as the recursive asymptotics property~\eqref{eqn::PPF_ASY} (see Proposition~\ref{prop::PPF_ASY}).
These imply the existence part in Theorem~\ref{thm::PPF}. 
The uniqueness part in Theorem~\ref{thm::PPF} 
is already well known by results of Flores and Kleban (see Lemma~\ref{lem::PFuniqueness}).
We verify the additional properties in Theorem~\ref{thm::PPF} in Section~\ref{sec::consequences}. 

\smallbreak

Before proceeding to the proofs, let us make a few remarks regarding the scope of these results.
\begin{itemize}[leftmargin=1.5em]
\item The BPZ PDE system~\eqref{eqn::PDE} is a linear hypoelliptic system, 
for which smoothness of solutions can be derived from H\"ormander's theory (cf.~\cite{Bony:Maximum_principle_Harnack_inequality_and_uniqueness,
Hormander:The_analysis_of_linear_partial_differential_operators_1,
Dubedat:SLE_and_Virasoro_representations_localization, Peltola-Wu:Global_and_local_multiple_SLEs_and_connection_probabilities_for_level_lines_of_GFF, AHSY:Conformal_welding_of_quantum_disks_and_multiple_SLE_the_non-simple_case, Karrila-Viitasaari:In_prep}.

\item The M\"{o}bius covariance~\eqref{eqn::COV} can be verified immediately from the definitions, 
as the involved functions are M\"{o}bius invariant/covariant (Lemma~\ref{lem::PPF_COV}).
We also prove a stronger version in Lemma~\ref{lem::PPF_COV_GEN}.

\item The power-law bound~\eqref{eqn::PLB_weak_upper} is immediate (Lemma~\ref{lem::PPF_PLB_weak}).
Positivity needs an argument (Lemma~\ref{lem::ppf_continuity_and_positivity}).

\item The recursive asymptotics property~\eqref{eqn::PPF_ASY} 
of the pure partition functions (addressed in Proposition~\ref{prop::PPF_ASY}) 
is not only their most important, characterizing property, 
but is also rather difficult to verify in the case where $\kappa \in (6,8)$. 
Indeed, the proof in~\cite[Section~6]{Wu:Convergence_of_the_critical_planar_ising_interfaces_to_hypergeometric_SLE} fails for this range, as also discussed in~\cite[Appendix~B]{Peltola:Towards_CFT_for_SLEs}.
The main reason for this is that it is hard to establish sufficient control for exchanging 
the limit in the asymptotics in~\eqref{eqn::PPF_ASY} with the expected value in the cascade relation 
in~\eqref{eqn::cascade_relation} in Lemma~\ref{lem::cascade}. 
We are not aware of any proof for the asymptotics property~\eqref{eqn::PPF_ASY} 
of the pure partition functions for $\kappa \in (6,8) \cap \QQ$ before the present work. 
In the case where $\kappa \in (6,8) \setminus \QQ$, the asymptotics property~\eqref{eqn::PPF_ASY} 
can be verified using the Coulomb gas formalism~\cite{Flores-Kleban:Solution_space_for_system_of_null-state_PDE4,
Kytola-Peltola:Pure_partition_functions_of_multiple_SLEs, 
Kytola-Peltola:Conformally_covariant_boundary_correlation_functions_with_quantum_group}, 
which however fails to ensure positivity of the functions (crucial for their probabilistic meaning).
Therefore, it seems to us that both the probabilistic arguments and the Coulomb gas are necessary.
\end{itemize}

\subsection{First properties of pure partition functions}

To prove the key properties of the pure partition functions, we proceed by induction on the number $N \geq 1$ of curves. The induction step will require a number of additional properties, gathered in the next lemmas.

\begin{lemma}
\label{lem::difpartition}
We have $\vec{\PartF}_{\alpha} = \cev{\PartF}_{\alpha}$ for all $\alpha\in\LP_N$ and $N \geq 1$. 
\end{lemma}

\begin{proof}
We prove the claim by induction on the number of curves in the course of this section. 
In the base case $N=1$, we have $\vec{\PartF}_{\vcenter{\hbox{\includegraphics[scale=0.2]{figures-arXiv/link-0.pdf}}}}(x_1,x_2) = \cev{\PartF}_{\vcenter{\hbox{\includegraphics[scale=0.2]{figures-arXiv/link-0.pdf}}}}(x_1,x_2) = (x_2-x_1)^{-2h}$.
The proof is completed in Section~\ref{subsec:rev}. 
\end{proof}

\begin{restatable}{lemma}{PPFusingCGI}
\label{lem::PPF_using_CGI}
Fix $N \geq 1$. We have
\begin{align}\label{eqn::PPF_using_CGI}
\PartF_{\alpha}(\bs x) = 
\vec{\PartF}_{\alpha}(\bs x) 
= \coulombGas_{\vcenter{\hbox{\includegraphics[scale=0.8]{figures-arXiv/link62.pdf}}}_N}(\bs x) \, 
\frac{\vec{\SLEmeasure}_N[\conn(\vec{\bs\eta}) = \alpha ]}{\meanderMatrix(\vcenter{\hbox{\includegraphics[scale=0.8]{figures-arXiv/link62.pdf}}}_N,\alpha)}
 , \qquad \bs{x}\in\chamber_{2N} , \; \alpha \in \LP_N .
\end{align}
\end{restatable}

Note that the formula~\eqref{eqn::PPF_using_CGI} is equivalent to the identity $\LK_{\vcenter{\hbox{\includegraphics[scale=0.8]{figures-arXiv/link62.pdf}}}_N} = \smash{\coulombGas_{\vcenter{\hbox{\includegraphics[scale=0.8]{figures-arXiv/link62.pdf}}}_N}}$ in~\eqref{eqn::def_KN}.

\begin{proof}
We prove this by induction on the number of curves in the course of this section. 
In the base case $N=1$, we have $\coulombGas_{\vcenter{\hbox{\includegraphics[scale=0.2]{figures-arXiv/link-0.pdf}}}}(x_1,x_2) = \fugacity(\kappa) \, (x_2-x_1)^{-2h} = \LK_{\vcenter{\hbox{\includegraphics[scale=0.2]{figures-arXiv/link-0.pdf}}}}(x_1,x_2)$.
The proof is completed in Section~\ref{subsec:proof_of_PPF_vs_CGI}. 
\end{proof}

\paragraph*{Induction hypothesis.}
Fix $N \geq 2$. From now on, suppose that for all $n \leq N-1$, the following hold: 
\begin{itemize}[leftmargin=6em]
\item[Lemma~\ref{lem::difpartition}:]
$\vec{\PartF}_{\alpha} = \cev{\PartF}_{\alpha}$ for all $\alpha\in\LP_n$ and $n \leq N-1$; 

\item[Lemma~\ref{lem::PPF_using_CGI}:]
$\LK_{\vcenter{\hbox{\includegraphics[scale=0.8]{figures-arXiv/link62.pdf}}}_n} = \coulombGas_{\vcenter{\hbox{\includegraphics[scale=0.8]{figures-arXiv/link62.pdf}}}_n}$ for all $n \leq N-1$.
\end{itemize}
Given this induction hypothesis, we will prove Lemmas~\ref{lem::difpartition} and~\ref{lem::PPF_using_CGI}, as well as other salient properties of the functions $\PartF_{\alpha}$. 
We begin with a fundamental ``cascade property'' in the next section.

\subsection{Cascade property}

The ``cascade property'' of the pure partition functions on polygons (Lemma~\ref{lem::cascade}) 
plays a crucial role not only in the verification of the defining properties of the pure partition functions, 
but also in applications of the pure partition functions in various problems in random geometry. 
We will prove it by induction on the number of curves, using also 
the induction hypothesis on Lemmas~\ref{lem::difpartition} and~\ref{lem::PPF_using_CGI}.

\begin{lemma}\label{lem::PPF_COV}
Fix $\kappa \in (4,8)$. 
The collection $\{\vec{\PartF}_{\alpha} \colon \alpha\in\LP_N\}$ of Definition~\ref{def::Zalpha} satisfies~\eqref{eqn::COV}. 
\end{lemma}

\begin{proof}
The M\"obius covariance property~\eqref{eqn::COV}   
follows immediately from the conformal invariance of $\smash{\SLEmeasure_N^{x_1\to x_{2\fixedindex}}}$ 
with the conformal covariance of $\hat{\LR}_{\omega_\fixedindex}$ from Lemma~\ref{lem:PDECOVASY_for_LR},
together with the factorization~\eqref{eqn::usefulexpression}.
\end{proof}

\begin{remark}
We extend the definition of any partition function $\vec{\PartF}_{\alpha}(\Omega; \bs{x})$ to general nice polygons $(\Omega; \bs{x})$ by sending $\Omega$ onto $\HH$ via 
any conformal map $\varphi$ with $\varphi(x_1)<\cdots<\varphi(x_{2N})$,
using the conformal covariance
\begin{align}\label{eqn::def_polygon_PPF}
\vec{\PartF}_{\alpha}(\Omega; \bs{x}) := 
\prod_{j=1}^{2N}|\varphi'(x_j)|^{h(\kappa)} \times \vec{\PartF}_{\alpha}(\Omega; \varphi(\bs{x})) .
\end{align}
Thanks to~\eqref{eqn::COV} in Lemma~\ref{lem::PPF_COV}, this definition is independent of the choice of $\varphi$.
\end{remark}

\begin{lemma}\label{lem::cascade}
Fix $\kappa \in (4,8)$. 
Fix $N\ge 1$, a nice polygon $(\Omega; \bs{x})$, and a link pattern $\alpha = \{ \{a_1,b_1\}, \ldots , \{a_N,b_N\}\}$ as in~\eqref{eq: link pattern ordering}. 
Fix a link $\{a,b\} \in \alpha$. 
Let $\eta$ be the chordal $\SLE_\kappa$ on $(\Omega;x_a,x_b)$.
Then, we have
\begin{align}\label{eqn::cascade_relation}
\vec{\PartF}_\alpha(\Omega;\bs{x}) 
= \vec{\PartF}_{\vcenter{\hbox{\includegraphics[scale=0.2]{figures-arXiv/link-0.pdf}}}}(\Omega;x_{a},x_{b}) 
\; 
\E(\Omega;x_a,x_b) \big[ \one\{ \LE_\alpha(\eta) \} \, \vec{\PartF}_{\alpha/\{a,b\}}(\hat{\Omega}_\eta; \bs{\ddot{x}}_{a,b} ) \big] ,
\end{align}
where $\E(\Omega;x_a,x_b)$ is the expected value under the chordal $\SLE_\kappa$ in $(\Omega;x_a, x_b)$,
and $\LE_{\alpha}(\eta)$ is the event that 
the curve $\eta$ does not partition $\Omega$ into components 
where some variables corresponding to a link in $\alpha$ would belong to different components, 
$\hat{\Omega}_\eta$ is the union of those connected components \textnormal{(}c.c\textnormal{)} $D$ of 
$\Omega \setminus \eta$ that contain some of the points $\{x_1, \ldots, x_{2N}\} \setminus \{x_a, x_b\}$ in $\overline{D}$, 
\begin{align*}
\hat{\Omega}_\eta := 
\bigsqcup_{\substack{ D \textnormal{ c.c of } \Omega \setminus \eta \\ \overline{D} \cap \{x_1, \ldots, x_{2N}\} \setminus \{x_a, x_b\} \neq \emptyset }} D ,
\end{align*}
and $\bs{\ddot{x}}_{a,b} = (x_1,\ldots,x_{a-1},x_{a+1}, \ldots,x_{b-1},x_{b+1},\ldots,x_{2N})$ are the remaining marked points and 
\begin{align}\label{eqn::cascade_PPF}
\vec{\PartF}_{\alpha/\{a,b\}}(\hat{\Omega}_\eta; \bs{\ddot{x}}_{a,b} )
:= 
\prod_{\substack{ D \textnormal{ c.c of } \Omega \setminus \eta \\ \overline{D} \cap \{x_1, \ldots, x_{2N}\} \setminus \{x_a, x_b\} \neq \emptyset }}
\vec{\PartF}_{\alpha^D}(D; \ldots) ,
\end{align}
with $\alpha^D$ the sub-link pattern of $\alpha/\{a,b\}$ associated to the component $D \subset \hat{\Omega}_\eta$ 
and $\vec{\PartF}_{\alpha^D}(D; \ldots)$ the pure partition function in the polygon $(D; \ldots)$ 
with some of the points among $\{x_{1},\ldots,x_{2N}\} \setminus \{x_a,x_b\}$ on $\partial D$. 
\end{lemma}

Note that $(\hat{\Omega}_\eta; \bs{\ddot{x}}_{a,b} )$ is not necessarily a polygon, but it can be thought of as a union of polygons obtained from the different connected components of $\Omega \setminus \eta$. 
The function $\smash{\vec{\PartF}_{\alpha/\{a,b\}}(\hat{\Omega}_\eta; \bs{\ddot{x}}_{a,b} )}$ 
is a partition function for a domain with multiple components, simply defined as the product of the partition functions on the various connected components.
This is a feature of the domain Markov property of $\SLE_\kappa$, 
which also results in the fundamental cascade property~\eqref{eqn::cascade_relation} of the pure partition functions.

Before we prove Lemma~\ref{lem::cascade}, we first collect some combinatorial facts about the meander matrix and the weights $w_s(\alpha)$ defined in~\eqref{eqn::defqj}. 

\begin{lemma} \label{lem::meander}
The meander matrix $\{\meanderMatrix(\alpha, \beta) \colon \alpha,\beta\in\LP_N\}$ defined in~\eqref{eqn::meandermatrix_def} satisfies the following factorization properties. 
For any $\alpha\in\LP_N$ such that $\{1,2\fixedindex\} \in \alpha$, we have
\begin{align} \label{eq: meander fac1}
\meanderMatrix(\vcenter{\hbox{\includegraphics[scale=0.8]{figures-arXiv/link62.pdf}}}_N,\alpha) = 
\begin{cases}
\meanderMatrix(\vcenter{\hbox{\includegraphics[scale=0.8]{figures-arXiv/link62.pdf}}}_{N-1},\alpha^L) , & \fixedindex=1 , \\[.5em]
\meanderMatrix(\vcenter{\hbox{\includegraphics[scale=0.8]{figures-arXiv/link61.pdf}}}_{\fixedindex-1},\alpha^R) \, \meanderMatrix(\vcenter{\hbox{\includegraphics[scale=0.8]{figures-arXiv/link62.pdf}}}_{N-\fixedindex},\alpha^L)  , 
& 2\le \fixedindex\le N-1 , \\[.5em]
\fugacity(\kappa) \, \meanderMatrix(\vcenter{\hbox{\includegraphics[scale=0.8]{figures-arXiv/link61.pdf}}}_{N-1},\alpha^R) , 
& \fixedindex=N ,
\end{cases}
\end{align}
where $\alpha^L \in \LP_{N-\fixedindex}$ is the link pattern induced by the links of $\alpha$ 
on $\{2\fixedindex+1,\ldots,2N\}$ outside of the link $\{1,2\fixedindex\}$, and
$\alpha^R \in \LP_{\fixedindex-1}$ the link pattern induced by the links of $\alpha$ 
on $\{2,3,\ldots,2\fixedindex-1\}$ nested by $\{1,2\fixedindex\}$.

\noindent
Combining~\eqref{eq: meander fac1} with the definition~\eqref{eqn::defqj} of $w_s(\alpha)$ shows that
\begin{itemize}
\item
if $\fixedindex=1$, we have
\begin{align} \label{eq: meander facb=1}
w_{\fixedindex}(\alpha)\frac{\fugacity(\kappa) }{\meanderMatrix(\vcenter{\hbox{\includegraphics[scale=0.8]{figures-arXiv/link62.pdf}}}_N,\alpha)}
=\frac{\fugacity(\kappa)}{ \meanderMatrix(\vcenter{\hbox{\includegraphics[scale=0.8]{figures-arXiv/link62.pdf}}}_{N-\fixedindex},\alpha^L)};
\end{align}
\item
if $2\le \fixedindex\le N-1$, we have
\begin{align} \label{eq: meander fac2bN-1}
w_{\fixedindex}(\alpha)\frac{\fugacity(\kappa) }{\meanderMatrix(\vcenter{\hbox{\includegraphics[scale=0.8]{figures-arXiv/link62.pdf}}}_N,\alpha)}
= \frac{\fugacity(\kappa)}{\meanderMatrix(\vcenter{\hbox{\includegraphics[scale=0.8]{figures-arXiv/link62.pdf}}}_{\fixedindex-1},\alpha^R)} \; \frac{\fugacity(\kappa)}{ \meanderMatrix(\vcenter{\hbox{\includegraphics[scale=0.8]{figures-arXiv/link62.pdf}}}_{N-\fixedindex},\alpha^L)} ;
\end{align}
\item
if $\fixedindex=N$, we have
\begin{align} \label{eq: meander facb=N}
w_{\fixedindex}(\alpha)\frac{\fugacity(\kappa) }{\meanderMatrix(\vcenter{\hbox{\includegraphics[scale=0.8]{figures-arXiv/link62.pdf}}}_N,\alpha)}
= \frac{\fugacity(\kappa)}{\meanderMatrix(\vcenter{\hbox{\includegraphics[scale=0.8]{figures-arXiv/link62.pdf}}}_{\fixedindex-1},\alpha^R)}.
\end{align}
\end{itemize}
\end{lemma}

\begin{proof}
The asserted properties are straightforward to check --- see Figure~\ref{fig:meander_cascade}.
\end{proof}

\begin{figure}[ht!]
\begin{subfigure}[b]{\textwidth}
\begin{center}
\includegraphics[width=0.7\textwidth]{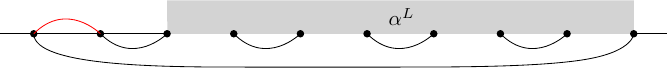}\\
$\downarrow$\\
\includegraphics[width=0.7\textwidth]{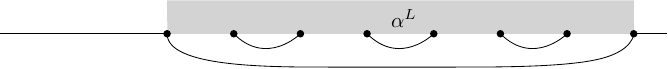}
\end{center}
\caption{Case: $\fixedindex=1$. }
\end{subfigure}
\vspace{0.5cm}

\begin{subfigure}[b]{\textwidth}
\begin{center}
\includegraphics[width=0.7\textwidth]{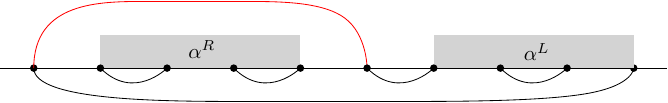}\\
$\downarrow$\\
\includegraphics[width=0.7\textwidth]{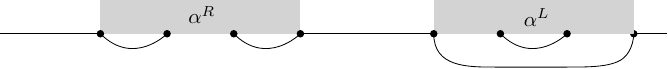}
\end{center}
\caption{Case: $2\le \fixedindex\le N-1$.}
\end{subfigure}
\vspace{0.5cm}

\begin{subfigure}[b]{\textwidth}
\begin{center}
\includegraphics[width=0.7\textwidth]{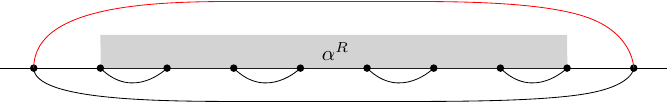}\\
$\downarrow$\\
\includegraphics[width=0.7\textwidth]{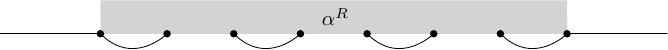}
\end{center}
\caption{Case: $\fixedindex=N$.}
\end{subfigure}

\caption{\label{fig:meander_cascade}
Illustration of the property~\eqref{eq: meander fac1} in Lemma~\ref{lem::meander}.
}
\end{figure}

\begin{proof}[Proof of Lemma~\ref{lem::cascade}]
Without loss of generality, we assume that $\Omega = \HH$ --- 
the formula for general polygons  follows by M\"obius covariance (Lemma~\ref{lem::PPF_COV}). 
Let $\{1,2\fixedindex\}\in\alpha$. 
Let us assume that $2\le \fixedindex\le N-1$; the other two cases can be proved similarly. 
On the one hand, by the construction in Definitions~\ref{def::Qn2b}~\&~\ref{def::QN}, 
and by Proposition~\ref{prop::Qn_2b_marginal}, using the notation from
Items~\ref{step: sample right curves}--\ref{step: sample left curves} in Definition~\ref{def::Qn2b}, 
we have 
\begin{align} 
\nonumber
\SLEmeasure^{x_1\to x_{2\fixedindex}}_N [\conn(\vec{\bs\eta}) = \alpha ]
= \; & 
\, \SLEmeasure^{x_1\to x_{2\fixedindex}}_N \bigg[ 
\Big(
\prod_{m=1}^r
\cev{\SLEmeasure}_{s^R_{m+1}-s^R_{m}} \big[ \conn(\bs\gamma_m^R) = \alpha_m^R \big]
\Big)
\Big( 
\prod_{m=1}^l
\vec{\SLEmeasure}_{s^L_{m+1}-s^L_{m}} \big[ \conn(\bs\gamma_m^L) = \alpha_m^L \big]
\Big)
\bigg] \\
\nonumber
= \; & 
\frac{(x_{2\fixedindex}-x_1)^{-2h}}{\hat{\LR}_{\omega_\fixedindex}(\bs{x})} \;
\E(\HH;x_1,x_{2\fixedindex}) \bigg[  \one\{ \LE_\alpha(\eta) \}  \, 
\Big(
\prod_{m=1}^r
\frac{\cev{S}_m \!\!\!\!\!{}^R}{\fugacity(\kappa)}
\Big)
\Big( 
\prod_{m=1}^l
\frac{\vec{S}_m^L}{\fugacity(\kappa)}
\Big)
\bigg] 
 \\
\label{eq: factorize P}
= \; & 
\frac{\vec{\PartF}_{\vcenter{\hbox{\includegraphics[scale=0.2]{figures-arXiv/link-0.pdf}}}}(\HH;x_{1},x_{2\fixedindex}) }{\hat{\LR}_{\omega_\fixedindex}(\bs{x})} \;
\E(\HH;x_1,x_{2\fixedindex}) \bigg[  \one\{ \LE_\alpha(\eta) \} \, 
\Big(
\prod_{m=1}^r
\frac{\cev{S}_m \!\!\!\!\!{}^R}{\fugacity(\kappa)}
\Big)
\Big( 
\prod_{m=1}^l
\frac{\vec{S}_m^L}{\fugacity(\kappa)}
\Big)
\bigg] ,
\end{align}
where $\{\alpha_1^R, \ldots , \alpha_r^R, \alpha_1^L, \ldots, \alpha_l^L\} = \{\alpha^D \colon D \textnormal{ c.c of } \HH \setminus \eta\}$, 
and where $\E(\HH;x_1,x_{2\fixedindex})$ is the expectation under the chordal $\SLE_\kappa$ in $(\HH;x_1,x_{2\fixedindex})$,
and we write 
\begin{align*} 
\cev{\SLEmeasure}_{s^R_{m+1}-s^R_{m}}\big( H_m^R; x_{2s^R_m},\ldots, x_{2s^R_{m+1}-1} \big) \big[ \conn(\bs\gamma_m^R) = \alpha_m^R \big] 
= \; & \frac{\cev{S}_m \!\!\!\!\!{}^R}{\coulombGas_{\!\! \vcenter{\hbox{\includegraphics[scale=0.8]{figures-arXiv/link62.pdf}}}_{s_{m+1}^R-s_m^R}}\big( H_m^R; x_{2s^R_m},\ldots, x_{2s^R_{m+1}-1} \big)} , \\
\vec{\SLEmeasure}_{s^L_{m+1}-s^L_{m}}\big( H_m^L; x_{2s^L_m+1},\ldots, x_{2s^L_{m+1}} \big) \big[ \conn(\bs\gamma_m^L) = \alpha_m^L \big] 
= \; & \frac{\vec{S}_m^L}{\coulombGas_{\!\! \vcenter{\hbox{\includegraphics[scale=0.8]{figures-arXiv/link62.pdf}}}_{s_{m+1}^L-s_m^L}}\big( H_m^L; x_{2s^L_m+1},\ldots, x_{2s^L_{m+1}} \big)} ,
\end{align*}
with $H_m^R = \smash{\phi^{-1}(\Omega_m^R)}$ and $H_m^L = \smash{\phi^{-1}(\Omega_m^L)}$, 
where $\phi$ is the M\"obius transformation~\eqref{eqn::Mobius_def} 
(depending on $\bs x$).
On the other hand, using the factorization formula 
\begin{align}\label{eqn::dec_leftright_link1}
\begin{split}
\frac{\fugacity(\kappa)}{\meanderMatrix(\vcenter{\hbox{\includegraphics[scale=0.8]{figures-arXiv/link62.pdf}}}_{\fixedindex-1},\alpha^R)}=\prod_{1\le m\le r}\frac{\fugacity(\kappa)}{\meanderMatrix(\vcenter{\hbox{\includegraphics[scale=0.8]{figures-arXiv/link62.pdf}}}_{s^R_{m+1}-s_m^R},\alpha_m^R)},\\
\frac{\fugacity(\kappa)}{\meanderMatrix(\vcenter{\hbox{\includegraphics[scale=0.8]{figures-arXiv/link62.pdf}}}_{N-\fixedindex},\alpha^L)}=\prod_{1\le m\le l}\frac{\fugacity(\kappa)}{\meanderMatrix(\vcenter{\hbox{\includegraphics[scale=0.8]{figures-arXiv/link62.pdf}}}_{s^L_{m+1}-s^L_m},\alpha_m^L)}.
\end{split}
\end{align}
(which is easy to verify, see Figure~\ref{fig:meander_split} for an illustration), 
we obtain
\begin{align*}
\vec{\PartF}_\alpha(\HH;\bs{x}) 
= \; & w_\fixedindex(\alpha) \, 
\frac{\fugacity(\kappa)}{\meanderMatrix(\vcenter{\hbox{\includegraphics[scale=0.8]{figures-arXiv/link62.pdf}}}_N,\alpha)} 
\, \hat{\LR}_{\omega_\fixedindex}(\bs x) \;  \SLEmeasure^{x_1\to x_{2\fixedindex}}_N [\conn(\vec{\bs\eta}) = \alpha ]
&& \textnormal{[by~\eqref{eqn::usefulexpression}]}
\\
= \; &
\frac{\fugacity(\kappa)}{\meanderMatrix(\vcenter{\hbox{\includegraphics[scale=0.8]{figures-arXiv/link62.pdf}}}_{\fixedindex-1},\alpha^R)} \; 
\frac{\fugacity(\kappa)}{ \meanderMatrix(\vcenter{\hbox{\includegraphics[scale=0.8]{figures-arXiv/link62.pdf}}}_{N-\fixedindex},\alpha^L)}
\; \hat{\LR}_{\omega_\fixedindex}(\bs x)
\; \SLEmeasure^{x_1\to x_{2\fixedindex}}_N [\conn(\vec{\bs\eta}) = \alpha ]
&& \textnormal{[by~\eqref{eq: meander fac2bN-1}]}
\\
= \; &
\vec{\PartF}_{\vcenter{\hbox{\includegraphics[scale=0.2]{figures-arXiv/link-0.pdf}}}}(\HH;x_{1},x_{2\fixedindex}) \; 
\E(\HH;x_1,x_{2\fixedindex}) \bigg[  \one\{ \LE_\alpha(\eta) \}  \, 
\Big(
\prod_{m=1}^r
\cev{\PartF}_{\alpha_m^R} 
\Big)
\Big( 
\prod_{m=1}^l
\vec{\PartF}_{\alpha_m^L} 
\Big)
\bigg] ,
\end{align*}
where to derive the last equality, 
we combine~\eqref{eq: factorize P} and~\eqref{eqn::dec_leftright_link1} and 
using the induction hypothesis that Lemma~\ref{lem::PPF_using_CGI} holds for all $n \le N-1$, writing 
\begin{align*}
\cev{\PartF}_{\alpha_m^R} 
= \; & \cev{\PartF}_{\alpha_m^R}\big( H_m^R ; x_{2s^R_m},\ldots, x_{2s^R_{m+1}-1} \big) 
= \frac{\cev{S}_m \!\!\!\!\!{}^R}{\meanderMatrix(\vcenter{\hbox{\includegraphics[scale=0.8]{figures-arXiv/link62.pdf}}}_{s_{m+1}^R-s_m^R},\alpha_m^R)}
&& \textnormal{[by~\eqref{eqn::PPF_using_CGI}]}
\\ 
\vec{\PartF}_{\alpha_m^L} 
= \; & \vec{\PartF}_{\alpha_m^L}\big( H_m^L ; x_{2s^L_m},\ldots, x_{2s^L_{m+1}-1} \big) 
= \frac{\vec{S}_m^L}{\meanderMatrix(\vcenter{\hbox{\includegraphics[scale=0.8]{figures-arXiv/link62.pdf}}}_{s_{m+1}^L-s_m^L},\alpha_m^L)} .
&& \textnormal{[by~\eqref{eqn::PPF_using_CGI}]}.
\end{align*}

Lastly, invoking the induction hypothesis that the reversibility Lemma~\ref{lem::difpartition} holds for all $n \le N-1$,
we have $\smash{\cev{\PartF}_{\alpha_m^R} = \vec{\PartF}_{\alpha_m^R}}$, 
so we conclude that the asserted formula~\eqref{eqn::cascade_relation} with $b = 2\fixedindex$ holds. 
Induction on the number of curves and the same arguments as in the proof of~\cite[Lemma~6.2]{Wu:Convergence_of_the_critical_planar_ising_interfaces_to_hypergeometric_SLE} 
or~\cite[Proposition~B.1]{Peltola:Towards_CFT_for_SLEs} 
then show that the cascade formula~\eqref{eqn::cascade_relation} 
also holds for the other links $\{a,b\} \in \alpha$.  
\end{proof}
\begin{figure}[ht!]
\includegraphics[width=0.9\textwidth]{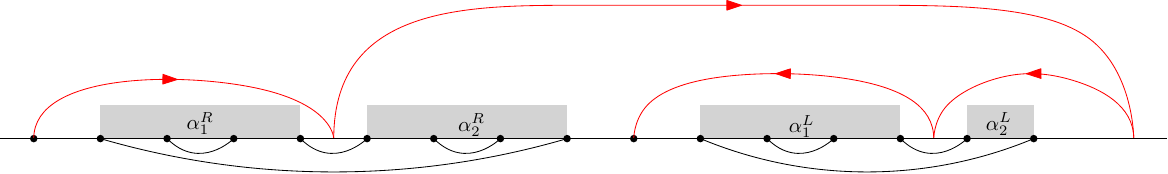}\\
$\downarrow$\\
\includegraphics[width=0.9\textwidth]{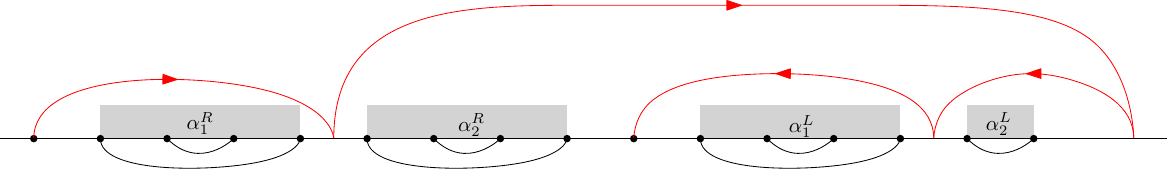}
\caption{\label{fig:meander_split}
Illustration of the relation~\eqref{eqn::dec_leftright_link1}.
}
\end{figure}

\subsection{General M\"obius covariance}

We next strengthen the M\"obius covariance (Lemma~\ref{lem::PPF_COV_GEN}) into 
a covariance property under any M\"obius map $\varphi \in \mathrm{PSL}(2,\R)$ preserving the upper half-plane, which may rotate the marked boundary points. 

\begin{lemma}\label{lem::PPF_COV_GEN}
Fix $\kappa \in (4,8)$ and $\alpha \in \LP_N$. The function $\smash{\vec{\PartF}_{\alpha}}$ of Definition~\ref{def::Zalpha} 
satisfies the M\"obius covariance 
\begin{align}\label{eqn::PPF_COV_GEN}
\vec{\PartF}_{\alpha} (\bs{x}) 
\; = \; \prod_{i=1}^{2N} |\varphi'(x_{i})|^{h} \; \times
\; \vec{\PartF}_{\varphi(\alpha)} (\varphi(\bs{x})) , 
\qquad \bs{x} \in \chamber_{2N} ,
\end{align}
for all M\"obius maps $\varphi \in \mathrm{PSL}(2,\R)$,  
where 
$\varphi(\bs{x}) := (\varphi(x_{j+1}) \ldots , \varphi(x_{2N}) , \varphi(x_{1}) , \ldots , \varphi(x_{j}))$, 
so that 
\begin{align*}
\varphi(x_{j+1}) < \varphi(x_{j+2}) < \cdots < \varphi(x_{2N}) < \varphi(x_{1}) < \varphi(x_{2}) < \cdots  < \varphi(x_{j})
\end{align*} 
with some $j \in \{0,1,2,\ldots,2N-1\}$,
and where $\varphi(\alpha) \in \LP_N$ is the link pattern obtained from $\alpha$ via permuting the indices according to the permutation of the boundary points induced by $\varphi$ and reordering. 
\end{lemma}

\begin{proof}
We prove~\eqref{eqn::PPF_COV_GEN} by induction on $N \geq 1$. 
The base case $N=1$ is clear for $\vec{\PartF}_{\vcenter{\hbox{\includegraphics[scale=0.2]{figures-arXiv/link-0.pdf}}}}(x_1,x_2) = (x_2-x_1)^{-2h}$.
Assuming that~\eqref{eqn::PPF_COV_GEN} holds for the collection 
$\{\vec{\PartF}_{\alpha} \colon \alpha\in\LP_n , \; n \le N-1\}$ and under 
the induction hypothesis that Lemmas~\ref{lem::difpartition} and~\ref{lem::PPF_using_CGI} hold for all $n \le N-1$,
we may apply the cascade property~\eqref{eqn::cascade_relation} from Lemma~\ref{lem::cascade} 
to obtain 
\begin{align*}
\vec{\PartF}_{\varphi(\alpha)} (\varphi(\bs{x}))
= \; & (\varphi(x_{2\fixedindex}) - \varphi(x_1))^{-2h} \; 
\E(\HH;\varphi(x_1),\varphi(x_{2\fixedindex})) \big[ \one\{ \LE_{\varphi(\alpha)}(\eta) \} \vec{\PartF}_{\varphi(\alpha/\{1,2\fixedindex\})}(\varphi(\hat{\HH}_\eta); \varphi(\bs{\ddot{x}}_{1,2\fixedindex}) ) \big]
\\
= \; & \prod_{i=1}^{2N} |\varphi'(x_{i})|^{-h} \; \times
(x_{2\fixedindex} - x_1)^{-2h} \; 
\E(\HH;x_1,x_{2\fixedindex}) \big[ \one\{ \LE_\alpha(\eta) \} \, \vec{\PartF}_{\alpha/\{1,2\fixedindex\}}(\hat{\HH}_\eta; \bs{\ddot{x}}_{1,2\fixedindex}) \big]
\\
= \; & \prod_{i=1}^{2N} |\varphi'(x_{i})|^{-h} \; \times \vec{\PartF}_{\alpha} (\bs{x}) ,
\qquad \alpha\in\LP_N ,\; \{1,2\fixedindex\} \in \alpha ,
\end{align*}
where $\E(\HH;x_1,x_{2\fixedindex})$ is the expected value under the chordal $\SLE_\kappa$ in $(\HH;x_1,x_{2\fixedindex})$.
\end{proof}

\subsection{Smoothness and hypoelliptic BPZ PDEs}
\label{subsec::PDEs_hypo}

We first verify the continuity and positivity of the functions $\vec{\PartF}_\alpha \colon \chamber_{2N} \to (0,\infty)$ (Lemma~\ref{lem::ppf_continuity_and_positivity}). 
We will promote the continuity to smoothness by hypoellipticity arguments. 

\begin{lemma}\label{lem::ppf_continuity_and_positivity}
The functions $\vec{\PartF}_\alpha \colon \chamber_{2N} \to (0,\infty)$ are positive and continuous\textnormal{:}
\begin{align}\label{eq::ppf_continuity}
\lim_{\substack{y_j \to x_j \\ \forall 1\le j\le 2N}} \vec{\PartF}_\alpha(\bs{y})
\; = \; \vec{\PartF}_\alpha(\bs{x}) 
 \qquad \textnormal{for any } \alpha\in \LP_N \textnormal{ and } \bs x \in \chamber_{2N} .
\end{align}
\end{lemma}

\begin{proof}
By Definition~\ref{def::Zalpha}, we have
\begin{align*}
\vec{\PartF}_{\alpha}(\bs{x}) 
= \; & \LK_{\vcenter{\hbox{\includegraphics[scale=0.8]{figures-arXiv/link62.pdf}}}_N} (\bs{x}) \, 
\frac{\vec{\SLEmeasure}_N[\conn(\vec{\bs\eta}) = \alpha ]}{\meanderMatrix(\vcenter{\hbox{\includegraphics[scale=0.8]{figures-arXiv/link62.pdf}}}_N,\alpha)}
, \qquad \alpha \in \LP_N , \\
\LK_{\vcenter{\hbox{\includegraphics[scale=0.8]{figures-arXiv/link62.pdf}}}_N} (\bs{x}) 
= \; & \fugacity(\kappa) \, 
\sum^N_{s=1} \hat{\LR}_{\omega_s}(\bs{x}) \, \sum_{\substack{\beta \in\LP_N  \\ \{1,2s\}\in\beta}} 
w_s(\beta) \, \SLEmeasure^{x_1\to x_{2s}}_N [\conn(\vec{\bs\eta}) = \beta ] 
, \qquad \bs{x}\in\chamber_{2N} .
\end{align*}
Now, the constants $\fugacity(\kappa) > 0$ and $w_s(\beta) > 0$ are positive in the current range $\kappa \in (4,8)$, and the functions $\hat{\LR}_{\omega_s}(\bs x) > 0$ are positive on $\chamber_{2N} \ni \bs x$ by Lemma~\ref{lem::another_CGI_with_conjugate_charge}. 
Moreover, we have $\smash{\vec{\SLEmeasure}_N[\conn(\vec{\bs\eta}) = \alpha ] > 0}$ 
and $\smash{\SLEmeasure^{x_1\to x_{2s}}_N [\conn(\vec{\bs\eta}) = \beta ] > 0}$ for every $\alpha\in\LP_N$ and $\bs{x}\in\chamber_{2N}$
respectively by Lemmas~\ref{lem::QN_positive} and~\ref{lem::Qn_2b_positive}. 
Hence, we have $\LK_{\vcenter{\hbox{\includegraphics[scale=0.8]{figures-arXiv/link62.pdf}}}_N}(\bs{x}) > 0$, which implies that $\vec{\PartF}_{\alpha}(\bs x) > 0$, 
for all $\bs x \in \chamber_{2N}$.

We prove the continuity~\eqref{eq::ppf_continuity} by induction on $N \geq 1$. 
The base case $\vec{\PartF}_{\vcenter{\hbox{\includegraphics[scale=0.2]{figures-arXiv/link-0.pdf}}}}(x_1,x_2) = (x_2-x_1)^{-2h}$ is clear.
Assuming that all functions in the collection $\{\vec{\PartF}_{\alpha} \colon \alpha\in\LP_n, \; n \le N-1\}$ are continuous on $\chamber_{2N}$, we use the cascade property (Lemma~\ref{lem::cascade}) to perform the induction step.
Fix $\alpha \in \LP_N$ and suppose $\{1,2\} \in \alpha$, without loss of generality by Lemma~\ref{lem::PPF_COV_GEN}.
Let $\varphi_{\bs y}$ be the M\"obius transformation such that $\varphi_{\bs y}(y_1) = x_1$, 
$\varphi_{\bs y}(y_{2}) = x_{2}$, and $\varphi_{\bs y}(\infty) = \infty$. 
Denote $\bs{\tilde{x}} = \varphi_{\bs y}(\bs y)$. 
Note that if $y_j\to x_j$ for all $1\le j\le 2N$, then $\varphi_{\bs y}(z)-z \to 0$ uniformly on $\overline\HH$,
and in particular $\varphi_{\bs y}'(z)\to 1$ on $\overline\HH$. 
Thus, by the M\"obius covariance in Lemma~\ref{lem::PPF_COV}, it suffices to prove that
\begin{align} \label{eq: limit tilde}
\lim_{\substack{\tilde x_j \to x_j \\ \forall 3\le j\le 2N}}
\vec{\PartF}(\bs{\tilde{x}}) = \vec{\PartF}_\alpha (\bs x) . 
\end{align}
To this end, using the cascade property~\eqref{eqn::cascade_relation} from Lemma~\ref{lem::cascade} and Fatou's lemma, we have
\begin{align} \label{eq: take limit of cascade}
\liminf_{\substack{\tilde x_j \to x_j \\ \forall 3\le j\le 2N}} \vec{\PartF}(\bs{\tilde{x}})
\geq \; & (x_2 - x_1)^{-2h} \; 
\E(\HH;x_1,x_2) \left[ \liminf_{\substack{\tilde x_j \to x_j \\ \forall 3\le j\le 2N}}
\one\{ \LE_\alpha(\eta) \} \, \vec{\PartF}_{\alpha/\{1,2\}}(\tilde{\Omega}_\eta; \tilde x_3 , \ldots , \tilde x_{2N} ) \right] ,
\end{align}
where $\tilde{\Omega}_\eta$ is the union of those connected components of 
$\HH \setminus \eta$ that contain some of the points $\{\tilde x_3, \ldots, \tilde x_{2N}\}$ on their boundary, 
and the event $\LE_\alpha(\eta)$ also implicitly depends on these marked points. 
By the induction hypothesis, we see that the right-hand side of~\eqref{eq: take limit of cascade} equals
\begin{align}\label{eq: take limit inf of cascade}
\textnormal{\eqref{eq: take limit of cascade}} = 
(x_2 - x_1)^{-2h} \; 
\E(\HH;x_1,x_2) \bigg[ \one\{ \LE_\alpha(\eta) \} \, \vec{\PartF}_{\alpha/\{1,2\}}(\hat{\Omega}_\eta; x_3 , \ldots , x_{2N} ) \bigg]
= \vec{\PartF}(\bs{x}) 
\end{align}
for the marked points $\{x_3, \ldots, x_{2N}\}$, again by Lemma~\ref{lem::cascade}.  
Hence, using the identity~\eqref{eqn::usefulexpression}, 
because the explicit function $\hat{\LR}_{\omega_1}$ appearing in Lemma~\ref{lem::another_CGI_with_conjugate_charge} is clearly continuous, we obtain
\begin{align*}
\hat{\LR}_{\omega_1}(\bs x) 
= \lim_{\substack{\tilde x_j \to x_j \\ \forall 3\le j\le 2N}}
\hat{\LR}_{\omega_1}(\bs{\tilde{x}}) 
= \; &  \lim_{\substack{\tilde x_j \to x_j \\ \forall 3\le j\le 2N}} 
\hat{\LR}_{\omega_1}(\bs{\tilde{x}}) 
\sum_{\substack{\alpha\in\LP_N \\ \{1,2\} \in \alpha}} \SLEmeasure^{x_1\to x_{2}}_N(\HH; \bs{\tilde{x}}) [\conn(\vec{\bs\eta}) = \alpha ] \\
= \; & \tfrac{1}{\fugacity(\kappa)} \, 
\limsup_{\substack{\tilde x_j \to x_j \\ \forall 3\le j\le 2N}}\sum_{\substack{\alpha\in\LP_N \\ \{1,2\} \in \alpha}} 
\meanderMatrix(\vcenter{\hbox{\includegraphics[scale=0.8]{figures-arXiv/link62.pdf}}}_N,\alpha) \;
  \vec{\PartF}_{\alpha}(\bs{\tilde{x}}) 
&& \textnormal{[by~\eqref{eqn::usefulexpression}~\&~\eqref{eqn::defqj}]} 
\\
\geq \; & \tfrac{1}{\fugacity(\kappa)} \, 
\sum_{\substack{\alpha\in\LP_N \\ \{1,2\} \in \alpha}} 
\meanderMatrix(\vcenter{\hbox{\includegraphics[scale=0.8]{figures-arXiv/link62.pdf}}}_N,\alpha) \;
\liminf_{\substack{\tilde x_j \to x_j \\ \forall 3\le j\le 2N}}  \vec{\PartF}_{\alpha}(\bs{\tilde{x}}) 
\\
\geq \; & \tfrac{1}{\fugacity(\kappa)} \, 
\sum_{\substack{\alpha\in\LP_N \\ \{1,2\} \in \alpha}} 
\meanderMatrix(\vcenter{\hbox{\includegraphics[scale=0.8]{figures-arXiv/link62.pdf}}}_N,\alpha) \;
\vec{\PartF}(\bs{x}) .
&& \textnormal{[by~\eqref{eq: take limit of cascade}~\&~\eqref{eq: take limit inf of cascade}]} 
\\
= \; & \hat{\LR}_{\omega_1}(\bs x) ,
&& \textnormal{[by~\eqref{eqn::usefulexpression}~\&~\eqref{eqn::defqj}]} 
\end{align*}
which implies~\eqref{eq: limit tilde} since $\meanderMatrix(\vcenter{\hbox{\includegraphics[scale=0.8]{figures-arXiv/link62.pdf}}}_N,\alpha) > 0$ for all $\alpha\in\LP_N$ (for the current range $\kappa \in (4,8)$. 
\end{proof}

Let us next record the intuitive fact that, if $\{1,2\} \in \alpha$, 
then upon shrinking the points $x_1$ and $x_2$ together,
$\SLEmeasure^{x_1\to x_{2}}_N [\conn(\vec{\bs\eta}) = \alpha ]$ tends to 
the corresponding probability for $\vec{\SLEmeasure}_{N-1}$.
In the proof, we make use of the continuity of the pure partition functions, Lemma~\ref{lem::ppf_continuity_and_positivity},
and the induction hypothesis that Lemma~\ref{lem::PPF_using_CGI} holds for all $n \le N-1$. 

\begin{lemma}\label{lem:proba_ASY}
Suppose $\{1,2\} \in \alpha \in \LP_N$ and write $\hat{\vec{\bs\eta}} = (\vec{\eta}_2, \ldots, \vec{\eta}_N)$.
Then, for all $\xi < x_3$, using the notation~\eqref{eqn::bs_notation}, we have
\begin{align*}
\lim_{x_1,x_2\to\xi} \SLEmeasure^{x_1\to x_{2}}_N(\HH; \bs{x}) [\conn(\vec{\bs\eta}) = \alpha ] 
\; = \;\; & \vec{\SLEmeasure}_{N-1}(\HH; \bs{\ddot{x}}_1) [\conn(\hat{\vec{\bs\eta}}) = \alpha/\{1,2\} ] .
\end{align*}
\end{lemma}

\begin{proof}
Write $\hat{\alpha} = \alpha/\{1,2\}$.
Using Proposition~\ref{prop::Qn_2b_marginal} (and Eq.~\eqref{eqn::RNderichordalSLE} from its proof), 
the construction of $\smash{\SLEmeasure^{x_1\to x_{2}}_N}$ in Definition~\ref{def::QN}, 
and the asymptotics property of $\hat{\LR}_{\omega_1}$ from Lemma~\ref{lem:R_ASY}, we obtain
\begin{align*}
\; & \lim_{x_1, x_2\to \xi} \SLEmeasure^{x_1\to x_{2}}_N(\HH; \bs{x})  [\conn(\vec{\bs\eta}) = \alpha ] \\
\ge \; & \lim_{x_1, x_2\to \xi}
\SLEmeasure^{x_1\to x_{2}}_N(\HH; \bs{x}) \big[ \vec\eta_1 \cap[x_3,+\infty) = \emptyset , \; \conn(\vec{\bs\eta}) = \alpha \big] \\
= \; & \lim_{x_1, x_2\to \xi} \frac{(x_{2}-x_1)^{-2h}}{\fugacity(\kappa) \, \hat{\LR}_{\omega_1}(\bs{x})} \; 
\E \Big[ 
\one\{\phi^{-1}(\cev\eta) \cap[x_3,+\infty) = \emptyset\} \; 
\coulombGas_{\!\! \vcenter{\hbox{\includegraphics[scale=0.8]{figures-arXiv/link62.pdf}}}_{N-1}}(H_\eta;\bs{\ddot{x}}_1)
\; \vec{\SLEmeasure}_{N-1}(H_\eta ; \bs{\ddot{x}}_1) [\conn(\hat{\vec{\bs\eta}}) = \hat{\alpha} ] \Big] 
\\
\ge \; & \E \bigg[ \liminf_{x_1, x_2\to \xi} \,  
\one\{\phi^{-1}(\cev\eta) \cap[x_3,+\infty) = \emptyset\} \; 
\frac{\coulombGas_{\!\! \vcenter{\hbox{\includegraphics[scale=0.8]{figures-arXiv/link62.pdf}}}_{N-1}}(H_\eta;\bs{\ddot{x}}_1)}{\coulombGas_{\!\! \vcenter{\hbox{\includegraphics[scale=0.8]{figures-arXiv/link62.pdf}}}_{N-1}}(\bs{\ddot{x}}_1)} \; 
\; \vec{\SLEmeasure}_{N-1}(H_\eta ; \bs{\ddot{x}}_1) [\conn(\hat{\vec{\bs\eta}}) = \hat{\alpha} ] \bigg] ,
\end{align*}
where $\E$ is the expected value for the chordal $\SLE_\kappa$ curve $\eta$ in $(\HH;0,\infty)$ 
and $H_\eta = \phi^{-1}(\Omega_\eta)$, where $\phi$ is the M\"obius transformation~\eqref{eqn::Mobius_def} 
(depending on $\bs x$) 
and $\Omega_\eta$ is the component of the complement $\HH\setminus\eta$ containing $[0,\infty)$.
We analyze the limits of the three factors inside the expectation separately.
\begin{itemize}
\item 
As the points $x_1$ and $x_2$ shrink together, 
we have $\PP$-almost surely 
$\one\{\phi^{-1}(\cev\eta) \cap[x_3,+\infty) = \emptyset\} \to 1$ 
because the curve $\phi^{-1}(\cev\eta)$ shrinks to a point.

\item 
Also, on this event the conformal covariance of $\smash{\coulombGas_{\!\! \vcenter{\hbox{\includegraphics[scale=0.8]{figures-arXiv/link62.pdf}}}_{N-1}}}$
(Corollary~\ref{cor: full Mobius covariance F} and Remark~\ref{rem: definition F polygon}) 
yields
\begin{align*}
\lim_{x_1, x_2\to \xi}
\frac{\coulombGas_{\!\! \vcenter{\hbox{\includegraphics[scale=0.8]{figures-arXiv/link62.pdf}}}_{N-1}}(H_\eta;\bs{\ddot{x}}_1)}{\coulombGas_{\!\! \vcenter{\hbox{\includegraphics[scale=0.8]{figures-arXiv/link62.pdf}}}_{N-1}}(\bs{\ddot{x}}_1)}
\; & = \lim_{x_1, x_2\to \xi} \;
\prod_{j=3}^{2N}|\varphi_\eta'(x_j)|^{h} \times 
\frac{\coulombGas_{\!\! \vcenter{\hbox{\includegraphics[scale=0.8]{figures-arXiv/link62.pdf}}}_{N-1}}(\varphi_\eta(\bs{\ddot{x}}_1))}{\coulombGas_{\!\! \vcenter{\hbox{\includegraphics[scale=0.8]{figures-arXiv/link62.pdf}}}_{N-1}}(\bs{\ddot{x}}_1)}
\; = \; 1 ,
\end{align*}
$\PP$-almost surely, where $\varphi_\eta$ is any conformal map from $H_\eta$ onto $\HH$ (so $\varphi_\eta$ converges to the identity map locally uniformly away from the boundary segment $[x_1,x_2]$).

\item 
Lastly, using the induction hypothesis that Lemma~\ref{lem::PPF_using_CGI} holds 
for $\coulombGas_{\!\! \vcenter{\hbox{\includegraphics[scale=0.8]{figures-arXiv/link62.pdf}}}_{N-1}}(\bs{\ddot{x}}_1) = \LK_{\!\! \vcenter{\hbox{\includegraphics[scale=0.8]{figures-arXiv/link62.pdf}}}_{N-1}}(\bs{\ddot{x}}_1)$, we have
\begin{align*}
\meanderMatrix(\vcenter{\hbox{\includegraphics[scale=0.8]{figures-arXiv/link62.pdf}}}_{N-1},\hat{\alpha}) \;
\lim_{\substack{y_j \to x_j \\ \forall 1\le j\le 2N-2}}
\frac{\vec{\PartF}_{\hat{\alpha}}(\bs y) }{\coulombGas_{\vcenter{\hbox{\includegraphics[scale=0.8]{figures-arXiv/link62.pdf}}}_{N-1}}(\bs y)}
= \lim_{\substack{y_j \to x_j \\ \forall 1\le j\le 2N-2}} \vec{\SLEmeasure}_{N-1}(\HH;\bs y)[\conn(\hat{\vec{\bs\eta}}) = \hat{\alpha} ] ,
\end{align*}
so it follows from Lemma~\ref{lem::ppf_continuity_and_positivity} 
and the continuity of the explicit function 
$\smash{\coulombGas_{\vcenter{\hbox{\includegraphics[scale=0.8]{figures-arXiv/link62.pdf}}}_N}}$ in~\eqref{eqn::CGI_def} that
\begin{align*}
\lim_{\substack{y_j \to x_j \\ \forall 1\le j\le 2N-2}}
\vec{\SLEmeasure}_{N-1}(\HH;\bs y) [\conn(\hat{\vec{\bs\eta}}) = \hat{\alpha} ]
\; = \; 
\vec{\SLEmeasure}_{N-1}(\HH;\bs x) [\conn(\hat{\vec{\bs\eta}}) = \hat{\alpha} ] 
, \qquad  \bs x \in \chamber_{2N} .
\end{align*}
Combining this with conformal invariance shows that, $\PP$-almost surely, we have
\begin{align*}
\lim_{x_1, x_2\to \xi} \vec{\SLEmeasure}_{N-1}(H_\eta ; \bs{\ddot{x}}_1) [\conn(\hat{\vec{\bs\eta}}) = \hat{\alpha} ]
\; = \; 
\lim_{x_1, x_2\to \xi} \vec{\SLEmeasure}_{N-1}(\HH ; \varphi_\eta(\bs{\ddot{x}}_1)) [\conn(\hat{\vec{\bs\eta}}) = \hat{\alpha} ]
\; = \; 
\vec{\SLEmeasure}_{N-1}(\HH ; \bs{\ddot{x}}_1) [\conn(\hat{\vec{\bs\eta}}) = \hat{\alpha} ] .
\end{align*}
\end{itemize}
We conclude that
\begin{align} \label{eq: limit of conn_proba}
\lim_{x_1, x_2\to \xi} \SLEmeasure^{x_1\to x_{2}}_N(\HH; \bs{x})  [\conn(\vec{\bs\eta}) = \alpha ] 
 \; \ge \; \vec{\SLEmeasure}_{N-1}(\HH ; \bs{\ddot{x}}_1) [\conn(\hat{\vec{\bs\eta}}) = \alpha/\{1,2\} ] .
\end{align}
To finish, note that the sum of total probability yields 
\begin{align*}
\sum_{\substack{\alpha\in\LP_N \\ \{1,2\}\in\alpha}}
\SLEmeasure^{x_1\to x_{2}}_N(\HH; \bs{x})  [\conn(\vec{\bs\eta}) = \alpha ]
\; = \; 1 
\; = \; 
\sum_{\substack{\alpha\in\LP_N \\ \{1,2\}\in\alpha}}
\vec{\SLEmeasure}_{N-1}(\HH ; \bs{\ddot{x}}_1) [\conn(\hat{\vec{\bs\eta}}) = \alpha/\{1,2\} ] ,
\end{align*}
which implies equality in~\eqref{eq: limit of conn_proba}. 
\end{proof}

\begin{proposition}\label{prop::PPF_PDE}
Fix $\kappa \in (4,8)$. 
The collection $\{\vec{\PartF}_{\alpha} \colon \alpha\in\LP_N\}$ of Definition~\ref{def::Zalpha} 
comprises smooth functions on $\chamber_{2N}$ which satisfy~\eqref{eqn::PDE}. 
\end{proposition}

\begin{proof}
Even though we already know that the functions $\{\hat{\LR}_{\omega_\fixedindex}(\bs{x}) \colon 1 \leq \fixedindex \leq N\}$ 
satisfy~\eqref{eqn::PDE} (Lemma~\ref{lem:PDECOVASY_for_LR}), 
because the probability measures $\smash{\SLEmeasure^{x_1\to x_{2\fixedindex}}_N}(\HH; \bs{x})$ 
and $\smash{\vec{\SLEmeasure}_N}(\HH; \bs{x})$ are not known to immediately satisfy any differential equations, 
or even to be differentiable, the verification of~\eqref{eqn::PDE} for the functions of Definition~\ref{def::Zalpha} requires a rather nontrivial argument. 
Under the induction hypothesis that Lemmas~\ref{lem::difpartition}~\&~\ref{lem::PPF_using_CGI} hold for all $n \le N-1$,
as detailed in the proof of~\cite[Lemma~B.4]{Peltola:Towards_CFT_for_SLEs}
or~\cite[Lemmas~6.3~\&~6.4]{Wu:Convergence_of_the_critical_planar_ising_interfaces_to_hypergeometric_SLE}, 
we may apply the M\"obius covariance from Lemma~\ref{lem::PPF_COV_GEN} 
and the cascade property~\eqref{eqn::cascade_relation} from Lemma~\ref{lem::cascade} 
to derive a local martingale for the evolution of the chordal $\SLE_\kappa$ curve.
Together with the continuity from Lemma~\ref{lem::ppf_continuity_and_positivity}, 
this shows that $\smash{\vec{\PartF}_{\alpha}}$ 
is a distributional solution to the PDE system~\eqref{eqn::PDE}
(see~\cite{Dubedat:SLE_and_Virasoro_representations_localization, 
Peltola-Wu:Global_and_local_multiple_SLEs_and_connection_probabilities_for_level_lines_of_GFF, AHSY:Conformal_welding_of_quantum_disks_and_multiple_SLE_the_non-simple_case, Karrila-Viitasaari:In_prep}).
Because the PDE system~\eqref{eqn::PDE} comprises hypoelliptic PDEs,  
we may conclude that any distributional solution to it is smooth~\cite{Hormander:The_analysis_of_linear_partial_differential_operators_1, Dubedat:SLE_and_Virasoro_representations_localization}. 
\end{proof}

\subsection{Relations of pure partition functions to Coulomb gas integrals}
\label{subsec:proof_of_PPF_vs_CGI}

Our proof of the asymptotics property in Section~\ref{subsec:ssymptotics_PPF} 
uses crucially the fact that the pure partition functions are intimately related to the Coulomb gas integrals discussed in Section~\ref{sec::CGI} and used in Section~\ref{sec::multipleSLEs}. 
The purpose of this section is to make this relation precise, that is, to verify Lemma~\ref{lem::PPF_using_CGI}.

We begin by checking the power-law bound~\eqref{eqn::PLB_weak_upper} for the pure partition functions (Lemma~\ref{lem::PPF_PLB_weak}).
This together with Lemma~\ref{lem::PPF_COV} and Proposition~\ref{prop::PPF_PDE} show that 
$\{\vec{\PartF}_{\alpha} \colon \alpha\in\LP_N\} \subset \mathcal{S}_N$ 
belong to the PDE solution space~\eqref{eq: solution space}, 
which enables us to invoke the uniqueness from Lemma~\ref{lem::PFuniqueness} for elements in $\mathcal{S}_N$.

\begin{lemma}\label{lem::PPF_PLB_weak}
Fix $\kappa \in (4,8)$. 
The collection $\{\vec{\PartF}_{\alpha} \colon \alpha\in\LP_N\}$ of Definition~\ref{def::Zalpha} satisfies the bound~\eqref{eqn::PLB_weak_upper}. 
\end{lemma}

\begin{proof} 
By Definition~\ref{def::Zalpha}, it suffices to bound $\LK_{\vcenter{\hbox{\includegraphics[scale=0.8]{figures-arXiv/link62.pdf}}}_N}$ in terms of 
the Coulomb gas integral functions $\hat{\LR}_{\omega_s}$, which satisfy the bound~\eqref{eqn::PLB_weak_upper} by construction (Lemma~\ref{lem:PDECOVASY_for_LR}):
\begin{align*}
\LK_{\vcenter{\hbox{\includegraphics[scale=0.8]{figures-arXiv/link62.pdf}}}_N}(\bs{x}) \; \lesssim \; \max_{1 \leq s \leq N}
\hat{\LR}_{\omega_s}(\bs{x}) ,
\end{align*}
where ``$\lesssim$'' holds up to some finite constant depending on $\kappa$ and $N$.
\end{proof}

Next, observe that by Eq.~(\ref{eqn::def_KN},~\eqref{eqn::usefulexpression}), we have
\begin{align}\label{eq: usefulexpression_for_K}
\LK_{\vcenter{\hbox{\includegraphics[scale=0.8]{figures-arXiv/link62.pdf}}}_N}(\bs{x}) 
= \; & \sum_{\alpha\in\LP_N} \meanderMatrix(\vcenter{\hbox{\includegraphics[scale=0.8]{figures-arXiv/link62.pdf}}}_N,\alpha) \, \vec{\PartF}_{\alpha} (\bs{x}) , \qquad \bs{x}\in\chamber_{2N} .
\end{align}
We analogously define
\begin{align}\label{eq: usefulexpression_for_L}
\LK_{\vcenter{\hbox{\includegraphics[scale=0.8]{figures-arXiv/link61.pdf}}}_N} (\bs{x})
:= \sum_{\alpha\in\LP_N} \meanderMatrix(\vcenter{\hbox{\includegraphics[scale=0.8]{figures-arXiv/link61.pdf}}}_N,\alpha) \, \vec{\PartF}_\alpha (\bs{x}) , \qquad \bs{x}\in\chamber_{2N} .
\end{align}
These functions are naturally related by rotation.

\begin{lemma}\label{lem::K_COV_GEN}
Fix $\kappa \in (4,8)$. 
The collections $\{\LK_{\vcenter{\hbox{\includegraphics[scale=0.8]{figures-arXiv/link62.pdf}}}_N} \colon N \in \bZnn\}$ and $\{\LK_{\vcenter{\hbox{\includegraphics[scale=0.8]{figures-arXiv/link61.pdf}}}_N} \colon N \in \bZnn\}$ 
in~\textnormal{(}\ref{eq: usefulexpression_for_K},~\ref{eq: usefulexpression_for_L}\textnormal{)} 
satisfy the M\"obius covariance 
\begin{align}\label{eqn::K_COV_GEN}
\LK_{\vcenter{\hbox{\includegraphics[scale=0.8]{figures-arXiv/link62.pdf}}}_N} (\bs{x}) 
\; = \; \prod_{i=1}^{2N} |\varphi'(x_{i})|^{h} \; \times \; 
\begin{cases}
\LK_{\vcenter{\hbox{\includegraphics[scale=0.8]{figures-arXiv/link61.pdf}}}_N} (\varphi(\bs{x})) ,
& \textnormal{if } j \in \{2,4,\ldots,2N-2 \} , \\
\LK_{\vcenter{\hbox{\includegraphics[scale=0.8]{figures-arXiv/link62.pdf}}}_N} (\varphi(\bs{x}))  ,
& \textnormal{if } j \in \{1,3,\ldots,2N-1 \} , 
\end{cases}
, \qquad \bs{x} \in \chamber_{2N} ,
\end{align}
for all M\"obius maps $\varphi \in \mathrm{PSL}(2,\R)$,  
where $\varphi(\bs{x}) := (\varphi(x_{j}) \ldots , \varphi(x_{2N}) , \varphi(x_{1}) , \ldots , \varphi(x_{j-1}))$, 
so that 
\begin{align*}
\varphi(x_{j}) < \varphi(x_{j+1}) < \cdots < \varphi(x_{2N}) < \varphi(x_{1}) < \varphi(x_{2}) < \cdots  < \varphi(x_{j-1}) .
\end{align*} 
\end{lemma}

\begin{proof}
For any M\"obius maps $\varphi \in \mathrm{PSL}(2,\R)$, using rotation symmetry of the meander matrix~\eqref{eqn::meandermatrix_def} and the M\"obius covariance~\eqref{eqn::PPF_COV_GEN} of $\vec{\PartF}_{\alpha}$ from Lemma~\ref{lem::PPF_COV_GEN}, using~\eqref{eq: usefulexpression_for_K}, we obtain
\begin{align*}
\LK_{\vcenter{\hbox{\includegraphics[scale=0.8]{figures-arXiv/link62.pdf}}}_N} (\bs{x}) 
= \; & \sum_{\alpha\in\LP_N} \meanderMatrix(\vcenter{\hbox{\includegraphics[scale=0.8]{figures-arXiv/link62.pdf}}}_N,\alpha) \, \vec{\PartF}_{\alpha} (\bs{x}) 
&& \textnormal{[by~\eqref{eq: usefulexpression_for_K}]} 
\\
= \; & \prod_{i=1}^{2N} | \varphi'(x_{i}) |^{h} \; \times
\sum_{\alpha\in\LP_N} \meanderMatrix(\vcenter{\hbox{\includegraphics[scale=0.8]{figures-arXiv/link62.pdf}}}_N,\alpha) \, \vec{\PartF}_{ \varphi(\alpha) } (\varphi(\bs{x})) 
&& \textnormal{[by~\eqref{eqn::PPF_COV_GEN}]} 
\\
= \; & \prod_{i=1}^{2N} | \varphi'(x_{i}) |^{h} \; \times
\sum_{\beta\in\LP_N} \meanderMatrix(\vcenter{\hbox{\includegraphics[scale=0.8]{figures-arXiv/link62.pdf}}}_N,\varphi^{-1}(\beta)) \, \vec{\PartF}_{\beta} (\varphi(\bs{x})) .
\end{align*}

\begin{itemize}
\item If $j$ is even, we have
$\meanderMatrix(\vcenter{\hbox{\includegraphics[scale=0.8]{figures-arXiv/link62.pdf}}}_N,\varphi^{-1}(\beta)) = \meanderMatrix(\varphi(\vcenter{\hbox{\includegraphics[scale=0.8]{figures-arXiv/link62.pdf}}})_N,\beta) = \meanderMatrix(\vcenter{\hbox{\includegraphics[scale=0.8]{figures-arXiv/link61.pdf}}}_N,\beta)$, and thus, 
\begin{align*}
\prod_{i=1}^{2N} | \varphi'(x_{i}) |^{-h} \; \times \;
\LK_{\vcenter{\hbox{\includegraphics[scale=0.8]{figures-arXiv/link62.pdf}}}_N} (\bs{x}) 
= \; & 
\sum_{\beta\in\LP_N} \meanderMatrix(\vcenter{\hbox{\includegraphics[scale=0.8]{figures-arXiv/link61.pdf}}}_N,\beta) \, \vec{\PartF}_{\beta} (\varphi(\bs{x})) 
\\
= \; & \LK_{\vcenter{\hbox{\includegraphics[scale=0.8]{figures-arXiv/link61.pdf}}}_N} (\varphi(\bs{x})) .
&& \textnormal{[by~\eqref{eq: usefulexpression_for_L}]} 
\end{align*}

\item If $j$ is odd, we have
$\meanderMatrix(\vcenter{\hbox{\includegraphics[scale=0.8]{figures-arXiv/link62.pdf}}}_N,\varphi^{-1}(\beta)) = \meanderMatrix(\varphi(\vcenter{\hbox{\includegraphics[scale=0.8]{figures-arXiv/link62.pdf}}})_N,\beta) = \meanderMatrix(\vcenter{\hbox{\includegraphics[scale=0.8]{figures-arXiv/link62.pdf}}}_N,\beta)$, and thus, 
\begin{align*}
\prod_{i=1}^{2N} | \varphi'(x_{i}) |^{-h} \; \times \;
\LK_{\vcenter{\hbox{\includegraphics[scale=0.8]{figures-arXiv/link62.pdf}}}_N} (\bs{x}) 
= \; & \sum_{\beta\in\LP_N} \meanderMatrix(\vcenter{\hbox{\includegraphics[scale=0.8]{figures-arXiv/link62.pdf}}}_N,\beta) \, \vec{\PartF}_{\beta} (\varphi(\bs{x})) 
\\
= \; & \LK_{\vcenter{\hbox{\includegraphics[scale=0.8]{figures-arXiv/link62.pdf}}}_N} (\varphi(\bs{x})) .
&& \textnormal{[by~\eqref{eq: usefulexpression_for_K}]} 
\end{align*}
\end{itemize}
This proves the asserted covariance property~\eqref{eqn::K_COV_GEN}.
\end{proof}

We are now ready to finish the induction step to prove Lemma~\ref{lem::PPF_using_CGI}.

\begin{lemma}
Fix $N \geq 2$. Suppose that the following hold: 
\begin{itemize}[leftmargin=6em]
\item[\textnormal{Lemma~\ref{lem::difpartition}:}]
$\vec{\PartF}_{\alpha} = \cev{\PartF}_{\alpha}$ for all $\alpha\in\LP_n$ and $n \leq N-1$\textnormal{;} 

\item[\textnormal{Lemma~\ref{lem::PPF_using_CGI}:}]
$\LK_{\vcenter{\hbox{\includegraphics[scale=0.8]{figures-arXiv/link62.pdf}}}_n} = \coulombGas_{\vcenter{\hbox{\includegraphics[scale=0.8]{figures-arXiv/link62.pdf}}}_n}$ for all $n \leq N-1$.
\end{itemize}
Then, we have
\begin{align*}
\PartF_{\alpha}(\bs x) = 
\vec{\PartF}_{\alpha}(\bs x) 
= \coulombGas_{\vcenter{\hbox{\includegraphics[scale=0.8]{figures-arXiv/link62.pdf}}}_N}(\bs x) \, 
\frac{\vec{\SLEmeasure}_N[\conn(\vec{\bs\eta}) = \alpha ]}{\meanderMatrix(\vcenter{\hbox{\includegraphics[scale=0.8]{figures-arXiv/link62.pdf}}}_N,\alpha)}
 , \qquad \bs{x}\in\chamber_{2N} , \; \alpha \in \LP_N .
\end{align*}
\end{lemma}

\begin{proof}
It suffices to prove that $\LK_{\vcenter{\hbox{\includegraphics[scale=0.8]{figures-arXiv/link62.pdf}}}_N} = \smash{\coulombGas_{\vcenter{\hbox{\includegraphics[scale=0.8]{figures-arXiv/link62.pdf}}}_N}}$ in~\eqref{eqn::def_KN}. 
To this end, we use the uniqueness from Lemma~\ref{lem::PFuniqueness} 
of elements in the PDE solution space $\mathcal{S}_N$~\eqref{eq: solution space} 
singled out by their asymptotic properties. 
Because $\smash{\coulombGas_{\vcenter{\hbox{\includegraphics[scale=0.8]{figures-arXiv/link62.pdf}}}_N}} \in \mathcal{S}_N$, to prove that $\LK_{\vcenter{\hbox{\includegraphics[scale=0.8]{figures-arXiv/link62.pdf}}}_N} = \smash{\coulombGas_{\vcenter{\hbox{\includegraphics[scale=0.8]{figures-arXiv/link62.pdf}}}_N}}$ it suffices to verify that $\LK_{\vcenter{\hbox{\includegraphics[scale=0.8]{figures-arXiv/link62.pdf}}}_N} \in \mathcal{S}_N$ 
and they satisfy the same recursive asymptotics property~\eqref{eqn::CGI_ASY}: 
\begin{align}
\label{eqn::CGI_ASY_for_K} 
\; & \lim_{x_j,x_{j+1}\to\xi} \frac{\LK_{\vcenter{\hbox{\includegraphics[scale=0.8]{figures-arXiv/link62.pdf}}}_N}(\bs{x})}{ (x_{j+1}-x_j)^{-2h} }
= 
\begin{cases}
\fugacity(\kappa) \, \LK_{\vcenter{\hbox{\includegraphics[scale=0.8]{figures-arXiv/link62.pdf}}}_{N-1}} (\bs{\ddot{x}}_j) ,
& \textnormal{if } j \in \{2,4,\ldots,2N-2 \} , \\
\LK_{\vcenter{\hbox{\includegraphics[scale=0.8]{figures-arXiv/link62.pdf}}}_{N-1}} (\bs{\ddot{x}}_j) ,
& \textnormal{if } j \in \{1,3,\ldots,2N-1 \} , 
\end{cases}
\end{align}
for all $\xi \in (x_{j-1}, x_{j+2})$ and for all $N \ge 1$ and $j \in \{1,2, \ldots, 2N-1 \}$. 
To this end, note first that $\vec{\PartF}_{\alpha} \in \mathcal{S}_N$ by Proposition~\ref{prop::PPF_PDE} and 
Lemmas~\ref{lem::PPF_COV}~\&~\ref{lem::PPF_PLB_weak}. 
Thus, we readily see that $\LK_{\vcenter{\hbox{\includegraphics[scale=0.8]{figures-arXiv/link62.pdf}}}_N} \in \mathcal{S}_N$ as a linear combination~\eqref{eq: usefulexpression_for_K}  of them.
To verify the asymptotics property~\eqref{eqn::CGI_ASY_for_K},
thanks to Lemma~\eqref{eqn::K_COV_GEN}, by conjugating by a suitable M\"obius map 
it suffices without loss of generality to check~\eqref{eqn::CGI_ASY_for_K} in the case $j=1$, 
and the following analogous property in the case $j=1$:
\begin{align}
\label{eqn::CGI_ASY_for_L} 
\; & \lim_{x_j,x_{j+1}\to\xi} \frac{\LK_{\vcenter{\hbox{\includegraphics[scale=0.8]{figures-arXiv/link61.pdf}}}_N}(\bs{x})}{ (x_{j+1}-x_j)^{-2h} }
= 
\begin{cases}
\LK_{\vcenter{\hbox{\includegraphics[scale=0.8]{figures-arXiv/link61.pdf}}}_{N-1}} (\bs{\ddot{x}}_j) ,
& \textnormal{if } j \in \{2,4,\ldots,2N-2 \} , \\
\fugacity(\kappa) \, \LK_{\vcenter{\hbox{\includegraphics[scale=0.8]{figures-arXiv/link61.pdf}}}_{N-1}} (\bs{\ddot{x}}_j) ,
& \textnormal{if } j \in \{1,3,\ldots,2N-1 \} .
\end{cases}
\end{align}

\begin{itemize}
\item For the first case (\eqref{eqn::CGI_ASY_for_K} with $j=1$), 
using the definition~\eqref{eqn::def_KN} 
and the asymptotics property of $\hat{\LR}_{\omega_s}$ from Lemma~\ref{lem:R_ASY}, we have
\begin{align*}
 \; & \lim_{x_1,x_2\to\xi} 
\frac{\LK_{\vcenter{\hbox{\includegraphics[scale=0.8]{figures-arXiv/link62.pdf}}}_N}(\bs x)}{ (x_2-x_1)^{-2h} } 
\\
= \; & \fugacity(\kappa) \, \lim_{x_1,x_2\to\xi} 
\sum^N_{s=1} \frac{\hat{\LR}_{\omega_s}(\bs x)}{ (x_2-x_1)^{-2h} } \, \sum_{\substack{\beta \in\LP_N  \\ \{1,2s\}\in\beta}} 
w_s(\beta) \, \SLEmeasure^{x_1\to x_{2s}}_N [\conn(\vec{\bs\eta}) = \beta ] 
&& \textnormal{[by~\eqref{eqn::def_KN}]}
\\
= \; & \fugacity(\kappa) \, \lim_{x_1,x_2\to\xi} \frac{\hat{\LR}_{\omega_1}(\bs x)}{ (x_2-x_1)^{-2h} }
\, \sum_{\substack{\beta \in\LP_N  \\ \{1,2\}\in\beta}} \SLEmeasure^{x_1\to x_{2}}_N [\conn(\vec{\bs\eta}) = \beta ] 
&& \textnormal{[by Lemma~\ref{lem:R_ASY} and~\eqref{eqn::defqj}]}
\\
= \; & \coulombGas_{\!\! \vcenter{\hbox{\includegraphics[scale=0.8]{figures-arXiv/link62.pdf}}}_{N-1}}(\bs{\ddot{x}}_1)
&& \textnormal{[by Lemma~\ref{lem:R_ASY}]}
\\
= \; & \LK_{\!\! \vcenter{\hbox{\includegraphics[scale=0.8]{figures-arXiv/link62.pdf}}}_{N-1}}(\bs{\ddot{x}}_1) ,
&& \textnormal{[by Lemma~\ref{lem::PPF_using_CGI}]}
\end{align*}
due to the induction hypothesis that Lemma~\ref{lem::PPF_using_CGI} holds 
for $\coulombGas_{\!\! \vcenter{\hbox{\includegraphics[scale=0.8]{figures-arXiv/link62.pdf}}}_{N-1}}(\bs{\ddot{x}}_1) = \LK_{\!\! \vcenter{\hbox{\includegraphics[scale=0.8]{figures-arXiv/link62.pdf}}}_{N-1}}(\bs{\ddot{x}}_1)$.

\item The second case (\eqref{eqn::CGI_ASY_for_L} with $j=1$) is slightly more elaborate.
 First, note that 
\begin{align} 
\label{eq: L limit 1}
\; & \lim_{x_1,x_2\to\xi} \frac{\LK_{\vcenter{\hbox{\includegraphics[scale=0.8]{figures-arXiv/link61.pdf}}}_N} (\bs{x})}{ (x_2-x_1)^{-2h} }  \\
\nonumber
:= \; & \lim_{x_1,x_2\to\xi}
\sum_{\alpha\in\LP_N} \meanderMatrix(\vcenter{\hbox{\includegraphics[scale=0.8]{figures-arXiv/link61.pdf}}}_N,\alpha) \, \frac{\vec{\PartF}_\alpha (\bs{x})}{ (x_2-x_1)^{-2h} }  
&& \textnormal{[by~\eqref{eq: usefulexpression_for_L}]}
\\
\nonumber
= \; & \fugacity(\kappa) \, \lim_{x_1,x_2\to\xi}
\sum^N_{s=1} \frac{\hat{\LR}_{\omega_s}(\bs x)}{ (x_2-x_1)^{-2h} } 
\, \sum_{\substack{\alpha \in\LP_N  \\ \{1,2s\}\in\alpha}} 
\frac{w_s(\alpha) \, \meanderMatrix(\vcenter{\hbox{\includegraphics[scale=0.8]{figures-arXiv/link61.pdf}}}_N,\alpha)}{\meanderMatrix(\vcenter{\hbox{\includegraphics[scale=0.8]{figures-arXiv/link62.pdf}}}_N,\alpha)}
\, \SLEmeasure^{x_1\to x_{2s}}_N [\conn(\vec{\bs\eta}) = \alpha ] .
&& \textnormal{[by~\eqref{eqn::usefulexpression}]}
\end{align}
By Lemma~\ref{lem:R_ASY}, only the $s=1$ term in the sum survives in the limit, with $w_1(\alpha) = 1$ by~\eqref{eqn::defqj}. 
From Lemmas~\ref{lem:R_ASY} and~\ref{lem:proba_ASY},
upon shrinking the points $x_1$ and $x_2$ together, the limits of interest are 
\begin{align} \label{eq: two limits}
\begin{split}
\lim_{x_1,x_2\to\xi} \SLEmeasure^{x_1\to x_{2}}_N(\HH; \bs{x}) [\conn(\vec{\bs\eta}) = \alpha ] 
\; = \;\; & \vec{\SLEmeasure}_{N-1}(\HH; \bs{\ddot{x}}_1) [\conn(\hat{\vec{\bs\eta}}) = \alpha/\{1,2\} ] ,  \\
\lim_{x_1,x_2\to\xi} \frac{\hat{\LR}_{\omega_1}(\bs x)}{ (x_2-x_1)^{-2h} } 
\; = \; \; & \tfrac{1}{\fugacity(\kappa)} \,
\coulombGas_{\!\! \vcenter{\hbox{\includegraphics[scale=0.8]{figures-arXiv/link62.pdf}}}_{N-1}}(\bs{\ddot{x}}_1) 
\; = \; \tfrac{1}{\fugacity(\kappa)} \,
\LK_{\!\! \vcenter{\hbox{\includegraphics[scale=0.8]{figures-arXiv/link62.pdf}}}_{N-1}}(\bs{\ddot{x}}_1) ,
\end{split}
\end{align}
where we write $\hat{\vec{\bs\eta}} = (\vec{\eta}_2, \ldots, \vec{\eta}_N)$, 
and use again the induction hypothesis that Lemma~\ref{lem::PPF_using_CGI} holds 
for $\coulombGas_{\!\! \vcenter{\hbox{\includegraphics[scale=0.8]{figures-arXiv/link62.pdf}}}_{N-1}}(\bs{\ddot{x}}_1) = \LK_{\!\! \vcenter{\hbox{\includegraphics[scale=0.8]{figures-arXiv/link62.pdf}}}_{N-1}}(\bs{\ddot{x}}_1)$. 
Thus, we obtain
\begin{align*}
\textnormal{\eqref{eq: L limit 1}}
= \; & \LK_{\!\! \vcenter{\hbox{\includegraphics[scale=0.8]{figures-arXiv/link62.pdf}}}_{N-1}}(\bs{\ddot{x}}_1)
\, \sum_{\substack{\alpha \in\LP_N  \\ \{1,2\}\in\alpha}} 
\frac{\meanderMatrix(\vcenter{\hbox{\includegraphics[scale=0.8]{figures-arXiv/link61.pdf}}}_N,\alpha)}{\meanderMatrix(\vcenter{\hbox{\includegraphics[scale=0.8]{figures-arXiv/link62.pdf}}}_N,\alpha)}
\, \vec{\SLEmeasure}_{N-1}(\HH; \bs{\ddot{x}}_1) [\conn(\hat{\vec{\bs\eta}}) = \alpha/\{1,2\} ] 
&& \textnormal{[by~\eqref{eq: two limits}]}
\\
= \; & \fugacity(\kappa) \,
\LK_{\!\! \vcenter{\hbox{\includegraphics[scale=0.8]{figures-arXiv/link62.pdf}}}_{N-1}}(\bs{\ddot{x}}_1)
\, \sum_{\beta \in\LP_{N-1}} 
\frac{\meanderMatrix(\vcenter{\hbox{\includegraphics[scale=0.8]{figures-arXiv/link61.pdf}}}_{N-1},\beta)}{\meanderMatrix(\vcenter{\hbox{\includegraphics[scale=0.8]{figures-arXiv/link62.pdf}}}_{N-1},\beta)}
\, \vec{\SLEmeasure}_{N-1}(\HH; \bs{\ddot{x}}_1) [\conn(\hat{\vec{\bs\eta}}) = \beta ] 
\\
= \; & \fugacity(\kappa) \,
\sum_{\beta \in\LP_{N-1}} 
\meanderMatrix(\vcenter{\hbox{\includegraphics[scale=0.8]{figures-arXiv/link61.pdf}}}_{N-1},\beta)
\, \vec{\PartF}_{\beta}(\bs{\ddot{x}}_1)
&& \textnormal{[by~\eqref{eqn::PPF_def_QN}]}
\\
= \; & \fugacity(\kappa) \, \LK_{\vcenter{\hbox{\includegraphics[scale=0.8]{figures-arXiv/link61.pdf}}}_{N-1}}(\bs{\ddot{x}}_1) .
&& \textnormal{[by~\eqref{eq: usefulexpression_for_L}]}
\end{align*}
\end{itemize}
This concludes the proof.
\end{proof}

\subsection{Asymptotics properties}
\label{subsec:ssymptotics_PPF}

\begin{proposition}\label{prop::PPF_ASY}
Fix $\kappa \in (4,8)$.
The collection $\{\vec{\PartF}_{\alpha} \colon \alpha\in\LP_N\}$ of Definition~\ref{def::Zalpha} 
satisfies the recursive asymptotics~\eqref{eqn::PPF_ASY}.
\end{proposition}

\begin{proof}
It suffices to show~\eqref{eqn::PPF_ASY} when $j=1$, that is, when $x_1,x_2\to\xi < x_3$,
because for the other cases the claim then follows after conjugating by a suitable M\"obius map and using Lemma~\ref{lem::PPF_COV_GEN}.

We first consider the case $\{1,2\}\notin \alpha$ in~\eqref{eqn::PPF_ASY}.  
By~\eqref{eqn::usefulexpression} and the asymptotics for $\hat{\LR}_{\omega_\fixedindex}$ from Lemma~\ref{lem:R_ASY}, 
\begin{align} \label{eq:limzero}
\lim_{x_1,x_2 \to\xi} \frac{\vec{\PartF}_{\alpha}(\bs{x})}{ (x_2-x_1)^{-2h} }
\leq \lim_{x_1,x_2 \to\xi}
\sum_{\substack{\beta \in \LP_N \\ \{1,2\}\notin\beta}} \frac{\vec{\PartF}_{\beta}(\bs x)}{ (x_2-x_1)^{-2h} } 
\lesssim \; & \lim_{x_1,x_2 \to\xi} \sum_{s=2}^{N} \frac{\hat{\LR}_{\omega_s}(\bs x)}{ (x_2-x_1)^{-2h} } 
= 0 ,
\end{align}
where ``$\lesssim$'' holds up to some finite constant depending on $\kappa$ and $N$.

We next consider the case $\{1,2\}\in \alpha$ in~\eqref{eqn::PPF_ASY}. 

Under the induction hypothesis that Lemma~\ref{lem::difpartition} holds for all $n \le N-1$,
and because we already verified Lemma~\ref{lem::PPF_using_CGI} in the previous section, 
we may apply the cascade property~\eqref{eqn::cascade_relation} from Lemma~\ref{lem::cascade},
induction on the number of curves, and the same arguments as in the proof of~\cite[Lemma~6.6]{Wu:Convergence_of_the_critical_planar_ising_interfaces_to_hypergeometric_SLE} 
or~\cite[Lemma~B.3 and Lemma~C.1]{Peltola:Towards_CFT_for_SLEs} 
to show that\footnote{It seems very hard to get equality of the limit in~\eqref{eq:liminf} 
without substantial additional work, since as far as we are aware, 
there are no known estimates available that would allow the exchange of the limit and expected value in the asymptotics~\eqref{eqn::PPF_ASY} and the cascade relation~\eqref{eqn::cascade_relation} --- see the discussion in~\cite[Appendix~B]{Peltola:Towards_CFT_for_SLEs}.} 

\begin{align} \label{eq:liminf}
\liminf_{x_1,x_2 \to\xi} \frac{\vec{\PartF}_{\alpha}(\bs{x})}{ (x_2-x_1)^{-2h} } 
\geq \; & \vec{\PartF}_{\alpha/\{1,2\}}(\bs{\ddot{x}}_1) .
\end{align}
We will accomplish the existence of the limit~\eqref{eqn::PPF_ASY}  
by using the relation of $\vec{\PartF}_{\alpha}$ with the Coulomb gas integral functions.
Indeed, by Lemma~\ref{lem::PPF_using_CGI}, we have 
$\LK_{\vcenter{\hbox{\includegraphics[scale=0.8]{figures-arXiv/link62.pdf}}}_N} = \smash{\coulombGas_{\vcenter{\hbox{\includegraphics[scale=0.8]{figures-arXiv/link62.pdf}}}_N}}$, 
so we obtain
\begin{align*}
\coulombGas_{\vcenter{\hbox{\includegraphics[scale=0.8]{figures-arXiv/link62.pdf}}}_{N-1}}(\bs{\ddot{x}}_1) 
= \; & 
\lim_{x_1,x_2 \to\xi} \frac{\coulombGas_{\vcenter{\hbox{\includegraphics[scale=0.8]{figures-arXiv/link62.pdf}}}_N}(\bs{x})}{ (x_2-x_1)^{-2h} } 
&& \textnormal{[by~\eqref{eqn::CGI_ASY}]} 
\\
= \; & \lim_{x_1,x_2 \to\xi} \frac{\LK_{\vcenter{\hbox{\includegraphics[scale=0.8]{figures-arXiv/link62.pdf}}}_N}(\bs{x})}{ (x_2-x_1)^{-2h} } 
&& \textnormal{[by Lemma~\ref{lem::PPF_using_CGI}]} 
\\
= \; & \limsup_{x_1,x_2 \to\xi} \sum_{\substack{\alpha\in\LP_N \\ \{1,2\} \in \alpha}} \meanderMatrix(\vcenter{\hbox{\includegraphics[scale=0.8]{figures-arXiv/link62.pdf}}}_N,\alpha) \, \frac{\vec{\PartF}_{\alpha} (\bs{x})}{ (x_2-x_1)^{-2h} } 
&& \textnormal{[by~\eqref{eqn::PPF_def_QN}~\&~\eqref{eq:limzero}]}
\\
\geq \; & \sum_{\substack{\alpha\in\LP_N \\ \{1,2\} \in \alpha}} \meanderMatrix(\vcenter{\hbox{\includegraphics[scale=0.8]{figures-arXiv/link62.pdf}}}_N,\alpha) \, \liminf_{x_1,x_2 \to\xi} \frac{\vec{\PartF}_{\alpha} (\bs{x})}{ (x_2-x_1)^{-2h} } 
\\
\geq \; & \sum_{\substack{\alpha\in\LP_N \\ \{1,2\} \in \alpha}} \meanderMatrix(\vcenter{\hbox{\includegraphics[scale=0.8]{figures-arXiv/link62.pdf}}}_N,\alpha) \, \vec{\PartF}_{\alpha/\{1,2\}}(\bs{\ddot{x}}_1)
&& \textnormal{[by~\eqref{eq:limzero}~\&~\eqref{eq:liminf}]}
\\
= \; & \sum_{\beta\in\LP_{N-1}} 
\meanderMatrix(\vcenter{\hbox{\includegraphics[scale=0.8]{figures-arXiv/link62.pdf}}}_{N-1},\beta) \, \vec{\PartF}_{\beta}(\bs{\ddot{x}}_1)
\\
= \; & \LK_{\vcenter{\hbox{\includegraphics[scale=0.8]{figures-arXiv/link62.pdf}}}_{N-1}}(\bs{\ddot{x}}_1)
&& \textnormal{[by~\eqref{eq: usefulexpression_for_K}]}
\\
= \; & \coulombGas_{\vcenter{\hbox{\includegraphics[scale=0.8]{figures-arXiv/link62.pdf}}}_{N-1}}(\bs{\ddot{x}}_1) ,
&& \textnormal{[by Lemma~\ref{lem::PPF_using_CGI}]} 
\end{align*}
Since $\meanderMatrix(\vcenter{\hbox{\includegraphics[scale=0.8]{figures-arXiv/link62.pdf}}}_N,\alpha) > 0$ for all $\alpha\in\LP_N$ (for the current range $\kappa \in (4,8)$, we conclude that 
\begin{align*}
\limsup_{x_1,x_2 \to\xi} \frac{\vec{\PartF}_{\alpha} (\bs{x})}{ (x_2-x_1)^{-2h} }
= \liminf_{x_1,x_2 \to\xi} \frac{\vec{\PartF}_{\alpha} (\bs{x})}{ (x_2-x_1)^{-2h} }
= \; & \vec{\PartF}_{\alpha/\{1,2\}}(\bs{\ddot{x}}_1) .
\end{align*}
This shows~\eqref{eqn::PPF_ASY} and finishes the proof.
\end{proof}

\subsection{Reversibility}
\label{subsec:rev}

In this section, we finish the induction step to prove Lemma~\ref{lem::difpartition}.
We work under the remaining induction hypothesis that Lemma~\ref{lem::difpartition} holds for all $n \leq N-1$.

\begin{lemma}
Fix $N \geq 2$. 
Suppose that 
$\vec{\PartF}_{\hat\alpha} = \cev{\PartF}_{\hat\alpha}$ for all $\hat\alpha\in\LP_n$ and $n \leq N-1$.
Then, we have $\vec{\PartF}_{\alpha} = \cev{\PartF}_{\alpha}$ for all $\alpha\in\LP_N$.
\end{lemma}

\begin{proof}
Using Lemmas~\ref{lem::PPF_using_CGI} and~\ref{lem::reflection} combined with the identity~\eqref{eqn::coulombgasintegral_Poch} for $\smash{\coulombGas_{\vcenter{\hbox{\includegraphics[scale=0.8]{figures-arXiv/link62.pdf}}}_N}}$, we have
\begin{align*}
\cev{\PartF}_{\alpha}(\bs x) 
= \; & \coulombGas_{\vcenter{\hbox{\includegraphics[scale=0.8]{figures-arXiv/link62.pdf}}}_N}(\bs x) \, 
\frac{\cev{\SLEmeasure}_N(\HH;\bs x) [\conn(\vec{\bs\eta}) = \alpha ]}{\meanderMatrix(\vcenter{\hbox{\includegraphics[scale=0.8]{figures-arXiv/link62.pdf}}}_N,\alpha)} 
&& \textnormal{[by Lemma~\ref{lem::PPF_using_CGI}]}
\\
= \; & \coulombGas_{-\vcenter{\hbox{\includegraphics[scale=0.8]{figures-arXiv/link62.pdf}}}_N}(-\bs x) \, 
\frac{\cev{\SLEmeasure}_N(\HH;\bs x) [\conn(\vec{\bs\eta}) = \alpha ]}{\meanderMatrix(\vcenter{\hbox{\includegraphics[scale=0.8]{figures-arXiv/link62.pdf}}}_N,\alpha)} 
&& \textnormal{[by~\eqref{eq::reflection}~\&~\eqref{eqn::coulombgasintegral_Poch}]}
\\
= \; & \coulombGas_{\vcenter{\hbox{\includegraphics[scale=0.8]{figures-arXiv/link62.pdf}}}_N}(-\bs x) \, 
\frac{\vec{\SLEmeasure}_N(\HH;- \bs x) [\conn(\vec{\bs\eta}) = -\alpha ]}{\meanderMatrix(\vcenter{\hbox{\includegraphics[scale=0.8]{figures-arXiv/link62.pdf}}}_N,-\alpha)} 
&& \textnormal{[by Definition of $\cev{\SLEmeasure}_N$]}
\\
= \; & \vec{\PartF}_{-\alpha}(-\bs x) 
&& \textnormal{[by~\eqref{eqn::PPF_using_CGI}]}
\end{align*}
where $-\bs x = (-x_{2N},\ldots,-x_{1})$ and 
$-\vcenter{\hbox{\includegraphics[scale=1.2]{figures-arXiv/link62.pdf}}}_N = \vcenter{\hbox{\includegraphics[scale=1.2]{figures-arXiv/link62.pdf}}}_N$ is reflection invariant (across the imaginary axis). 
Hence, in particular, $\cev{\PartF}_{\alpha}$ satisfies 
the properties~\eqref{eqn::PDE},~\eqref{eqn::COV}, and~\eqref{eqn::PLB_weak_upper}
--- so $\cev{\PartF}_{\alpha} \in \mathcal{S}_N$. 
Therefore, 
we can prove that $\cev{\PartF}_{\alpha} = \vec{\PartF}_{\alpha}$ 
by invoking the uniqueness of the PDE boundary value problem on $\mathcal{S}_N$ from Lemma~\ref{lem::PFuniqueness}. 
Indeed, both collections $\{\vec{\PartF}_{\alpha} \colon \alpha\in\LP_N\}$ and $\{\cev{\PartF}_{\alpha} \colon \alpha\in\LP_N\}$ 
satisfy the same recursive asymptotics property~\eqref{eqn::PPF_ASY}: 
for the former collection this was the content of Proposition~\ref{prop::PPF_ASY}, 
and for the latter, the above identity $\cev{\PartF}_{\alpha}(\bs x) = \vec{\PartF}_{-\alpha}(-\bs x)$ implies that 
\begin{align*}
\lim_{x_j,x_{j+1}\to\xi} \frac{\cev{\PartF}_{\alpha}(\bs{x})}{ |x_{j+1}-x_j|^{-2h} }
= \; & \lim_{x_j,x_{j+1}\to\xi} \frac{\vec{\PartF}_{-\alpha}(-\bs{x})}{ |x_{j+1}-x_j|^{-2h} }
\\
= \; & 
\begin{cases}
\vec{\PartF}_{-\alpha/\{2N-j,2N-j+1\}}(-\bs{\ddot{x}}_{2N-j}),
& \textnormal{if }\{j, j+1\}\in\alpha , \\
0 ,
& \textnormal{if }\{j, j+1\} \not\in \alpha ,
\end{cases}
&& \textnormal{[by~\eqref{eqn::PPF_ASY}]}
\\
= \; & 
\begin{cases}
\cev{\PartF}_{\alpha/\{j,j+1\}} (\bs{\ddot{x}}_j), 
& \textnormal{if }\{j, j+1\}\in\alpha , \\
0 ,
& \textnormal{if }\{j, j+1\} \not\in \alpha ,
\end{cases}
\end{align*}
finishing the proof 
using the induction hypothesis that $\vec{\PartF}_{\alpha} = \cev{\PartF}_{\alpha}$ holds for all $\alpha\in\LP_n$ and $n \leq N-1$.
\end{proof}

\subsection{Global multiple $\SLE_\kappa$ for $\kappa\in (4,8)$: proof of Theorem~\ref{thm::existenceglobalSLE}}
\label{subsec::global_existence}

The goal of this section is to construct a probability measure on the curve space $X_0(\Omega; \bs{x})$ defined in~\eqref{eq:curve_space} which satisfies the defining property of a global multiple $\SLE_\kappa$ measure (Definition~\ref{def::NSLE}), and to prove Theorem~\ref{thm::existenceglobalSLE}. 
Because the construction is already well-known for the range $\kappa\in (0,4]$~\cite{Kozdron-Lawler:Configurational_measure_on_mutually_avoiding_SLEs, 
Peltola-Wu:Global_and_local_multiple_SLEs_and_connection_probabilities_for_level_lines_of_GFF}, 
and not very meaningful when $\kappa \geq 8$ (cf.~\cite{LPW:UST_in_topological_polygons_partition_functions_for_SLE8_and_correlations_in_logCFT}), 
we will focus on the case $\kappa\in (4,8)$. 
Very recently a construction in the case $\kappa\in (4,8)$ was established using welding of LQG surfaces in ~\cite{AHSY:Conformal_welding_of_quantum_disks_and_multiple_SLE_the_non-simple_case}, 
and the uniqueness of such measures was proven in~\cite{Zhan:Existence_and_uniqueness_of_nonsimple_multiple_SLE}. 
In this article, we obtain a different construction and derive the marginal law 
using the relation of Coulomb gas integrals and pure partition functions.

\bigskip

Fix $\kappa\in (4,8)$ and $N\ge 1$. 
Recall that in Section~\ref{subsec::QN} we constructed a conformally invariant probability measure $\vec{\SLEmeasure}_N(\Omega; \bs{x})$ (Definition~\ref{def::QN}) on (oriented) curves 
$\vec{\bs\eta} := (\vec\eta_1,\ldots,\vec\eta_N) \in X_0(\Omega; \bs{x})$ 
such that $\vec\eta_s$ starts from $x_{2s-1}$ for each $1\le s\le N$. 
Similarly, the probability measure $\cev{\SLEmeasure}_N(\Omega;\bs{x})$ is obtained from $\vec{\SLEmeasure}_N(\Omega;\bs{x})$ 
by sending $\Omega$ onto $\Omega$ via an anti-conformal map $\anticonf$: 
if one samples $(\vec\eta_1,\ldots,\vec\eta_N) \sim \vec{\SLEmeasure}_N(\Omega;\anticonf(x_{2N}),\ldots,\anticonf(x_1))$, 
then $\cev{\SLEmeasure}_N(\Omega;\bs{x})$ is the law of the (oriented) curves 
$(\cev\eta_1, \ldots,\cev\eta_N) := (\anticonf^{-1}(\vec\eta_N), \ldots,\anticonf^{-1}(\vec\eta_1))$.

Now, Definition~\ref{def::NSLE} states that a global $N$-$\SLE_{\kappa}$ associated to 
a given link pattern 
\begin{align*}
\alpha=\{\{a_1, b_1\}, \ldots, \{a_N, b_N\}\}\in\LP_N 
\end{align*}
must be supported 
on unoriented curves $\bs\eta := (\eta_1, \ldots, \eta_N) \in X_0^{\alpha}(\Omega; \bs x)$  
and satisfy the resampling property. 
To construct it, we consider the regular conditional laws\footnote{Note that the probabilities $\vec{\SLEmeasure}_N(\Omega;\bs{x}) [\conn(\vec{\bs\eta}) = \alpha]$
and $\cev{\SLEmeasure}_N(\Omega;\bs{x}) [\conn(\cev{\bs\eta}) = \alpha]$ are positive by Lemma~\ref{lem::QN_positive}.}
\begin{align}\label{eqn::defQalpha}
\begin{split}
\vec{\SLEmeasure}_\alpha(\Omega;\bs{x}) 
:= \; & \vec{\SLEmeasure}_N(\Omega;\bs{x}) [\; \cdot \, \cond \conn(\vec{\bs\eta}) = \alpha] , \\
\cev{\SLEmeasure}_\alpha(\Omega;\bs{x})
:= \; & 
\cev{\SLEmeasure}_N(\Omega;\bs{x}) [\; \cdot \, \cond \conn(\cev{\bs\eta}) = \alpha] , \qquad \alpha \in \LP_N  .
\end{split}
\end{align}
We denote by $\vec{\SLEmeasureEx}_\alpha(\Omega;\bs{x})$ and $\cev{\SLEmeasureEx}_\alpha(\Omega;\bs{x})$ the corresponding expected values.

\begin{definition} \label{def: Qkalpha}
Fix $\kappa\in (4,8)$  
and $\alpha=\{\{a_1, b_1\}, \ldots, \{a_N, b_N\}\}\in\LP_N$.
Here, we index the link endpoints not as in~\eqref{eq: link pattern ordering}, 
but we rather assume that each $a_r$ is odd and each $b_r$ is even. 
For each $s \in \{1,2,\ldots,N\}$, using the notation from Lemma~\ref{lem::cascade}, we construct a probability measure 
$\smash{\SLEmeasure_\alpha^{(s)}}(\Omega;\bs{x})$ in the following manner. 

\begin{enumerate}[leftmargin=2em]
\item 
Sample a random curve $\eta_s$ 
as a chordal $\SLE_\kappa$ curve $\eta$ on $(\Omega;x_{a_{s}},x_{b_{s}})$ weighted by (cf.~\eqref{eqn::cascade_relation})
\begin{align*}
\one\{\LE_\alpha(\eta)\} \, 
\frac{\vec{\PartF}_{\vcenter{\hbox{\includegraphics[scale=0.2]{figures-arXiv/link-0.pdf}}}}(\Omega;x_{a_{s}},x_{b_{s}})}
{\vec{\PartF}_\alpha(\Omega;\bs{x})} \; 
\vec{\PartF}_{\alpha/\{a_{s},b_{s}\}}(\hat{\Omega}_\eta; \bs{\ddot{x}}_{a_{s},b_{s}} )  .
\end{align*}

\item 
Given $\eta_s$, sample the other curves $\{\eta_1, \ldots, \eta_N\}\setminus\{\eta_s\}$ independently under the product measure of
\begin{align*}
\vec{\SLEmeasure}_{\alpha^D}(D; \ldots) 
\qquad \textnormal{or} \qquad
\cev{\SLEmeasure}_{\alpha^D}(D; \ldots) ,
\end{align*}
with $\alpha^D$ the sub-link pattern of $\alpha/\{a_{s},b_{s}\}$ associated to a component $D \subset \hat{\Omega}_\eta$ 
with some of the points among $\{x_{1},\ldots,x_{2N}\} \setminus \{x_{a_{s}},x_{b_{s}}\}$ on $\partial D$, 
where the orientation is chosen is such a way that curves in the components 
left of the curve $\gamma_s$ are sampled from
$\vec{\SLEmeasure}_{\alpha^D}(D; \ldots)$ 
and curves in the components right of the curve $\gamma_s$ are sampled from $\cev{\SLEmeasure}_{\alpha^D}(D; \ldots)$.
\end{enumerate}
\end{definition}

We first note that $\smash{\SLEmeasure_\alpha^{(s)}}(\Omega;\bs{x})$ is independent of $s$.

\begin{lemma}\label{lem::idQalphak}
Fix $\kappa\in (4,8)$, $N\ge 1$, $\alpha\in\LP_N$, and a polygon $(\Omega; \bs x)$. 
For all $1\le s \neq r \le N$, we have 
\begin{align*}
\SLEmeasure_\alpha^{(s)}(\Omega; \bs x) = \SLEmeasure_\alpha^{(r)}(\Omega; \bs x) .
\end{align*}
\end{lemma}

\begin{proof}
This can be proven by induction on the number $N \geq 1$ of curves similarly to~\cite[Lemma~6.2]{Wu:Convergence_of_the_critical_planar_ising_interfaces_to_hypergeometric_SLE},
by applying the cascade property~\eqref{eqn::cascade_relation} from Lemma~\ref{lem::cascade}
and the arguments in its proof.
\end{proof}

Next, we show that $\smash{\SLEmeasure_\alpha^{(s)}}(\Omega;\bs{x})$ coincides with $\vec{\SLEmeasure}_\alpha(\Omega;\bs{x})$ defined in~\eqref{eqn::defQalpha}. 

\begin{lemma}\label{lem::idtwomeasures}
Fix $\kappa\in (4,8)$, $N\ge 1$, $\alpha\in\LP_N$, and a polygon $(\Omega; \bs x)$. 
For all $1\le s \le N$, we have 
\begin{align*}
\SLEmeasure_\alpha^{(s)}(\Omega; \bs x) = \vec{\SLEmeasure}_\alpha(\Omega;\bs{x}) .
\end{align*}
\end{lemma}

\begin{proof}
Without loss of generality, we assume that $\Omega = \HH$ --- 
the formula for general polygons follows by conformal invariance. 
Also, thanks to Lemma~\ref{lem::idQalphak} it suffices to prove 
the asserted identity for $s$ such that $\{a_s, b_s\} = \{1,2\fixedindex\}\in\alpha$. 
Now, it suffices to prove that the marginal law of $\eta_s$ under 
$\vec{\SLEmeasure}_\alpha = \vec{\SLEmeasure}_\alpha(\HH;\bs{x})$ equals the marginal law of $\eta_s$ 
under $\smash{\SLEmeasure_\alpha^{(s)} = \SLEmeasure_\alpha^{(s)}}(\HH;\bs{x})$. 
For any bounded continuous function $F$, we have 
\begin{align*}
\vec{\SLEmeasureEx}_\alpha[F(\eta_s)] 
= \; & \vec{\SLEmeasureEx}_N [F(\eta_s) \cond \conn(\vec{\bs\eta}) = \alpha] 
\; = \; \frac{\vec{\SLEmeasureEx}_N [F(\eta_s) \one\{\conn(\vec{\bs\eta}) = \alpha\}]}{\vec{\SLEmeasure}_N [ \conn(\vec{\bs\eta}) = \alpha ] }
&& \textnormal{[by~\eqref{eqn::defQalpha}]}
\\
= \; & \SLEmeasureEx_N^{x_1\to x_{2\fixedindex}} [ F(\eta_s) \one\{\conn(\vec{\bs\eta}) = \alpha\} \cond \conn(\vec{\bs\eta}) = \alpha]
 \; = \; \frac{\SLEmeasureEx_N^{x_1\to x_{2\fixedindex}} [ F(\eta_s) \one\{\conn(\vec{\bs\eta}) = \alpha\} ]}{\SLEmeasure_N^{x_1\to x_{2\fixedindex}} [ \conn(\vec{\bs\eta}) = \alpha]} 
 && \textnormal{[by~\eqref{eqn::defQN}]}
\end{align*}
using Definition~\ref{def::QN}. 
The computations made in the proof of Lemma~\ref{lem::cascade} imply that  
\begin{align*}
\frac{\SLEmeasureEx_N^{x_1\to x_{2\fixedindex}} [ F(\eta_s) \one\{\conn(\vec{\bs\eta}) = \alpha\} ]}{\SLEmeasure_N^{x_1\to x_{2\fixedindex}} [ \conn(\vec{\bs\eta}) = \alpha]} 
= \; &
\frac{\vec{\PartF}_{\vcenter{\hbox{\includegraphics[scale=0.2]{figures-arXiv/link-0.pdf}}}}(\HH;x_{1},x_{2\fixedindex})}{\vec{\PartF}_\alpha(\HH;\bs{x})}
\;
\E(\HH;x_1,x_{2\fixedindex}) \bigg[ F(\eta) \one\{ \LE_\alpha(\eta) \}  \, 
\Big(
\prod_{m=1}^r
\cev{\PartF}_{\alpha_m^R} 
\Big)
\Big( 
\prod_{m=1}^l
\vec{\PartF}_{\alpha_m^L} 
\Big)
\bigg] 
\\
= \; & 
\SLEmeasureEx_\alpha^{(s)} [F(\eta_s)] ,
\end{align*}
where $\smash{\SLEmeasureEx_\alpha^{(s)}}$ is the expected value under the measure $\smash{\SLEmeasureEx_\alpha^{(s)}}$.
This shows that $\smash{\vec{\SLEmeasureEx}_\alpha[F(\eta_s)] = \SLEmeasureEx_\alpha^{(s)} [F(\eta_s)]}$. 
\end{proof}

\begin{corollary}\label{cor::global_existence}
Fix $\kappa \in (4,8)$.
For any $N\ge 1$ and $\alpha\in\LP_N$, the measure $\vec{\SLEmeasure}_\alpha(\Omega;\bs{x})$ defined in~\eqref{eqn::defQalpha} 
is a global $N$-$\SLE_{\kappa}$ associated to $\alpha$
\textnormal{(}as in Definition~\ref{def::NSLE}\textnormal{)}.
\end{corollary}

\begin{proof}
When $N =1$, the law 
$\vec{\SLEmeasure}_{\vcenter{\hbox{\includegraphics[scale=0.2]{figures-arXiv/link-0.pdf}}}}(\Omega;x_1,x_2) 
= \vec{\SLEmeasure}_1(\Omega;x_1,x_2)$ is just the law of chordal $\SLE_\kappa$ in $(\Omega;x_1,x_2)$.
The claim for $N\ge 2$ then follows by induction: 
indeed, the resampling property follows directly from construction
for the measures $\smash{\SLEmeasure_\alpha^{(s)}(\Omega; \bs x)}$, 
which all coincide with $\smash{\vec{\SLEmeasure}_\alpha(\Omega;\bs{x})}$ by Lemma~\ref{lem::idtwomeasures}.
\end{proof}

\begin{proof}[Proof of Theorem~\ref{thm::existenceglobalSLE}]
Existence is given by Corollary~\ref{cor::global_existence} for $\kappa \in (4,8)$ 
and~\cite{Kozdron-Lawler:Configurational_measure_on_mutually_avoiding_SLEs} and~\cite[Theorem~1.3]{Peltola-Wu:Global_and_local_multiple_SLEs_and_connection_probabilities_for_level_lines_of_GFF} for $\kappa \in (0,4]$;
and uniqueness follows from Zhan's recent work~\cite[Theorem~4.2]{Zhan:Existence_and_uniqueness_of_nonsimple_multiple_SLE} for $\kappa \in (4,8)$ 
and~\cite[Theorem~1.2]{BPW:On_the_uniqueness_of_global_multiple_SLEs} for $\kappa \in (0,4]$. 
It remains to derive the marginal law of $\eta_s$ for each $1 \leq s \leq N$,
that is, to verify that it is a Loewner chain associated to the pure partition function $\PartF_{\alpha}$.

By symmetry, it suffices to consider the case where $\Omega=\HH$ and $\{a_s, b_s\}=\{1, 2\fixedindex\}\in\alpha$. 
Let $\eta$ be a chordal $\SLE_\kappa$ curve $\eta$ on $(\HH;x_1,x_{2\fixedindex})$, with Loewner driving function $W$. 
Define $\tau_\eps$ to be the hitting time of $\eta$ at the $\eps$-neighborhood of $(x_2,+\infty)$. 
We use the same notations as in Lemma~\ref{lem::cascade}. 
For any bounded continuous function $F$, using Lemma~\ref{lem::idtwomeasures}, Lemma~\ref{lem::cascade}, and conformal covariance, 
we obtain
\begin{align*}
\; & \vec{\SLEmeasureEx}_\alpha(\HH;\bs x) [F(\eta_s[0,\tau_\eps])] \\
= \; & \SLEmeasureEx_\alpha^{(s)}(\HH;\bs x) [F(\eta_s[0,\tau_\eps])] 
\\
= \; & \frac{\vec{\PartF}_{\vcenter{\hbox{\includegraphics[scale=0.2]{figures-arXiv/link-0.pdf}}}}(\HH;x_{1},x_{2\fixedindex})}{\vec{\PartF}_\alpha(\HH;\bs{x})}
\;
\E(\HH;x_1,x_{2\fixedindex}) \bigg[ F(\eta[0,\tau_\eps]) \one\{ \LE_\alpha(\eta) \}  \, 
\Big(
\prod_{m=1}^r
\cev{\PartF}_{\alpha_m^R} 
\Big)
\Big( 
\prod_{m=1}^l
\vec{\PartF}_{\alpha_m^L} 
\Big)
\bigg]
\\
= \; & \frac{(x_{2\fixedindex}-x_1)^{-2h}}{\vec{\PartF}_\alpha(\HH;\bs{x})}
\;
\E(\HH;x_1,x_{2\fixedindex}) \bigg[ F(\eta[0,\tau_\eps]) \one\{ \LE_\alpha(\eta) \}  \, 
\E(\HH;x_1,x_{2\fixedindex}) \bigg[ \Big(
\prod_{m=1}^r
\cev{\PartF}_{\alpha_m^R} 
\Big)
\Big( 
\prod_{m=1}^l
\vec{\PartF}_{\alpha_m^L} 
\Big)
\cond 
\eta[0,\tau_\eps]
\bigg]
\bigg]
\\
= \; & 
\E(\HH;x_1,x_{2\fixedindex}) \bigg[ F(\eta[0,\tau_\eps]) \one\{ \LE_\alpha(\eta) \}  \, 
\frac{|x_{2\fixedindex}-x_1|^{-2h}}{|g_{\tau_\eps}(x_{2\fixedindex})-W_{\tau_\eps}|^{-2h}} \;
\frac{\PartF_\alpha(\HH;W_{\tau_\eps},g_{\tau_{\eps}}(\bs{\dot{x}}_{1})}{\vec{\PartF}_\alpha(\HH;\bs{x})} \;
\prod_{\substack{1 \leq j \leq 2N \\ j \neq 1, 2\fixedindex}} g'_{\tau_\eps}(x_j)^h 
\bigg] ,
\end{align*}
where $g_t \colon \HH \setminus K_t \to \HH$ 
are the conformal maps normalized such that $g_t(x_{2\fixedindex}) = x_{2\fixedindex}$ for all $t \leq \tau_\eps$,
with $\HH \setminus K_t$ being the unbounded component of $\HH \setminus \eta[0,t]$ 
(containing in particular the target point $x_{2\fixedindex}$).

Now, recall from, e.g.,~\cite{Schramm-Wilson:SLE_coordinate_changes} that the law of $\eta$ up to the stopping time $\tau_\eps$ 
is the same as that of the chordal $\SLE_\kappa$ curve $\tilde{\eta}$ on $(\HH;x_1,\infty)$ weighted by 
$\{M_t/M_0\}_{0\le t\le\tau_\eps}$ with the martingale
\begin{align*}
M_t := |\tilde g_t(x_{2\fixedindex}) - \tilde W_t|^{-2h} \; \tilde{g}'_t(x_{2\fixedindex})^h ,
\end{align*}
where $(\tilde g_t \colon t\ge 0)$ are the conformal maps corresponding to $\tilde{\eta}$ 
normalized at infinity\footnote{That is, solutions to~\eqref{eq:LE} with $\underset{z\to\infty}{\lim} \, |\tilde{g}_t(z)-z|=0$.} 
and $\tilde W$ its driving function. 
Hence, we conclude that
the law of $\eta_s$ under $\vec{\SLEmeasure}_\alpha(\HH;\bs x)$ is the same as the law of $\tilde{\eta}$ weighted by 
\begin{align*}
\prod_{j=2}^{2N} \tilde{g}'_{\tau_\eps}(x_j)^h \; 
\times \; \frac{\PartF_\alpha(\HH;\tilde W_{\tau_\eps}, \tilde{g}_{\tau_{\eps}}(\bs{\dot{x}}_{1})}{\PartF_\alpha(\HH;\bs x)} .
\end{align*}
In other words, the driving function of $\eta_s$ under $\vec{\SLEmeasure}_\alpha(\HH;\bs x)$ solves the SDE system~\eqref{eq:Loewner_SDE} with $\PartF = \PartF_\alpha$ and $i=1$, that is, the Loewner chain associated to $\PartF_{\alpha}$.
\end{proof}

\begin{corollary}\label{cor::reversibility}
Fix $\kappa \in (4,8)$.
For any $N\ge 1$, 
the measures $\vec{\SLEmeasure}_N$ and $\cev{\SLEmeasure}_N$ are the same.
\end{corollary}

\begin{proof}
On the one hand, by Lemma~\ref{lem::difpartition} and Definition~\ref{def::Zalpha}, we have 
\begin{align*}
\vec{\SLEmeasure}_N[\conn(\vec{\bs\eta}) = \alpha ]=\cev{\SLEmeasure}_N[\conn(\cev{\bs\eta}) = \alpha ]
 , \qquad \alpha \in \LP_N .
\end{align*} 
On the other hand, we have
\begin{align*}
\vec{\SLEmeasure}_\alpha(\Omega;\bs{x}) 
:= \; & \vec{\SLEmeasure}_N(\Omega;\bs{x}) [\; \cdot \, \cond \conn(\vec{\bs\eta}) = \alpha] 
= \cev{\SLEmeasure}_N(\Omega;\bs{x}) [\; \cdot \, \cond \conn(\cev{\bs\eta}) = \alpha] 
=: \cev{\SLEmeasure}_\alpha(\Omega;\bs{x}) , \qquad \alpha \in \LP_N ,
\end{align*}
since by Corollary~\ref{cor::global_existence} 
(and similar arguments for $\cev{\SLEmeasure}_\alpha$), 
both measures $\vec{\SLEmeasure}_\alpha$ and $\cev{\SLEmeasure}_\alpha$ defined in~\eqref{eqn::defQalpha} 
are global $N$-$\SLE_{\kappa}$ measures associated to $\alpha$ --- 
but such measures are unique by~\cite[Theorem~4.2]{Zhan:Existence_and_uniqueness_of_nonsimple_multiple_SLE}. 
\end{proof}

\begin{remark}\label{rem::NSLE_sanity}
Let us argue that the $\SLE_\kappa$ partition functions $\tilde{\LZ}_{\alpha}$ constructed in~\cite{AHSY:Conformal_welding_of_quantum_disks_and_multiple_SLE_the_non-simple_case} 
are the same as pure partition functions $\LZ_{\alpha}$ in Definition~\ref{def::PPF_general}, up to a multiplicative constant. In~\cite{AHSY:Conformal_welding_of_quantum_disks_and_multiple_SLE_the_non-simple_case}, the authors prove that 
any curve in a global multiple $\SLE_\kappa$ has the same law as
the Loewner chain associated to their partition functions $\tilde{\LZ}_{\alpha}$. 
Combining with the uniqueness of global multiple $\SLE_\kappa$ measures~\cite{Zhan:Existence_and_uniqueness_of_nonsimple_multiple_SLE} 
and Theorem~\ref{thm::existenceglobalSLE}, 
we see that $\partial_i\log\tilde{\LZ}_{\alpha}$ is the same as $\partial_i\log\LZ_{\alpha}$ in~\eqref{eq:Loewner_SDE} for all $i\in\{1, 2, \ldots, 2N\}$. This implies that $\tilde{\LZ}_{\alpha}$ is the same as $\LZ_{\alpha}$ up to a multiplicative constant. 
\end{remark}


\bigskip{}
\section{Tying up loose ends}
\label{sec::consequences}
In this section, we investigate the solutions to the PDEs in the space $\mathcal{S}_N$ defined in~\eqref{eq: solution space}. 
Specifically, we finish the proofs of the main results: key properties of the pure partition functions with $\kappa \in (0,8)$  (main part of Theorem~\ref{thm::PPF}) 
in Sections~\ref{subsec::pf_lin_indep}--\ref{subsec::pf_continuity}; 
key properties of the Coulomb gas integral functions (Theorem~\ref{thm::CGI}) 
and their relation to the pure partition functions (Proposition~\ref{prop::CGI_PPF}) 
in Section~\ref{subsec:finish_CGI_proof};
and the renormalization of these functions for special values of $\kappa$ 
(Propositions~\ref{prop::CGI_odd},~\ref{prop::CGI_8},
and Item~\ref{item:PPF_continuity8} of Theorem~\ref{thm::PPF}) 
in Section~\ref{subsec::renormalized_pf}.

\subsection{Pure partition functions: linear independence}
\label{subsec::pf_lin_indep}

In this section, we prove the linear independence of the pure partition functions (which is similar to~\cite[Proposition~4.5]{Peltola-Wu:Global_and_local_multiple_SLEs_and_connection_probabilities_for_level_lines_of_GFF}).
By~\cite[Lemma~1]{Flores-Kleban:Solution_space_for_system_of_null-state_PDE2}, 
the asymptotics properties can be used to uniquely characterize elements in $\mathcal{S}_N$.
To analyze them, we consider linear functionals $\FKdual_\beta \colon \mathcal{S}_N \to \C$ defined in terms of iterated limits (allowable sequences of limits in the terminology of~\cite{Flores-Kleban:Solution_space_for_system_of_null-state_PDE1, Kytola-Peltola:Pure_partition_functions_of_multiple_SLEs}).
For a link pattern 
\begin{align} \label{eq: link pattern non-ordering}
\beta = \{ \{a_1,b_1\},  \{a_2,b_2\},\ldots , \{a_N,b_N\}\} , 
\qquad a_r < b_r , \textnormal{ for all } 1 \leq r \leq N , 
\end{align}
where the links $\{a_r,b_r\}$ are not necessarily ordered as in~\eqref{eq: link pattern ordering}, 
we say that the ordering of links is \emph{allowable} if all links of $\beta$ as in~\eqref{eq: link pattern non-ordering}  
can be removed in the order $\{a_1,b_1\},  \{a_2,b_2\},\ldots , \{a_N,b_N\}$ in such a way that at each step, 
the link to be removed connects two consecutive indices 
(see~\cite[Figure~2.1]{Peltola-Wu:Global_and_local_multiple_SLEs_and_connection_probabilities_for_level_lines_of_GFF} for an illustration and, e.g.,~\cite[Section~3.5]{Kytola-Peltola:Pure_partition_functions_of_multiple_SLEs} for a more formal definition).
Suppose the ordering~\eqref{eq: link pattern non-ordering} of the links in $\beta$ is allowable. 
Fix points $\xi_r \in (x_{a_r-1}, x_{b_r+1})$ for all $r \in \{ 1,\ldots,N \}$, with the convention that $x_0 = -\infty$ and  $x_{2N+1} = +\infty$. 
It was proved in~\cite[Lemma~10]{Flores-Kleban:Solution_space_for_system_of_null-state_PDE2} 
that the following sequence of limits exists and is finite for any solution $F \in \mathcal{S}_N$:
\begin{align}\label{eq: limit operation}
\FKdual_\beta (F )
:= \lim_{x_{a_N},x_{b_{N}}\to\xi_{N}}
\cdots
\lim_{x_{a_{1}},x_{b_{1}}\to\xi_{1}}
(x_{b_{N}}-x_{a_{N}})^{2h(\kappa)}
\cdots
(x_{b_{1}}-x_{a_{1}})^{2h(\kappa)} \, F(x_1,\ldots,x_{2N}) .
\end{align}
Furthermore, by~\cite[Lemma~12]{Flores-Kleban:Solution_space_for_system_of_null-state_PDE2},  
any other allowable ordering of the links of $\beta$ gives the same limit~\eqref{eq: limit operation}.
Therefore, for each $\beta \in \LP_N$ with any choice of allowable ordering of links,~\eqref{eq: limit operation} defines a linear functional $\FKdual_\beta \colon \mathcal{S}_N \to \C$. 
Finally, it was proved in~\cite[Theorem~8]{Flores-Kleban:Solution_space_for_system_of_null-state_PDE3} that, 
for any $\kappa \in (0,8)$, the collection
$\{\FKdual_\beta \colon \beta \in \LP_N\}$ is a basis for the dual space $\mathcal{S}_N^*$ of the solution space $\mathcal{S}_N$.

\begin{corollary}\label{lem:PPF_lin_indep}
For $\kappa\in (0,8)$, 
the collection $\{\smash{\PartF_{\alpha}^{(\kappa)}} \colon \alpha \in \LP_N\}$ is linearly independent.
\end{corollary}
\begin{proof}
By Proposition~\ref{prop::PPF_ASY}, 
we have $\FKdual_\beta ( \PartF_{\alpha}^{(\kappa)} ) = \delta_{\alpha,\beta}$ for all $\alpha, \beta \in \LP_N$. 
This shows that the set $\{\smash{\PartF_{\alpha}^{(\kappa)}} \colon \alpha \in \LP_N\}$ is linearly independent as a dual basis 
of the linear operators $\{\FKdual_\beta \colon \beta \in \LP_N\}$.
\end{proof}

\subsection{Pure partition functions: real-analyticity in $\kappa$}
\label{subsec::pf_continuity}

The main goal of this section is to show the continuity in $\kappa$ of pure partition functions (Item~\ref{item:PPF_continuity} in Theorem~\ref{thm::PPF}). 
This is easier to derive for $\kappa\in (0,6)$ (Lemma~\ref{lem::PPF_continuity_06}), 
because in that case the pure partition functions admit 
a very strong upper bound~\eqref{eqn::PPF_bound_polygon},
which is significantly better than~\eqref{eqn::PLB_weak_upper}.
The range $\kappa\in (4,8)$ can be handled using Coulomb gas integrals and their relation to the pure partition functions 
from Proposition~\ref{prop::CGI_PPF} 
(see Lemmas~\ref{lem::CGI_PPF_good}~\&~\ref{lem::PPF_continuity_48}). 
We are not aware of a direct proof only using the probabilistic construction for the range $\kappa\in (4,8)$.

\begin{lemma}\label{lem::PPF_continuity_06}
The pure partition functions $\{\smash{\PartF_{\alpha}^{(\kappa)}} \colon \alpha \in \LP\}$ are continuous in $(0,6)\ni \kappa$. 
\end{lemma}

\begin{proof}
We can use induction, the cascade property, and dominated convergence to verify the continuity of 
$\smash{\PartF_{\alpha}^{(\kappa)}}$ in $\kappa$.
Indeed, the cascade property~\eqref{eqn::cascade_relation}
holds for all $\kappa \in (0,8)$ by Lemma~\ref{lem::cascade} (for $\kappa\in (4,8)$),
by~\cite[Proposition~3.5]{Peltola-Wu:Global_and_local_multiple_SLEs_and_connection_probabilities_for_level_lines_of_GFF}
(for $\kappa\in (0,4]$),
and by~\cite[Corollary~6.8]{Wu:Convergence_of_the_critical_planar_ising_interfaces_to_hypergeometric_SLE} (for $\kappa\in (0,6]$). 
Moreover, when $\kappa\in (0,6]$, the pure partition functions admit the strong upper bound~\cite[Lemma~4.1]{Peltola-Wu:Global_and_local_multiple_SLEs_and_connection_probabilities_for_level_lines_of_GFF}
(for $\kappa\in (0,4]$),~\cite[Theorem~1.7]{Wu:Convergence_of_the_critical_planar_ising_interfaces_to_hypergeometric_SLE}
(for $\kappa\in (0,6]$)
\begin{align}\label{eqn::PPF_bound_polygon}
\PartF_{\alpha}^{(\kappa)}(\Omega; x_1, \ldots, x_{2N}) 
\le \prod_{\{a,b\}\in\alpha} H(\Omega; x_a, x_b)^{h(\kappa)} , 
\end{align}
where $H(\Omega; x, y)$ is the (boundary) Poisson kernel.
This upper bound plays an essential role in the following proof of continuity, which is obstructed for $\kappa \geq 6$ because we do not have such a strong bound available, 
nor could we use the Poisson kernel monotonicity~\eqref{eqn::Poisson_mono} for $h(\kappa) \leq 0$.

\smallbreak

For the base case $N=1$, the pure partition function is explicit,  
$\smash{\PartF^{(\kappa)}_{\vcenter{\hbox{\includegraphics[scale=0.2]{figures-arXiv/link-0.pdf}}}}(x_1, x_2)=(x_2-x_1)^{-2h(\kappa)}}$, 
and obviously continuous.
Fix $N \geq 2$, and suppose the members in the collection $\{\smash{\PartF^{(\kappa)}_{\alpha}} \colon \alpha\in\LP_n , \; n \le N-1 \}$ are continuous in $(0,6)\ni \kappa$.  
Fix $\kappa\in (0,6)$ and $\alpha\in\LP_N$, with $\{1,2\fixedindex\}\in\alpha$. 
Let $\kappa_m\in (0,6)$ be a converging sequence, $\kappa_m\to\kappa$ as $m\to\infty$. 
Denote by $\eta$ the $\SLE_{\kappa}$ curve from $x_1$ to $x_{2\fixedindex}$ and by $\eta_m$ the $\SLE_{\kappa_m}$ curve from $x_1$ to $x_{2\fixedindex}$ for each $m\ge 1$. 
By~\cite[Theorem 1.10]{Kemppainen-Smirnov:Random_curves_scaling_limits_and_Loewner_evolutions}, 
we know that the law of $\eta_m$ converges to the law of $\eta$ weakly under the metric~\eqref{eqn::curve_metric}. 
By Skorokhod’s representation theorem, there is a coupling of $\eta$ and $\{\eta_m\}_{m\ge 1}$, which we denote by $\PP$, such that under $\PP$, random curve $\eta_m$ converges to $\eta$ almost surely as $m\to\infty$. 
Using the cascade property~\eqref{eqn::cascade_relation}  or~\cite[Corollary~6.8]{Wu:Convergence_of_the_critical_planar_ising_interfaces_to_hypergeometric_SLE}, 
we may thus conclude that
\begin{align} \label{eqn::before_limit}
\PartF_\alpha^{(\kappa_m)}(\bs{x}) 
= \PartF_{\vcenter{\hbox{\includegraphics[scale=0.2]{figures-arXiv/link-0.pdf}}}}^{(\kappa_m)}(x_{1},x_{2\fixedindex}) 
\; 
\E \big[ \one\{ \LE_\alpha(\eta_m) \} \, \PartF_{\alpha/\{1,2\fixedindex\}}^{(\kappa_m)}(\hat{\HH}_{\eta_m}; \bs{\ddot{x}}_{1,2\fixedindex} ) \big] ,
\end{align}
where $\E$ is the expectation with respect to $\PP$.  
Because the Poisson kernel is monotone by~\eqref{eqn::Poisson_mono}, 
and $h(\kappa_m) > 0$ for all $\kappa_m\in (0,6)$, 
we see from~\eqref{eqn::PPF_bound_polygon} that the functions inside the expected value are bounded:
\begin{align*}
\PartF_{\alpha/\{1,2\fixedindex\}}^{(\kappa_m)}(\hat{\HH}_{\eta_m}; \bs{\ddot{x}}_{1,2\fixedindex} )
\; & \leq 
\prod_{\substack{ D \textnormal{ c.c of } \HH \setminus \eta_m \\ \overline{D} \cap \{x_1, \ldots, x_{2N}\} \setminus \{x_1, x_{2\fixedindex}\} \neq \emptyset }}
\prod_{\{a,b\}\in\alpha^D} H(D; x_a, x_b)^{h(\kappa)} 
\\
\; & \leq 
\prod_{\substack{ D \textnormal{ c.c of } \HH \setminus \eta_m \\ \overline{D} \cap \{x_1, \ldots, x_{2N}\} \setminus \{x_1, x_{2\fixedindex}\} \neq \emptyset }}
\prod_{\{a,b\}\in\alpha^D} H(\HH; x_a, x_b)^{h(\kappa)} .
\end{align*}
Therefore, combined with the cascade property~\eqref{eqn::cascade_relation} or~\cite[Corollary~6.8]{Wu:Convergence_of_the_critical_planar_ising_interfaces_to_hypergeometric_SLE},   
the dominated convergence theorem and the induction hypothesis guarantee that
\begin{align*}
\lim_{m\to\infty} 
\PartF_\alpha^{(\kappa_m)}(\bs{x}) 
= \; & \PartF_{\vcenter{\hbox{\includegraphics[scale=0.2]{figures-arXiv/link-0.pdf}}}}^{(\kappa)}(x_{1},x_{2\fixedindex}) 
\; 
\E \big[ \one\{ \LE_\alpha(\eta) \} \, 
\PartF_{\alpha/\{1,2\fixedindex\}}^{(\kappa)}(\hat{\HH}_{\eta_m}; \bs{\ddot{x}}_{1,2\fixedindex} ) \big] 
\; = \; \PartF_\alpha^{(\kappa)}(\bs{x}) ,
\end{align*}
which proves that $\PartF_\alpha^{(\kappa)}(\bs{x})$ is continuous on $(0,6) \ni \kappa$. 
\end{proof}

\begin{lemma}\label{lem::CGI_PPF_good}
The relation~\eqref{eqn::CGI_PPF} holds for all $\kappa\in (0,8) \setminus \big\{ \tfrac{8}{m} \colon m \in \bZpos \big\}$. 
\end{lemma}

\begin{proof}
We use the uniqueness from Lemma~\ref{lem::PFuniqueness} 
of elements in the PDE solution space $\mathcal{S}_N$~\eqref{eq: solution space} 
singled out by their asymptotic properties. 
On the one hand, we already know from Corollary~\ref{cor::CGI_PDECOVASY} that 
$\{ \smash{\coulombGas_{\beta}^{(\kappa)}} \colon \beta \in \LP_N \} \subset \mathcal{S}_N$, 
and that they satisfy the recursive asymptotic property~\eqref{eqn::CGI_ASY}. 
On the other hand, we also know that $\{ \smash{\PartF_{\alpha}^{(\kappa)}} \colon \alpha \in \LP_N \} \subset \mathcal{S}_N$,
thanks to Proposition~\ref{prop::PPF_PDE}, Lemma~\ref{lem::PPF_COV}, and Lemma~\ref{lem::PPF_PLB_weak}.
It thus remains to check that the right-hand side of~\eqref{eqn::CGI_PPF},
\begin{align}\label{eqn::CGI_PPF_RHS}
\LG_{\beta}^{(\kappa)}(\bs{x}) := \sum_{\alpha\in\LP_N} \meanderMatrix(\alpha, \beta) \,  \PartF_{\alpha}^{(\kappa)}(\bs{x}), \qquad  \bs{x} \in \chamber_{2N} , \; \beta \in \LP_N ,
\end{align}
satisfy the same recursive asymptotics property property~\eqref{eqn::CGI_ASY} as the left-hand side $\smash{\coulombGas_{\beta}^{(\kappa)}}$:
\begin{align*}
\lim_{x_j, x_{j+1}\to \xi}\frac{\LG_{\beta}^{(\kappa)}(\bs{x}) }{(x_{j+1}-x_j)^{-2h(\kappa)}} 
= \; & \sum_{\alpha\in\LP_N}\meanderMatrix(\alpha,\beta) \;
\lim_{x_j, x_{j+1}\to \xi}\frac{\PartF_{\alpha}^{(\kappa)}(\bs{x})}{(x_{j+1}-x_j)^{-2h(\kappa)}} 
&& \textnormal{[by~\eqref{eqn::CGI_PPF}]}
\\
= \; & \sum_{\substack{\alpha\in\LP_N \\\{j,j+1\}\in\alpha}}
\meanderMatrix(\alpha,\beta) \;
\PartF_{\alpha/\{j, j+1\}}^{(\kappa)}(\bs{\ddot{x}}_j) 
&& \textnormal{[by~\eqref{eqn::PPF_ASY}]}
\\
= \; & 
\begin{cases}
\fugacity(\kappa) \,
\underset{\gamma\in\LP_{N-1}}{\sum} 
\meanderMatrix(\gamma, \beta/\{j,j+1\}) \; 
\PartF_{\gamma}^{(\kappa)}(\bs{\ddot{x}}_j),
& \textnormal{if }\{j, j+1\}\in\beta , \\[.5em]
\underset{\gamma\in\LP_{N-1}}{\sum} 
\meanderMatrix(\gamma, \wp_j(\beta)/\{j,j+1\}) \; 
\PartF_{\gamma}^{(\kappa)}(\bs{\ddot{x}}_j),
& \textnormal{if }\{j, j+1\} \not\in \beta ,
\end{cases}
\\
= \; & 
\begin{cases}
\fugacity(\kappa) \,
\LG_{\beta/\{j,j+1\}}^{(\kappa)}(\bs{\ddot{x}}_j),
& \textnormal{if }\{j, j+1\}\in\beta , \\
\LG_{\wp_j(\beta)/\{j,j+1\}}^{(\kappa)}(\bs{\ddot{x}}_j) ,
& \textnormal{if }\{j, j+1\} \not\in \beta ,
\end{cases}
&& \textnormal{[by~\eqref{eqn::CGI_PPF}]}
\end{align*}
using the bijection $\alpha \leftrightarrow \alpha/\{j,j+1\}$ 
between $\{\alpha\in\LP_N \colon \{j,j+1\}\in\alpha\}$ and $\LP_{N-1}$ 
and the facts that  
$\meanderMatrix(\alpha,\beta) = \fugacity(\kappa) \,\meanderMatrix(\alpha/\{j,j+1\}, \beta/\{j,j+1\})$ 
and
$\meanderMatrix(\alpha,\beta) = \meanderMatrix(\alpha/\{j,j+1\}, \wp_j(\beta)/\{j,j+1\})$.
\end{proof}

\begin{lemma}\label{lem::PPF_continuity_48}
The pure partition functions $\{\smash{\PartF_{\alpha}^{(\kappa)}} \colon \alpha \in \LP\}$ are continuous in $(4,8)\ni\kappa$.
\end{lemma}

\begin{proof}
From Lemma~\ref{lem::CGI_PPF_good}, for $\kappa\in (4,8)$, we have
\begin{align}\label{eqn::CGI_PPF_48}
\coulombGas_{\beta}^{(\kappa)}(\bs{x}) = \sum_{\alpha\in\LP_N} \meanderMatrix(\alpha, \beta) \,  \PartF_{\alpha}^{(\kappa)}(\bs{x}), \qquad  \bs{x} \in \chamber_{2N} , \; \beta \in \LP_N. 
\end{align}
This relation plays an essential role in the following proof of continuity by induction on $N \geq 1$.

\smallbreak

For the base case $N=1$, the pure partition function is explicit,  
$\smash{\PartF^{(\kappa)}_{\vcenter{\hbox{\includegraphics[scale=0.2]{figures-arXiv/link-0.pdf}}}}(x_1, x_2)=(x_2-x_1)^{-2h(\kappa)}}$, 
and obviously continuous.
Fix $N \geq 2$, and suppose the members in the collection $\{\smash{\PartF^{(\kappa)}_{\alpha}} \colon \alpha\in\LP_n , \; n \le N-1 \}$ are continuous in $(4,8) \ni \kappa$.  
Fix $\kappa\in (4,8)$ and $\alpha\in\LP_N$, with $\{1,2\fixedindex\}\in\alpha$. 
Let $\kappa_m\in (4,8)$ be a converging sequence, $\kappa_m\to\kappa$ as $m\to\infty$. 
With~\eqref{eqn::CGI_PPF_48} and the cascade property~\eqref{eqn::before_limit} (Lemma~\ref{lem::cascade}), 
we obtain:
\begin{itemize}
\item On the one hand, for the quantities on the right hand-side of~\eqref{eqn::before_limit}, 
by the induction hypothesis, we have, on the event 
$\underset{i\ge 0}{\bigcup} \; \underset{m\ge i}{\bigcap} \LE_\alpha(\eta_m) \supset \LE_\alpha(\eta)$, 
\begin{align*}
\lim_{m\to\infty} 
\PartF_{\alpha/\{1,2\fixedindex\}}^{(\kappa_m)}(\hat{\HH}_{\eta_m}; \bs{\ddot{x}}_{1,2\fixedindex} )
= \PartF_{\alpha/\{1,2\fixedindex\}}^{(\kappa)}(\hat{\HH}_{\eta}; \bs{\ddot{x}}_{1,2\fixedindex} ) ,
\end{align*}
so Fatou's lemma yields
\begin{align}\label{eqn::liminf_kappa}
\liminf_{m\to\infty}
\PartF_\alpha^{(\kappa_m)}(\bs{x}) 
\ge 
\PartF_\alpha^{(\kappa)}(\bs{x}) .
\end{align}

\item On the other hand, 
as $\smash{\coulombGas_{\beta}^{(\kappa)}}$ is continuous in $\kappa\in (4,8)$, we conclude from~\eqref{eqn::CGI_PPF_48} that 
\begin{align}\label{eqn::limsup_kappa}
\limsup_{m\to\infty} \sum_{\alpha\in\LP_N} \mathcal{M}_{\nu(\kappa_m)}(\alpha,\beta) \; 
\PartF^{(\kappa_m)}_{\alpha}(\bs{x}) 
= \sum_{\alpha\in\LP_N} \meanderMatrix(\alpha, \beta) \; \PartF^{(\kappa)}_{\alpha} (\bs{x})  .
\end{align}
\end{itemize}
Combining~\eqref{eqn::liminf_kappa} and~\eqref{eqn::limsup_kappa} and the fact that $\meanderMatrix(\alpha,\beta)>0$ when $\kappa\in (4,8)$, we obtain 
\begin{align*}
\lim_{m\to\infty}
\PartF_\alpha^{(\kappa_m)}(\bs{x}) 
= 
\PartF_\alpha^{(\kappa)}(\bs{x}) ,
\end{align*}
which proves that $\PartF_\alpha^{(\kappa)}(\bs{x})$ is continuous on $(4,8) \ni \kappa$. 
\end{proof}

\begin{lemma}\label{lem::meandermatrix_invertible}
The meander matrix~\eqref{eqn::meandermatrix_def} is invertible if and only if $\kappa$ is not one of the exceptional values
\begin{align} \label{eq:bad_fugacity}
\kappa=\frac{4p}{p'} , \quad p,p' \in \bZpos \textnormal{ coprime and } 2 \le p \leq N + 1 .
\end{align}
\end{lemma}

\begin{proof}
By~\cite[Eq.~(5.6)]{DGG:Meanders_and_TL_algebra}, the determinant of the meander matrix~\eqref{eqn::meandermatrix_def} is given by 
\begin{align*} 
\det \meanderMatrix 
= \prod_{n=1}^N U_n(\fugacity)^{a_{N, n}} , 
\qquad \textnormal{where} \quad 
a_{N, n} = \binom{2N}{N-n} - 2\binom{2N}{N-n-1} + \binom{2N}{N-n-2} ,
\end{align*}
and $U_n(\cdot)$ are Chebyshev polynomials of the second kind. From~\cite[Eq.~(4.2)]{DGG:Meanders_and_TL_algebra},
we obtain
\begin{align*}
U_n(z+z^{-1}) = \frac{z^{n+1}-z^{-n-1}}{z-z^{-1}}. 
\end{align*}
Now, recall that $\fugacity = \fugacity(\kappa) = -\exp(4 \pi \ii / \kappa) - \exp(-4 \pi \ii / \kappa)$. 
Therefore, we have $\det\meanderMatrix=0$ if and only if 
\begin{align*} 
\exp(8(n+1) \pi \ii / \kappa) = 1 
\qquad \textnormal{for some } \; 
n \in \{1, \ldots, N\} ,
\end{align*}
which is equivalent to~\eqref{eq:bad_fugacity}.
\end{proof}

\begin{corollary} \label{cor::analyticity_PPF}
The pure partition functions $\{\smash{\PartF_{\alpha}^{(\kappa)}} \colon \alpha \in \LP\}$ are real-analytic in $\kappa\in (0,8)$. 
In particular, for each fixed $\alpha\in\LP$ and $\kappa_0\in (0,8)$, there exists $\epsilon>0$ such that $\PartF_{\alpha}^{(\kappa)}$ can be analytically extended to $B_{\epsilon}(\kappa_0):= \{ \upsilon \in \C \colon  |\upsilon - \kappa_0|<\epsilon\}  \ni  \kappa$.
\end{corollary}

\begin{proof}
Fix $N \geq 1$, $\alpha\in \LP_N$, and $\kappa_0\in (0,8)$. 
It follows from Lemma~\ref{lem::meandermatrix_invertible} that there exists $\epsilon_1>0$ such that the meander matrix~\eqref{eqn::meandermatrix_def} is invertible for all $\kappa\in B_{\epsilon_1}(\kappa_0)\setminus \{\kappa_0\}$. 
Clearly, there exists $\epsilon\in (0,\epsilon_1)$ such that $\big(B_{\epsilon}(\kappa_0)\setminus \{\kappa_0\}\big)\cap \big\{ \tfrac{8}{m} \colon m \in \bZpos \big\} =\emptyset$. 
It then follows from~\eqref{eqn::CGI_PPF_48} and the analyticity of $\coulombGas_{\beta}^{(\kappa)}$ in $\kappa\in B_{\epsilon}(\kappa_0)\setminus \{\kappa_0\}$ 
for all $\beta\in \LP_N$ that $\PartF_{\alpha}^{(\kappa)}$ can be analytically extended to this set. 
Finally, the continuity of $\kappa \mapsto \PartF_{\alpha}^{(\kappa)}$ in $\kappa\in (0,8)$ (Lemmas~\ref{lem::PPF_continuity_06}~\&~\ref{lem::PPF_continuity_48})
implies that $\kappa_0$ is a removable singularity. 
\end{proof}

\subsection{Finishing the proof of Theorem~\ref{thm::CGI} and Proposition~\ref{prop::CGI_PPF}}
\label{subsec:finish_CGI_proof}

The primary goal of this section is to derive the relation~\eqref{eqn::CGI_PPF} 
between the pure partition functions and the Coulomb gas integrals
(which is similar to~\cite[Proposition~1.10]{FPW:Connection_probabilities_of_multiple_FK_Ising_interfaces}) for all $\kappa \in (0,8)$. We also prove the continuity properties of the Coulomb gas integrals, thereby completing the proof of Theorem~\ref{thm::CGI}.

\begin{proof}[Proof of Theorem~\ref{thm::CGI} and Proposition~\ref{prop::CGI_PPF}] 
As in~\eqref{eqn::CGI_PPF_RHS}, denote by $\smash{\LG_{\beta}^{(\kappa)}}$ the right-hand side of~\eqref{eqn::CGI_PPF}. 
Since both $\meanderMat_{\fugacity(\kappa)}(\alpha,\beta)$ and $\smash{\PartF_{\alpha}^{(\kappa)}}$ are real-analytic in $(0,8)\ni\kappa$ by Corollary~\ref{cor::analyticity_PPF}, we see that 
$\smash{\LG_{\beta}^{(\kappa)}}$ is real-analytic in $(0,8)\ni\kappa$. 
Therefore, the function $\smash{\coulombGas_{\beta}^{(\kappa)}=\LG_{\beta}^{(\kappa)}}$ can be real-analytically extended to $\kappa\in (0,8)$ and the extensions satisfy~\eqref{eqn::PDE}, \eqref{eqn::COV}, \eqref{eqn::PLB_weak_upper}, 
\eqref{eqn::CGI_ASY}, and~\eqref{eqn::CGI_PPF} for all $\kappa\in (0,8)$. 
\begin{itemize}[leftmargin=3.5em]
\item[\textnormal{(POS):}]
For all $\kappa\in \underset{m\in \bZpos}{\bigcup} \big(\tfrac{8}{4m-1}, \tfrac{8}{4m-3}\big)$, since $\meanderMatrix(\alpha,\beta)>0$ and $\smash{\PartF_{\alpha}^{(\kappa)}} > 0$, 
we have 
\begin{align*}
\coulombGas_{\beta}^{(\kappa)}(\bs{x})= \sum_{\alpha\in\LP_N} \meanderMatrix(\alpha, \beta) \,  \PartF_{\alpha}^{(\kappa)}(\bs{x}) 
\; > \;  0, \qquad  \bs{x} \in \chamber_{2N} .
\end{align*}

\item[\textnormal{(NUL):}] 
For all $\kappa\in \big\{ \tfrac{8}{2m+1} \colon m \in \bZpos \big\}$, since $\meanderMatrix(\alpha,\beta)=0$, we have
\begin{align*}
\coulombGas_{\beta}^{(\kappa)}(\bs{x})= \sum_{\alpha\in\LP_N} \meanderMatrix(\alpha, \beta) \,  \PartF_{\alpha}^{(\kappa)}(\bs{x})
\;  = \; 0, \qquad  \bs{x} \in \chamber_{2N} .
\end{align*}

\item[\textnormal{(SGN):}] 
Note that when $\kappa\in \underset{m\in \bZpos}{\bigcup} \big(\tfrac{8}{4m+1}, \tfrac{8}{4m-1}\big)$, the fugacity is $\fugacity(\kappa)<0$.  
We claim that there are 
\begin{align*}
\bs{x}^+=\bs{x}^+(\kappa,\beta)\in\chamber_{2N}, \qquad
\bs{x}^0=\bs{x}^0(\kappa,\beta)\in\chamber_{2N}, \qquad
\bs{x}^-=\bs{x}^-(\kappa,\beta)\in\chamber_{2N},
\end{align*}
such that
\begin{align} \label{eqn::CGI_nonpos_again}
\coulombGas_{\beta}^{(\kappa)}(\bs{x}^+)>0,\qquad
\coulombGas_{\beta}^{(\kappa)}(\bs{x}^0)=0, \qquad
\coulombGas_{\beta}^{(\kappa)}(\bs{x}^-)<0 .
\end{align}
\end{itemize}
We prove~\eqref{eqn::CGI_nonpos_again} by induction on $N \geq 2$. 
The base case $N=2$ is governed by Proposition~\ref{prop::F4points_pos} in Appendix~\ref{app::examples}.
Assuming that~\eqref{eqn::CGI_nonpos} holds for 
$\{\smash{\coulombGas_{\beta}^{(\kappa)}} \colon \beta\in\LP_n , \; n \le N-1 \}$,  
fix $\beta \in \LP_N$. 
Without much loss of generality, suppose $\{1,2\}\in\beta$
(the case $\{1,2\}\not\in\beta$ is similar). 
The induction hypothesis for $\beta/\{1,2\}\in\LP_{N-1}$ guarantees that there exists a point
$\bs{y}^{+}=\bs{y}^+(\kappa,\beta/\{1,2\})\in\chamber_{2N-2}$
such that
$\smash{\coulombGas_{\beta/\{1,2\}}^{(\kappa)}}(\bs{y}^+)>0$. 
Now, the asymptotics property~\eqref{eqn::CGI_ASY} shows that
\begin{align*}
\lim_{x_1, x_2\to \xi}\frac{\coulombGas_{\beta}^{(\kappa)}(\bs{x})}{(x_2-x_1)^{-2h(\kappa)}} 
= \coulombGas_{\beta/\{1,2\}}^{(\kappa)}(\bs{y}^+) 
\; > \; 0 ,
\end{align*}
where $\bs{x}=(x_1, x_2, y_1, \ldots, y_{2N-2})\in\chamber_{2N}$
and $\bs{y}^+=(y_1, \ldots, y_{2N-2}) \in \chamber_{2N-2}$.
Therefore, there exists a point 
$\bs{x}^+=\bs{x}^+(\kappa, \beta)\in\chamber_{2N}$ such that $\smash{\coulombGas_{\beta}^{(\kappa)}}(\bs{x}^+)>0$. 
In a similar manner, we see that there exists a point  
$\bs{x}^-=\bs{x}^-(\kappa,\beta)\in\chamber_{2N}$ such that $\smash{\coulombGas_{\beta}^{(\kappa)}}(\bs{x}^-)<0$. 
However, the function $\bs{x}\mapsto \smash{\coulombGas_{\beta}^{(\kappa)}}(\bs{x})$ is continuous on $\chamber_{2N}$.
Thus, there exists a point $\bs{x}^0\in\chamber_{2N}$ such that 
$\smash{\coulombGas_{\beta}^{(\kappa)}}(\bs{x}^0)=0$. 
This completes the proof of (SGN)~\eqref{eqn::CGI_nonpos_again}. 

\smallbreak

Finally, $\{\smash{\coulombGas_{\beta}^{(\kappa)}} \colon \beta\in\LP_N\}$ is linearly independent for $\kappa$ in~\eqref{eqn::meander_inv} 
 because  $\{\smash{\PartF_{\alpha}^{(\kappa)}} \colon \alpha\in\LP_N\}$ is linearly independent 
(by Lemma~\ref{lem:PPF_lin_indep}) 
 and the meander matrix $\meanderMatrix$ is invertible for these $\kappa$ (by Lemma~\ref{lem::meandermatrix_invertible}). 
\end{proof}

The above results imply a partial improvement of the power-law bound~\eqref{eqn::PLB_weak_upper}. When $\kappa\le 6$, the pure partition functions satisfy a strong power-law bound~\eqref{eqn::PPF_bound_polygon}. 
Unfortunately, the proof of~\eqref{eqn::PPF_bound_polygon} in~\cite{Wu:Convergence_of_the_critical_planar_ising_interfaces_to_hypergeometric_SLE} does not apply to the case $\kappa\in (6,8)$ --- note that $h(\kappa) < 0$ when $\kappa > 6$.
The bound in~\eqref{eqn::PLB_middle} in the below Corollary~\ref{cor::PLB_middle} is weaker than~\eqref{eqn::PPF_bound_polygon}, but stronger than~\eqref{eqn::PLB_weak_upper}, possibly sufficient for certain applications.

\begin{corollary}\label{cor::PLB_middle}
Fix $\kappa\in (4,8)$. The pure partition functions $\{\smash{\PartF_{\alpha}^{(\kappa)}} \colon \alpha \in \LP\}$ satisfy the following refined power-law bound: 
there exists constant $C=C(N,\kappa) \in (0,\infty)$ such that for all $\bs x = (x_1, \ldots, x_{2N})\in\chamber_{2N}$, 
\begin{align}\label{eqn::PLB_middle}
0<\PartF_{\alpha}^{(\kappa)}(\bs x) \le C \prod_{1\le i<j\le 2N}|x_j-x_i|^{\mu_{ij}}, 
\quad\textnormal{where } 
\mu_{ij} 
= \begin{cases}
10/\kappa, &\textnormal{if }|x_j-x_i|\ge 1,\\
-6/\kappa, &\textnormal{if }|x_j-x_i|< 1. 
		\end{cases}
	\end{align}
\end{corollary}

\begin{proof} 
By Proposition~\ref{prop::CGI_PPF} with $\beta = {\hbox{\includegraphics[scale=1]{figures-arXiv/link61.pdf}}}_N$, we have
\begin{align*}
\smash{\coulombGas^{(\kappa)}_{\!\! \vcenter{\hbox{\includegraphics[scale=0.8]{figures-arXiv/link61.pdf}}}_N}}( \bs x )  
= \sum_{\alpha\in\LP_N}\meanderMatrix(\alpha, \vcenter{\hbox{\includegraphics[scale=0.8]{figures-arXiv/link61.pdf}}}_N)\, \PartF_{\alpha}^{(\kappa)} ( \bs x ) .
\end{align*}
It suffices to show that there exists constant $C=C(N,\kappa) \in (0,\infty)$ such that 
\begin{align}\label{eqn::CGI_PLB_middle}
0<\smash{\coulombGas^{(\kappa)}_{\!\! \vcenter{\hbox{\includegraphics[scale=0.8]{figures-arXiv/link61.pdf}}}_N}}( \bs x ) 
\le C \prod_{1\le i<j\le 2N}|x_j-x_i|^{\mu_{ij}}, \qquad \bs x \in \chamber_{2N} .
\end{align}
Definition~\ref{def::CGI_def} and a straightforward contour deformation computation 
(see also Section~\ref{subsec:Contour_deformations} and~\cite{FPW:Connection_probabilities_of_multiple_FK_Ising_interfaces}, for instance)
show that, when $\kappa\in (4,8)$, we have
\begin{align}\label{eqn::CGI_PLB_middle_aux1}
0<\smash{\coulombGas^{(\kappa)}_{\!\! \vcenter{\hbox{\includegraphics[scale=0.8]{figures-arXiv/link61.pdf}}}_N}}( \bs x ) 
:= \bigg(\frac{- 2 \, \cos(4\pi/\kappa) \, \Gamma(2-8/\kappa)}{\Gamma(1-4/\kappa)^2}\bigg)^N \; 
\int_{x_{1}}^{x_{2}}\cdots\int_{x_{2N-1}}^{x_{2N}}f(\bs{x};u_1,\ldots,u_N) \; \ud u_1\cdots\ud u_N,
\end{align}
where each $\ud u_r$ is integrated from $x_{2r-1}$ to $x_{2r}$, for $1\le r\le N$, and the integrand is
\begin{align}\label{eqn::CGI_PLB_middle_aux2}
f(\bs{x};u_1,\ldots,u_N) := \prod_{1\le i<j\le 2N}(x_j-x_i)^{2/\kappa}
\prod_{1\le r<s\le N}|u_s-u_r|^{8/\kappa} 
\prod_{\substack{1\le j\le 2N, \\ 1\le t\le N}} |u_t-x_j|^{-4/\kappa}.
\end{align}
The third and second terms in the integrand~\eqref{eqn::CGI_PLB_middle_aux2} can be easily bounded from above:  
\begin{align*}
\prod_{1\le j\le 2N}|u_t-x_j|^{-4/\kappa} 
\le \; &|u_t-x_{2t-1}|^{-4/\kappa}|u_t-x_{2t}|^{-4/\kappa} \prod_{1\le j<2t-1}|x_{2t-1}-x_j|^{-4/\kappa} 
\prod_{2t<j\le 2N}|x_j-x_{2t}|^{-4/\kappa} ,
\\
|u_s-u_r|\le \; & |x_{2s}-x_{2r-1}| , \qquad r,s,t \in \{1,2,\ldots,N\} , \, r<s .
\end{align*}
Note now that using the change of variables $v=\frac{u_t-x_{2t-1}}{x_{2t}-x_{2t-1}}$, we have 
\begin{align*}
\int_{x_{2t-1}}^{x_{2t}}|u_t-x_{2t-1}|^{-4/\kappa}|u_t-x_{2t}|^{-4/\kappa}\ud u_t
= \; & (x_{2t}-x_{2t-1})^{1-8/\kappa} \, \frac{\Gamma(1-4/\kappa)^2}{\Gamma(2-8/\kappa)}. 
\end{align*}
Plugging these three observations into~\eqref{eqn::CGI_PLB_middle_aux1}, we see that 
\begin{align}\label{eqn::CGI_PLB_middle_aux3}
\smash{\coulombGas^{(\kappa)}_{\!\! \vcenter{\hbox{\includegraphics[scale=0.8]{figures-arXiv/link61.pdf}}}_N}}( \bs x ) 
\le \; & C \prod_{1\le i<j\le 2N}(x_j-x_i)^{2/\kappa} 
\; \prod_{1\le r<s\le N}|x_{2s}-x_{2r-1}|^{8/\kappa} \\
\; & \times \; |x_{2t}-x_{2t-1}|^{1-8/\kappa} 
\; \prod_{1\le t\le N} \bigg( \prod_{1\le j<2t-1}|x_{2t-1}-x_j|^{-4/\kappa} 
\; \prod_{2t<j\le 2N}|x_j-x_{2t}|^{-4/\kappa} \bigg) , \notag
\end{align}
where $C = C(N,\kappa)$ only depends on $N$ and $\kappa$.
The sought bound~\eqref{eqn::CGI_PLB_middle} follows by collecting the exponents on the right-hand side of~\eqref{eqn::CGI_PLB_middle_aux3}: the accumulated exponents for the various terms are:
\begin{itemize}
\item $2/\kappa+1-8/\kappa=1-6/\kappa$ for $|x_{2t}-x_{2t-1}|$ with $1\le t\le N$;

\item $2/\kappa+8/\kappa=10/\kappa$ for $|x_{2s}-x_{2r-1}|$ with $1\le r<s\le N$;

\item $2/\kappa-4/\kappa=-2/\kappa$ for $|x_{2t-1}-x_{2t-2}|$ with $2\le t\le N$;

\item $2/\kappa-4/\kappa-4/\kappa=-6/\kappa$ for $|x_{2s+1}-x_{2r}|$ with $1\le r<s\le N-1$;

\item $2/\kappa-4/\kappa=-2/\kappa$ for $|x_{2s}-x_{2t}|$ with $1\le t<s\le N$;

\item $2/\kappa-4/\kappa=-2/\kappa$ for $|x_{2t-1}-x_{2r-1}|$ with $1\le r<t\le N$.
\end{itemize}
\qedhere
\end{proof}

\subsection{Renormalized partition functions}
\label{subsec::renormalized_pf}

We will show in this section that null partition functions can be positive and finite after a suitable renormalization~\eqref{eqn::CGI_odd}, \eqref{eqn::CGI_8}, and~\eqref{eqn::PPF_continuity8}. 
Moreover, the renormalized partition functions are linearly independent and have asymptotic properties analogous to those of the generic ones (see Lemma~\ref{lem::ASY_CGI_hat}).

\begin{lemma}\label{lem::meandermatrix_invertible_re}
The renormalized meander matrix~\eqref{eqn::meandermatrix_renormalized} is invertible. 
\end{lemma}
\begin{proof}
This is immediate from~\cite[Eq.~(5.18)]{DGG:Meanders_and_TL_algebra}, since $\det\meanderRenorm=(\partial_{\fugacity}\det\meanderMatrix)\;\big|_{\nu=0}\neq0$.
\end{proof}

\begin{proof}[Proof of Proposition~\ref{prop::CGI_odd}]
From~\eqref{eqn::CGI_PPF}, for all $\kappa'\in (0,8)$, we have
\begin{align*}
\coulombGas_{\beta}^{(\kappa')}(\bs{x}) 
= \sum_{\alpha\in\LP_N} \meanderMatrixPrime(\alpha, \beta) \,  \PartF_{\alpha}^{(\kappa')}(\bs{x}), \qquad  \bs{x} \in \chamber_{2N} , \; \beta \in \LP_N .
\end{align*}
Thus, for $\kappa\in \big\{ \tfrac{8}{2m+1} \colon m \in \bZpos\big\}$, we obtain
\begin{align*}
\lim_{\kappa'\to \kappa}
\tfrac{1}{\fugacity(\kappa')} \, \coulombGas_{\beta}^{(\kappa')}(\bs{x}) 
= \; & \sum_{\alpha\in\LP_N} \Big( 
\lim_{\kappa'\to\kappa} 
\tfrac{1}{\fugacity(\kappa')} \, \meanderMatrixPrime(\alpha, \beta) \Big) 
\; \PartF_{\alpha}^{(\kappa')}(\bs{x}) \\
= \; & \sum_{\alpha\in\LP_N} \meanderRenorm(\alpha, \beta) \, \PartF_{\alpha}^{(\kappa)}(\bs{x}) ,
&& \textnormal{[by~\eqref{eqn::meandermatrix_renormalized} and Lemmas~\ref{lem::PPF_continuity_06}~\&~\ref{lem::PPF_continuity_48}]} 
\end{align*}
using the definition of $\meanderRenorm$ and the continuity of pure partition functions in $\kappa$. 
Now, note that 
\begin{align*}
\lim_{\kappa'\to \kappa}\frac{(\kappa'-\kappa)}{\fugacity(\kappa')}=\frac{-\kappa^2}{8\pi\sin(4\pi/\kappa)}.
\end{align*}
This proves Proposition~\ref{prop::CGI_odd}.
\end{proof}

Moreover, from the above proof, we obtain
\begin{align} \label{eqn::CGI_PPF_odd_bis}
\coulombGasRenorm_{\beta}^{(\kappa)}(\bs{x}) 
\; = \; \sum_{\alpha\in\LP_N} \meanderRenorm(\alpha, \beta) \, \PartF_{\alpha}^{(\kappa)} (\bs{x}) \; \in \; (0,\infty)
, \qquad \bs{x}\in\chamber_{2N} , \; \beta \in \LP_N ,
\end{align}
where the positivity follows from the facts that 
the renormalized meander matrix has non-negative entries~\eqref{eqn::meandermatrix_renormalized}  
and is invertible (Lemma~\ref{lem::meandermatrix_invertible_re}) 
and from the positivity of $\smash{\PartF_{\alpha}^{(\kappa)}}$ (Lemma~\ref{lem::ppf_continuity_and_positivity}). 
We can also verify that the renormalized partition functions $\smash{\coulombGasRenorm_{\beta}^{(\kappa)}}$ satisfy recursive asymptotics analogous to~\eqref{eqn::CGI_ASY}.

\begin{lemma}\label{lem::ASY_CGI_hat}
Fix $\kappa\in\{\tfrac{8}{2m+1}: m\in \bZpos\}$. For the renormalized Coulomb gas integral~\eqref{eqn::CGI_odd}, we have 
\begin{align*}
\coulombGasRenorm_{\emptyset}^{(\kappa)} \equiv 1 
\qquad  \textnormal{and} \qquad 
\coulombGasRenorm_{\includegraphics[scale=0.2]{figures-arXiv/link-0}}^{(\kappa)}(x_1, x_2)=(x_2-x_1)^{-2h(\kappa)},
\end{align*}
and for any $N \ge 2$ and $j \in \{1,2, \ldots, 2N-1 \}$ and for all $\xi \in (x_{j-1}, x_{j+2})$,  
\begin{align}
\label{eqn::CGI_ASY_hat} 
\; & \lim_{x_j,x_{j+1}\to\xi} 
\frac{\coulombGasRenorm_{\beta}^{(\kappa)}(\bs{x})}{ (x_{j+1}-x_j)^{-2h(\kappa)} }
= 
\begin{cases}
0 ,
& \textnormal{if }\{j, j+1\}\in\beta , \\
\coulombGasRenorm_{\wp_j(\beta)/\{j,j+1\}}^{(\kappa)}(\bs{\ddot{x}}_j),
& \textnormal{if }\{j, j+1\} \not\in \beta .
\end{cases}
\tag{$\coulombGasRenorm$\textnormal{-ASY}}
\end{align}
\end{lemma}

\begin{proof}
Fix $N\ge 2$. 
The asymptotics~\eqref{eqn::CGI_ASY_hat} follow from the asymptotics of the right-hand side of~\eqref{eqn::CGI_PPF_odd_bis}, 
which is obtained easily from the asymptotics properties~\eqref{eqn::PPF_ASY} 
of $\smash{\PartF_{\alpha}^{(\kappa)}}$ 
from~\cite[Lemma~4.3]{Peltola-Wu:Global_and_local_multiple_SLEs_and_connection_probabilities_for_level_lines_of_GFF} 
(for $\kappa \in (0,4]$) 
and from Proposition~\ref{prop::PPF_ASY} (for $\kappa \in (4,8)$):
\begin{itemize}
\item if $\{j,j+1\}\in\beta$, we have
\begin{align*}
\lim_{x_j, x_{j+1}\to \xi}\frac{\coulombGasRenorm_{\beta}^{(\kappa)}(\bs{x})}{(x_{j+1}-x_j)^{-2h(\kappa)}} 
= \; & \sum_{\substack{\alpha\in\LP_N \\\{j,j+1\}\not\in\alpha}}\meanderRenorm(\alpha,\beta) \; 
\lim_{x_j, x_{j+1}\to \xi}\frac{\PartF_{\alpha}^{(\kappa)}(\bs{x})}{(x_{j+1}-x_j)^{-2h(\kappa)}}
\; = \; 0 \textnormal{;} 
&& \textnormal{[by~\eqref{eqn::PPF_ASY}]}
\end{align*}

\item if $\{j,j+1\}\not\in\beta$, we have
\begin{align*}
\lim_{x_j, x_{j+1}\to \xi}\frac{\coulombGasRenorm_{\beta}^{(\kappa)}(\bs{x}) }{(x_{j+1}-x_j)^{-2h(\kappa)}} 
= \; & \sum_{\substack{\alpha\in\LP_N \\\{j,j+1\}\in\alpha}}\meanderRenorm(\alpha,\beta) \; 
\lim_{x_j, x_{j+1}\to \xi}\frac{\PartF_{\alpha}^{(\kappa)}(\bs{x})}{(x_{j+1}-x_j)^{-2h(\kappa)}} \\
= \; & \sum_{\substack{\alpha\in\LP_N \\\{j,j+1\}\in\alpha}}\meanderRenorm(\alpha,\beta) \; \PartF_{\alpha/\{j, j+1\}}^{(\kappa)}(\bs{\ddot{x}}_j) 
&& \textnormal{[by~\eqref{eqn::PPF_ASY}]}
\\
= \; & 
\coulombGasRenorm_{\wp_j(\beta)/\{j,j+1\}}^{(\kappa)}(\bs{\ddot{x}}_j) ,
&& \textnormal{[by~\eqref{eqn::CGI_PPF_odd_bis}]}
\end{align*}
using the bijection $\alpha \leftrightarrow \alpha/\{j,j+1\}$ 
between $\{\alpha\in\LP_N \colon \{j,j+1\}\in\alpha\}$ and $\LP_{N-1}$ 
and the fact that  
$\meanderRenorm(\alpha,\beta) = \meanderRenorm(\alpha/\{j,j+1\}, \wp_j(\beta)/\{j,j+1\})$ to arrive at the last expression.
\qedhere
\end{itemize}
\end{proof}

\begin{proof}[Proof of Proposition~\ref{prop::CGI_8}]
Recall from~\eqref{eq: normalization_cst} and~\eqref{eqn::coulombgasintegral_Poch} that 
\begin{align*}
\LH_\beta^{(\kappa')} (\bs{x}) 
:= \; & \ointclockwise_{\acycle^\beta_1}  \ud u_1 \ointclockwise_{\acycle^\beta_2}  \ud u_2 \cdots \ointclockwise_{\acycle^\beta_N} \ud u_N 
\; f_\beta^{(\kappa')}(\bs{x};\bs{u}), 
\end{align*}
and 
\begin{align*}
\coulombGas_{\beta}^{(\kappa')}(\bs{x})=\cst(\kappa')^N\LH_\beta^{(\kappa')} (\bs{x}) , 
\qquad \textnormal{where }
\cst(\kappa') = \frac{\fugacity(\kappa')}{4\sin^2(4\pi/\kappa')} 
\, \frac{\Gamma(2-8/\kappa')}{\Gamma(1-4/\kappa')^2}.
\end{align*}
Note that 
\begin{align*}
\lim_{\kappa'\to 8-}\frac{(8-\kappa')}{\fugacity(\kappa')} 
= \frac{8}{\pi} 
\qquad \textnormal{and} \qquad 
\lim_{\kappa'\to 8-}\frac{\cst(\kappa')}{\fugacity(\kappa')}= \frac{\Gamma(1)}{4\Gamma(1/2)^2}=\frac{1}{4\pi} ,
\end{align*}
and that $\LH_\beta^{(\kappa')}$ is continuous in $\kappa'$ 
and the limit $\LH_{\beta}^{(8)} = \underset{\kappa'\to 8-}{\lim}\LH_\beta^{(\kappa')}$ is positive and finite. 
Hence, we have 
\begin{align*}
\widehat{\coulombGas}_{\beta}^{(8)}(\bs{x}):=
\left(\frac{\pi}{8}\right)^{N-1}\lim_{\kappa'\to 8-}\frac{\coulombGas_{\beta}^{(\kappa')}(\bs{x})}{\fugacity(\kappa')^N}
=\frac{8}{32^N\pi}\LH_{\beta}^{(8)}(\bs{x}) 
 \; \in \; (0,\infty) . 
\end{align*}
This proves Proposition~\ref{prop::CGI_8}.
\end{proof}

\subsection{Finishing the proof of Theorem~\ref{thm::PPF}}
\label{subsec:finish_PPF_proof}

\begin{proof}[Proof of Theorem~\ref{thm::PPF}]
The existence and uniqueness of the collection 
$\{\smash{\PartF_{\alpha}^{(\kappa)}} \colon \alpha \in \LP\}$ 
follows from~\cite[Theorem~1.1]{Peltola-Wu:Global_and_local_multiple_SLEs_and_connection_probabilities_for_level_lines_of_GFF} 
(for $\kappa \in (0,4]$) 
and from the results of Section~\ref{sec::PPF} (for $\kappa \in (4,8)$):
\begin{itemize}
\item
the uniqueness follows by results of Flores and Kleban (see Lemma~\ref{lem::PFuniqueness}); and

\item 
the PDE system~\eqref{eqn::PDE} is given by Proposition~\ref{prop::PPF_PDE}; 

\item 
the M\"{o}bius covariance~\eqref{eqn::COV} is given by  Lemma~\ref{lem::PPF_COV};

\item 
the recursive asymptotics property~\eqref{eqn::PPF_ASY} is given by  Proposition~\ref{prop::PPF_ASY}; 

\item 
the power-law bound~\eqref{eqn::PLB_weak_upper} is given by  Lemma~\ref{lem::PPF_PLB_weak};

\item 
and the positivity is given by  Lemma~\ref{lem::ppf_continuity_and_positivity}.
\end{itemize}
The linear independence is the content of Lemma~\ref{lem:PPF_lin_indep}.
The continuity (Item~\ref{item:PPF_continuity} of Theorem~\ref{thm::PPF}) 
is the content of Lemmas~\ref{lem::PPF_continuity_06}~\&~\ref{lem::PPF_continuity_48}. 
The limit as $\kappa \to 8-$ (Item~\ref{item:PPF_continuity8} of Theorem~\ref{thm::PPF}) can be shown using the relation~\eqref{eqn::CGI_PPF} (proven in Lemma~\ref{lem::CGI_PPF_good}): 
for all $\kappa'\in (0,8)$, 
we have
\begin{align}\label{eqn::CGI_PPF_bis}
\coulombGas_{\beta}^{(\kappa')}(\bs{x}) 
= \sum_{\alpha\in\LP_N} \meanderMatrixPrime(\alpha, \beta) \,  \PartF_{\alpha}^{(\kappa')}(\bs{x}), \qquad  \bs{x} \in \chamber_{2N} , \; \beta \in \LP_N .
\end{align}
As the renormalized meander matrix is invertible by Lemma~\ref{lem::meandermatrix_invertible_re}, 
and there exists $\kappa_0=\kappa_0(N)\in (4,8)$ such that
the meander matrix $\{\meanderMatrixPrime(\alpha,\beta) \colon \alpha,\beta\in\LP_N\}$ is invertible for all $\kappa'\in (\kappa_0, 8)$ 
by Lemma~\ref{lem::meandermatrix_invertible}, 
Cramer's rule shows that 
\begin{align}\label{eqn::invertible_matrix_limit}
\lim_{\kappa'\to 8-} \Big( \tfrac{1}{\fugacity(\kappa')}\meanderMatrixPrime\Big)^{-1}(\alpha,\beta) 
= \meanderRenorm^{-1}(\alpha,\beta),\qquad \textnormal{for all }\alpha,\beta\in\LP_N.
\end{align}
Plugging this into~\eqref{eqn::CGI_PPF_bis}, we obtain \eqref{eqn::PPF_continuity8}: 
\begin{align*}
\widehat{\PartF}_{\alpha}^{(8)}(\bs{x}) 
:= \; & \lim_{\kappa'\to 8-}\frac{\PartF_{\alpha}^{(\kappa')}(\bs{x})}{(8-\kappa')^{N-1}} 
\; = \;
\left(\frac{\pi}{8}\right)^{N-1} 
\lim_{\kappa'\to 8-}\frac{\PartF^{(\kappa')}_{\alpha}(\bs{x})}{\fugacity(\kappa')^{N-1}} 
\\
= \; & 
\left(\frac{\pi}{8}\right)^{N-1} 
\lim_{\kappa'\to 8-}\sum_{\beta\in\LP_N}
\Big( \tfrac{1}{\fugacity(\kappa')}\meanderMatrixPrime\Big)^{-1}(\alpha,\beta)  \; \frac{\coulombGas^{(\kappa')}_{\beta}(\bs{x})}{\fugacity(\kappa')^{N}}
&& \textnormal{[by~\eqref{eqn::CGI_PPF_bis}]} 
\\
= \; & 
\sum_{\beta\in\LP_N} \meanderRenorm^{-1}(\alpha,\beta) \; \widehat{\coulombGas}_{\beta}^{(8)} \; \in \; (0,\infty) ,
&& \textnormal{[by~\eqref{eqn::CGI_8} and~\eqref{eqn::invertible_matrix_limit}]} 
\end{align*}
where the positivity follows from~
\cite[Theorem~1.2]{LPW:UST_in_topological_polygons_partition_functions_for_SLE8_and_correlations_in_logCFT}.
\end{proof}

Moreover, from the above proof, we obtain
\begin{align}\label{eqn::CGI_PPF_8bis}
\widehat{\PartF}_{\alpha}^{(8)}(\bs{x}) 
= \sum_{\beta\in\LP_N}\meanderRenorm^{-1}(\alpha,\beta) \; \widehat{\coulombGas}_{\beta}^{(8)}.
\end{align} 

\begin{proof}[Proof of Proposition~\ref{prop::CGI_PPF_renormalized}]
The first claim~\eqref{eqn::CGI_PPF_odd} follows from~\eqref{eqn::CGI_PPF_odd_bis},  
and the linear independence follows from the fact that the renormalized meander matrix is invertible (Lemma~\ref{lem::meandermatrix_invertible_re}) 
and the pure partition functions are linearly independent (Lemma~\ref{lem:PPF_lin_indep}). 
The second claim~\eqref{eqn::CGI_PPF_8} follows from~\eqref{eqn::CGI_PPF_8bis}, 
and the linear independence is proved in~\cite[Section~4.3]{LPW:UST_in_topological_polygons_partition_functions_for_SLE8_and_correlations_in_logCFT}.\end{proof}


\bigskip{}
\section{Frobenius series}
\label{sec:Frobenius_series_proofs}
In this section, we refine the asymptotic properties considered in Section~\ref{subsec:ASY} to derive 
the Frobenius series of the Coulomb gas integral functions. 
We thus finish the proofs of Propositions~\ref{prop::Frobenius_F},~\ref{prop::CGI8over3_Frobenius}~\&~\ref{prop::CGI8_Frobenius}.

It was proven in~\cite{Dubedat:SLE_and_Virasoro_representations_fusion} and~\cite{Flores-Kleban:Solution_space_for_system_of_null-state_PDE4} that solutions to the second order BPZ PDEs~\eqref{eqn::PDE} 
have a Frobenius series expansion. In this section, our aim is to make these expansions explicit for the Coulomb gas integral functions $\coulombGas_{\beta}$. 
Using Proposition~\ref{prop::CGI_PPF}, one can then derive the explicit Frobenius series expansions also for the pure partition functions $\PartF_\alpha$.  
From the algebraic points of view, it is known that the pure partition functions form a distinguished basis in a tensor product module of the quantum group (Hopf algebra) $U_q(\mathfrak{sl}(2,\C))$ with $q = \exp(4 \pi \ii / \kappa)$ 
(cf.~Remark~\ref{rem:quantum_group} and~\cite{Kytola-Peltola:Pure_partition_functions_of_multiple_SLEs}) 
--- dual to the basis carrying the natural action of 
the Temperley-Lieb algebra, 
see, e.g.,~\cite{Flores-Peltola:Standard_modules_radicals_and_the_valenced_TL_algebra, 
Flores-Peltola:Higher_spin_QSW, 
GSNJRS:Global_symmetry_and_conformal_bootstrap_in_the_two-dimensional_On_model} 
(which for the functions is realized as iterated limits, or ``fusions'', cf.~Section~\ref{subsec::pf_lin_indep}).
In light of the anticipated algebraic structure, the factors in 
the Frobenius series in the generic cases addressed in Proposition~\ref{prop::Frobenius_F} 
completely match our expectations (the cases where $\kappa = 8/m$ for some integer $m$ 
correspond to more complicated, non-semisimple representation theory, 
with logarithmic corrections to the Frobenius series in Propositions~\ref{prop::CGI8over3_Frobenius}~\&~\ref{prop::CGI8_Frobenius}).

\subsection{Generic case: proof of Proposition~\ref{prop::Frobenius_F}}
\label{subsec: Proof of Frob generic}

In this section, we use the notation from Proposition~\ref{prop::Frobenius_F}.
The proof of Proposition~\ref{prop::Frobenius_F} comprises the two 
Lemmas~\ref{lem::Frobenius_F_aux1}~\&~\ref{lem::coeff_frob_bad_axu2},
together with the analysis of the exceptional case $\kappa=4$ in Appendix~\ref{app::Frob}: 

\begin{proof}[Proof of Proposition~\ref{prop::Frobenius_F}]
When $\{1,2\}\in \beta$, the asserted expansion~\eqref{eqn::CGI_Frobenius_paired} is proven in Lemma~\ref{lem::Frobenius_F_aux1}. 
When $\{1,2\}\notin \beta$, the asserted expansion~\eqref{eqn::CGI_Frobenius_tying} is proven in Lemma~\ref{lem::coeff_frob_bad_axu2} for the generic case
$\kappa\in (8/3,8)\setminus \{4\}$, and in Proposition~\ref{prop::CGI4_Frobenius} 
in Appendix~\ref{app::Frob} for the exceptional case $\kappa=4$. 
\end{proof}

\begin{lemma} \label{lem::Frobenius_F_aux1}
Fix $\kappa\in (0,8)\setminus \big\{ \tfrac{8}{2m+1} \colon m\in \bZpos \big\}$. 
Fix $\beta\in\LP_N$ and suppose $\{1,2\}\in\beta$. 
Then, for all $\xi\in (-\infty, x_3)$, we have
\begin{align} \label{eqn::F_Fro_good_1}
\coulombGas_{\beta}(\bs{x}) 
= \; & \fugacity(\kappa) \, 
\coulombGas_{\beta/\{1,2\}}(\bs{\ddot{x}}_1)
\, (x_{2}-x_1)^{-2h(\kappa)}
\; + \; 
O\big( (x_{2}-x_1)^{-2h(\kappa)+2} \big) , \qquad \textnormal{as } x_1, x_{2}\to \xi  .
\end{align}
\end{lemma}

As a consequence of Lemma~\ref{lem::Frobenius_F_aux1}, 
we see that Eq.~\eqref{eqn::CGI_Frobenius_paired} 
in Proposition~\ref{prop::Frobenius_F} holds for all $\kappa\in (8/3,8)$. 
Indeed, when $\kappa > 8/3$, we have $-2h(\kappa)+2 = (3\kappa - 6)/\kappa > 2/\kappa$, so~\eqref{eqn::CGI_Frobenius_paired} follows from~\eqref{eqn::F_Fro_good_1}.

\begin{proof}
Using the change of variables $u=\frac{u_1-x_1}{x_{2}-x_1}$, we obtain 
\begin{align} \label{eqn::H_change_variables_good}
\begin{split}
\coulombGasH_{\beta}(\bs{x}) 
= \; & (x_{2}-x_1)^{-2h(\kappa)}\ointclockwise_{\acycle_2^{\beta}}\ud u_2\cdots  \ointclockwise_{\acycle^{\beta}_{N}}\ud u_N 
\; f_{\beta/\{1,2\}}(\bs{\ddot{x}}_1;\bs{\dot{u}}_1)  \\
\; & \times \ointclockwise_{\acycle(0,1)} \ud u \; (u-1)^{-4/\kappa}u^{-4/\kappa} 
\; g(\bs{\ddot{x}}_1;\bs{\dot{u}}_1;x_1,x_{2},u) ,
\end{split}
\end{align}
where 
\begin{align*} 
g(x_1,x_2,u; \bs{\ddot{x}}_1;\bs{\dot{u}}_1) 
:= \; & 
\prod_{s=2}^{N} \Big( \big(u_s- (x_{2}-x_1)u-x_1\big)^{8/\kappa} 
\, (u_s-x_1)^{-4/\kappa}(u_s-x_{2})^{-4/\kappa} \Big) 
 \\
 \; & \times \prod_{i=3}^{2N} \big(x_i-(x_{2}-x_1)u-x_1 \big)^{-4/\kappa} 
\, \prod_{i=3}^{2N} \big( (x_i-x_1)^{2/\kappa} (x_i-x_2)^{2/\kappa}\big) ,
\end{align*}
and its branch is chosen to be real and positive when $0<\Re (u)<1$ and $x_{a_s}<\Re (u_s)<x_{a_s+1}$ for all $s \in \{3,4,\ldots,N\}$. 
We expand the function $g$ in $(x_{2}-x_1)$ exactly as
\begin{align} \label{eqn::expansion_g}
\begin{split}
g(x_1,x_2,u; \bs{\ddot{x}}_1;\bs{\dot{u}}_1) 
= \; & 1 \; + \; (x_{2}-x_1)(2u-1) 
\bigg(\sum_{s=2}^{N} 
\frac{4/\kappa}{x_1 - u_s} 
\; - \; \sum_{i=3}^{2N} \frac{2/\kappa}{x_i-x_1}\bigg)  \\
\; & + \; (x_2-x_1)^2 \; \hat{g}^{(\kappa)} (x_1,x_2,u; \bs{\ddot{x}}_1,\bs{\dot{u}}_1),
\end{split}
\end{align}
where the function $\smash{\hat{g}^{(\kappa)}(x_1,x_2,u; \bs{\ddot{x}}_1,\bs{\dot{u}}_1)}$ is uniformly bounded in 
$x_1,x_2\in (\xi-(x_3-\xi)/3,\xi+(x_3-\xi)/3)$ and  
$\kappa\in \big(\tfrac{8}{2m+1},\tfrac{8}{2m-1} \big)$ 
for each fixed $m\in \bZpos$, and $\bs{\ddot{x}}_1 \in \chamber_{2N-2}$, and $\xi \in (-\infty,x_3)$.
Therefore, for each fixed $m\in \bZpos$, 
plugging Eq.~\eqref{eqn::expansion_g} and 
Eq.~(\ref{eqn::beta_1},~\ref{eqn::beta_vanish}) from Lemma~\ref{lem::beta_acycle} into~\eqref{eqn::H_change_variables_good} shows that
\begin{align}
\nonumber
& \; \coulombGas_{\beta}(\bs{x}) 
- \fugacity(\kappa) \, \coulombGas_{\beta/\{1,2\}}(\bs{\ddot{x}}_1) \, (x_2-x_1)^{-2h(\kappa)} 
\\
\nonumber
= \; & (x_2-x_1)^{-2h(\kappa)+2} \,  \cst(\kappa)^N 
\; \ointclockwise_{\acycle_2^{\beta}}\ud u_2\cdots  \ointclockwise_{\acycle^{\beta}_{N}}\ud u_N  
\; f_{\beta/\{1,2\}}(\bs{\ddot{x}}_1;\bs{\dot{u}}_1) 
\\
\nonumber
\; &  \qquad\qquad\qquad\qquad\qquad\qquad\qquad \times
\; \ointclockwise_{\acycle(0,1)} \ud u  \; (u-1)^{-4/\kappa}\, u^{-4/\kappa} 
\; \hat{g}^{(\kappa)}(x_1,x_2,u; \bs{\ddot{x}}_1;\bs{\dot{u}}_1) 
\\
\label{eqn::Fro_pairing_aux1}
= \; & O\big( (x_{2}-x_1)^{-2h(\kappa)+2} \big) , \qquad \textnormal{as } x_1, x_{2}\to \xi ,
\end{align}
for any $\kappa\in \big( \tfrac{8}{2m+1},\tfrac{8}{2m-1} \big) \setminus \big\{\tfrac{4}{m}\big\}$. 
To extend~\eqref{eqn::F_Fro_good_1} to all parameters $\kappa\in (0,8)\setminus \big\{ \tfrac{8}{2m+1} \colon m\in \bZpos \big\}$, let us fix $\bs{\ddot{x}}_1\in \chamber_{2N-2}$, and $\xi\in (-\infty,x_3)$, and $m\in \bZpos$,
and consider the limit $\kappa' \to \kappa_m := 4/m$. 
Similarly as in the proof of Lemma~\ref{lem::Poch_line_relation}, in the integral~\eqref{eqn::Fro_pairing_aux1} over 
$\smash{f_{\beta/\{1,2\}}^{(\kappa')} (\bs{\ddot{x}}_1;\bs{\dot{u}}_1) \, \hat{g}^{(\kappa')}(x_1,x_{2},u;\bs{\ddot{x}}_1;\bs{\dot{u}}_1)}$, 
we can replace each Pochhammer contour $\acycle(x_{a_r}, x_{b_r})$ by the linear combination of the three contours
$\smash{(1-e^{8\pi\ii/\kappa}) \, (1-e^{-8\pi\ii/\kappa}) \, \landupint_{x_{a}+\epsilon}^{x_{b}-\epsilon} \ud u_r}$, and 
$(1-e^{8\pi\ii/\kappa}) \, \smash{\ointclockwise_{\partial B(x_{a},\epsilon)}\ud u_r}$, and 
$(1-e^{-8\pi\ii/\kappa}) \, \smash{\ointclockwise_{\partial B(x_{b},\epsilon)}\ud u_r}$;
and similarly, 
we can replace the Pochhammer contour $\acycle(0, 1)$ by the linear combination of the three contours
$\smash{(1-e^{8\pi\ii/\kappa}) \, (1-e^{-8\pi\ii/\kappa}) \, \int_{\epsilon}^{1-\epsilon} \ud u}$, and 
$(1-e^{8\pi\ii/\kappa}) \, \smash{\ointclockwise_{\partial B(0,\epsilon)}\ud u}$, and 
$(1-e^{-8\pi\ii/\kappa}) \, \smash{\ointclockwise_{\partial B(1,\epsilon)}\ud u}$.
Therefore, since
\begin{align*}
\textnormal{both }
\; \cst(\kappa) \, (1-e^{8\pi\ii/\kappa}) \; 
\textnormal{ and }
\; \cst(\kappa) \, (1-e^{-8\pi\ii/\kappa}) \; 
\textnormal{ are bounded in }
\; \kappa\in \big( \tfrac{8}{2m+1},\tfrac{8}{2m-1} \big) \setminus \big\{\tfrac{4}{m}\big\} ,
\end{align*}
we conclude that 
\begin{align*}
\cst(\kappa)^N 
\; \ointclockwise_{\acycle_2^{\beta}}\ud u_2\cdots  \ointclockwise_{\acycle^{\beta}_{N}}\ud u_N  
\; f_{\beta/\{1,2\}}(\bs{\ddot{x}}_1;\bs{\dot{u}}_1) 
\; \ointclockwise_{\acycle(0,1)} \ud u  \; (u-1)^{-4/\kappa}\, u^{-4/\kappa} 
\; \hat{g}^{(\kappa)}(x_1,x_2,u; \bs{\ddot{x}}_1;\bs{\dot{u}}_1) 
\end{align*}
is uniformly bounded in $\kappa\in \big( \tfrac{8}{2m+1},\tfrac{8}{2m-1} \big)\setminus \big\{\tfrac{4}{m}\big\}$ and $x_1,x_2\in (\xi-(x_3-\xi)/3, \xi+(x_3-\xi)/3)$. Finally, taking the limit $\kappa\to \tfrac{4}{m}$ in~\eqref{eqn::Fro_pairing_aux1} and using the continuity from Theorem~\ref{thm::CGI}, 
we get~\eqref{eqn::F_Fro_good_1} for $\kappa = \tfrac{4}{m}$.
\end{proof}

\begin{lemma}\label{lem::coeff_frob_bad_axu2}
Fix $\kappa\in (8/3,8)\setminus \{4\}$ and $N \geq 2$. 
Fix $\beta \in \LP_N$ and suppose $\{1,2\}\not\in \beta$. 
Then, for all $\xi\in (-\infty, x_3)$, we have
\begin{align} \label{eqn::CGI_Frobenius_tying_again}
\begin{split}
\coulombGas_{\beta}(\bs{x}) 
= \; & \coulombGas_{\wp_1(\beta)/\{1,2\}}(\bs{\ddot{x}}_1)
\, (x_{2}-x_1)^{-2h(\kappa)} 
\; + \; \coulombGas_{\beta/\vee_1}(\xi, \bs{\ddot{x}}_1) 
\,  (x_{2}-x_1)^{2/\kappa} \\
\; & + \; o\big((x_{2}-x_1)^{2/\kappa}\big) , \qquad \textnormal{as } x_1, x_{2}\to \xi ,
\end{split}
\end{align}
where \textnormal{(}using the branch choice~\eqref{eqn::integrand_CGI_fused_generic}\textnormal{)} 
\begin{align*}
\coulombGas_{\beta/\vee_1}(\xi, \bs{\ddot{x}}_1) 
=  \; &  
\frac{\cst(\kappa)^{N-2} \,(\fugacity(\kappa)^2-1)}{\hF\big(\tfrac{4}{\kappa}, 1-\tfrac{4}{\kappa}, \tfrac{8}{\kappa}; 1\big)} 
\; \ointclockwise_{\acycle(x_{a_3},x_{b_3})}\ud u_3 \cdots \ointclockwise_{\acycle(x_{a_N},x_{b_N})} \ud u_N 
\; \hat{f}_{\beta/\vee_1} (\xi, \bs{\ddot{x}}_1; \bs{\ddot{u}}_1) , 
\\
\hat{f}_{\beta/\vee_1} (\xi, \bs{\ddot{x}}_1; \bs{\ddot{u}}_1) 
= \; & \prod_{r=3}^N (u_r-\xi)^{16/\kappa-2} 
\prod_{j=3}^{2N} (x_j-\xi)^{1-8/\kappa} 
\; f (\bs{\ddot{x}}_1; \bs{\ddot{u}}_1) .
\end{align*}
\end{lemma}

Note that Eq.~\eqref{eqn::CGI_Frobenius_tying_again} is exactly the 
second claim~\eqref{eqn::CGI_Frobenius_tying} in Proposition~\ref{prop::Frobenius_F} 
for $\kappa\in (8/3,8)\setminus \{4\}$. 

\begin{proof}
Using Proposition~\ref{prop::remove_integration}, 
we obtain\footnote{The way we remove the integral surrounding $x_1$ and $x_{b_1}$ here is different from that in the proof of Proposition~\ref{prop::H_ASY}: 
in the latter, we use the integrand $\smash{\hat{f}_\beta^b}$ with conjugate charge at $x_{b_1}$, while here, we use $\smash{\hat{f}_\beta^a}$ with conjugate charge at $x_1$.}
\begin{align*}
\coulombGasH_{\beta}(\bs{x})
= \; & \frac{\fugacity(\kappa) }{\cst(\kappa)} \,
\ointclockwise_{\acycle^\beta_2} \ud u_2
\cdots \ointclockwise_{\acycle^\beta_{N}} \ud u_{N} 
\; \hat{f}_\beta^a (\bs{x};\bs{\dot{u}}_1) \\
= \; & \frac{\fugacity(\kappa) }{\cst(\kappa)}  \,
(x_{2}-x_1)^{-2h(\kappa)} \, 
(x_{b_1}-x_1)^{-2h(\kappa)} \, 
(x_{b_1}-x_{2})^{2/\kappa} \, 
(x_{b_2}-x_{1})^{-2h(\kappa)} \, 
(x_{b_2}-x_{2})^{2/\kappa} \, 
\\
\; & \times \prod_{\substack{3\leq i\leq 2N \\ i \notin \{b_1,b_2\}}} \big( (x_i-x_1)^{-2h(\kappa)} \, (x_i-x_{2})^{2/\kappa}\big)
\\ 
\; & \times 
\int_{\Gamma_{\hat{\beta}}} 
\ud \bs{\ddot{u}}_1 \; f(\bs{\ddot{x}}_1;\bs{\ddot{u}}_1) 
\; \ointclockwise_{\acycle(x_{2},x_{b_2})}\ud u_2 
(u_2-x_{b_1})^{-4/\kappa} \, (u_2-x_{b_2})^{-4/\kappa} \, (u_2-x_1)^{4h(\kappa)} \, (u_2-x_{2})^{-4/\kappa} 
\\
\; & \times 
\prod_{\substack{3\leq i\leq 2N \\ i \notin \{b_1,b_2\}}} (u_2-x_i)^{-4/\kappa} 
 \prod_{3\leq s\leq N }(u_s-u_2)^{8/\kappa}
 \prod_{3\leq s\leq N }\big( (u_s-x_1)^{4h(\kappa)} \, (u_s-x_{2})^{-4/\kappa}\big) ,
\end{align*}
where the integration $\smash{\Gamma_{\hat{\beta}}}$ comprises all of the other contours 
than the two Pochhammer contours surrounding the pairs $\{x_{1},x_{b_1}\}$ and $\{x_{2},x_{b_2}\}$, 
and the integrand $\smash{\hat{f}_{\beta}^a}$ is defined in~\eqref{eq: integrand with one less screening variable a} with $r=1$ and the conjugate charge at $x_{1}$, 
and $f$ is the integrand function in~\eqref{eq: integrand_gen}, 
with the branch of $f(\bs{\ddot{x}}_1;\cdot)$ and the rest of the products in the integrand chosen to be real and positive when $x_{a_s}<\Re (u_s)<x_{a_s}+1$ for all $s\in \{3,4,\ldots, N\}$. 
Using the change of variables $w=\CR(x_{b_2},x_{b_1},x_2,u_2)$, 
where $\CR$ is the cross-ratio~\eqref{eqn::crossratio}, 
we obtain
\begin{align} \label{eqn::H_ASY_change_of_variables2}
\begin{split}
\coulombGasH_{\beta}(\bs{x})
= \; & \frac{\fugacity(\kappa) }{\cst(\kappa)} \; 
(x_{b_1}-x_{b_2})^{1-8/\kappa} \, 
\bigg( \frac{(x_2-x_1) (x_{b_1}-x_2) (x_{b_2}-x_{1})}{(x_{b_1}-x_1) (x_{b_2}-x_2)} \bigg)^{-2h(\kappa)} 
\\
\; & \times  
\int_{\Gamma_{\hat{\beta}}}\ud \bs{\ddot{u}}_1 
\; f(\bs{\ddot{x}}_1;\bs{\ddot{u}}_1)
\; \ointclockwise_{\acycle(0,1)}\ud w 
\; \tilde{g}(x_1,x_2,w; \bs{\ddot{x}}_1;\bs{\ddot{u}}_1) 
\, \hat{k}(\chi,w) ,
\end{split}
\end{align}
writing $\chi=\CR(x_1,x_2,x_{b_2},x_{b_1})$, and 
$\hat{k}(\chi,w) := \big( w+\tfrac{\chi}{1-\chi} \big)^{4h(\kappa)} (1-w)^{-4/\kappa} \, w^{-4/\kappa}$, 
and
\begin{align} 
\label{eqn::def_ghat}
\tilde{g}(x_1,x_2,w; \bs{\ddot{x}}_1;\bs{\ddot{u}}_1) 
:= \; & \prod_{\substack{3\leq i\leq 2N \\ i \notin \{b_1,b_2\}}}
\Big| x_{b_1}-(x_{b_1}-x_2) 
\Big( 1 + w\frac{x_{b_2}-x_2}{x_{b_1}-x_{b_2}} \Big)^{-1}-x_i \Big|^{-4/\kappa}
\\
\nonumber
& \; \times \prod_{3\leq s\leq N}\Big(u_s-x_{b_1}+(x_{b_1}-x_2) 
\Big( 1 + w\frac{x_{b_2}-x_2}{x_{b_1}-x_{b_2}}\Big)^{-1}\Big)^{8/\kappa}
\\ 
\nonumber
& \;  \times \prod_{3\leq s\leq N } \big( (u_s-x_1)^{4h(\kappa)} \, (u_s-x_{2})^{-4/\kappa}
\big) 
\prod_{\substack{3\leq i\leq 2N \\ i\notin \{b_1,b_2\}}} \big( (x_i-x_1)^{-2h(\kappa)} \, (x_i-x_{2})^{2/\kappa} \big) ,
\end{align}
whose branch is chosen to be real and positive when $0<\Re (w)<1$ and $x_{a_s}<\Re (u_s)<x_{a_s+1}$ for all $s\in \{3,4,\ldots,N\}$. 
Note that $\tilde{g}$ stays bounded in the limit $x_1, x_{2}\to \xi$, and $\chi \to 0$ in this limit.

Let us prove~\eqref{eqn::CGI_Frobenius_tying_again} for $\kappa\in (4,8)$.  
By~\eqref{eqn::H_ASY_change_of_variables2}, it suffices to consider the limit of the integral over $\acycle(0,1)$ as $x_1, x_{2}\to \xi$.
Similarly as in Lemma~\ref{lem::Poch_line_relation}, we see that
\begin{align*}
\; &  \ointclockwise_{\acycle(0,1)}\ud w 
\; \tilde{g}(x_1,x_2,w; \bs{\ddot{x}}_1;\bs{\ddot{u}}_1) 
\, \hat{k}(\chi,w) \\
= \; & 
(1 - e^{8 \pi\ii/\kappa }) (1 - e^{-8 \pi\ii/\kappa})
\underbrace{\int_{\epsilon \tfrac{\chi}{1-\chi}}^{1-\epsilon} \ud w
\; \tilde{g}(x_1,x_2,w; \bs{\ddot{x}}_1;\bs{\ddot{u}}_1) 
\, \hat{k}(\chi,w)}_{=:\; I_1}
\\
\; & 
\; + \; 
(1 - e^{8 \pi\ii /\kappa}) 
\underbrace{\ointclockwise_{\partial B(0,\epsilon \tfrac{\chi}{1-\chi})}\ud w
\; \tilde{g}(x_1,x_2,w; \bs{\ddot{x}}_1;\bs{\ddot{u}}_1) 
\, \hat{k}(\chi,w)}_{=:\; I_2}
\\
\; & 
\; + \; 
 (1 - e^{-8 \pi\ii /\kappa})
\underbrace{\ointclockwise_{\partial B(1,\epsilon)}\ud w 
\; \tilde{g}(x_1,x_2,w; \bs{\ddot{x}}_1;\bs{\ddot{u}}_1) 
\, \hat{k}(\chi,w)}_{=:\; I_3} ,
\end{align*}
when $\epsilon > 0$ is small and fixed.
We derive the limit as $x_1, x_{2}\to \xi$ by analyzing the terms $I_1$, $I_2$, and $I_3$ one by one. 
The second term vanishes in the limit:
since $|\hat{k}(\chi,w)| \sim \chi^{8/\kappa-2}$ for all $w \in \partial B(0,\epsilon \tfrac{\chi}{1-\chi})$, we have
\begin{align*}
I_2 = \; & \ointclockwise_{\partial B(0,\epsilon \tfrac{\chi}{1-\chi})}\ud w
\; \tilde{g}(x_1,x_2,w; \bs{\ddot{x}}_1;\bs{\ddot{u}}_1) 
\, \hat{k}(\chi,w)
\quad \overset{x_1, x_{2}\to \xi}{\longrightarrow} \quad  0 .
\end{align*}
The third term has the limit
\begin{align*}
I_3 = \; & \ointclockwise_{\partial B(1,\epsilon)}\ud w 
\; \tilde{g}(x_1,x_2,w; \bs{\ddot{x}}_1;\bs{\ddot{u}}_1) 
\, \hat{k}(\chi,w) 
\quad \overset{x_1, x_{2}\to \xi}{\longrightarrow} \quad \; 
 \ointclockwise_{\partial B(1,\epsilon)}\ud w 
\; \tilde{g}(\xi,\xi,w; \bs{\ddot{x}}_1;\bs{\ddot{u}}_1) 
\, w^{8/\kappa-2} (1-w)^{-4/\kappa} .
\end{align*}
The first term has the limit
\begin{align*}
I_1 = \; & \int_{\epsilon \tfrac{\chi}{1-\chi}}^{1-\epsilon} \ud w
\; \tilde{g}(x_1,x_2,w; \bs{\ddot{x}}_1;\bs{\ddot{u}}_1) 
\, \hat{k}(\chi,w) 
\quad \overset{x_1, x_{2}\to \xi}{\longrightarrow} \quad \; 
\int_{0}^{1-\epsilon} \ud w
\; \tilde{g}(\xi,\xi,w; \bs{\ddot{x}}_1;\bs{\ddot{u}}_1) 
\, w^{8/\kappa-2} (1-w)^{-4/\kappa} .
\end{align*}
Collecting these facts, we conclude that
\begin{align*}
\frac{1}{\fugacity(\kappa)} \, \coulombGasH_{\wp_1(\beta)/\{1,2\}}(\bs{\ddot{x}}_1)
= \; & \frac{\cst(\kappa)}{\fugacity(\kappa)}
\lim_{x_1,x_2\to \xi }\frac{\coulombGasH_{\beta}(\bs{x})}{(x_2-x_1)^{-2h(\kappa)}} 
 \\
=\; & (1-e^{-8\pi\ii/\kappa}) (x_{b_1}-x_{b_2})^{1-8/\kappa}\int_{\Gamma_{\hat{\beta}}} \ud \bs{\ddot{u}}_1 f(\bs{\ddot{x}}_1;\bs{\ddot{u}}_1) \\
& \; \times \bigg( 
(1 - e^{8 \pi\ii/\kappa }) \, \int_{0}^{1-\epsilon} \ud w
\; + \; \ointclockwise_{\partial B(1,\epsilon)}\ud w 
\bigg)
\; \tilde{g}(\xi,\xi,w; \bs{\ddot{x}}_1;\bs{\ddot{u}}_1) 
\, w^{8/\kappa-2} (1-w)^{-4/\kappa} ,
\end{align*}
using also Proposition~\ref{prop::H_ASY} to identify the left-hand side.
Note that $-2(\kappa)+1=2-6/\kappa>2/\kappa$ when $\kappa\in (4,8)$. 
Consequently, 
\begin{align}
\nonumber
&\frac{\coulombGasH_{\beta}(\bs{x})
- \;  (x_2-x_1)^{-2h(\kappa)} \, \frac{1}{\cst(\kappa)}
\, \coulombGasH_{\wp_1(\beta)/\{1,2\}}(\bs{\ddot{x}}_1)}{\frac{\fugacity(\kappa)}{\cst(\kappa)}(x_2-x_1)^{-2h(\kappa)}}\\
\nonumber
=
\; & \bigg(\frac{(x_{b_1}-\xi)(x_{b_2}-\xi)}{x_{b_2}-x_{b_1}}\bigg)^{8/\kappa-1} \, \int_{\Gamma_{\hat{\beta}}} \ud \bs{\ddot{u}}_1 \hat{f}_{\beta/\vee_1}(\xi,\bs{\ddot{x}}_1;\bs{\ddot{u}}_1)\, \Big(\underbrace{(1-e^{8\pi\ii/\kappa})\ointclockwise_{\partial B(0,\epsilon\frac{\chi}{1-\chi})} \ud w\, \big(w+\tfrac{\chi}{1-\chi}\big)^{4h(\kappa)}w^{-4/\kappa}}_{=:\;I_4}
\\
\nonumber
\; &\qquad\qquad +\underbrace{(1-e^{-8\pi\ii/\kappa})\int_{\epsilon\frac{\chi}{1-\chi}}^{10\frac{\chi}{1-\chi}} \ud w  \big(w+\tfrac{\chi}{1-\chi}\big)^{4h(\kappa)} w^{-4/\kappa}-(1-e^{-8\pi\ii})\int_{0}^{10\frac{\chi}{1-\chi}}\ud w \, w^{8/\kappa-2}}_{=:\;I_5}
\\
\nonumber
\; &\qquad\qquad+\underbrace{4\sin^2(4\pi/\kappa)\int_{10\frac{\chi}{1-\chi}}^{1-\epsilon} \ud w \left(\big(w+\tfrac{\chi}{1-\chi}\big)^{4h(\kappa)}w^{-4/\kappa}-w^{8/\kappa-2}\right)}_{=:\; I_6}\Big)\\
\; &+O(x_2-x_1) 
, \qquad \textnormal{as } x_1, x_{2}\to \xi ,
\label{eqn::H_expansion}
\end{align}
where $O(x_2-x_1)=o\big((x_2-x_1)^{8/\kappa-1}\big)$ since $1>8\kappa-1$ when $\kappa\in (4,8)$.
For the terms $I_4$ and $I_5$, using the change of variables $u=w\frac{1-\chi}{\chi}$, we see that
\begin{align*}
I_4+I_5=\chi^{8/\kappa-1} C_1(\kappa),
\end{align*}
where 
\[C_1(\kappa):=(1-e^{8\pi\ii/\kappa})\ointclockwise_{\partial B(0,\epsilon)} \ud u(u+1)^{4h(\kappa)} u^{-4/\kappa}+(1-e^{-8\pi\ii}) \big(\int_{\epsilon}^{10}(u+1)^{4h(\kappa)}u^{-4/\kappa}-\int_0^{10}\ud u \, u^{8/\kappa-2}\big)\]
is a finite constant depending only on $\kappa$ (and not on $N$). 
For the term $I_6$, we Taylor expand
\begin{align*}
\Big( w + \frac{\chi}{1-\chi} \Big)^{12/\kappa-2} - w^{12/\kappa-2}  
=  w^{12/\kappa-2} \; \sum_{n=1}^{\infty}\frac{(\tfrac{12}{\kappa}-2)(\tfrac{12}{\kappa}-3)\cdots(\tfrac{12}{\kappa}-n-1)}{n!}
\Big(\frac{\chi}{w(1-\chi)}\Big)^{n} \; 
\end{align*}
for $w\in \big(\frac{10\chi}{1-\chi},\frac{9}{10} \big)$. 
Since $1>8/\kappa-1$ when $\kappa\in (4,8)$, we see that
\begin{align*}
I_6
= \Big(\frac{\chi}{1-\chi}\Big)^{8/\kappa-1} \, C_2(\kappa) 
\; + \; o\big((x_2-x_1)^{8/\kappa-1}\big) ,
\end{align*}
where 
\begin{align*}
C_2(\kappa):=-4\sin^2(4\pi/\kappa)\sum_{n=1}^{\infty}\frac{(\tfrac{12}{\kappa}-2)(\tfrac{12}{\kappa}-3)\cdots(\tfrac{12}{\kappa}-n-1)}{n! \, (\tfrac{8}{\kappa}-n-1)}\; 10^{8/\kappa-1-n}
\end{align*}
is a finite constant depending only on $\kappa$ (and not on $N$). 

Collecting the above expansions, 
on the one hand, plugging them into~\eqref{eqn::H_expansion} gives 
\begin{align} \label{eqn::coeff_frob_bad_axu1}
\begin{split}
\coulombGasH_{\beta}(\bs{x})
= \; & (x_2-x_1)^{-2h(\kappa)}  \, \frac{1}{\cst(\kappa)} \, \coulombGasH_{\wp_1(\beta)/\{1,2\}}(\bs{\ddot{x}}_1)
\; + \; 
o\big( (x_2-x_1)^{2/\kappa}\big) 
\\
\; & + (x_2-x_1)^{2/\kappa} \, \big(C_1(\kappa)+C_2(\kappa)\big) \, 
\frac{\nu(\kappa)}{\cst(\kappa)}
\, \int_{\Gamma_{\hat{\beta}}}\ud \bs{\ddot{u}}_1 
\;  \hat{f}_{\beta/\vee_1}(\xi, \bs{\ddot{x}}_1; \bs{\ddot{u}}_1) ,
\end{split}
\end{align}
On the other hand, when $N=2$, using the relation~\eqref{eqn::F_Fro_bad_fourpoint} from Corollary~\ref{cor::F_Fro_bad_fourpoint}, we have, as $x_1,x_2\to \xi$, 
\begin{align} \label{eqn::coeff_frob_bad_axu2}
\coulombGasH_{_{\includegraphics[scale=0.2]{figures-arXiv/link-2}}}(x_1,x_2,x_3,x_4)
= \; & (x_2-x_1)^{-2h(\kappa)} \, \frac{1}{\cst(\kappa)} \,  \coulombGasH_{_{\includegraphics[scale=0.2]{figures-arXiv/link-0}}}(x_3,x_4)+\; o\big((x_{2}-x_1)^{2/\kappa}\big) 
\notag
\\
\; & + (x_2-x_1)^{2/\kappa} \, \frac{1}{\cst(\kappa)^2} \, \frac{\fugacity(\kappa)^2-1}{\hF\big(\tfrac{4}{\kappa}, 1-\tfrac{4}{\kappa}, \tfrac{8}{\kappa}; 1\big)} 
\; \bigg( \frac{(x_4-x_3)^{2/\kappa}}{(x_4-\xi)^{8/\kappa-1}(x_3-\xi)^{8/\kappa-1}} \bigg) .
\end{align} 
Comparing~\eqref{eqn::coeff_frob_bad_axu1} with~\eqref{eqn::coeff_frob_bad_axu2} then shows that
\begin{align*}
\left(C_1(\kappa)+C_2(\kappa)\right) \, \frac{\nu(\kappa)}{\cst(\kappa)}
= \frac{1}{\cst(\kappa)^2} \, \frac{\fugacity(\kappa)^2-1}{\hF\big(\tfrac{4}{\kappa}, 1-\tfrac{4}{\kappa}, \tfrac{8}{\kappa}; 1\big)} ,
\end{align*}
and finally, plugging this into~\eqref{eqn::coeff_frob_bad_axu1} and using the relation~\eqref{eq::F_as_C_times_H} 
gives~\eqref{eqn::CGI_Frobenius_tying_again} for $\kappa\in (4,8)$.

Let us then discuss the case $\kappa\in (8/3,4)$, where the analysis is essentially the same, 
with one important difference: in this case, the term of order $(x_2-x_1)^{-2h(\kappa)+1}$ 
in the expansion of $\coulombGasH_{\beta}(\bs{x})$ vanishes.  
Indeed, it follows from~\eqref{eqn::H_ASY_change_of_variables} and the analysis after it in the proof of 
Proposition~\ref{prop::H_ASY} that
\begin{align}
\begin{split}\label{eqn::first_order_vanish}
\coulombGasH_{\beta}^{(\kappa)}(\bs{x})
= \; & \frac{1}{\cst(\kappa)} \, \coulombGasH_{\wp_1(\beta)/\{1,2\}}^{(\kappa)}(\bs{\ddot{x}}_1)  \, (x_{2}-x_1)^{-2h(\kappa)}
\; + \fugacity(\kappa)\; o\big((x_2-x_1)^{-2h(\kappa)+1}\big)
\\
\; & + (x_2-x_1)^{-2h(\kappa)+1} \, \frac{\fugacity(\kappa)}{\cst(\kappa)} \, (1-e^{8\pi\ii/\kappa})
\; \int_{\Gamma_{\hat{\beta}}}\tilde{f}(\bs{\ddot{x}}_1;\bs{\ddot{u}}_1)\ud \bs{\ddot{u}}_1
\\
\; & \times \ointclockwise_{\intloop(0,\infty)} \ud v 
\; v^{-4/\kappa} (1+v)^{-4/\kappa}  \Big(\tfrac{6}{\kappa}-1+\tfrac{4}{\kappa}(\tfrac{v}{1+v}+v) 
+ (2v+1) \Big(\sum_{2\leq s\leq N}\frac{4/\kappa}{\xi-u_s}+\sum_{\substack{3\leq i\leq 2N \\ i\notin \{b_1,b_2\}}} \frac{2/\kappa}{\xi-x_i}\Big)\Big) \\
= \; & \frac{1}{\cst(\kappa)} \, \coulombGasH_{\wp_1(\beta)/\{1,2\}}^{(\kappa)}(\bs{\ddot{x}}_1)  \, (x_{2}-x_1)^{-2h(\kappa)}
\; + \fugacity(\kappa)\; o\big((x_2-x_1)^{-2h(\kappa)+1}\big) 
, \qquad \textnormal{as } x_1, x_{2}\to \xi ,
\end{split}
\end{align}
thanks to the identities~\eqref{eqn::beta_vanish_2} and~\eqref{eqn::beta_vanish_3} from Lemma~\ref{lem::beta_rholoop}.
\end{proof}

\subsection{$\kappa=8/3$: proof of Proposition~\ref{prop::CGI8over3_Frobenius}}
\label{subsec: Proof of Frob 83}

In this section, we fix $\kappa=8/3$ and use the notation from Proposition~\ref{prop::CGI8over3_Frobenius}. 

\begin{proof}[Proof of Proposition~\ref{prop::CGI8over3_Frobenius}]
When $\{1,2\}\in \beta$, the asserted expansion~\eqref{eqn::CGIRe_Frobenius_paired} is proven in Lemma~\ref{lem::CGI8over3_Frobenius_aux1}. When $\{1,2\}\notin \beta$, the asserted expansion~\eqref{eqn::CGIRe_Frobenius_tying} is proven in Lemma~\ref{lem::CGI8over3_Frobenius_aux2}. 
\end{proof}

\begin{lemma}  \label{lem::CGI8over3_Frobenius_aux1}
Fix $\beta\in \LP_N$ and suppose $\{1,2\}\in \beta$. 
Then, for all $\xi\in(-\infty,x_3)$, we have
\begin{align} \label{eqn::CGIRe_Frobenius_paired_again}
\coulombGasRenorm_{\beta}^{(8/3)}(\bs{x})
= \; & \coulombGasRenorm_{\beta/\vee_1}^{(8/3)}(\xi,\bs{\ddot{x}}_1) \, (x_2-x_1)^{3/4}
\; + \; o\big((x_2-x_1)^{3/4}\big),\qquad \textnormal{as }x_1, x_{2}\to \xi ,
\end{align}
where \textnormal{(}using the branch choice~\eqref{eqn::integrand_CGI_fused_generic}\textnormal{)} 
\begin{align*}
\coulombGasRenorm_{\beta/\vee_1}^{(8/3)}(\xi, \bs{\ddot{x}}_1)
= \; & \frac{\cst(8/3)^{N-2}}{ \hF\big(\tfrac{3}{2},-\tfrac{1}{2},3;1\big)}
\; \ointclockwise_{\acycle(x_{a_3},x_{b_3})}\ud u_3 \cdots \ointclockwise_{\acycle(x_{a_N},x_{b_N})} \ud u_N 
\; \hat{f}_{\beta/\vee_1}^{(8/3)} (\xi, \bs{\ddot{x}}_1; \bs{\ddot{u}}_1) , \\
\hat{f}_{\beta/\vee_1}^{(8/3)} (\xi, \bs{\ddot{x}}_1; \bs{\ddot{u}}_1)
= \; & \prod_{r=3}^N (u_r-\xi)^{4} 
\prod_{j=3}^{2N} (x_j-\xi)^{-2} 
\; f (\bs{\ddot{x}}_1; \bs{\ddot{u}}_1).
\end{align*}
\end{lemma}

Note that Eq.~\eqref{eqn::CGIRe_Frobenius_paired_again} is exactly the 
first claim~\eqref{eqn::CGIRe_Frobenius_paired} in Proposition~\ref{prop::CGI8over3_Frobenius}.

\begin{proof}
According to Proposition~\ref{prop::remove_integration}, we have 
\begin{align} \label{eqn::CGI8over3_Frobenius_aux11}
\coulombGasHRenorm_{\beta}^{(8/3)}(\bs{x})
:= \; & \lim_{\kappa'\to 8/3} \frac{1}{\fugacity(\kappa')} \, \coulombGasH_{\beta}^{(\kappa')} (\bs{x})
\notag 
\\
= \; & \frac{1}{\cst(8/3)} \, \ointclockwise_{\acycle^\beta_2} \ud u_2
\cdots \ointclockwise_{\acycle^\beta_{N}} \ud u_{N} 
\; \hat{f}_\beta^a (\bs{x};\bs{\dot{u}}_1)
\notag 
\\
= \; & \frac{(x_2-x_1)^{-5/4}}{\cst(8/3)} \, \int_{\Gamma_{\hat{\beta}}}\ud \bs{\ddot{u}}_1 \; f^{(8/3)}(\bs{\ddot{x}}_1;\bs{\ddot{u}}_1) 
\ointclockwise_{\acycle(x_{3},x_{b_2})} \ud u_2 
\; g_1(x_1,x_2,u_2;\bs{\ddot{x}}_1,\bs{\ddot{u}}_1),
\end{align}
where the integration $\smash{\Gamma_{\hat{\beta}}}$ comprises all of the other contours 
than the two Pochhammer contours surrounding the pairs $\{x_{1},x_{2}\}$ and $\{x_{3},x_{b_2}\}$, 
and the integrand $\smash{\hat{f}_{\beta}^a}$ is defined in~\eqref{eq: integrand with one less screening variable a} 
with $r=1$ and the conjugate charge at $x_{1}$, 
and $f$ is the integrand function in~\eqref{eq: integrand_gen} with $\kappa=8/3$, 
and 
\begin{align*}
g_1(x_1,x_2,u_2;\bs{\ddot{x}}_1,\bs{\ddot{u}}_1)
:= \; & (u_2-x_1)^{5/2}(u_2-x_2)^{-3/2}\prod_{3\leq s\leq N }(u_s-u_2)^{3} \prod_{3\leq i\leq 2N}(u_2-x_i)^{-3/2}\\
&\times \prod_{3\leq s\leq N} \big((u_s-x_1)^{5/2}(u_s-x_2)^{-3/2}\big) \prod_{3\leq i\leq 2N} \big((x_i-x_1)^{5/2}(x_i-x_2)^{-3/2}\big),
\end{align*}
with the branch of $f(\bs{\ddot{x}}_1;\cdot)$ and the rest of the products in the integrand chosen to be real and positive when $x_{a_s}<\Re (u_s)<x_{a_s}+1$ for all $s\in \{2,\ldots, N\}$. 
Note that the function $u_2 \mapsto g_1(x_1,x_2,u_2;\bs{\ddot{x}}_1,\bs{\ddot{u}}_1)$ is of order $|u_2|^{-2}$ as $|u_2|\to \infty$, thus holomorphic at infinity. 
Using Corollary~\ref{cor::pochhammer_to_loop_8overodd} (for $\kappa=8/3$) 
to transform all of the Pochhammer contours $\acycle(x_{a_r}, x_{b_r})$ in the integrals in~\eqref{eqn::CGI8over3_Frobenius_aux11} 
into simple clockwise loops $\intloop(x_{a_r}, x_{b_r})$ 
(see also Figure~\ref{fig::loop_deformation_loop}),
and then using Lemma~\ref{lem::circle_integral_general} 
to pull the loop $\intloop_2 = \intloop(x_{3},x_{b_2})$ surrounding the points $x_{3}$ and $x_{b_2}$ (and originally, no other marked points) across the other loops so that it becomes a large loop $\Gloop_2$
surrounding all of the points $\{x_{3}, x_4, \ldots, x_{2N}\}$
(see also Figure~\ref{fig::loop_deformation_pulling}),
and then pulling that loop across the point at infinity to become a loop $\intloop(x_1,x_2)$ only surrounding the points $x_1$ and $x_2$, and lastly, transforming all of the loops back to Pochhammer contours
(so that in particular $\intloop(x_1,x_2)$ becomes $-\acycle(x_1,x_2)$),
we see that 
\begin{align} \label{eqn::CGI8over3_Frobenius_aux12}
\coulombGasHRenorm_{\beta}^{(8/3)}(\bs{x}) 
= \frac{(x_2-x_1)^{-5/4}}{\cst(8/3)} \int_{\Gamma_{\hat{\beta}}}\ud \bs{\ddot{u}}_1 
\; f^{(8/3)}(\bs{\ddot{x}}_1;\bs{\ddot{u}}_1)  
\ointclockwise_{\acycle(x_1,x_2)}\ud u_2 \;g_1(x_1,x_2,u_2;\bs{\ddot{x}}_1,\bs{\ddot{u}}_1),
\end{align}
where the branch of the product in the integrand is chosen to be real and positive when $x_{1}<\Re u_2<x_{2}$ and $x_{a_s}<u_s<x_{a_s+1}$ for all $s\in \{3,4,\ldots,N\}$
(so that the change in the branch choice cancels the minus sign from $-\acycle(x_1,x_2)$). 
To finish, thanks to~\eqref{eqn::CGI8over3_Frobenius_aux12}, we can use the change of variables $u=\frac{u_2-x_1}{x_2-x_1}$ to obtain
\begin{align*}
\coulombGasHRenorm_{\beta}^{(8/3)}(\bs{x}) 
= \; & \frac{(x_2-x_1)^{3/4}}{\cst(8/3)} \, \int_{\Gamma_{\hat{\beta}}} \ud \bs{\ddot{u}}_1 
\; \hat{f}_{\beta/\vee_1}^{(8/3)} (\xi,\bs{\ddot{x}}_1; \bs{\ddot{u}}_1) 
\ointclockwise_{\acycle (0,1)} \ud u \;  \big(u^{5/2}(u-1)^{-3/2}+o(1)\big),\qquad \textnormal{as }x_1,x_2\to \xi.
\end{align*}
Using~\eqref{eqn::beta_3} from Lemma~\ref{lem::beta_acycle} to evaluate the integral
\begin{align*}
\ointclockwise_{\acycle(0,1)} u^{5/2} (u-1)^{-3/2}\ud u = 4 \frac{\Gamma(7/2)\Gamma(-1/2)}{\Gamma(3)}=-\frac{15}{2}\pi,
\end{align*}
we conclude that 
\begin{align*}
\lim_{x_1,x_2\to \xi }\frac{\coulombGasRenorm_{\beta}^{(8/3)}(\bs{x})}{(x_2-x_1)^{3/4}}
= \; & \cst(8/3)^{N} \, \lim_{x_1,x_2\to \xi} \frac{\coulombGasHRenorm_{\beta}^{(8/3)}(\bs{x})}{(x_2-x_1)^{3/4}} \\
= \; & \cst(8/3)^{N-2} \, \frac{15\pi}{32} \, \int_{\Gamma_{\hat{\beta}}}\ud \bs{\ddot{u}}_1 \hat{f}_{\beta/\vee_1}^{(8/3)} (\xi,\bs{\ddot{x}}_1; \bs{\ddot{u}}_1)
\; = \; \coulombGasRenorm_{\beta/\vee_1}^{(8/3)}(\xi,\bs{\ddot{x}}_1),
\end{align*}
which proves the asserted expansion~\eqref{eqn::CGIRe_Frobenius_paired_again}.
\end{proof}

\begin{remark} \label{rem::Frobenius_coeffi_equiv}
Note that in the proof of Lemma~\ref{lem::CGI8over3_Frobenius_aux1}, one can replace 
$\{x_3,x_{b_2}\}=\{x_{a_2},x_{b_2}\}$ by any $\{x_{a_s},x_{b_s}\}$ for $s\in \{2,3,\ldots,N\}$. 
The only difference is that one will get an extra $(-1)^{(a_s-1)/2}$ phase factor on the right-hand side of~\eqref{eqn::CGI8over3_Frobenius_aux12} for a general $s\in \{2,3,\dots,N\}$:  
\begin{align*}
\coulombGasRenorm_{\beta/\vee_1}^{(8/3)}(\xi, \bs{\ddot{x}}_1)
= \; & (-1)^{(a_s-1)/2} \, \frac{\cst(8/3)^{N-2}}{ \hF\big(\tfrac{3}{2},-\tfrac{1}{2},3;1\big)}
\; \ointclockwise_{\acycle(x_{a_2},x_{b_2})}\ud u_2\cdots \ointclockwise_{\acycle(x_{a_{s-1}},x_{b_{s-1}})}\ud u_{s-1}\\
\;&\times\ointclockwise_{\acycle(x_{a_{s+1}},x_{b_{s+1}})} \ud u_{s+1} \cdots \ointclockwise_{\acycle(x_{a_N},x_{b_N})} \ud u_N 
\prod_{2\leq r\leq N\atop r\neq s} (u_r-\xi)^{4} 
\prod_{j=3}^{2N} (x_j-\xi)^{-2} 
\; f (\bs{\ddot{x}}_1; \bs{\ddot{u}}_{1,s}),
\end{align*}
where $\bs{\ddot{u}}_{1,s}:=(u_2,\ldots,u_{s-1},u_{s+1},\ldots,u_N)$ and the branch of the integrand is chosen to be real and positive when $x_{a_r}<\Re u_r<x_{a_r+1}$ for $r\in \{2,3,\ldots,N\}\setminus \{s\}$. 
\end{remark}

\begin{lemma}\label{lem::CGI8over3_Frobenius_aux2}
Fix $\beta\in \LP_N$ and suppose $\{1,2\}\notin \beta$. 
Then, for all $\xi\in(-\infty,x_3)$, we have 
\begin{align}\label{eqn::CGIRe_Frobenius_tying_again}
\begin{split}
\coulombGasRenorm_{\beta}^{(8/3)}(\bs{x})
= & \;
\coulombGasRenorm_{\wp_1(\beta)/\{1,2\}}^{(8/3)}(\bs{\ddot{x}}_1) \, (x_2-x_1)^{-5/4}
\; - \; \frac{1}{\pi} (-1)^{(b_2-1)/2} \, \coulombGasRenorm_{\wp_1(\beta)/\vee_1}^{(8/3)}(\xi, \bs{\ddot{x}}_1) \, (x_2-x_1)^{3/4} \, |\log(x_2-x_1)| \\
& \; + \; O\big((x_2-x_1)^{3/4}\big) , \qquad \textnormal{as } x_1, x_2 \to \xi .
\end{split}
\end{align}
\end{lemma}

Note that Eq.~\eqref{eqn::CGIRe_Frobenius_tying_again} is exactly the 
second claim~\eqref{eqn::CGIRe_Frobenius_tying} in Proposition~\ref{prop::CGI8over3_Frobenius}.

\begin{proof}
For $\coulombGasHRenorm_{\beta}^{(8/3)}(\bs{x})$ defined in~\eqref{eqn::CGI8over3_Frobenius_aux11}, 
one can proceed as in the proof of~\eqref{eqn::H_expansion} to show that 
\begin{align} \label{eqn::H_expansion_8over3}
\begin{split}
\coulombGasHRenorm_{\beta}^{(8/3)}(\bs{x})
= \; & \frac{1}{\cst(8/3)} \, \coulombGasHRenorm_{\wp_1(\beta)/\{1,2\}}^{(8/3)}(\bs{\ddot{x}}_1)
\; + \; g_{2}(\xi,\bs{\ddot{x}}_1) \, (x_2-x_1)^{-1/4}
\; + \; O\big((x_2-x_1)^{3/4} \big) 
\\
\; & + \frac{4}{\cst(8/3)}(x_2-x_1)^{-5/4} \, (x_{b_1}-x_{b_2})^{-2} \, (x_{b_1}-\xi)^{2} \, (x_{b_2}-\xi)^{2}  \, 
\int_{\Gamma_{\hat{\beta}}}\ud \bs{\ddot{u}}_1 
\; \hat{f}_{\beta/\vee_1}^{(8/3)} (\xi,\bs{\ddot{x}}_1; \bs{\ddot{u}}_1)
\\
\; & \qquad\qquad\qquad\qquad\qquad\qquad\times
\int_{10\frac{\chi}{1-\chi}}^{\frac{9}{10}}\ud w \Big(\big(w+\tfrac{\chi}{1-\chi})^{5/2}-w^{5/2}\Big) w^{-3/2} (1-w)^{-3/2} ,
\end{split}
\end{align}
as $x_1, x_2\to \xi$, 
where the integration $\smash{\Gamma_{\hat{\beta}}}$ comprises all of the other contours 
than the two Pochhammer contours surrounding the pairs $\{x_{1},x_{b_1}\}$ and $\{x_{2},x_{b_2}\}$, and $\chi=\CR(x_1,x_2,x_{b_2},x_{b_1})$.
The Taylor expansion
\begin{align*} 
\Big(w+\tfrac{\chi}{1-\chi} \Big)^{5/2}-w^{5/2} 
= w^{5/2} \; \Big(\tfrac{5}{2}\tfrac{\chi}{w(1-\chi)} \, + \, \tfrac{15}{8}\big(\tfrac{\chi}{w(1-\chi)}\big)^2
\; + \; O\big(\big(\tfrac{\chi}{w(1-\chi)}\big)^3\big)\Big), \qquad \textnormal{for }w\in (10\tfrac{\chi}{1-\chi},\tfrac{9}{10}),
\end{align*}
shows that there exists a function $\hat{g}_{2}$ such that 
\begin{align}\label{eqn::H_expansion_8over4}
\coulombGasHRenorm_{\beta}^{(8/3)}(\bs{x})
= \; & \frac{1}{\cst(8/3)} \, \coulombGasHRenorm_{\wp_1(\beta)/\{1,2\}}^{(8/3)}(\bs{\ddot{x}}_1)
\; + \; \hat{g}_2(\xi,\bs{\ddot{x}}_1) \, (x_2-x_1)^{-1/4}
\; + \; O\big((x_2-x_1)^{3/4} \big) 
\notag 
\\
 \; & + \frac{15}{2} \, \frac{1}{\cst(8/3)} \, (x_2-x_1)^{3/4} 
 \int_{\Gamma_{\hat{\beta}}}\ud \bs{\ddot{u}}_1
 \; \hat{f}_{\beta/\vee_1}^{(8/3)} (\xi,\bs{\ddot{x}}_1; \bs{\ddot{u}}_1) 
 \int_{10\frac{\chi}{1-\chi}}^{\frac{9}{10}} \frac{\ud w}{w} 
 \notag \\
= \; & \frac{1}{\cst(8/3)} \, \coulombGasHRenorm_{\wp_1(\beta)/\{1,2\}}^{(8/3)}(\bs{\ddot{x}}_1)
\; + \; \hat{g}_2(\xi,\bs{\ddot{x}}_1) \, (x_2-x_1)^{-1/4}
\; + \; O\big((x_2-x_1)^{3/4} \big) 
\notag 
\\
 \; & + \frac{15}{2} \, \frac{1}{\cst(8/3)} \, (x_2-x_1)^{3/4} \, |\log(x_2-x_1)|
 \int_{\Gamma_{\hat{\beta}}}\ud \bs{\ddot{u}}_1 
 \; \hat{f}_{\beta/\vee_1}^{(8/3)} (\xi,\bs{\ddot{x}}_1; \bs{\ddot{u}}_1) ,
\end{align}
as $x_1,x_2\to \xi$. However, letting $\kappa\to 8/3$ in Eq.~\eqref{eqn::first_order_vanish} shows that $\hat{g}_2(\xi,\bs{\dot{x}}_1)\equiv 0$. 
Combining this observation with~\eqref{eqn::H_expansion_8over4} and Remark~\ref{rem::Frobenius_coeffi_equiv} gives 
the asserted expansion~\eqref{eqn::CGIRe_Frobenius_tying_again}. 
\end{proof}

\subsection{$\kappa=8$: proof of Proposition~\ref{prop::CGI8_Frobenius}}
\label{subsec::USTlog}

In this section, we fix $\kappa=8$ and use the notation from Proposition~\ref{prop::CGI8_Frobenius}. 

\begin{proof}[Proof of Proposition~\ref{prop::CGI8_Frobenius}]
When $\{1,2\}\in \beta$, the asserted expansion~\eqref{eqn::UST_paired} is proven in~\cite[Proposition~2.9]{LPW:UST_in_topological_polygons_partition_functions_for_SLE8_and_correlations_in_logCFT}. 
When $\{1,2\}\notin \beta$, the asserted expansion~\eqref{eqn::UST_log} is proven in Lemma~\ref{lem::CGI8_Frobenius}. 
\end{proof}

\begin{lemma} \label{lem::CGI8_Frobenius}
Fix $\beta\in\LP_N$ and suppose $\{1,2\}\not\in\beta$. 
Then, for all $\xi\in(-\infty,x_3)$, we have 
\begin{align}\label{eqn::UST_log_again}
\begin{split}
\coulombGasRenorm_{\beta}^{(8)}(\bs{x})
= &\; \frac{1}{8} \, \coulombGasRenorm_{\wp_1(\beta)/\{1,2\}}^{(8)}(\bs{\ddot{x}}_1) 
\, (x_2-x_1)^{1/4} 
\, |\log(x_{j+1}-x_j)| \; + \; O\big((x_2-x_1)^{1/4} \, \big) , \qquad \textnormal{as } x_1, x_2 \to \xi .
\end{split}
\end{align}
\end{lemma}

Note that Eq.~\eqref{eqn::UST_log_again} is the second claim~\eqref{eqn::UST_log} in Proposition~\ref{prop::CGI8_Frobenius}.

\begin{proof}
Using Proposition~\ref{prop::CGI_8} and the identity~\eqref{eqn::H_ASY_change_of_variables2} from the proof of Lemma~\ref{lem::coeff_frob_bad_axu2}, we obtain 
\begin{align*}
\coulombGasRenorm_{\beta}^{(8)}(\bs{x})= \; &\lim_{\kappa'\to 8-} \frac{8}{\pi} \, 
\frac{\coulombGas_{\beta}^{(\kappa')}(\bs{x})}{(8-\kappa')^{N}} 
\; =\; \frac{8}{32^N\pi}\LH_{\beta}^{(8)}(\bs{x}) 
\; =\; \frac{8}{32^N\pi}\lim_{\kappa'\to 8-} \coulombGasH^{(\kappa')}_{\beta}(\bs{x})\\
=\; &\frac{1}{32^{N-1}} \, 
\bigg( \frac{(x_2-x_1) (x_{b_1}-x_2) (x_{b_2}-x_{1})}{(x_{b_1}-x_1) (x_{b_2}-x_2)} \bigg)^{1/4} 
\\
\; & \times  
\int_{\Gamma_{\hat{\beta}}}\ud \bs{\ddot{u}}_1 
\; f(\bs{\ddot{x}}_1;\bs{\ddot{u}}_1)
\ointclockwise_{\acycle(0,1)}\ud w 
\; \tilde{g}(x_1,x_2,w; \bs{\ddot{x}}_1;\bs{\ddot{u}}_1) 
\, \Big( w+\frac{\chi}{1-\chi} \Big)^{-1/2} (1-w)^{-1/2} \, w^{-1/2},
\end{align*}
where the integration $\smash{\Gamma_{\hat{\beta}}}$ comprises all of the other contours 
than the two Pochhammer contours surrounding the pairs $\{x_{1},x_{b_1}\}$ and $\{x_{2},x_{b_2}\}$, and $\chi=\CR(x_1,x_2,x_{b_2},x_{b_1})$, and where $\tilde{g}$ is defined in~\eqref{eqn::def_ghat} with $\kappa=8$. 
One can proceed as in the proof of~\eqref{eqn::H_expansion} to show that
\begin{align} \label{eqn::CGI8_Frobenius_aux1}
\begin{split}
\coulombGasRenorm_{\beta}^{(8)}(\bs{x})
= \; & O\big((x_2-x_1)^{1/4}\big)\\
\; & +\frac{4}{32^{N-1}} \, 
(x_2-x_1)^{1/4}
\int_{\Gamma_{\hat{\beta}}}\ud \bs{\ddot{u}}_1 \; f(\bs{\ddot{x}}_1;\bs{\ddot{u}}_1)\; \int_{10\frac{\chi}{1-\chi}}^{\frac{9}{10}} \, \big( w+\tfrac{\chi}{1-\chi} \big)^{-1/2} (1-w)^{-1/2} \, w^{-1/2} 
\end{split}
\end{align}
 as $x_1,x_2\to \xi$.
On the one hand, plugging the Taylor expansion
\begin{align*}
\big( w+\tfrac{\chi}{1-\chi} \big)^{-1/2}
= \; & w^{-1/2}\big(1+O\big(\tfrac{\chi}{w(1-\chi)}\big)\big), \qquad \textnormal{for }w\in (10\tfrac{\chi}{1-\chi},\tfrac{9}{10}),
\end{align*}
into~\eqref{eqn::CGI8_Frobenius_aux1}
shows that 
\begin{align}\label{eqn::CGI8_Frobenius_aux4}
\begin{split}
\coulombGasRenorm_{\beta}^{(8)}(\bs{x})
=\; & O\big((x_2-x_1)^{1/4}\big)
\; + \; (x_2-x_1)^{1/4} \, | \log(x_2-x_1)|  \,\frac{4}{32^{N-1}} \,\int_{\Gamma_{\hat{\beta}}}\ud \bs{\ddot{u}}_1 \; f(\bs{\ddot{x}}_1;\bs{\ddot{u}}_1).
\end{split}
\end{align}
On the other hand, according to Proposition~\ref{prop::remove_integration}, we have 
\begin{align}\label{eqn::CGI8_Frobenius_aux2}
\coulombGasH^{(\kappa')}_{\wp_1(\beta)/\{1,2\}}(\bs{\ddot{x}}_1)=\; & \frac{\fugacity(\kappa') }{\cst(\kappa')} \,
\int_{\Gamma_{\hat{\beta}}} \ud \bs{\ddot{u}}_1
\; \hat{f}_\beta^a (\bs{\ddot{x}}_1;\bs{\ddot{u}}_1),\quad \textnormal{for all }\kappa'\in (0,8)\setminus \{\tfrac{8}{m}:m\in\bZpos\},
\end{align}
where $\hat{f}_\beta^a$ is defined in~\eqref{eq: integrand with one less screening variable a} 
(with $\kappa$ there replaced by $\kappa'$). 
Taking the limit $\kappa'\to 8-$ in~\eqref{eqn::CGI8_Frobenius_aux2} then gives 
\begin{align}\label{eqn::CGI8_Frobenius_aux3}
\int_{\Gamma_{\hat{\beta}}}\ud \bs{\ddot{u}}_1 \; f(\bs{\ddot{x}}_1;\bs{\ddot{u}}_1)=\frac{1}{4\pi} \coulombGasH^{(8)}_{\wp_1(\beta)/\{1,2\}}=32^{N-2}\coulombGasRenorm_{\wp_1(\beta)/\{1,2\}}^{(8)}(\bs{\ddot{x}}_1).
\end{align}
Combining~\eqref{eqn::CGI8_Frobenius_aux4} with~\eqref{eqn::CGI8_Frobenius_aux3} gives~\eqref{eqn::UST_log_again}, as desired.
\end{proof}

\subsection{Third order BPZ equations: proof of Proposition~\ref{prop::third_order_PDE}}\label{subsec::3rdorderBPZ}

The goal of this section is to prove that the ``fused'' Coulomb gas integral function $\smash{\coulombGas_{\beta/\vee_1}^{(\kappa)}}$ defined in~\eqref{eqn::CGI_fused_generic} 
for $\kappa\in (8/3,8)\setminus \{4\}$ and $\beta\in \LP_N$ with $\{1,2\}\notin \beta$ 
satisfies the third order BPZ PDE~\eqref{eqn::third_order_PDE}. 
We first relate the function $\smash{\coulombGas_{\beta/\vee_1}^{(\kappa)}(\xi, \bs{\ddot{x}}_1)}$ 
to another Coulomb gas integral function. 
For $\kappa \in \C \setminus \{0\}$, we define 
\begin{align} \label{eqn::CGI_fused_generic_thirdorder}
\LG_{\beta/\vee_1}^{(\kappa)}(\xi, \bs{\ddot{x}}_1) 
:=  \; &  \ointclockwise_{\acycle(\xi,x_{b_1})}\ud u_1\ointclockwise_{\acycle(\xi,x_{b_2})} \ud u_2
\ointclockwise_{\acycle(x_{a_3},x_{b_3})}\ud u_3 \cdots \ointclockwise_{\acycle(x_{a_N},x_{b_N})} \ud u_N 
\; \hat{g}_{\beta/\vee_1}^{(\kappa)} (\xi, \bs{\ddot{x}}_1; \bs{{u}}),
\end{align}
where $\bs{\ddot{x}}_1=(x_3,\ldots,x_{2N})\in \chamber_{2N-2}$ and $\xi \in (-\infty,x_3)$, and
where $\bs{{u}} := (u_1,u_2,\ldots,u_N)$ and
\begin{align} \label{eqn::integrand_CGI_fused_generic_thirdorder}
\hat{g}_{\beta/\vee_1}^{(\kappa)} (\xi, \bs{\ddot{x}}_1; \bs{{u}}) 
:= \; &  (u_2-u_1)^{8/\kappa} \; 
\prod_{r=1}^N (u_r-\xi)^{-8/\kappa} \prod_{j=3}^{2N}(x_j-\xi)^{4/\kappa}
\; \times f^{(\kappa)} (\bs{\ddot{x}}_1; \bs{\ddot{u}}_1), 
\end{align}
with $\bs{\ddot{u}}_1 = (u_3,u_4,\ldots,u_N)$, 
whose branch is chosen to be real and positive when 
$\xi< \Re (u_1)< \Re (u_2)<x_{3}$ and $x_{a_r}<\Re (u_r)< x_{a_r+1}$ for $r\in \{3,4,\ldots, N\}$, 
and where $\acycle(\xi,x_{b_1})$ and $\acycle(\xi,x_{b_2})$ are two disjoint Pochhammer contours surrounding 
the two pairs of points $(\xi, x_{b_1})$ and $(\xi, x_{b_2})$, respectively, and such that $\acycle(\xi,x_{b_2})$ surrounds $\acycle(\xi, x_{b_1})$ 
and both contours avoid the other contours $\acycle(x_{a_r},x_{b_r})$ for $r\in \{3,4,\ldots, N\}$. 
See Figure~\ref{fig:Pochhammer_fusion} for an illustration of the two Pochhammer contours $\acycle(\xi,x_{b_1})$ and $\acycle(\xi, x_{b_2})$.

\begin{figure}
	\includegraphics[width=0.8\textwidth]{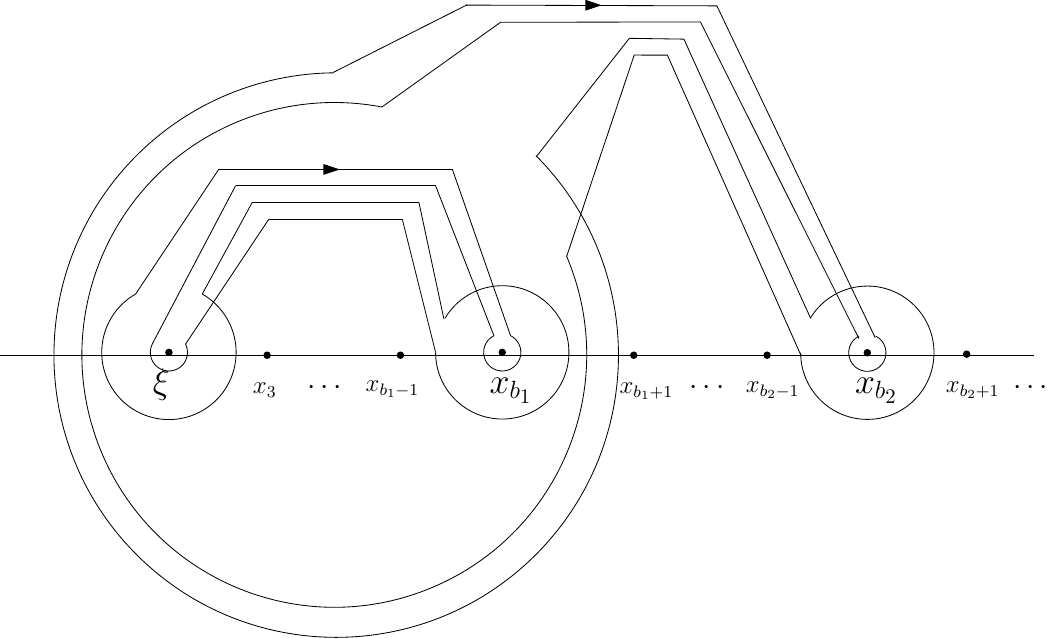}
	$\qquad$
	\caption{\label{fig:Pochhammer_fusion}
The Pochhammer contours $\acycle(\xi,x_{b_1})$ and $\acycle(\xi,x_{b_2})$ used to define the function $\LG_{\beta/\vee_1}^{(\kappa)}(\xi, \bs{\ddot{x}}_1)$ in Eq.~\eqref{eqn::CGI_fused_generic_thirdorder}. 
Importantly, the contour $\acycle(\xi,x_{b_2})$ surrounds the contour $\acycle(\xi, x_{b_1})$.}
\end{figure}

\begin{lemma} \label{lem::G_fusion_properties}
Fix $\kappa \in \C \setminus \{0\}$ and $\beta\in \LP_N$ with $\{1,2\}\notin \beta$.
The function $\smash{\LG_{\beta/\vee_1}^{(\kappa)}}$ defined in~\eqref{eqn::CGI_fused_generic_thirdorder} satisfies 
the following properties.
\begin{itemize}[leftmargin=3.5em]
\item[\textnormal{(COV)}] 
For all M\"obius maps $\varphi$ of the upper half-plane $\HH$ such that $\varphi(\xi)<\varphi(x_3)<\cdots <\varphi(x_{2N})$, we have
\begin{align} \label{eqn::loop_removal_fusion_aux1}
\LG^{(\kappa)}_{\beta/\vee_1}(\varphi(\xi),\varphi(\bs{\ddot{x}}_1))=\varphi'(\xi)^{(8-\kappa)/\kappa} \prod_{j=3}^{2N} \varphi'(x_j)^{2h(\kappa)}  \times \LG^{(\kappa)}_{\beta/\vee_1}(\xi, \bs{\ddot{x}}_1) .
\end{align} 

\item[\textnormal{(PDE)}] 
The third order BPZ PDE~\eqref{eqn::third_order_PDE} holds: 
\begin{align*}
\Big[	\frac{\partial^3}{\partial \xi^3} -\frac{16}{\kappa}\LL_{-2}\frac{\partial}{\partial \xi} +\frac{8(8-\kappa)}{\kappa^2} \LL_{-3}\Big] \LG^{(\kappa)}_{\beta/\vee_1}(\xi, \bs{\ddot{x}}_1)=0,
\end{align*}
where $\LL_{-2}$ and $\LL_{-3}$ are defined in~\eqref{eqn::L_operators}.
\end{itemize}
\end{lemma}

\begin{proof}
The same arguments as in~\cite[Proposition~4.15]{Kytola-Peltola:Conformally_covariant_boundary_correlation_functions_with_quantum_group} show the M\"obius covariance~\eqref{eqn::loop_removal_fusion_aux1}.
(We used such an argument already in the proof of Proposition~\ref{prop: full Mobius covariance} for an analogous covariance for the case where the function has the homogeneity of a correlation function of $2N$ degenerate fields at level two; 
instead, here $\smash{\LG_{\beta/\vee_1}^{(\kappa)}}$ has the homogeneity of a correlation function of one degenerate field at level three and $2N-2$ degenerate fields at level two.) 
The key point is that the integration contours in~\eqref{eqn::CGI_fused_generic_thirdorder} are homologically closed and the homogeneity degree of the function is appropriate (corresponding to the eigevalue equations $E.v=0$ and $K.v=v$ for a vector $v$ in a suitable quantum group $U_q(\mathfrak{sl}(2,\C))$ module in~\cite{Kytola-Peltola:Conformally_covariant_boundary_correlation_functions_with_quantum_group}.

The third order BPZ PDE~\eqref{eqn::third_order_PDE} can be verified via the arguments used in~\cite[Proposition~4.12 and Corollary~4.11]{Kytola-Peltola:Conformally_covariant_boundary_correlation_functions_with_quantum_group},
with the following input (a version of~\cite[Lemma~4.9]{Kytola-Peltola:Conformally_covariant_boundary_correlation_functions_with_quantum_group}\footnote{The proof of~\cite[Lemma~4.9]{Kytola-Peltola:Conformally_covariant_boundary_correlation_functions_with_quantum_group} presented in~\cite[Appendix~B]{Kytola-Peltola:Conformally_covariant_boundary_correlation_functions_with_quantum_group} 
uses a fusion argument which assumes $\kappa \notin \QQ$ due to the usage of Virasoro representation theory. However, since everything is explicit, we can use analytic continuation in $\kappa$.}):
the explicit function
\begin{align*}
\LG^{(\kappa)}_0(\xi, \bs{\ddot{x}}_1) 
:= \prod_{j=3}^{2N}(x_j-\xi)^{4/\kappa} \; \times \prod_{3\leq i<j\leq 2N}(x_{j}-x_{i})^{2/\kappa} 
\end{align*}
satisfies the desired PDE~\eqref{eqn::third_order_PDE}; 
indeed, this follows from~\cite[Lemma~4.9]{Kytola-Peltola:Conformally_covariant_boundary_correlation_functions_with_quantum_group} and analytic continuation in $\kappa$.
Hence, one can bootstrap the PDE~\eqref{eqn::third_order_PDE} 
from $\smash{\LG^{(\kappa)}_0}$ to $\LG^{(\kappa)}_{\beta/\vee_1}$ similarly as in~\cite[Proposition~4.12 and Corollary~4.11]{Kytola-Peltola:Conformally_covariant_boundary_correlation_functions_with_quantum_group}. 
(See also~\cite[Proposition~2.3]{FPW:Connection_probabilities_of_multiple_FK_Ising_interfaces}, 
where such an argument was used as well.)
\end{proof}

\begin{proposition} \label{prop::loop_removal_fusion}
Fix $\beta\in \LP_N$ with $\{1,2\}\notin \beta$. 
Then, for $\kappa\in (8/3,8)\setminus \{4\}$, 
we have\footnote{Note that for $\kappa\in\{3,6\}$, both sides of Eq.~\eqref{eqn::loop_removal_fusion} are zero.} 
\begin{align} \label{eqn::loop_removal_fusion}
\coulombGas_{\beta/\vee_1}^{(\kappa)}(\xi, \bs{\ddot{x}}_1) 
=  \; &  
\frac{\cst(\kappa)^{N-2} \,(\fugacity(\kappa)^2-1)}{C_*(\kappa) \, \hF\big(\tfrac{4}{\kappa}, 1-\tfrac{4}{\kappa}, \tfrac{8}{\kappa}; 1\big)}\, \LG_{\beta/\vee_1}^{(\kappa)}(\xi, \bs{\ddot{x}}_1) ,
\qquad 
\begin{cases}
\bs{\ddot{x}}_1=(x_3,\ldots,x_{2N})\in \chamber_{2N-2} , \\ 
\xi \in (-\infty,x_3) ,
\end{cases}
\end{align}
where $\smash{\coulombGas_{\beta/\vee_1}^{(\kappa)}}$ is defined in~\eqref{eqn::CGI_fused_generic}, 
$\smash{\LG_{\beta/\vee_1}^{(\kappa)}}$ in~\eqref{eqn::CGI_fused_generic_thirdorder}, 
and $C_*(\kappa)$ is the explicit constant 
\begin{align*}
C_*(\kappa) := \frac{32 \, e^{-8\pi\ii/\kappa} \big( 1+e^{8\pi\ii/\kappa} \big)}{(1-\frac{4}{\kappa})
} \sin(8\pi/\kappa)\sin^3(4\pi/\kappa) \, \frac{\Gamma(1 -4/\kappa)^2}{\Gamma(2 -8/\kappa)}. 
\end{align*}
\end{proposition}

\begin{proof}
Fix $\bs{\ddot{x}}_1=(x_3,\ldots,x_{2N})\in\chamber_{2N-2}$ and $\xi \in (-\infty, x_3)$. 
As in the proof of Proposition~\ref{prop::remove_integration}, we may assume that $\smash{\LG_{\beta/\vee_1}^{(\kappa)}(\xi, \bs{\ddot{x}}_1) \neq 0}$, 
and consider the ratio $\coulombGas_{\beta/\vee_1}^{(\kappa)}(\xi, \bs{\ddot{x}}_1) /\LG_{\beta/\vee_1}^{(\kappa)}(\xi, \bs{\ddot{x}}_1) 
$ of the left and right-hand sides. 
For large $R >0$, let $\varphi_R$ be the M\"obius map of the upper half-plane $\HH$ such that
\begin{align*}
\begin{cases}
-\infty<\varphi_R(\xi)<\varphi_R(x_3)<\cdots<\varphi_R(x_{2N})<\infty,\\
\varphi_R(x_3)=x_3,\quad \varphi_R(x_{2N})=x_{2N},\quad \varphi_R(\xi)=-R.
\end{cases}
\end{align*} 
Denote $\varphi_R(\bs{\ddot{x}}_1):=(\varphi_R(x_3),\ldots,\varphi_R(x_{2N}))$ and 
$\realpt_j := \underset{R \to \infty}{\lim} \varphi_R(x_j)$ for $3\le j\le 2N$. Note that 
\begin{align*}
-\infty<-R<x_3=\realpt_3<\cdots< \realpt_{2N}=x_{2N}<\infty. 
\end{align*}  

Let us first check that the function $\smash{\coulombGas_{\beta/\vee_1}^{(\kappa)}}$ satisfies the same M\"obius covariance property~\eqref{eqn::loop_removal_fusion_aux1} as $\smash{\LG^{(\kappa)}_{\beta/\vee_1}}$. 
Indeed, 
Eq.~\eqref{eqn::CGI_Frobenius_tying} and the M\"obius covariance~\eqref{eqn::COV} of $\smash{\coulombGas_{\beta}^{(\kappa)}}$ in Theorem~\ref{thm::CGI} together imply that
\begin{align*}
\coulombGas_{\beta/\vee_1}^{(\kappa)}(\varphi(\xi),\varphi(\bs{\ddot{x}}_1))
=\; &   \lim_{x_1,x_2\to \xi } \frac{\coulombGas_{\beta}^{(\kappa)}(\varphi(x_1),\varphi(x_2),\varphi(\bs{\ddot{x}}_1))-\coulombGas_{\beta/\vee_1}^{(\kappa)}(\varphi(\xi), \varphi(\bs{\ddot{x}}_1)) 
\,   \left(\varphi(x_2)-\varphi(x_1)\right)^{-2h(\kappa)}}{\left(\varphi(x_2)-\varphi(x_1)\right)^{2/\kappa} }\\
=\; &   \varphi'(\xi)^{(8-\kappa)/\kappa} \prod_{j=3}^{2N} \varphi'(x_j)^{2h(\kappa)}  \times  \lim_{x_1,x_2\to \xi } \frac{\coulombGas_{\beta}^{(\kappa)}(x_1,x_2,\bs{\ddot{x}}_1)-\coulombGas_{\beta/\vee_1}^{(\kappa)}(\xi, \bs{\ddot{x}}_1) 
\,   \left(x_2-x_1\right)^{-2h(\kappa)}}{\left(x_2-x_1\right)^{2/\kappa} }\\
=\; &  \varphi'(\xi)^{(8-\kappa)/\kappa} \prod_{j=3}^{2N} \varphi'(x_j)^{2h(\kappa)}  \times \coulombGas_{\beta/\vee_1}^{(\kappa)}(\xi, \bs{\ddot{x}}_1) 
, \qquad \kappa\in (8/3,8)\setminus \{4\} .
\end{align*} 
Thus, we may compute the ratio in the form
\begin{align} \label{eqn::loop_removal_fusion_aux2}
\; & \frac{\hF\big(\tfrac{4}{\kappa}, 1-\tfrac{4}{\kappa}, \tfrac{8}{\kappa}; 1\big)} {\cst(\kappa)^{N-2} \,(\fugacity(\kappa)^2-1)} \;
\frac{\coulombGas_{\beta/\vee_1}^{(\kappa)}(\xi, \bs{\ddot{x}}_1)}{\LG_{\beta/\vee_1}^{(\kappa)}(\xi, \bs{\ddot{x}}_1)} 
\; = \; \frac{\hF\big(\tfrac{4}{\kappa}, 1-\tfrac{4}{\kappa}, \tfrac{8}{\kappa}; 1\big)} {\cst(\kappa)^{N-2} \,(\fugacity(\kappa)^2-1)} \;
\lim_{R\to\infty} \frac{\coulombGas_{\beta/\vee_1}^{(\kappa)}\big(\varphi_R(\xi), \varphi_R(\bs{\ddot{x}}_1)\big)}{\LG_{\beta/\vee_1}^{(\kappa)}\big(\varphi_R(\xi), \varphi_R(\bs{\ddot{x}}_1)\big)} 
\; 
\nonumber
\\
= \; & 
\lim_{R\to \infty} \frac{R^{16/\kappa-2}  \; 
\ointclockwise_{\realacycle_3^\beta}\ud u_3\cdots \ointclockwise_{\realacycle_{N}^{\beta}}\ud u_{N} 
\; \hat{f}^{(\kappa)}_{\beta/\vee_1} (\varphi(\xi),\varphi_R(\bs{\ddot{x}}_1);\bs{\ddot{u}}_1)}{R^{16/\kappa-2} \; \ointclockwise_{\realacycle_1^\beta}\ud u_1\cdots \ointclockwise_{\realacycle_{N}^{\beta}}\ud u_{N}\; \hat{g}^{(\kappa)}_{\beta/\vee_1} (\varphi(\xi),\varphi_R(\bs{\ddot{x}}_1);\bs{{u}})}, 
\end{align} 
where the Pochhammer contours 
$\smash{\realacycle_3^{\beta},\ldots,\realacycle_N^{\beta}}$ are all disjoint, each $\smash{\realacycle_s^{\beta}}$ 
surrounds the points $\realpt_{a_s},\realpt_{b_s}$ and does not surround any other points among $\{-R,\realpt_{3},\ldots,\realpt_{2N}\}$, and where 
$\realacycle_1^{\beta}$ and $\realacycle_2^{\beta}$ are two disjoint Pochhammer contours, 
also disjoint from the other contours, 
surrounding the two pairs of points $-R,\realpt_{b_1}$ and $-R,\realpt_{b_2}$, respectively, 
and such that $\realacycle_2^{\beta}$ surrounds $\realacycle_1^{\beta}$. 
Now, after collecting the powers
$0 = 16/\kappa-2+(N-2)(16/\kappa-2)-(2N-2)(8/\kappa-1)$ of the blow-up parameter $R$, 
the numerator of~\eqref{eqn::loop_removal_fusion_aux2} equals
\begin{align} \label{eqn::loop_removal_fusion_aux3}
\ointclockwise_{\realacycle_3^\beta}\ud u_3\cdots \ointclockwise_{\realacycle_{N}^{\beta}}\ud u_{N} 
\; f^{(\kappa)} (\realpt_3,\ldots,\realpt_{2N};\bs{\ddot{u}}_1)
\end{align}
where $\smash{f^{(\kappa)}}$ 
is the integrand function~\eqref{eq: integrand_gen} 
with $d=2N-2$ and $\ell=N-2$,
and with branch chosen to be real and positive when $\realpt_{a_s} < \Re (u_s) < x_{a_s+1}$ for all $3\leq s\leq N$.

To evaluate the denominator of~\eqref{eqn::loop_removal_fusion_aux2}, we first assume that $\kappa > 8$ (note that Eq.~\eqref{eqn::CGI_fused_generic_thirdorder} is well-defined and continuous for all $\kappa \in \C \setminus \{0\}$). 
Then, relating the contour integrals to line integrals (see, e.g.,~\cite[Proof of Lemma~2.1]{FPW:Connection_probabilities_of_multiple_FK_Ising_interfaces}) and using the changes of variables $u:=\frac{u_1+R}{\realpt_{b_1}+R}$ and $v:=\frac{u_2+R}{\realpt_{b_2}+R}$, we obtain
\begin{align}  \label{eqn::loop_removal_fusion_aux4}
\; & \lim_{R\to \infty}R^{16/\kappa-2} \; \ointclockwise_{\realacycle_1^\beta}\ud u_1\cdots \ointclockwise_{\realacycle_{N}^{\beta}}\ud u_{N}\; \hat{g}^{(\kappa)}_{\beta/\vee_1} (\varphi(\xi),\varphi_R(\bs{\ddot{x}}_1);\bs{u})
\notag 
\\
=\; & 4\sin^2(8\pi/\kappa)\big(4\sin^2(4\pi/\kappa)\big)^{N-1} \, e^{-8\pi\ii/\kappa} \notag\\
\;&\qquad\qquad\times \lim_{R\to \infty} R^{16/\kappa-2} \landupint_{-R}^{\realpt_{b_1}} \ud u_1\landupint_{-R}^{\realpt_{b_2}}\ud u_2\landupint_{\realpt_{a_3}}^{\realpt_{b_3}}\ud u_3\cdots \landupint_{\realpt_{a_N}}^{\realpt_{b_N}} \ud u_N \; \hat{g}^{(\kappa)}_{\beta/\vee_1} (\varphi(\xi),\varphi_R(\bs{\ddot{x}}_1);\bs{u}) \notag\\
=\; & 16\sin^{2}(8\pi/\kappa)\sin^2(4\pi/\kappa) \, e^{-8\pi\ii/\kappa} \landupint_{0}^1 \ud u \; u^{-8/\kappa} (1-u)^{-4/\kappa} \landupint_{0}^1 \ud v \; (v-u)^{8/\kappa} v^{-8/\kappa} (1-v)^{-4/\kappa}
\notag\\
\;&\qquad\qquad\times \ointclockwise_{\realacycle_3^\beta}\ud u_3\cdots \ointclockwise_{\realacycle_{N}^{\beta}}\ud u_{N} 
\; f^{(\kappa)} (\realpt_3,\ldots,\realpt_{2N};\bs{\ddot{u}}_1).
\end{align}
To evaluate the prefactor, we use the Selberg integral~\cite[Chapter~4]{Forrester:Log_gases_and_random_matrices}
\begin{align*}
\landupint_{0}^1 \ud u \; u^{-8/\kappa} (1-u)^{-4/\kappa} \landupint_{0}^1 \ud v \; |v-u|^{8/\kappa} v^{-8/\kappa} (1-v)^{-4/\kappa} = \frac{\Gamma(1 -8/\kappa) \, \Gamma(1 -4/\kappa)^2 \, \Gamma(1+8/\kappa)}{\Gamma(2 -8/\kappa) \, \Gamma(2 -4/\kappa) \, \Gamma(1+4/\kappa)} .
\end{align*}
The identity $\Gamma(a+1) = a \Gamma(a)$, the identity $\Gamma(1-a) = \frac{\pi}{\Gamma(a) \sin(\pi a)}$, and simple trigonometry gives
\begin{align*}
\Gamma(2 -4/\kappa) = \Big( 1-\frac{4}{\kappa} \Big) \, \Gamma(1-4/\kappa)
\qquad \textnormal{and}\qquad
\frac{\Gamma( 1 + 8/\kappa )}{\Gamma( 1 + 4/\kappa)}
= 2 \, \frac{\Gamma(8/\kappa )}{\Gamma(4/\kappa)} 
= 2 \, 
\frac{\sin(4\pi/\kappa) \, \Gamma(1-4/\kappa) }{\sin(8\pi/\kappa) \, \Gamma(1-8/\kappa)} .
\end{align*}
After collecting these identities, we find that Eq.~\eqref{eqn::loop_removal_fusion_aux4} takes the form 
\begin{align} \label{eqn::loop_removal_fusion_aux5}
\frac{32 \, e^{-8\pi\ii/\kappa} \big( 1+e^{8\pi\ii/\kappa} \big)}{(1-\frac{4}{\kappa})} 
\sin(8\pi/\kappa)\sin^3(4\pi/\kappa) \, 
\frac{\Gamma(1 -4/\kappa)^2}{\Gamma(2 -8/\kappa)}
\ointclockwise_{\realacycle_3^\beta}\ud u_3\cdots \ointclockwise_{\realacycle_{N}^{\beta}}\ud u_{N} 
\; f^{(\kappa)} (\realpt_3,\ldots,\realpt_{2N};\bs{\ddot{u}}_1).
\end{align}
Plugging~(\ref{eqn::loop_removal_fusion_aux3},~\ref{eqn::loop_removal_fusion_aux5}) into Eq.~\eqref{eqn::loop_removal_fusion_aux2} and using continuity for $\kappa\in (8/3,8)\setminus \{4\}$, we obtain~\eqref{eqn::loop_removal_fusion}. 
\end{proof}

\begin{proof}[Proof of Proposition~\ref{prop::third_order_PDE}]
By Lemma~\ref{lem::G_fusion_properties}, 
the function $\smash{\LG_{\beta/\vee_1}^{(\kappa)}}$ satisfies the third order BPZ PDE~\eqref{eqn::third_order_PDE}. 
Combining this fact with Proposition~\ref{prop::loop_removal_fusion} gives 
the same PDE~\eqref{eqn::third_order_PDE} for $\smash{\coulombGas_{\beta/\vee_1}^{(\kappa)}}$ when $\kappa\in (8/3,8)\setminus\{4\}$.  
The remaining case $\kappa=4$ is treated in Proposition~\ref{prop::CGI4_Frobenius} in Appendix~\ref{app::Frob}. 
\end{proof}

\appendix


\bigskip{}
\section{Some integral identities}
\label{app::identities}
 The identities gathered in this appendix are used in Coulomb gas computations.

\begin{lemma} \label{lem::beta_acycle}
We have 
\begin{align}
\label{eqn::beta_1} 
& \ointclockwise_{\acycle(0,1)} u^{-4/\kappa} (u-1)^{-4/\kappa} \, \ud u
\; = \; \frac{\fugacity(\kappa)}{\cst(\kappa)} , 
\quad &&\textnormal{for }\kappa\in \C \setminus\big( \big\{ \tfrac{8}{m} \colon m \in \bZpos \big\}\cup \{0\}\big);
\\
& \ointclockwise_{\acycle(0,1)} u^{-4/\kappa} (u-1)^{-4/\kappa} (2u-1) \, \ud u 
\; = \; 0 , 
\quad &&\textnormal{for }\kappa\in \C\setminus \{0\};
\label{eqn::beta_vanish} 
\end{align}
and
\begin{align}\label{eqn::beta_3}
		\begin{split}
			\oint_{\acycle(0,1)}u^{12/\kappa-2} (1-u)^{-4/\kappa} \ud u
			= \; & 4\sin^2(4\pi/\kappa) \, (1-\fugacity(\kappa)^2) \, e^{16\pi\ii/\kappa} \, \frac{\Gamma(12/\kappa-1) \, \Gamma(1-4/\kappa)}{\Gamma(8/\kappa)},\\
			& \qquad\qquad\qquad \textnormal{for } \kappa\in \C\setminus\big( \big\{\tfrac{8}{m} \colon m\in (2\bZpos)\cup (\bZneg) \big\}\cup \{0\}\big).
		\end{split} 
\end{align}
\end{lemma}

\begin{proof}
We first verify~\eqref{eqn::beta_1}:  for $\kappa\in (4,8)$, we have
\begin{align*}
\ointclockwise_{\acycle(0,1)} u^{-4/\kappa} (u-1)^{-4/\kappa} \, \ud u 
= \; & (1-e^{8\pi\ii/\kappa}+1-e^{-8\pi\ii/\kappa}) \, 
\int_0^1 |u(u-1)|^{-4/\kappa} \, \ud u \\
= \; & 4\sin^2(4\pi/\kappa) \, \frac{\Gamma(1-4/\kappa)^2}{\Gamma(2-8/\kappa)}
\; = \; \frac{\fugacity(\kappa)}{\cst(\kappa)}. 
\end{align*}
The identity~\eqref{eqn::beta_1} 
for $\kappa\in \C \setminus \big( \big\{ \tfrac{8}{m} \colon m \in \bZpos \big\}\cup \{0\}\big)$ 
follows by analytic continuation. 

Let us next verify~\eqref{eqn::beta_vanish}:  for $\kappa\in (4,8)$, we have 
\begin{align*}
		\ointclockwise_{\acycle(0,1)} u^{-4/\kappa} (u-1)^{-4/\kappa} (2u-1)\, \ud u 
		= \; & (1-e^{8\pi\ii/\kappa}+1-e^{-8\pi\ii/\kappa}) \, 
		\int_0^1 |u(u-1)|^{-4/\kappa}(2u-1) \, \ud u \\
		= \; & 4\sin^2(4\pi/\kappa) \, \Big(2\frac{\Gamma(2-4/\kappa)\Gamma(1-4/\kappa)}{\Gamma(3-8/\kappa)}-\frac{\Gamma(1-4/\kappa)^2}{\Gamma(2-8/\kappa)}\Big)
		\; \\
		=\; & 4\sin^2(4\pi/\kappa)\frac{\Gamma(1-4/\kappa)^2}{\Gamma(2-8/\kappa)} \left(1-1\right)=0.
\end{align*}
The identity~\eqref{eqn::beta_vanish} 
for $\kappa\in \C\setminus \{0\}$ follows by analytic continuation. 

Finally, we prove~\eqref{eqn::beta_3}:  for $\kappa\in (4,8)$, we have
\begin{align*}
\ointclockwise_{\acycle(0,1)}u^{12/\kappa-2}(u-1)^{-4/\kappa} \ud u
=\; & (1-e^{8\pi\ii/\kappa}+e^{32\pi\ii/\kappa}-e^{24\pi\ii/\kappa}) 
\; \int_0^1 |u|^{12/\kappa-1} |u-1|^{4/\kappa} \ud u\\
=\; & 4\sin^2(4\pi/\kappa) (1-\fugacity(\kappa)^2) \, e^{16\pi\ii/\kappa} \,  \frac{\Gamma(12/\kappa-1) \, \Gamma(1-4/\kappa)}{\Gamma(8/\kappa)} .
\end{align*}
The identity~\eqref{eqn::beta_3} for $\kappa\in \C\setminus\big( \big\{\tfrac{8}{m} \colon m\in (2\bZpos)\cup (\bZneg) \big\}\cup \{0\}\big)$ follows by analytic continuation.
\end{proof}

\begin{lemma}\label{lem::beta_rholoop}
We have
\begin{align}
\label{eqn::beta_2}
\begin{split}
 \int_{\intloop(0,\infty)} u^{-4/\kappa} (1+u)^{-4/\kappa} \, \ud u
\; =& \; \frac{(1-e^{-8\pi\ii/\kappa})}{\fugacity(\kappa)}  \, \frac{\Gamma(1-4/\kappa)^2}{\Gamma(2-8/\kappa)} , \\
&\textnormal{for }\kappa\in \{\kappa'\in \C: 0<|\kappa'|<8\}\setminus \{\tfrac{8}{m} \colon m\in \bZpos\} ;
\end{split}
\end{align}
and
\begin{align}
\label{eqn::beta_vanish_2} 
& \int_{\intloop(0,\infty)} u^{-4/\kappa} (1+u)^{-4/\kappa} 
\Big( \frac{6}{\kappa}-1+\frac{4}{\kappa} \Big( \frac{u}{1+u}+u \Big) \Big) \, \ud u
\; = \; 0 , 
\quad &&\textnormal{for }0<|\kappa| < 8 ;
\\
\label{eqn::beta_vanish_3}
& \int_{\intloop(0,\infty)} u^{-4/\kappa} (1+u)^{-4/\kappa} (2u+1)\, \ud u
\; = \; 0  ,
\quad &&\textnormal{for }0<|\kappa| < 8 ;
\end{align}
where $\intloop(0,\infty)$ is a contour which starts from $\infty$, 
then winds around $0$ once in 
clockwise direction and finally goes back to $\infty$,
\begin{align} \label{eq: rho contour}
\intloop(0,\infty) := \vcenter{\hbox{\includegraphics[scale=0.5]{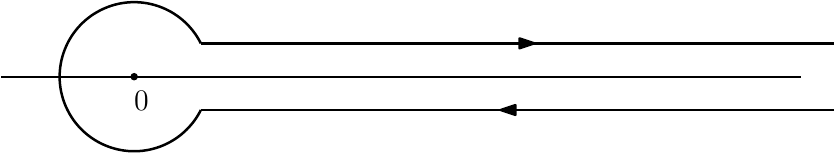}}} 
\end{align}
and the branch of $u \mapsto  u^{-4/\kappa} (1+u)^{-4/\kappa}$ 
is chosen to be real and positive when $0<\Re (u)<\infty$. 
\end{lemma}

\begin{proof}
First, we prove~\eqref{eqn::beta_2}: for $\kappa\in (4,8)$, we have
\begin{align*}
\int_{\intloop(0,\infty)} u^{-4/\kappa} (1+u)^{-4/\kappa} \, \ud u
= \; & (1-e^{-8\pi\ii/\kappa}) \, \int_0^\infty u^{-4/\kappa} (1+u)^{-4/\kappa} \, \ud u
\\
= \; & (1-e^{-8\pi\ii/\kappa}) \, \frac{\Gamma(1-4/\kappa)\Gamma(8/\kappa-1)}{\Gamma(4/\kappa)}\\
=\; & \frac{(1-e^{-8\pi\ii/\kappa})}{\fugacity(\kappa)}  \, \frac{\Gamma(1-4/\kappa)^2}{\Gamma(2-8/\kappa)} ,
\qquad \textnormal{since} \qquad 
\Gamma(1-z)\Gamma(z) = \frac{\pi}{ \sin(\pi z) } .
\end{align*}
The identity~\eqref{eqn::beta_2} for $\kappa\in \{\kappa'\in \C \colon 0<|\kappa'|<8\}\setminus \{\tfrac{8}{m} \colon m\in \bZpos\}$ follows by analytic continuation.

Next, we prove~(\ref{eqn::beta_vanish_2},~\ref{eqn::beta_vanish_3}): 
for $\kappa\in (4,8)$, we have
\begin{align*}
\frac{1}{4\sin^2(4\pi/\kappa)} \; & \; 
\int_{\intloop(0,\infty)} u^{-4/\kappa} (1+u)^{-4/\kappa} 
\Big( \frac{6}{\kappa}-1+\frac{4}{\kappa} \Big( \frac{u}{1+u}+u \Big) \Big) \, \ud u
\\
= & \; \Big( \frac{6}{\kappa}-1\Big) \, \frac{\Gamma(1-4/\kappa)\Gamma(8/\kappa-1)}{\Gamma(4/\kappa)} 
\; + \; \frac{4}{\kappa} \, \Big(\frac{\Gamma(2-4/\kappa)\Gamma(8/\kappa-1)}{\Gamma(1+4/\kappa)}
\; + \; \frac{\Gamma(2-4/\kappa)\Gamma(8/\kappa-2)}{\Gamma(4/\kappa)}\Big) 
\\
= & \; \frac{\Gamma(1-4/\kappa)\Gamma(8/\kappa-1)}{\Gamma(4/\kappa)} \, 
\Big(\frac{6}{\kappa}-1+1-\frac{4}{\kappa}-\frac{2}{\kappa}\Big) 
\; = \; 0 ;
\\ 
\frac{1}{4\sin^2(4\pi/\kappa)} \; & \; 
\int_{\intloop(0,\infty)} u^{-4/\kappa} (1+u)^{-4/\kappa} (2u+1)\, \ud u
\\
= & \; \frac{2 \, \Gamma(2-4/\kappa)\Gamma(8/\kappa-2)}{\Gamma(4/\kappa)} 
\; + \; \frac{\Gamma(1-4/\kappa)\Gamma(8/\kappa-1)}{\Gamma(4/\kappa)} 
\\
= & \; \frac{\Gamma(1-4/\kappa)\Gamma(8/\kappa-1)}{\Gamma(4/\kappa)} \; (-1+1)
\; = \; 0 ,
\end{align*}
and for $0<|\kappa| < 8$ the identities~(\ref{eqn::beta_vanish_2},~\ref{eqn::beta_vanish_3}) follow by analytic continuation. 
\end{proof}

\bigskip{}
\section{Additional formulas for Coulomb gas integrals}
\label{app::additional_formulas}
This section gathers some useful formulas for Sections~\ref{sec::multipleSLEs}--\ref{sec::PPF}. 
We will consider the Coulomb gas type integrands with $d \geq 2$ variables and $\ell \geq 0$ integration variables (with any branch choice), 
\begin{align*} 
f (\bs{x};\bs{u}) := \; &
\prod_{1\leq i<j\leq d}(x_{j}-x_{i})^{2/\kappa} 
\prod_{1\leq r<s\leq \ell}(u_{s}-u_{r})^{8/\kappa} 
\prod_{\substack{1\leq i\leq d \\ 1\leq r\leq \ell}}
(u_{r}-x_{i})^{-4/\kappa} 
\end{align*}
having negative screening charges $\alpha_- = - 2/\sqrt{\kappa}$ at the integration variables $\bs{u} = (u_1, u_2, \ldots, u_\ell)$, 
and symmetric charges 
$\alpha_j = 1/\sqrt{\kappa}$ for all $j \in \{1, 2, \ldots, d\}$ 
at all of the variables $\bs{x} = (x_1, x_2, \ldots, x_d)$, 
and 
\begin{align}
\begin{split} \label{eq: integrand_gen_conjugate_a}
\hat{f} (\bs{x};\bs{v})
:= \; &
\prod_{2\leq i<j\leq d}(x_{j}-x_{i})^{2/\kappa} 
\prod_{2\leq j\leq d}(x_{j}-x_1)^{1-6/\kappa} \\
\; & \times \; 
\prod_{1\leq t<s\leq \ell-1}(v_{s}-v_{t})^{8/\kappa} 
\prod_{\substack{2\leq i\leq d \\ 1\leq s\leq \ell-1}}
(v_{s}-x_{i})^{-4/\kappa} 
\prod_{\substack{1\leq s\leq \ell-1}}
(v_{s}-x_1)^{12/\kappa-2} ,
\end{split}
\end{align} 
having the conjugate charge $\alpha_1 = \frac{\kappa - 6}{2\sqrt{\kappa}}$ at $x_1$ and charges 
$\alpha_j = 1/\sqrt{\kappa}$ for all other $j \in \{2, 3, \ldots, d\}$.
When $d = 2N$ and $\ell = N$, the latter integrand satisfies the charge neutrality~\cite{DMS:CFT,
Flores-Kleban:Solution_space_for_system_of_null-state_PDE3}:
\begin{align} \label{eq: charge neutrality}
\Big( \sum_{j=2}^{2N} \alpha_j \Big) + \Big( \frac{\kappa - 6}{2\sqrt{\kappa}} \Big) + (N-1) \alpha_- = 2 \alpha_0 .
\end{align}

\begin{remark} \label{rem:holo}
Note that when $\ell = m/2$, the residue of $\hat{f}$ at infinity is well-defined and equals zero: 
\begin{align*}
\underset{v_r = \infty}{\mathrm{Res}} \hat{f}(\bs{x};\bs{v})
= - \lim_{|v_r| \to \infty} v_r \hat{f}(\bs{x};\bs{v}) 
= - \lim_{|v_r| \to \infty} v_r^{-1}
= 0 , \qquad 1 \leq r \leq \ell-1 .
\end{align*}
In particular, the integrand function $v_r \mapsto \hat{f}(\bs{x};\bs{v})$ is holomorphic at infinity for all $1 \leq r \leq \ell-1$.
\end{remark}

\begin{lemma} \label{lem::RN_equivalent_def2}
Fix $\kappa\in (4,8)$. 
We have 
\begin{align}\label{eq::RN_equivalent_def2}
\begin{split}
\coulombGas_{\!\! \vcenter{\hbox{\includegraphics[scale=0.8]{figures-arXiv/link62.pdf}}}_N}(\bs{x})
= \; & e^{4 \pi \ii/\kappa} \, \fugacity(\kappa) \, \hat{\cst}(\kappa)^N 
\; \int_{x_{2}}^{x_{3}} \ud v_1 \cdots \int_{x_{2N-2}}^{x_{2N-1}} \ud v_{N-1} \int_{x_{2N}}^{+\infty} \ud v_N \; f_{\!\! \vcenter{\hbox{\includegraphics[scale=0.8]{figures-arXiv/link62.pdf}}}_N}(\bs{x};\bs{v}) \\
= \; & \fugacity(\kappa) \, \hat{\cst}(\kappa)^N
\; \int_{x_{2}}^{x_{3}} \ud v_1 \cdots \int_{x_{2N-2}}^{x_{2N-1}} \ud v_{N-1} \int_{x_{2N}}^{+\infty} \ud v_N \; |f(\bs{x};\bs{v})| 
, \qquad \bs{x} \in \chamber_{2N} ,
\end{split}
\end{align}
and 
\begin{align}\label{eq::RN_equivalent_def3}
\begin{split}
\coulombGas_{\!\! \vcenter{\hbox{\includegraphics[scale=0.8]{figures-arXiv/link61.pdf}}}_N}(\bs{x})
= \; & e^{-4 \pi \ii/\kappa} \, \fugacity(\kappa) \, \hat{\cst}(\kappa)^N
\;  \int_{-\infty}^{x_{1}} \ud v_1 \int_{x_{3}}^{x_{4}} \ud v_2 \cdots \int_{x_{2N-1}}^{x_{2N}} \ud v_N \; f_{\!\! \vcenter{\hbox{\includegraphics[scale=0.8]{figures-arXiv/link61.pdf}}}_N}(\bs{x};\bs{v}) \\
= \; & \fugacity(\kappa) \, \hat{\cst}(\kappa)^N
\;  \int_{-\infty}^{x_{1}} \ud v_1 \int_{x_{3}}^{x_{4}} \ud v_2 \cdots \int_{x_{2N-1}}^{x_{2N}} \ud v_N  \; |f(\bs{x};\bs{v})| 
, \qquad \bs{x} \in \chamber_{2N} ,
\end{split}
\end{align}
where \textnormal{(}see Figure~\ref{fig:cchat}\textnormal{)}
\begin{align*}
\hat{\cst}(\kappa) := 4\sin^2(4\pi/\kappa) \, \cst(\kappa) 
= \fugacity(\kappa) \, \frac{\Gamma(2-8/\kappa)}{\Gamma(1-4/\kappa)^2} .
\end{align*}
In particular, we have $\coulombGas_{\!\! \vcenter{\hbox{\includegraphics[scale=0.8]{figures-arXiv/link62.pdf}}}_N}(\bs{x}) > 0$ and 
$\coulombGas_{\!\! \vcenter{\hbox{\includegraphics[scale=0.8]{figures-arXiv/link61.pdf}}}_N}(\bs{x}) > 0$ for all $\bs{x} \in \chamber_{2N}$. 
\end{lemma}

\begin{figure}[ht!]
\includegraphics[scale=.5]{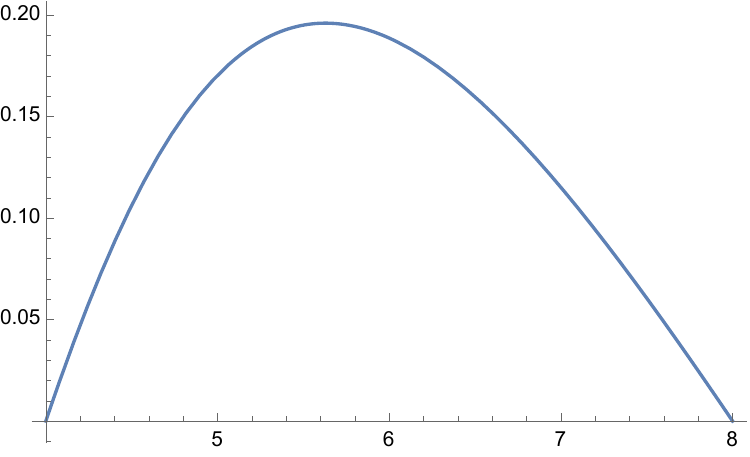}
\qquad\qquad
\includegraphics[scale=.5]{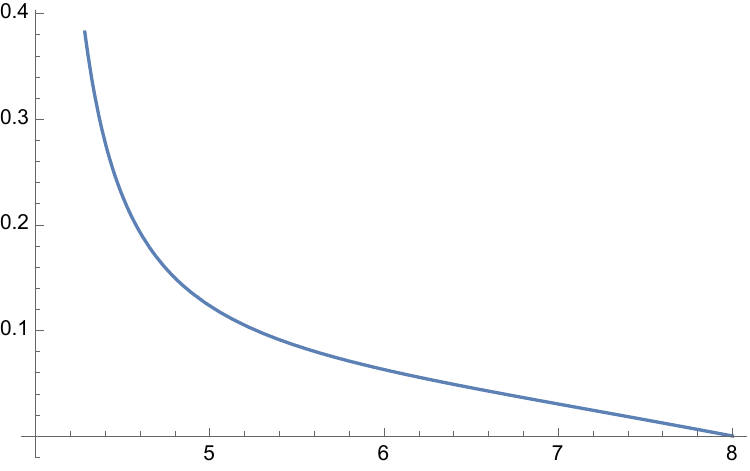}
\caption{\label{fig:cchat}
The left panel depicts a plot of the function $\kappa \mapsto \hat{\cst}(\kappa)$. 
Note that $\hat{\cst}(4) = 0 = \hat{\cst}(8)$.
The right panel depicts a plot of the function $\kappa \mapsto \cst(\kappa)$. 
Note that $\cst(8) = 0$ and $\cst(\kappa)$ diverges as $\kappa \to 4+$. 
}
\end{figure}

\begin{proof}
The two lines on the right-hand side of~\eqref{eq::RN_equivalent_def2} (resp.~\eqref{eq::RN_equivalent_def3}) 
are obtained from each other by the branch choice~\eqref{eq: branch choice} of the integrand $f_{\!\! \vcenter{\hbox{\includegraphics[scale=0.8]{figures-arXiv/link62.pdf}}}_N}$ 
(resp.~$f_{\!\! \vcenter{\hbox{\includegraphics[scale=0.8]{figures-arXiv/link61.pdf}}}_N}$). 
To prove the identity~\eqref{eq::RN_equivalent_def2}, we note on the one hand 
that, by the limit property~\eqref{eqn::CGI_ASY_at_infinity} in Corollary~\ref{cor::CGI_PDECOVASY}, 
\begin{align*}
\lim_{R \to \infty} R^{2h(\kappa)}
\, \coulombGas_{\!\! \vcenter{\hbox{\includegraphics[scale=0.8]{figures-arXiv/link61.pdf}}}_{N+1}}(-R,\bs{x},R)
= \coulombGas_{\!\! \vcenter{\hbox{\includegraphics[scale=0.8]{figures-arXiv/link62.pdf}}}_N}(\bs{x}) .
\end{align*}
On the other hand, using Proposition~\ref{prop::remove_integration} for $\beta = \smash{\vcenter{\hbox{\includegraphics[scale=1.2]{figures-arXiv/link61.pdf}}}_{N+1}}$ and 
$r=1$ in~\eqref{eq: Dub consistency general a} so that $a_1 = 1$ and $b_1=2$ and writing the Pochhammer integrations as line integrals\footnote{If necessary, we use the integration symbols $\smash{\landupint_{x_{a_s}}^{x_{b_s}}} \ud u_s$ to indicate that the integration of the variable $u_s$ is performed from $x_{a_s}$ to $x_{b_s}$ in the upper half-plane.} by taking $\epsilon \to 0$ in Lemma~\ref{lem::Poch_line_relation},
\begin{align} \label{eq:: loop vs interval}
\ointclockwise_{\acycle(x_{a_s},x_{b_s})} \ud u_s \; \hat{f}(\bs{x};\bs{u}) 
= \; & 4 \sin^2(4 \pi / \kappa) \landupint_{x_{a_s}}^{x_{b_s}} \ud u_s \; \hat{f}(\bs{x};\bs{u}) ,
\end{align}  
which holds for all $s$ such that neither $x_{a_s}$ nor $x_{b_s}$ carries the conjugate charge (i.e., $s \in \{2,3,\ldots,\ell\}$ since $\hat{f}$ in~\eqref{eq: integrand_gen_conjugate_a} has the 
conjugate charge at $x_1$), and taking note of the branch choice of $\smash{\hat{f}^a_{\!\! \vcenter{\hbox{\includegraphics[scale=0.8]{figures-arXiv/link61.pdf}}}_{N+1}}}$ in~\eqref{eq: integrand with one less screening variable a}, 
we obtain
\begin{align*}
\coulombGas_{\!\! \vcenter{\hbox{\includegraphics[scale=0.8]{figures-arXiv/link61.pdf}}}_{N+1}}(-R,\bs{x},R)
= \; & \fugacity(\kappa) \, \cst(\kappa)^{N}
\ointclockwise_{\acycle(x_2,x_3)} \ud u_2 
\cdots \ointclockwise_{\acycle(x_{2N},R)} \ud u_{N+1} 
\; \hat{f}^a_{\!\! \vcenter{\hbox{\includegraphics[scale=0.8]{figures-arXiv/link61.pdf}}}_{N+1}}(-R,\bs{x},R;\bs{\dot{u}}_1) 
&& \textnormal{[by~\eqref{eq: integrand with one less screening variable a}]}
\\
= \; & \fugacity(\kappa) \, \hat{\cst}(\kappa)^N
\; \int_{x_2}^{x_3} \ud u_2 \cdots \int_{x_{2N}}^{R} \ud u_{N+1} 
\; |\hat{f}(-R,\bs{x},R;\bs{\dot{u}}_1)| .
&& \textnormal{[by~\eqref{eq:: loop vs interval}]}
\end{align*}
Multiplying this by $(2R)^{2h(\kappa)}$ and taking $R \to \infty$ gives 
\begin{align*}
\lim_{R \to \infty} (2R)^{2h(\kappa)}
\, \coulombGas_{\!\! \vcenter{\hbox{\includegraphics[scale=0.8]{figures-arXiv/link61.pdf}}}_{N+1}}(-R,\bs{x},R)
= \; & \fugacity(\kappa) \, \hat{\cst}(\kappa)^N
\; \int_{x_2}^{x_3} \ud v_1 \cdots \int_{x_{2N}}^{+\infty} \ud v_N 
\; |f(\bs{x};\bs{v})| ,
\end{align*}
using the definition~\eqref{eq: integrand_gen_conjugate_a}.
Hence, we obtain 
\begin{align*}
\coulombGas_{\!\! \vcenter{\hbox{\includegraphics[scale=0.8]{figures-arXiv/link62.pdf}}}_N}(\bs{x}) 
= \fugacity(\kappa) \, \hat{\cst}(\kappa)^N
\; \int_{x_2}^{x_3} \ud v_1 \cdots \int_{x_{2N}}^{+\infty} \ud v_N 
\; |f(\bs{x};\bs{v})| ,
\end{align*}
which proves~\eqref{eq::RN_equivalent_def2}.  
The other asserted identity~\eqref{eq::RN_equivalent_def3}  
can be proven similarly.
\end{proof}

The following auxiliary Coulomb gas integral functions will be useful in Sections~\ref{sec::multipleSLEs}--\ref{sec::PPF}. 
They are integrals of the integrand function~\eqref{eq: integrand_gen_conjugate_a}, 
which has the conjugate charge at $x_1$, and they differ from the functions in~\eqref{eq: Dub consistency general a} in the sense that in~\eqref{eq: Dub consistency general a}, 
no integral contour surrounds the point $x_1$ carrying the conjugate charge, 
while in the functions $\hat{\LR}_{\omega_s}$ below, an integral contour does surround the point $x_1$. 
We consider special link patterns $\omega_s$ indexed by the pair $2s-1 \in \{1,3,5,\ldots,2N-1\}$ of the last point $2N$,
\begin{align*}
\omega_1 := \; &
\{\{1,2N\}, \{2,3\}, \{4,5\}, \ldots, \{2N-2, 2N-1\}\} 
\; = \; \vcenter{\hbox{\includegraphics[scale=1.2]{figures-arXiv/link62.pdf}}}_N , \\
\omega_2 := \; & \{ \{1,2\}, \{3,2N\}, \{4,5\}, \{6,7\}, \ldots, \{2N-2, 2N-1\}\} , \\
\vdots \; &\\
\omega_N := \; &
\{\{1,2\}, \{3,4\}, \ldots, \{2N-1, 2N\}\} 
\; = \; \vcenter{\hbox{\includegraphics[scale=1.2]{figures-arXiv/link61.pdf}}}_N .
\end{align*}
Let the branch of $\hat{f}_{\omega_s} = \hat{f}^a_{\omega_s}$ be chosen as in~\eqref{eq: integrand with one less screening variable a} with $a=1$. 
For $\kappa\in (0,8) \setminus \big\{ \tfrac{8}{m} \colon m \in \bZpos \big\}$, 
we define\footnote{When $\kappa=6$, we have $\fugacity(\kappa)^2 = 1$, but also, 
the integrals give zero by~(\ref{eq:: loop vs interval special1},~\ref{eq:: loop vs interval special2}), 
and $\hat{\LR}_{\omega_s}$ are well-defined by continuity.}
\begin{align*}
\hat{\LR}_{\omega_1}(\bs{x})
:= \; & \frac{e^{-4 \pi \ii/\kappa} \, \cst(\kappa)^{N-1}}{\fugacity(\kappa)^2 - 1} \, 
\ointclockwise_{\acycle(x_1,x_{2N})} \ud u_1 
\ointclockwise_{\acycle(x_4,x_5)} \ud u_3 
\cdots
\ointclockwise_{\acycle(x_{2N-2}, x_{2N-1})} \ud u_N 
\; \hat{f}^a_{\omega_1}(\bs{x};\bs{\dot{u}}_2) , \\
\hat{\LR}_{\omega_s}(\bs{x})
:= \; & - \frac{e^{-20 \pi \ii/\kappa} \, \cst(\kappa)^{N-1}}{(\fugacity(\kappa)^2 - 1)^2} \, 
\ointclockwise_{\acycle(x_1,x_2)} \ud u_1 
\cdots
\ointclockwise_{\acycle(x_{2s-3},x_{2s-2})} \ud u_{s-1} 
\ointclockwise_{\acycle(x_1,x_{2N})} \ud u_{s} \\
\; & \qquad\qquad\qquad \times 
\ointclockwise_{\acycle(x_{2s+2},x_{2s+3})} \ud u_{s+2} 
\cdots
\ointclockwise_{\acycle(x_{2N-2},x_{2N-1})} \ud u_N 
\; \hat{f}^a_{\omega_s}(\bs{x};\bs{\dot{u}}_{s+1}) , \qquad 2 \leq s \leq N-1 , \\
\hat{\LR}_{\omega_N}(\bs{x})
:= \; & - \frac{e^{-16 \pi \ii/\kappa} \, \cst(\kappa)^{N-1}}{\fugacity(\kappa)^2 - 1} \, 
\ointclockwise_{\acycle(x_1,x_2)} \ud u_1 
\ointclockwise_{\acycle(x_3,x_4)} \ud u_2
\cdots
\ointclockwise_{\acycle(x_{2N-3}, x_{2N-2})} \ud u_{N-1} 
\; \hat{f}^a_{\omega_N}(\bs{x};\bs{\dot{u}}_N) ,
\end{align*}
where the $(s+1)$:st integral is omitted and 
for $2 \leq s \leq N-1$ the integration contours are given by $\omega_s$ except for the $s$:th integral, 
which is a Pochhammer contour\footnote{The integrals involving $x_1$ are convergent, 
for the integrand $\hat{f}$ in~\eqref{eq: integrand_gen_conjugate_a} has the conjugate charge at $x_1$, 
which implies that the divergence as both integration variables $u_1, u_s$ are $\varepsilon$-close to $x_1$ is of order $\varepsilon^{1 + 2(12/\kappa-2)} \to 0$ as $\varepsilon\to 0$ for all $\kappa < 8$.} 
$\acycle(x_1,x_{2N})$ around the points $x_1,x_{2N}$.

Now, because the integration contours for $\hat{\LR}_{\omega_s}$ are closed in a suitable homology 
and satisfy the charge neutrality~\eqref{eq: charge neutrality}~\cite{DMS:CFT, Dubedat:Euler_integrals_for_commuting_SLEs, 
Flores-Kleban:Solution_space_for_system_of_null-state_PDE3}, 
we readily conclude that they provide solutions to the BPZ PDEs and are M\"obius covariant, 
so they lie in the solution space $\mathcal{S}_N$ defined in~\eqref{eq: solution space}.

\begin{lemma} \label{lem:PDECOVASY_for_LR}
The functions $\{\hat{\LR}_{\omega_s} \colon 1 \leq s \leq N\}$ 
satisfy the properties~\eqref{eqn::PDE},~\eqref{eqn::COV}, and~\eqref{eqn::PLB_weak_upper}. 
\end{lemma}

\begin{proof}
\eqref{eqn::PDE}~\&~\eqref{eqn::COV} are well-known~\cite[Lemma~2]{Dubedat:Euler_integrals_for_commuting_SLEs}. 
The bound~\eqref{eqn::PLB_weak_upper} is clear. 
\end{proof}

We write $\int_{x_{2N}}^{x_{1}} \ud u = \int_{-\infty}^{x_{1}} \ud u + \int_{x_{2N}}^{+\infty} \ud u$ (cf.~Remark~\ref{rem:holo}).  

\begin{lemma} \label{lem::another_CGI_with_conjugate_charge}
Fix $\kappa\in (4,8)$. 
We have 
\begin{align*}
\hat{\LR}_{\omega_1}(\bs{x})
= \; & \hat{\cst}(\kappa)^{N-1} \, 
\int_{x_{2N}}^{x_{1}} \ud u_1 \int_{x_4}^{x_5} \ud u_3 
\cdots
\int_{x_{2N-2}}^{x_{2N-1}} \ud u_N 
\; |\hat{f}(\bs{x};\bs{\dot{u}}_2)| , \\
\hat{\LR}_{\omega_s}(\bs{x})
= \; & \hat{\cst}(\kappa)^{N-1} \, 
\; \int_{x_{1}}^{x_{2}} \ud u_1 
\cdots
\int_{x_{2s-3}}^{x_{2s-2}} \ud u_{s-1} 
\int_{x_{2N}}^{x_{1}} \ud u_{s} 
\\
\; & \qquad\qquad\qquad\qquad\qquad\qquad\qquad \times 
\int_{x_{2s+2}}^{x_{2s+3}} \ud u_{s+2} 
\cdots
\int_{x_{2N-2}}^{x_{2N-1}} \ud u_N 
\; |\hat{f}(\bs{x};\bs{\dot{u}}_{s+1})| , \\
\hat{\LR}_{\omega_N}(\bs{x})
= \; & \hat{\cst}(\kappa)^{N-1} \,
\int_{x_{1}}^{x_{2}} \ud u_1 
\cdots
\int_{x_{2N-3}}^{x_{2N-2}} \ud u_{N-1} 
\; |\hat{f}(\bs{x};\bs{\dot{u}}_{N})| , \qquad \bs{x} \in \chamber_{2N} ,
\end{align*}
In particular, we have $\hat{\LR}_{\omega_s}(\bs{x}) > 0$ for all $s \in \{1,2,\ldots,N\}$ and $\bs{x} \in \chamber_{2N}$.
\end{lemma}

\begin{proof}
This follows using~\eqref{eq:: loop vs interval} for the integrals not involving the point $x_1$, and the analogous formulas
\begin{align} 
\label{eq:: loop vs interval special1} 
\ointclockwise_{\acycle(x_{1},x_{2})} \ud u_1 \; \hat{f}(\bs{x};\bs{u})
= \; & - 4 \sin^2(4 \pi / \kappa) \, \big( \fugacity(\kappa)^2 - 1 \big) \, e^{16 \pi \ii / \kappa} \,
\int_{x_{1}}^{x_{2}} \ud u_1 \; \hat{f}(\bs{x};\bs{u})  , \\
\label{eq:: loop vs interval special2} 
\ointclockwise_{\acycle(x_{1},x_{2N})} \ud u_1 \; \hat{f}(\bs{x};\bs{u})
= \; & 4 \sin^2(4 \pi / \kappa) \, \big( \fugacity(\kappa)^2 - 1 \big) \, e^{-8 \pi \ii / \kappa} \,
\int_{x_{2N}}^{x_{1}} \ud u_1 \; \hat{f}(\bs{x};\bs{u})  ,
\end{align} 
for the integrals involving the point $x_1$
obtained by taking $\epsilon \to 0$ in Lemma~\ref{lem::Poch_line_relation}.
\end{proof}

\begin{lemma} \label{lem:R_ASY}
The functions $\{\hat{\LR}_{\omega_s} \colon 1 \leq s \leq N\}$ 
satisfy the asymptotics property 
\begin{align} \label{eqn::R_ASY}
\; & \lim_{x_1,x_2 \to\xi} \frac{\hat{\LR}_{\omega_s}(\bs{x})}{ (x_2-x_1)^{-2h(\kappa)} }
= 
\begin{cases}
\tfrac{1}{\fugacity(\kappa)} \, \coulombGas_{\!\! \vcenter{\hbox{\includegraphics[scale=0.8]{figures-arXiv/link62.pdf}}}_{N-1}}(\bs{\ddot{x}}_1) , 
& \textnormal{if } s = 1 , \\
0 ,
& \textnormal{if } s \neq 1 ,
\end{cases}
\end{align}
for any $\xi < x_3 < x_4 < \cdots < x_{2N}$.
\end{lemma}

\begin{proof}
First, we treat the case $s=1$. 
Thanks to Lemma~\ref{lem::another_CGI_with_conjugate_charge}, a direct computation shows that
\begin{align*}
	\; &  \hat{\LR}(\bs{\ddot{x}}_1;\xi ) 
	:= \lim_{x_1,x_2\to \xi }\frac{\hat{\LR}_{\omega_1}(\bs{x})}{(x_2-x_1)^{-2h(\kappa)}}\\
	= \; & \hat{\cst}(\kappa)^{N-1} \prod_{j=3}^{2N}(x_j-\xi)^{1-4/\kappa} \int_{x_{2N}}^{\xi } \ud  u_1 \int_{x_4}^{x_5} \ud u_3\cdots \int_{x_{2N-2}}^{x_{2N-1}} \ud u_{N}  \prod_{\substack{1\leq t\leq N\\ t\neq 2}} | u_t -\xi |^{8/\kappa-2} 
	\; | f(\bs{\ddot{x}}_1;\bs{\dot{u}}_2) | .
\end{align*}
If $\varphi$ is a M\"obius map of the upper half-plane $\HH$ with $\varphi(\xi)<\varphi(x_3)<\cdots<\varphi(x_{2N})$, then the M\"obius covariance of $\hat{\LR}_{\omega_1}$ in Lemma~\ref{lem:PDECOVASY_for_LR} and the simple identity 
$\tfrac{\varphi(x_2)-\varphi(x_2)}{x_2-x_1}=\sqrt{\varphi'(x_1)}\sqrt{\varphi'(x_2)}$ imply that
	\begin{align} \label{eqn::COV_limit}
\hat{\LR} (\varphi(\bs{\ddot{x}}_1);\varphi(\xi))
= \; & \lim_{x_1,x_2\to \xi}  \frac{\hat{\LR}_{\omega_1}(\varphi(x_1),\varphi(x_2),\varphi(\bs{\ddot{x}}_1))}{(\varphi(x_2)-\varphi(x_1))^{-2h(\kappa)}}
=\lim_{x_1,x_2\to \xi} \frac{\hat{\LR}_{\omega_1}(\bs{x})}{(x_2-x_1)^{-2h(\kappa)}}\times \prod_{j=3}^{2N}  \varphi'(x_j)^{-h}\notag\\
= \; & \hat{\LR}_{\omega_1}(\bs{\ddot{x}}_1,\xi) \times \prod_{j=3}^{2N}  \varphi'(x_j)^{-h},
	\end{align}
	where $\varphi(\bs{\ddot{x}}_1)=(\varphi(x_3),\ldots,\varphi(x_{2N}))$. 
For large enough $R>0$, take $\varphi_R$ to be the M\"obius transformation of $\HH$ such that 
\begin{align*}
\varphi_R(\xi )=-R, \quad \varphi_R(x_3)=x_3,\quad \text{and}\quad \varphi_R(x_{2N})=x_{2N}.
\end{align*}
Denote  $\varphi_R(\bs{\ddot{x}}_1):=(\varphi_R(x_3),\ldots,\varphi_R(x_{2N}))$. 
Then, we have 	
\begin{align*}
	\frac{\hat{\LR}(\bs{\ddot{x}}_1;\xi)}{\mathcal{F}_{\includegraphics[scale=0.8]{figures-arXiv/link62.pdf}_{N-1}}(\bs{\ddot{x}}_1)} =\lim_{R\to\infty}	\frac{\hat{\LR}(\varphi_R(\bs{\ddot{x}}_1);\varphi_R(\xi))}{\mathcal{F}_{\includegraphics[scale=0.8]{figures-arXiv/link62.pdf}_{N-1}}(\varphi_R(\bs{\ddot{x}}_1))}=\frac{1}{\fugacity(\kappa)},
\end{align*}	
thanks to the M\"obius covariance of $\smash{\coulombGas_{\includegraphics[scale=0.8]{figures-arXiv/link62.pdf}_{N-1}}}$ from Corollary~\ref{cor: full Mobius covariance F}, 
the M\"obius covariance of $\hat{\LR}$ in~\eqref{eqn::COV_limit}, 
and the line-integral formula~\eqref{eq::RN_equivalent_def2} 
for $\smash{\coulombGas_{\includegraphics[scale=0.8]{figures-arXiv/link62.pdf}_{N-1}}}$ 
from Lemma~\ref{lem::RN_equivalent_def2}.
This shows~\eqref{eqn::R_ASY} when $s=1$.

Second, in the case where $s\neq 1$,  using Lemma~\ref{lem::another_CGI_with_conjugate_charge} and the change of variables $u=\frac{u_1-x_1}{x_2-x_1}$ gives 
\begin{align*}
	\frac{\hat{\LR}_{\omega_s}(\bs{x})}{(x_2-x_1)^{-2h(\kappa)}} = \; & (x_2-x_1)^{8/\kappa-1}\int_{0}^{1} \ud u \cdots  \int_{x_{2s-3}}^{x_{2s-2}} \ud u_{s-1} \int_{x_{2N}}^{x_1} \ud u_s \\
	\; & \qquad\qquad\times \int_{x_{2s+2}}^{x_{2s+3}} \ud u_{s+2} \cdots\int_{x_{2N-2}}^{x_{2N-1}} \ud u_N \, 
	|u|^{12/\kappa-2} |u-1|^{-4/\kappa}  \; 
	|f(\bs{\ddot{x}}_1;\bs{\ddot{u}})| \; 
	\hat{g}_s(\bs{x};\bs{\ddot{u}},u) , \\
	\quad \overset{x_1,x_2\to \xi}{\longrightarrow} \; & \quad 0 ,
\end{align*}
where $\bs{\ddot{u}}:= (u_2,u_3,\ldots,u_{s},u_{s+2},\ldots,u_N)$ 
and the function $\hat{g}_s(\bs{x};\bs{\ddot{u}},u)$ is bounded as $x_2,x_1\to \xi$.
\end{proof}

\bigskip{}
\section{Examples with $N=1$ and $N=2$}
\label{app::examples}
In this appendix, we investigate the Coulomb gas integrals for $N=1$ and $N=2$. 
We denote
\begin{align*}
\vcenter{\hbox{\includegraphics[scale=0.4]{figures-arXiv/link-0.pdf}}}
 \quad = \quad \{\{1,2\}\} , \qquad 
\vcenter{\hbox{\includegraphics[scale=0.4]{figures-arXiv/link-1.pdf}}}
 \quad = \quad  \{\{1,2\},\{3,4\}\} , \qquad
\vcenter{\hbox{\includegraphics[scale=0.4]{figures-arXiv/link-2.pdf}}}
 \quad = \quad \{\{1,4\},\{2,3\}\}  . 
\end{align*}
It is straightforward to calculate $\smash{\LFtwo^{(\kappa)}}$ for $N=1$ (Lemma~\ref{lem::2points}). 
We derive asymptotics properties of the Coulomb gas integrals with $N=2$ in Section~\ref{subsec::H4points_asy},
and investigate their positivity in the case $N=2$ in Section~\ref{subsec::F4points_pos}. 
Lastly, in Section~\ref{subsec::Z4points} 
we consider the Frobenius series of the Coulomb gas integrals 
and the pure partition functions in the case $N=2$.

\begin{lemma} \label{lem::2points}
For $\kappa\in (0,8)\setminus \big\{\frac{8}{m} \colon m\in \bZpos \big\}$, we have 
\begin{align} \label{eqn::H2points}
\LHtwo^{(\kappa)}(x_1,x_2) = & \; \frac{\fugacity(\kappa)}{\cst(\kappa)} \, (x_2-x_1)^{-2h(\kappa)},\qquad  (x_1,x_2)\in \chamber_2. 
\end{align}
Consequently, for $\kappa\in (0,8)$, we have
\begin{align*} 
\LFtwo^{(\kappa)}(x_1,x_2) = & \; \fugacity(\kappa) \, (x_2-x_1)^{-2h(\kappa)}
,\qquad  (x_1,x_2)\in \chamber_2. 
\end{align*}
\end{lemma}

\begin{proof}
This follows by a change of variables,
Eq.~\eqref{eqn::beta_1} from Lemma~\ref{lem::beta_acycle}, 
and continuity in $(0,8) \ni \kappa$.
\end{proof}

\subsection{Asymptotics for $N=2$}
\label{subsec::H4points_asy}

This section contains an example computation for the asymptotics in Proposition~\ref{prop::H_ASY} when $N=2$.

\begin{proposition}\label{prop::H4points_asy}
Fix $\kappa\in \! (0,8)\setminus \big\{\frac{8}{m} \colon m\in \bZpos \big\}$. 
The following asymptotics hold for {$(x_1, x_2, x_3, x_4)\in\chamber_4$.}
\begin{itemize}
\item  For any $\xi \in (-\infty,x_3)$, we have 
\begin{align} 
\label{eqn::H4points_asy1}
\lim_{x_1,x_2\to \xi} \frac{\LHfoura^{(\kappa)}(x_1,x_2,x_3,x_4)}{(x_2-x_1)^{-2h(\kappa)}} 
= \frac{\fugacity(\kappa)}{\cst(\kappa)} \; \coulombGasH^{(\kappa)}_{\includegraphics[scale=0.2]{figures-arXiv/link-0}}(x_3,x_4) ; \\ 
\lim_{x_1,x_2\to \xi} \frac{\LHfourb^{(\kappa)}(x_1,x_2,x_3,x_4)}{(x_2-x_1)^{-2h(\kappa)}} 
= \frac{1}{\cst(\kappa)} \; \coulombGasH^{(\kappa)}_{\includegraphics[scale=0.2]{figures-arXiv/link-0}}(x_3,x_4) .
\label{eqn::H4points_asy1b}
\end{align}

\item For any $\xi \in (x_1,x_4)$, we have 
\begin{align} 
\label{eqn::H4points_asy2}
\lim_{x_2,x_3\to \xi} \frac{\LHfoura^{(\kappa)}(x_1,x_2,x_3,x_4)}{(x_3-x_2)^{-2h(\kappa)}}=\frac{1}{\cst(\kappa)}
 \; \coulombGasH^{(\kappa)}_{\includegraphics[scale=0.2]{figures-arXiv/link-0}}(x_1,x_4); \\
\lim_{x_2,x_3\to \xi} \frac{\LHfourb^{(\kappa)}(x_1,x_2,x_3,x_4)}{(x_3-x_2)^{-2h(\kappa)}}=\frac{\fugacity(\kappa)}{\cst(\kappa)}
 \; \coulombGasH^{(\kappa)}_{\includegraphics[scale=0.2]{figures-arXiv/link-0}}(x_1,x_4) . 
\label{eqn::H4points_asy2b}
\end{align}
\item For any $\xi \in (x_2, \infty)$, we have 
\begin{align} 
\label{eqn::H4points_asy3}
\lim_{x_3,x_4\to \xi} \frac{\LHfoura^{(\kappa)}(x_1,x_2,x_3,x_4)}{(x_4-x_3)^{-2h(\kappa)}}=\frac{\fugacity(\kappa)}{\cst(\kappa)}
 \; \coulombGasH^{(\kappa)}_{\includegraphics[scale=0.2]{figures-arXiv/link-0}}(x_1,x_2);\\
\lim_{x_3,x_4\to \xi} \frac{\LHfourb^{(\kappa)}(x_1,x_2,x_3,x_4)}{(x_4-x_3)^{-2h(\kappa)}}=\frac{1}{\cst(\kappa)}
 \; \coulombGasH^{(\kappa)}_{\includegraphics[scale=0.2]{figures-arXiv/link-0}}(x_1,x_2). 
\label{eqn::H4points_asy3b}
\end{align}
\end{itemize}
\end{proposition}

To prove Proposition~\ref{prop::H4points_asy}, we will use the following expressions for $\LHfoura^{(\kappa)}$ in Lemma~\ref{lem::reduction_twopoint}. 

\begin{lemma} \label{lem::reduction_twopoint}
Fix $N=2$, and $\includegraphics[scale=0.4]{figures-arXiv/link-1}=\{\{1,2\},\{3,4\}\}$, and $\kappa\in (0,8)\setminus \{\frac{8}{m} \colon m\in \bZpos\}$. 
We have
\begin{align}
\LHfoura^{(\kappa)}(x_1,x_2,x_3,x_4)= \; & \frac{\fugacity(\kappa)}{\cst(\kappa)} 
\; (x_3-x_1)^{-2/\kappa}(x_3-x_2)^{2/\kappa}(x_4-x_1)^{2/\kappa}(x_4-x_3)^{1-6/\kappa}(x_4-x_2)^{-2/\kappa}\notag\\ \; & \times(x_2-x_1)^{1-6/\kappa} \ointclockwise_{\acycle(0,1)}  \big(w(1-w)(1-\chi w)\big)^{-4/\kappa}\ud w\label{eqn::reduction_beta_1}\\
= \; & \frac{\fugacity(\kappa)}{\cst(\kappa)}
 \; (x_3-x_1)^{-2/\kappa}(x_3-x_2)^{2/\kappa}(x_4-x_1)^{2/\kappa}(x_4-x_3)^{1-6/\kappa}(x_4-x_2)^{-2/\kappa}\notag\\
 \; & \times (x_2-x_1)^{1-6/\kappa}(1-\chi)^{1-8/\kappa} \int_{\acycle(0,\frac{1}{1-\chi})} \big((1+\chi v)v(1+(\chi-1)v)\big)^{-4/\kappa}\ud v, \label{eqn::reduction_beta_2}
\end{align}
where $\chi=\CR(x_1,x_2,x_3,x_4)$ is the cross-ratio~\eqref{eqn::crossratio},
the branch of $\smash{w\mapsto \big(w(1-w)(1-\chi w)\big)^{-4/\kappa}}$ is chosen to be real and positive  when $0<\Re (w)<1$, 
and the branch of $\smash{v\mapsto \big((1+\chi v)v(1+(\chi-1)v)\big)^{-4/\kappa}}$ is chosen to be real and positive when $0<\Re (v)<\frac{1}{1-\chi}$. 
\end{lemma}

\begin{proof}
First, we prove~(\ref{eqn::reduction_beta_1},~\ref{eqn::reduction_beta_2}) for $\kappa\in (4,8)$.
According to Proposition~\ref{prop::remove_integration}, we have
\begin{align*}
\LHfoura^{(\kappa)}(x_1,x_2,x_3,x_4)= \; & \frac{\fugacity(\kappa)}{\cst(\kappa)} 
 \; \big((x_3-x_1)(x_2-x_1)(x_3-x_2)\big)^{2/\kappa}\big((x_4-x_1)(x_4-x_3)(x_4-x_2)\big)^{1-6/\kappa}\notag\\
 \; & \times \ointclockwise_{\acycle(x_1,x_2)} (u_1-x_1)^{-4/\kappa}(u_1-x_2)^{-4/\kappa}(u_1-x_3)^{-4/\kappa}(u_1-x_4)^{12/\kappa-2}\ud u_1,
\end{align*}
where the branch of $(u_1-x_1)^{-4/\kappa}(u_1-x_2)^{-4/\kappa}(u_1-x_3)^{-4/\kappa}(u_1-x_4)^{12/\kappa-2}$ is chosen to be real and positive when $x_1<\Re (u_1)<x_2$.
Thus, for $\kappa\in (4,8)$, we have 
\begin{align*}
\LHfoura^{(\kappa)}(x_1,x_2,x_3,x_4)= \; & \frac{\fugacity(\kappa)}{\cst(\kappa)} 
 \; \big((x_3-x_1)(x_2-x_1)(x_3-x_2)\big)^{2/\kappa}\big((x_4-x_1)(x_4-x_3)(x_4-x_2)\big)^{1-6/\kappa}\notag\\
 \; & \times4\sin^2(4\pi/\kappa)  \int_{x_1}^{x_2} \big|(u_1-x_1)(u_1-x_2)(u_1-x_3)\big|^{-4/\kappa}|u_1-x_4|^{12/\kappa-2}\ud u_1.
\end{align*}
We then use the change of variables $u_1=x_4-(x_4-x_1)\big(1+w\frac{x_2-x_1}{x_4-x_2}\big)^{-1}$ to obtain
\begin{align}
\LHfoura^{(\kappa)}(x_1,x_2,x_3,x_4)= \; & \frac{\fugacity(\kappa)}{\cst(\kappa)}
 \; (x_3-x_1)^{-2/\kappa}(x_3-x_2)^{2/\kappa}(x_4-x_1)^{2/\kappa}(x_4-x_3)^{1-6/\kappa}(x_4-x_2)^{-2/\kappa}\notag\\ \; & \times(x_2-x_1)^{1-6/\kappa} \, 4\sin^2(4\pi/\kappa) 
 \; \int_0^1  \big|w(1-w)(1-\chi w)\big|^{-4/\kappa}\ud w \label{eqn::reduction_beta_2_aux1}\\
= \; & \frac{\fugacity(\kappa)}{\cst(\kappa)}
 \; (x_3-x_1)^{-2/\kappa}(x_3-x_2)^{2/\kappa}(x_4-x_1)^{2/\kappa}(x_4-x_3)^{1-6/\kappa}(x_4-x_2)^{-2/\kappa}\notag\\ \; & \times(x_2-x_1)^{1-6/\kappa} 
 \; \ointclockwise_{\acycle(0,1)}  \big(w(1-w)(1-\chi w)\big)^{-4/\kappa}\ud w, \notag
\end{align}
which gives~\eqref{eqn::reduction_beta_1} for $\kappa\in(4,8)$. From~\eqref{eqn::reduction_beta_2_aux1}, we use the change of variables $v(1-\chi)=1-w$ to obtain
\begin{align}
\LHfoura^{(\kappa)}(x_1,x_2,x_3,x_4)= \; & \frac{\fugacity(\kappa)}{\cst(\kappa)}
 \; (x_3-x_1)^{-2/\kappa}(x_3-x_2)^{2/\kappa}(x_4-x_1)^{2/\kappa}(x_4-x_3)^{1-6/\kappa}(x_4-x_2)^{-2/\kappa}\notag\\
 \; & \times (x_2-x_1)^{1-6/\kappa}(1-\chi)^{1-8/\kappa} \, 4\sin^2(4\pi/\kappa) 
 \; \int_0^{\frac{1}{1-\chi}} \big|(1+\chi v)v(1+(\chi-1)v)\big|^{-4/\kappa}\ud v\notag\\
= \; & \frac{\fugacity(\kappa)}{\cst(\kappa)}
 \; (x_3-x_1)^{-2/\kappa}(x_3-x_2)^{2/\kappa}(x_4-x_1)^{2/\kappa}(x_4-x_3)^{1-6/\kappa}(x_4-x_2)^{-2/\kappa}\notag\\
 \; & \times (x_2-x_1)^{1-6/\kappa}(1-\chi)^{1-8/\kappa} 
 \; \ointclockwise_{\acycle(0,\frac{1}{1-\chi})} \big((1+\chi v)v(1+(\chi-1)v)\big)^{-4/\kappa}\ud v, \notag
\end{align}
We just proved that~(\ref{eqn::reduction_beta_1},~\ref{eqn::reduction_beta_2}) hold for $\kappa\in (4,8)$,
and~(\ref{eqn::reduction_beta_1},~\ref{eqn::reduction_beta_2}) for $\kappa\in (0,8)\setminus \{\frac{8}{m} \colon m\in\bZpos\}$ follow by analytic continuation. 
\end{proof}

Now, we are ready to prove Proposition~\ref{prop::H4points_asy}. 
\begin{proof}[Proof of Proposition~\ref{prop::H4points_asy}]

First, we show~\eqref{eqn::H4points_asy1}. 
By~\eqref{eqn::reduction_beta_1} in Lemma~\ref{lem::reduction_twopoint},
 for any $\xi\in (-\infty,x_3)$, we have 
\begin{align*}
\lim_{x_1,x_2\to \xi} \frac{\LHfoura^{(\kappa)}(x_1,x_2,x_3,x_4)}{(x_2-x_1)^{-2h(\kappa)}}
= \; & \frac{\fugacity(\kappa)}{\cst(\kappa)}
 \; (x_4-x_3)^{-2h(\kappa)}\; \lim_{\chi\to 0+} \ointclockwise_{\acycle(0,1)}\big(w(1-w)(1 - \chi w)\big)^{-4/\kappa}\ud w 
 \\
= \; & \frac{\fugacity(\kappa)}{\cst(\kappa)}
 \; (x_4-x_3)^{-2h(\kappa)}\;\ointclockwise_{\acycle(0,1)}\big(w(1-w)\big)^{-4/\kappa}\ud w 
 \\
= \; & \Big(\frac{\fugacity(\kappa)}{\cst(\kappa)}\Big)^2 
 \; (x_4-x_3)^{-2h(\kappa)}
\; = \; \frac{\fugacity(\kappa)}{\cst(\kappa)}   \; \LHtwo^{(\kappa)}(x_3,x_4) ,
\qquad\qquad \textnormal{[by~(\ref{eqn::beta_1}~\&~\ref{eqn::H2points})]}
\end{align*}
Second, we show~\eqref{eqn::H4points_asy2}. 
By~\eqref{eqn::reduction_beta_2} in Lemma~\ref{lem::reduction_twopoint}, for any $\xi \in(x_1,x_4)$, we have
\begin{align*}
\; & \lim_{x_2,x_3\to \xi} \frac{\LHfoura^{(\kappa)}(x_1,x_2,x_3,x_4)}{(x_3-x_2)^{-2h(\kappa)}}
\\
= \; &  \frac{\fugacity(\kappa)}{\cst(\kappa)}
 \; (x_4-x_3)^{-2h(\kappa)} 
 \; \lim_{\chi\to 1-} \ointclockwise_{\acycle(0,\frac{1}{1-\chi})} \big((1 + \chi v)v(1+(\chi-1)v)\big)^{-4/\kappa}\ud v\\
= \; & \frac{\fugacity(\kappa)}{\cst(\kappa)}
 \; (x_4-x_3)^{-2h(\kappa)}\, (1-e^{8\pi\ii/\kappa})
 \; \int_{\intloop(0,\infty)} (1+v)^{-4/\kappa} v^{-4/\kappa} \ud v\\
= \; & \frac{\fugacity(\kappa)}{\cst(\kappa)}
 \; (x_4-x_3)^{-2h(\kappa)}\, (1-e^{8\pi\ii/\kappa})\, (1-e^{-8\pi\ii/\kappa}) \, \frac{\Gamma(1-4/\kappa)^2}{\Gamma(2-8/\kappa)} \, \frac{1}{\fugacity(\kappa)}
&& \textnormal{[by~\eqref{eqn::beta_2}]}\\
= \; &  \frac{1}{\cst(\kappa)}
 \; \LHtwo^{(\kappa)}(x_1,x_4) .
&& \textnormal{[by~\eqref{eqn::beta_1}]}
\end{align*}
Third, the proof of~\eqref{eqn::H4points_asy3} is essentially the same as that of~\eqref{eqn::H4points_asy1}. 
Finally, the identities~(\ref{eqn::H4points_asy1b},~\ref{eqn::H4points_asy2b},~\ref{eqn::H4points_asy3b})
follow using~(\ref{eqn::H4points_asy1},~\ref{eqn::H4points_asy2},~\ref{eqn::H4points_asy3}) 
and Proposition~\ref{prop: full Mobius covariance}. 
\end{proof}

\subsection{Positivity for $N=2$}
\label{subsec::F4points_pos}
We shall next investigate the signs of the two Coulomb gas integral functions $\smash{\LFfoura^{(\kappa)}}$ and $\smash{\LFfourb^{(\kappa)}}$ 
when $\bs{x} = (x_1, x_2, x_3, x_4) \in \chamber_4$, using their relation with the pure partition functions.
Interestingly enough, we observe that when $\kappa < 8/3$, these functions are not always positive 
(thus, they are not true $\SLE_\kappa$ partition functions in that case) --- 
unlike the pure partition functions $\smash{\LZfoura^{(\kappa)}}$ and $\smash{\LZfourb^{(\kappa)}}$.

\begin{proposition}\label{prop::F4points_pos}
Fix $\kappa\in (0,8)$. The following identities hold.
\begin{align}\label{eqn::F4_PPF4_a}
\LFfoura^{(\kappa)}(\bs{x})= \; & \fugacity(\kappa)^2 \, \PartF_{\includegraphics[scale=0.2]{figures-arXiv/link-1}}^{(\kappa)}(\bs{x})
\, + \, \fugacity(\kappa) \, \PartF_{\includegraphics[scale=0.2]{figures-arXiv/link-2}}^{(\kappa)}(\bs{x}), \qquad \bs{x}\in\chamber_4; \\
\LFfourb^{(\kappa)}(\bs{x})= \; & \fugacity(\kappa) \,  \PartF_{\includegraphics[scale=0.2]{figures-arXiv/link-1}}^{(\kappa)}(\bs{x})
\,  + \, \fugacity(\kappa)^2 \, \PartF_{\includegraphics[scale=0.2]{figures-arXiv/link-2}}^{(\kappa)}(\bs{x}), \qquad \bs{x}\in\chamber_4; 
\label{eqn::F4_PPF4_b}
\end{align}
where 
\begin{align}\label{eqn::PPF_four}
\begin{split}
\LZfoura^{(\kappa)}(x_1, x_2, x_3, x_4)= \; & (x_2-x_1)^{-2h(\kappa)} \, (x_4-x_3)^{-2h(\kappa)} \, (1-\chi)^{2/\kappa}
\; \frac{\hF(\tfrac{4}{\kappa}, 1-\tfrac{4}{\kappa}, \tfrac{8}{\kappa}; 1-\chi)}{\hF(\tfrac{4}{\kappa}, 1-\tfrac{4}{\kappa}, \tfrac{8}{\kappa}; 1)},\\
\LZfourb^{(\kappa)}(x_1, x_2, x_3, x_4)= \; & (x_4-x_1)^{-2h(\kappa)} \, (x_3-x_2)^{-2h(\kappa)} \, \chi^{2/\kappa}
\; \frac{\hF(\tfrac{4}{\kappa}, 1-\tfrac{4}{\kappa}, \tfrac{8}{\kappa}; \chi)}{\hF(\tfrac{4}{\kappa}, 1-\tfrac{4}{\kappa}, \tfrac{8}{\kappa}; 1)},
\end{split}
\end{align}
and $\chi=\CR(x_1,x_2,x_3,x_4)$ is the cross-ratio~\eqref{eqn::crossratio}.

\noindent 
Moreover, the positivity of $\LFfoura^{(\kappa)}(\bs{x})$ and $\LFfourb^{(\kappa)}(\bs{x})$ depends on $\kappa$ and $\bs{x}$ in the following manner\textnormal{:} 
\begin{itemize}[leftmargin=3em]
\item[\textnormal{(POS)}] When $\fugacity(\kappa) > 0$, i.e., 
$\kappa\in \underset{m\in\bZpos}{\bigcup} \big(\frac{8}{4m-1},\frac{8}{4m-3}\big)$, 
we have $\LFfoura^{(\kappa)}(\bs{x})>0$ and $\LFfourb^{(\kappa)}(\bs{x})>0$ on $\chamber_4$. 

\smallskip

\item[\textnormal{(NUL)}] 
When $\fugacity(\kappa) = 0$, i.e., 
$\kappa\in \big\{\frac{8}{2m+1} \colon m\in\bZpos\big\}$, 
we have $\LFfoura^{(\kappa)}(\bs{x})=\LFfourb^{(\kappa)}(\bs{x})=0$  on $\chamber_4$. 

\smallskip 
 
\item[\textnormal{(SGN)}] 
When $\fugacity(\kappa) < 0$, i.e., 
$\kappa\in \underset{m\in\bZpos}{\bigcup} \big(\frac{8}{4m+1},\frac{8}{4m-1}\big)$, there exists $z(\kappa)\in (0,1)$ such that 
\begin{align}\label{eqn::F4pointsa_pos}
\begin{cases}
\LFfoura^{(\kappa)}(\bs{x})>0,  \; & \textnormal{ if }\CR(\bs{x})<z(\kappa);\\[.5em]
\LFfoura^{(\kappa)}(\bs{x})=0,  \; & \textnormal{ if }\CR(\bs{x})=z(\kappa);\\[.5em]
\LFfoura^{(\kappa)}(\bs{x})<0,  \; & \textnormal{ if }\CR(\bs{x})>z(\kappa);
\end{cases}
\qquad\begin{cases}
\LFfourb^{(\kappa)}(\bs{x})>0,  \; & \textnormal{ if }\CR(\bs{x})>1-z(\kappa);\\[.5em]
\LFfourb^{(\kappa)}(\bs{x})=0,  \; & \textnormal{ if }\CR(\bs{x})=1-z(\kappa);\\[.5em]
\LFfourb^{(\kappa)}(\bs{x})<0,  \; & \textnormal{ if }\CR(\bs{x})<1-z(\kappa).
\end{cases}
\end{align}

\medskip

Moreover, we have
\begin{align}\label{eqn::F4points_zero_semiclassical}
\limsup_{\kappa\to 0}z(\kappa)=1/2 
\qquad \textnormal{and} \qquad 
\liminf_{\kappa\to 0}z(\kappa)=0.
\end{align}
\end{itemize}
\end{proposition}

\begin{proof}[Proof of Proposition~\ref{prop::F4points_pos}]
By continuity in $\kappa \in (0,8)$, it suffices to check~(\ref{eqn::F4_PPF4_a},~\ref{eqn::F4_PPF4_b}) 
when $\kappa \neq 8/m$ for $m\in \bZpos$. 
From Lemma~\ref{lem::2points} and Proposition~\ref{prop::H4points_asy}, 
we see that 
both sides of~(\ref{eqn::F4_PPF4_a},~\ref{eqn::F4_PPF4_b})
satisfy~\eqref{eqn::PDE}, \eqref{eqn::COV},~\eqref{eqn::PLB_weak_upper}, and~\eqref{eqn::CGI_ASY} respectively,
and Lemma~\ref{lem::PFuniqueness} guarantees that 
they are equal.

Next, let us consider the sign of $\LFfoura^{(\kappa)}(\bs{x})$. 
The cases $\fugacity(\kappa)>0$ and $\fugacity(\kappa)=0$ are clear from~\eqref{eqn::F4_PPF4_a}. 
It remains to investigate the case when $\fugacity(\kappa)<0$ and to prove~\eqref{eqn::F4pointsa_pos} and~\eqref{eqn::F4points_zero_semiclassical}. Suppose $\fugacity(\kappa)<0$. 
Then, we have $\LFfoura^{(\kappa)}(\bs{x})>0$ if and only if 
$-\fugacity(\kappa) > \PartF_{\includegraphics[scale=0.2]{figures-arXiv/link-2}}^{(\kappa)}(\bs{x})/\PartF_{\includegraphics[scale=0.2]{figures-arXiv/link-1}}^{(\kappa)}(\bs{x})$. 
Using~\eqref{eqn::PPF_four}, we see that
\begin{align}\label{eqn::F4pointsa_pos_aux1}
\LFfoura^{(\kappa)}(\bs{x})>0 ,
\qquad \textnormal{if and only if} \qquad 
2\cos(4\pi/\kappa) > \frac{\chi^{8/\kappa-1} \, \hF\big(\frac{4}{\kappa}, 1-\frac{4}{\kappa}, \frac{8}{\kappa}; \chi\big)}{(1-\chi)^{8/\kappa-1} \, \hF\big(\frac{4}{\kappa}, 1-\frac{4}{\kappa}, \frac{8}{\kappa}; 1-\chi\big)},
\end{align}
where $\chi=\CR(\bs{x})$ is the cross-ratio~\eqref{eqn::crossratio}. 
Using the identity~\cite[Eq.~15.3.3]{Abramowitz-Stegun:Handbook}, 
\begin{align*}
\hF\big(\tfrac{4}{\kappa}, 1-\tfrac{4}{\kappa}, \tfrac{8}{\kappa}; z\big)=(1-z)^{8/\kappa-1}\hF\big(\tfrac{4}{\kappa}, \tfrac{12}{\kappa}-1, \tfrac{8}{\kappa}; z\big) ,
\end{align*}
we obtain 
\begin{align}\label{eqn::F4pointsa_pos_aux2}
\LFfoura^{(\kappa)}(\bs{x})>0,
\qquad \textnormal{if and only if}\qquad 
2\cos(4\pi/\kappa) > \frac{\hF\big(\tfrac{4}{\kappa}, \tfrac{12}{\kappa}-1, \tfrac{8}{\kappa}; \chi\big)}{\hF\big(\tfrac{4}{\kappa}, \tfrac{12}{\kappa}-1, \tfrac{8}{\kappa}; 1-\chi\big)}=:G(\chi).
\end{align}
Next note that $z\mapsto \hF\big(\tfrac{4}{\kappa}, \tfrac{12}{\kappa}-1, \tfrac{8}{\kappa}; z\big)$ is increasing for $z\in (0,1)$ and we have 
\begin{align*}
\hF\big(\tfrac{4}{\kappa}, \tfrac{12}{\kappa}-1, \tfrac{8}{\kappa}; 0\big)=1
\qquad \textnormal{and} \qquad 
 \lim_{z\to 1-}\hF\big(\tfrac{4}{\kappa}, \tfrac{12}{\kappa}-1, \tfrac{8}{\kappa}; z\big)=\infty.
\end{align*}
Thus $z\mapsto G(z)$ is also increasing, 
and $\smash{\underset{z\to 0+}{\lim} \, G(z)=0}$ and 
$\smash{\underset{z\to 1-}{\lim} \, G(z)=\infty}$. 
Therefore, there exists a unique point $z(\kappa)\in (0,1)$ such that 
\begin{align}\label{eqn::F4points_zero}
G(z(\kappa))=2\cos(4\pi/\kappa). 
\end{align}
Combining this with~\eqref{eqn::F4pointsa_pos_aux2}, we obtain the conclusion (SGN) 
for $\LFfoura^{(\kappa)}$ in~\eqref{eqn::F4pointsa_pos}. 
The proof for $\LFfourb^{(\kappa)}$ in~\eqref{eqn::F4pointsa_pos} is similar. 

It remains to prove~\eqref{eqn::F4points_zero_semiclassical}. 
Note that $\kappa\mapsto z(\kappa)$ is a continuous function and we may extend its definition to $\kappa\in (0,8/3]$ by setting $z(\kappa)=0$ when $\fugacity(\kappa)\ge 0$, 
and such an extension $\kappa\in (0,8/3]\mapsto z(\kappa)$ is still continuous. 
Because $G(0)=0$ and $G(1/2)=1$, we have
\begin{align}\label{eqn::F4points_zero_semiclassical_aux1}
z\big(\tfrac{8}{2m+1}\big)=0,\qquad z\big(\tfrac{12}{6m\pm 1}\big)=1/2,\qquad \textnormal{for }m\in\bZpos.
\end{align}
Let us consider the the asymptotics of $G(z)$ as $\kappa\to 0$. 
From direct analysis of the hypergeometric function $\hF$ 
(see also~\cite[Section~7.3]{ABKM:Pole_dynamics_and_an_integral_of_motion_for_multiple_SLE0}), we obtain
\begin{align*}
U(z) := \; & \lim_{\kappa\to 0}\kappa\log G(z) 
\; = \; 4\log\bigg(\frac{(1-z+z^2)^{3/2}+(1+z)(z-1/2)(2-z)}{(1-z+z^2)^{3/2}-(1+z)(z-1/2)(2-z)}\bigg). 
\end{align*}
For any $\eps>0$, we have $U(\frac{1}{2}+\eps)>0$. Thus $G(\frac{1}{2}+\eps)\ge 2$ when $\kappa$ is small enough. 
Therefore, $z(\kappa)\le \frac{1}{2}+\eps$ when $\kappa$ is small enough. 
In other words, we have
$\smash{\underset{\kappa\to 0}{\limsup} \, z(\kappa)\le \frac{1}{2}+\eps}$. 
Sending $\eps\to 0$, we obtain 
\begin{align*}
\limsup_{\kappa\to 0}z(\kappa) \le \frac{1}{2}.
\end{align*}
Combining this with~\eqref{eqn::F4points_zero_semiclassical_aux1}, we obtain~\eqref{eqn::F4points_zero_semiclassical} and complete the proof. 
\end{proof}

\subsection{Frobenius series for $N=2$}
\label{subsec::Z4points}
Using the explicit formulas~\eqref{eqn::PPF_four} 
in terms of the hypergeometric function $\hF$, 
we derive in this section the Frobenius series of the pure partition functions
$\smash{\LZfourb^{(\kappa)}}$ (Lemma~\ref{lem::PPF_fourb_expansion})
and 
$\smash{\LZfoura^{(\kappa)}}$ (Lemma~\ref{lem::PPF_foura_expansion}),
and from this, also obtain the Frobenius series of the Coulomb gas integral functions
$\smash{\LFfoura^{(\kappa)}}$ and $\smash{\LFfourb^{(\kappa)}}$ (Corollary~\ref{cor::F_Fro_bad_fourpoint}) for $\kappa\in (8/3,8)$.
We also address the case where $\kappa=8/3$ in Lemma~\ref{lem::CGI_8over3_log} --- it involves a logarithmic correction to the asymptotics, which is a special case of Proposition~\ref{prop::CGI8over3_Frobenius}.

\begin{lemma} \label{lem::PPF_fourb_expansion}
For $\kappa\in (0,8)$ and $\xi<x_3<x_4$, we have
\begin{align}\label{eqn::PPF_fourb_expansion}
\LZfourb^{(\kappa)}(x_1, x_2, x_3, x_4)
= (x_2-x_1)^{2/\kappa}\LZthree^{(\kappa)}(\xi, x_3, x_4)
\; + \; o \big((x_2-x_1)^{2/\kappa}\big) ,\qquad \textnormal{as }x_1, x_2\to \xi.
\end{align}
where
\begin{align}\label{eqn::PPF_three}
\LZthree^{(\kappa)}(\xi, x_3, x_4) 
:= \frac{1}{\hF\big(\tfrac{4}{\kappa}, 1-\tfrac{4}{\kappa}, \tfrac{8}{\kappa}; 1\big)}
\; \frac{(x_4-x_3)^{2/\kappa}}{(x_3-\xi)^{8/\kappa-1}(x_4-\xi)^{8/\kappa-1}} .
\end{align}
\end{lemma}

\begin{proof}
Denoting by $\chi = \CR(x_1, x_2, x_3, x_4)$ the cross-ratio~\eqref{eqn::crossratio},  
the asserted~\eqref{eqn::PPF_fourb_expansion} follows from~\eqref{eqn::PPF_four}:
\begin{align*}
\frac{\LZfourb^{(\kappa)}(x_1, x_2, x_3, x_4)}{(x_3-x_1)^{-2h(\kappa)}(x_4-x_2)^{-2h(\kappa)}}
= \; & \chi^{2/\kappa}(1-\chi)^{-2h(\kappa)}
\frac{\hF\big(\tfrac{4}{\kappa}, 1-\tfrac{4}{\kappa}, \tfrac{8}{\kappa}; \chi\big)}{\hF\big(\tfrac{4}{\kappa}, 1-\tfrac{4}{\kappa}, \tfrac{8}{\kappa}; 1\big)}
\\
= \; & \chi^{2/\kappa}\frac{(1+O(\chi))}{\hF\big(\tfrac{4}{\kappa}, 1-\tfrac{4}{\kappa}, \tfrac{8}{\kappa}; 1\big)} 
, \qquad\textnormal{as }\chi\to 0. 
\end{align*}
\end{proof}

\begin{lemma}\label{lem::PPF_foura_expansion}
For $\kappa\in (8/3,8)$ and $\xi<x_3<x_4$, we have
\begin{align}\label{eqn::PPF_foura_expansion}
\LZfoura^{(\kappa)}(x_1, x_2, x_3, x_4)
= \; & (x_2-x_1)^{-2h(\kappa)} \, \LZtwo^{(\kappa)}(x_3, x_4)
\; - \; \frac{1}{\fugacity(\kappa)} \, (x_2-x_1)^{2/\kappa} \, \LZthree^{(\kappa)}(\xi, x_3, x_4) 
\notag
\\
\; & \; + \; o\big((x_2-x_1)^{2/\kappa}\big), \qquad\textnormal{as }x_1, x_2\to\xi .
\end{align}
\end{lemma}

\begin{proof}
In the case where $\kappa=4$, the expansion~\eqref{eqn::PPF_foura_expansion} is immediate:
\begin{align*}
\LZfoura^{(4)}(x_1,x_2,x_3,x_4)=(x_2-x_1)^{-1/2}(x_4-x_3)^{-1/2}(1-\chi)^{1/2} ,
\end{align*}
where $\chi = \CR(x_1, x_2, x_3, x_4)$. 
Hence, in the rest of the proof, we assume $\kappa\in (8/3,4)\cup(4,8)$. 
From the explicit formula~\eqref{eqn::PPF_four}, we obtain
\begin{align}\label{eqn::PPF_foura_aux}
\frac{\LZfoura^{(\kappa)}(x_1, x_2, x_3, x_4)}{(x_3-x_1)^{-2h(\kappa)}(x_4-x_2)^{-2h(\kappa)}}
= \; & \chi^{-2h(\kappa)}(1-\chi)^{2/\kappa}
\; \frac{\hF(\tfrac{4}{\kappa}, 1-\tfrac{4}{\kappa}, \tfrac{8}{\kappa};1-\chi)}{\hF(\tfrac{4}{\kappa}, 1-\tfrac{4}{\kappa}, \tfrac{8}{\kappa};1)} ,
\end{align}
and we can expand the hypergeometric function using~\cite[Eq.~15.3.6]{Abramowitz-Stegun:Handbook} as
\begin{align*}
\; & \frac{\hF(\tfrac{4}{\kappa}, 1-\tfrac{4}{\kappa}, \tfrac{8}{\kappa};1-\chi)}{\hF(\tfrac{4}{\kappa}, 1-\tfrac{4}{\kappa}, \tfrac{8}{\kappa};1)}
= \hF\big(\tfrac{4}{\kappa}, 1-\tfrac{4}{\kappa}, 2-\tfrac{8}{\kappa}; \chi\big)+\chi^{8/\kappa-1}\frac{\Gamma(12/\kappa-1)\Gamma(1-8/\kappa)}{\Gamma(1-4/\kappa)\Gamma(8/\kappa-1)}\hF\big(\tfrac{4}{\kappa}, \tfrac{12}{\kappa}-1, \tfrac{8}{\kappa}; \chi\big)\\
= \; & 1+\tfrac{2}{\kappa}\chi+O(\chi^2)+\chi^{8/\kappa-1}\frac{\Gamma(12/\kappa-1)\Gamma(1-8/\kappa)}{\Gamma(1-4/\kappa)\Gamma(8/\kappa-1)}\big(1+\tfrac{12-\kappa}{2\kappa}\chi+O(\chi^2)\big),\qquad \textnormal{as }\chi\to 0.
\end{align*}
Plugging this into~\eqref{eqn::PPF_foura_aux}, since $\kappa>8/3$, 
we obtain
\begin{align*}
\frac{\LZfoura^{(\kappa)}(x_1, x_2, x_3, x_4)}{(x_3-x_1)^{-2h(\kappa)}(x_4-x_2)^{-2h(\kappa)}}
= \; & \chi^{-2h(\kappa)}\big(1+\chi^{8/\kappa-1}\frac{\Gamma(12/\kappa-1)\Gamma(1-8/\kappa)}{\Gamma(1-4/\kappa)\Gamma(8/\kappa-1)}+O(\chi^{8/\kappa})+O(\chi^2)\big)\\
= \; & \chi^{-2h(\kappa)}+\chi^{2/\kappa}\frac{\Gamma(12/\kappa-1)\Gamma(1-8/\kappa)}{\Gamma(1-4/\kappa)\Gamma(8/\kappa-1)}+o(\chi^{2/\kappa}),\qquad \textnormal{as }\chi\to 0 .
\end{align*}
This gives~\eqref{eqn::PPF_foura_expansion} for $\kappa\in (8/3,4)\cup(4,8)$ because
the constants can be written in the form
\begin{align*}
\frac{\Gamma(12/\kappa-1)\Gamma(1-8/\kappa)}{\Gamma(1-4/\kappa)\Gamma(8/\kappa-1)}
= \; & \frac{\Gamma(12/\kappa-1)\Gamma(4/\kappa)\Gamma(8/\kappa)\Gamma(1-8/\kappa)}{\Gamma(1-4/\kappa)\Gamma(4/\kappa)\Gamma(8/\kappa)\Gamma(8/\kappa-1)} 
\\
= \; & \frac{\Gamma(12/\kappa-1)\Gamma(4/\kappa)\sin(4\pi/\kappa)}{\Gamma(8/\kappa)\Gamma(8/\kappa-1)\sin(8\pi/\kappa)}
&& \textnormal{[by~\cite[Eq.~6.1.17]{Abramowitz-Stegun:Handbook}]}
\\
= \; & \frac{\Gamma(12/\kappa-1)\Gamma(4/\kappa)}{\Gamma(8/\kappa)\Gamma(8/\kappa-1)}\frac{1}{2\cos(4\pi/\kappa)}
\\
= \; & \frac{1}{\hF\big(\tfrac{4}{\kappa}, 1-\tfrac{4}{\kappa}, \tfrac{8}{\kappa}; 1\big)}\frac{1}{2\cos(4\pi/\kappa)} .
&& \textnormal{[by~\cite[Eq.~15.1.20]{Abramowitz-Stegun:Handbook}]}
\end{align*}
\end{proof}

\begin{corollary} \label{cor::F_Fro_bad_fourpoint}
For $\kappa\in (8/3,8)$ and $\xi<x_3<x_4$, we have
\begin{align}
\LFfoura^{(\kappa)}(x_1, x_2, x_3, x_4)
= \; & (x_2-x_1)^{-2h(\kappa)} \, \fugacity(\kappa)^2 \, 
\LZtwo^{(\kappa)}(x_3, x_4)
\; + \; o\big((x_2-x_1)^{2/\kappa}\big), 
\\[.5em]
\LFfourb^{(\kappa)}(x_1, x_2, x_3, x_4)
= \; & (x_2-x_1)^{-2h(\kappa)} \, \fugacity(\kappa) \,  \, 
\LZtwo^{(\kappa)}(x_3, x_4)+(x_2-x_1)^{2/\kappa}\big(\fugacity(\kappa)^2-1\big)\LZthree^{(\kappa)}(\xi, x_3, x_4)\notag\\
\; & \; + \; o\big((x_2-x_1)^{2/\kappa}\big),\qquad \textnormal{as }x_1, x_2\to\xi. \label{eqn::F_Fro_bad_fourpoint}
\end{align}
\end{corollary}

\begin{proof}
These are immediate from~(\ref{eqn::F4_PPF4_a},~\ref{eqn::F4_PPF4_b},~\ref{eqn::PPF_foura_expansion},~\ref{eqn::PPF_fourb_expansion}). 
\end{proof}

\begin{lemma}\label{lem::CGI_8over3_log}
Fix $\kappa=8/3$ and $\xi <x_3<x_4$. For the renormalized Coulomb gas integrals 
in Proposition~\ref{prop::CGI_PPF_renormalized}, we have 
\begin{align}
\LFfouraRenorm^{(8/3)}(x_1,x_2,x_3,x_4)
= \; &  (x_2-x_1)^{3/4} \, \LZthree^{(8/3)}(\xi, x_3, x_4)
\; + \; o\big((x_2-x_1)^{3/4}\big), 
\label{eqn::asy_four_SAW_good}
\\[.5em]
\LFfourbRenorm^{(8/3)}(x_1,x_2,x_3,x_4)
= \; & (x_2-x_1)^{-5/4} \, \coulombGasRenorm_{\includegraphics[scale=0.2]{figures-arXiv/link-0}}^{(8/3)}(x_3, x_4) -\frac{1}{\pi}(x_2-x_1)^{3/4}\vert\log(x_2-x_1)\vert\LZthree^{(8/3)}(\xi, x_3, x_4)\notag\\
\; & \; + \; O\big((x_2-x_1)^{3/4}\big) ,\qquad \textnormal{as }x_1, x_2\to \xi. 
\label{eqn::asy_four_SAW_bad}
\end{align}
\end{lemma}

\begin{proof}
First, we prove~\eqref{eqn::asy_four_SAW_good} by using a limiting argument. 
Using Lemma~\ref{lem::reduction_twopoint}, we obtain
\begin{align}
\LFfouraRenorm^{(8/3)}(x_1,x_2,x_3,x_4)
= \; &  \lim_{\kappa'\to 8/3}\frac{1}{\fugacity(\kappa')} \, \LFfoura^{(\kappa')}(x_1,x_2,x_3,x_4)\notag\\
= \; &  \cst(8/3) \, (x_3-x_1)^{-3/4}(x_3-x_2)^{3/4}(x_4-x_1)^{3/4}(x_4-x_3)^{-5/4}(x_4-x_2)^{-3/4} \notag\\
 \; &  \times (x_2-x_1)^{-5/4}\ointclockwise_{\acycle(0,1)}\big(w(1-w)(1-\chi w)\big)^{-3/2}\ud w, \label{eqn::asy_four_SAW_good_aux1}
\end{align}
where $\chi=\CR(x_1, x_2, x_3, x_4)$. 
On the one hand, we have the expansion
\begin{align} \label{eqn::asy_four_SAW_good_aux2}
\big(w(1-w)(1-\chi w)\big)^{-3/2}
= \; & \big(w(1-w)\big)^{-3/2} \; \big(1+\tfrac{3}{2}\chi w+\tfrac{15}{8}\chi^2w^2+o(\chi^2)\big),\qquad \textnormal{as }x_1,x_2\to \xi.
\end{align}
On the other hand, using Eq.~\eqref{eqn::beta_1}  from Lemma~\ref{lem::beta_acycle}, 
we have
\begin{align}\label{eqn::asy_four_SAW_good_aux3}
\ointclockwise_{\acycle(0,1)} \big(w(1-w)\big)^{-3/2}\ud w 
= \lim_{\kappa'\to 8/3} \ointclockwise_{\acycle(0,1)} \big(w(1-w)\big)^{-4/\kappa'} \ud w
= \lim_{\kappa'\to 8/3} \frac{\fugacity(\kappa')}{\cst(\kappa')} = 0 ,
\end{align}

Plugging~(\ref{eqn::asy_four_SAW_good_aux2},~\ref{eqn::asy_four_SAW_good_aux3}) into~\eqref{eqn::asy_four_SAW_good_aux1} gives the asserted~\eqref{eqn::asy_four_SAW_good}.

To prove~\eqref{eqn::asy_four_SAW_bad}, one can proceed as in the proof of~\eqref{eqn::reduction_beta_2} 
to show that 
\begin{align*}
\LFfourb^{(\kappa')} (x_1,x_2,x_3,x_4)
= \; &  \fugacity(\kappa') \, \cst(\kappa') \, (x_4-x_2)^{-2/\kappa'} (x_2-x_1)^{2/\kappa'}(x_4-x_3)^{2/\kappa'}(x_3-x_2)^{1-6/\kappa'}(x_3-x_1)^{-2/\kappa'}\\
 \; & \times (x_4-x_1)^{1-6/\kappa'} \chi^{1-8/\kappa'} 
 \; \ointclockwise_{\acycle(0,\frac{1}{\chi})} \big((1+(1-\chi)v)v(1-\chi v)\big)^{-4/\kappa} \ud v 
\end{align*}
when 
$\kappa'\notin \{\frac{8}{m} \colon m\in \mathbb{Z}_{>0}\}$,  
which implies that 
\begin{align} \label{eqn::asy_four_SAW_bad_aux1}
\frac{\LFfourbRenorm^{(8/3)}(x_1,x_2,x_3,x_4)}{(x_3-x_1)^{-5/4}(x_4-x_2)^{-5/4}}
= \; & \lim_{\kappa'\to 8/3} \frac{\LFfourb^{(\kappa')}(x_1,x_2,x_3,x_4)}{\fugacity(\kappa') \, (x_3-x_1)^{-2h(\kappa')}(x_4-x_2)^{-2h(\kappa')}}\notag\\
= \; & \cst(8/3) \, (1-\chi)^{-5/4} \chi^{-5/4} 
\; \ointclockwise_{\acycle(0,\frac{1}{\chi})} \big((1+(1-\chi)v)v(1-\chi v)\big)^{-3/2} \ud v.
\end{align}
Expanding 
\begin{align}\label{eqn::asy_four_SAW_bad_aux2}
\begin{split}
 \; & 
(1-\chi)^{-5/4}\big((1+(1-\chi)v)v(1-\chi v)\big)^{-3/2} 
\\
= \; &  \big((1+v)v\big)^{-3/2}\; \big(1+\tfrac{5}{4}\chi+\tfrac{3}{2}\tfrac{\chi v}{1+v}+\tfrac{3}{2}\chi v+\tfrac{15}{8}\chi^2v^2+O(\chi^2)+O(\chi^3v^3)\big),
\end{split}
\end{align}
as $x_1,x_2\to \xi$, with the constants in $O(\chi^2)$ and $O(\chi^3v^3)$ uniformly bounded in $v$,
and using the change of variables $w=\chi v$, we obtain
\begin{align}\label{eqn::asy_four_SAW_bad_aux3}
\ointclockwise_{\acycle(0,\tfrac{1}{\chi})} \chi^3v^3\big((1+v)v\big)^{-3/2}\ud v=\chi^2\ointclockwise_{\acycle(0,1)} t^3\big(t(\chi+t)\big)^{-3/2}\ud t=O(\chi^2),\qquad \textnormal{as }x_1,x_2\to \xi. 
\end{align}
Plugging~(\ref{eqn::beta_vanish_2},~\ref{eqn::asy_four_SAW_bad_aux2},~\ref{eqn::asy_four_SAW_bad_aux3})  into~\eqref{eqn::asy_four_SAW_bad_aux1}, we obtain the asserted~\eqref{eqn::asy_four_SAW_bad}. 
\end{proof}

\bigskip{}
\section{Frobenius series for $\kappa=4$}
\label{app::Frob}
In this appendix, we derive Frobenius series for Coulomb gas integrals $\smash{\coulombGas_{\beta}^{(4)}}$ when $\kappa=4$. 
\begin{proposition}\label{prop::CGI4_Frobenius}
Fix $\kappa=4$ and $\beta\in\LP_N$. 
Let $\bs{x}\in\chamber_{2N}$ and $\xi\in(-\infty,x_3)$, 
and use the notation~\eqref{eqn::bs_notation}\textnormal{:} 
\begin{align*}
\begin{split}
\bs{x} = \; & (x_1, \ldots, x_{2N}) \in \chamber_{2N} 
\qquad 
\textnormal{and}
\qquad 
\bs{\ddot{x}}_1 = (x_3,x_4, \ldots,  x_{2N}) \in \chamber_{2N-2}. 
\end{split}
\end{align*}

\begin{itemize}
\item If $\{1,2\}\in\beta$, then
\begin{align}\label{eqn::CGI4_paired}
\coulombGas_{\beta}^{(4)}(\bs{x}) 
= 2\coulombGas_{\beta/\{1,2\}}^{(4)}(\bs{\ddot{x}}_1) \, (x_2-x_1)^{-1/2} \; + \; o\big((x_2-x_1)^{1/2}\big),\qquad\textnormal{as } x_1, x_2\to \xi.
\end{align}

\item If $\{1,2\}\not\in\beta$, then
\begin{align}\label{eqn::CGI4_notpaired}
\begin{split}
\coulombGas_{\beta}^{(4)}(\bs{x}) 
= \; & \coulombGas_{\wp_1(\beta)/\{1,2\}}^{(4)}(\bs{\ddot{x}}_1) \, (x_2-x_1)^{-1/2} \; + \; \coulombGas_{\beta/\vee_1}^{(4)}(\xi, \bs{\ddot{x}}_1) \, (x_2-x_1)^{1/2}\\
& \; \; + \; o\big((x_2-x_1)^{1/2}\big), 
\qquad\textnormal{as }x_1, x_2\to \xi, 
\end{split}
\end{align}
where
\begin{align}\label{eqn::CGI4_fused}
\coulombGas_{\beta/\vee_1}^{(4)}(\xi, \bs{\ddot{x}}_1)
:= \; & \sum_{\substack{\alpha\in\LP_N \\ \{1,2\}\in\alpha}} \meanderMat_2(\alpha,\beta) 
\, \PartF_{\alpha/\vee_1}^{(4)}(\xi, \bs{\ddot{x}}_1) 
+ \sum_{\substack{\alpha\in\LP_N \\ \{1,2\}\not\in\alpha}} \meanderMat_2(\alpha,\beta) \, \PartF_{\alpha/\times_1}^{(4)}(\xi, \bs{\ddot{x}}_1),
\end{align}
and $\PartF_{\alpha/\vee_1}^{(4)}$ and $\PartF_{\alpha/\times_1}^{(4)}$ are defined 
in Eq.~\eqref{eqn::PPF4_vee} and~\eqref{eqn::PPF4_times}, respectively. 
\end{itemize}
Furthermore, the function $\smash{\coulombGas_{\beta/\vee_1}^{(4)}}$ defined in~\eqref{eqn::CGI4_fused} satisfies the third order PDE~\eqref{eqn::third_order_PDE} with $\kappa=4$. 
\end{proposition}

When $\kappa=4$, the pure partition functions have explicit formulas, 
which we can utilize to prove Proposition~\ref{prop::CGI4_Frobenius}. 
To state the useful formulas, we first need to introduce the ``conformal block functions''
$\{\LU_{\gamma} \colon \gamma\in\LP_N\}$ (Definition~\ref{def:coblo})
and a certain incidence matrix $\{\LK(\alpha, \gamma) \colon \alpha,\gamma\in\LP_N\}$ (Definition~\ref{def:incidence}).

\begin{definition} \label{def:coblo}
For a link pattern $\gamma=\{\{a_1, b_1\}, \ldots, \{a_N, b_N\}\}\in\LP_N$ 
ordered as in~\eqref{eq: link pattern ordering}, 
we define the \emph{conformal block function} 
$\LU_{\gamma} \colon \chamber_{2N}\to (0,\infty)$ as 
\begin{align}\label{eqn::LU_def}
\LU_{\gamma}(x_1, \ldots, x_{2N}) 
:= \; & \prod_{1\le i<j\le 2N}(x_j-x_i)^{\tfrac{1}{2}\lambda_{\gamma}(i,j)},
\\
\nonumber
\qquad \textnormal{where} \qquad 
\lambda_{\gamma}(i,j)
:= \; &
\begin{cases}
+1, &\textnormal{if }i,j\in\{a_1, \ldots, a_N\}\textnormal{ or }i,j\in\{b_1, \ldots, b_N\},\\
-1, &\textnormal{otherwise}.
\end{cases}
\end{align}
\end{definition}

See~\cite{Peltola-Wu:Global_and_local_multiple_SLEs_and_connection_probabilities_for_level_lines_of_GFF,
Liu-Wu:Scaling_limits_of_crossing_probabilities_in_metric_graph_GFF} 
for probabilistic interpretations of the functions $\{\LU_{\gamma} \colon \gamma\in\LP_N\}$ in terms of the Gaussian free field. 
See also~\cite{LPR:Fused_Specht_polynomials_and_c_equals_1_degenerate_conformal_blocks} for alternative formulas for these conformal blocks and their fusion.

\begin{lemma}\label{lem::LU_Frobenius}
Fix $\gamma\in\LP_N$. Let $\bs{x}\in\chamber_{2N}$ and $\xi\in(-\infty,x_3)$, and use the notation~\eqref{eqn::bs_notation}. 
\begin{itemize}
\item If $\{1,2\}\in\gamma$, then
\begin{align}\label{eqn::LU_paired}
\begin{split}
\LU_{\gamma}(\bs{x}) 
= & \; \LU_{\gamma/\{1,2\}}(\bs{\ddot{x}}_1) \, (x_2-x_1)^{-1/2}
\; +\; \LU_{\gamma/\{1,2\}}(\bs{\ddot{x}}_1) \times \bigg(\sum_{i=3}^{2N}\frac{\tfrac{1}{2}\lambda_{\gamma}(i,1)}{x_i-\xi}\bigg) \, (x_2-x_1)^{1/2} \\
& \; \; + \; o\big((x_2-x_1)^{1/2}\big), \qquad\textnormal{as }x_1, x_2\to\xi.
\end{split}
\end{align}

\item If $\{1,2\}\not\in\gamma$, then
\begin{align}\label{eqn::LU_notpaired}
\LU_{\gamma}(\bs{x})
= & \; \LU_{\gamma/\times_1}(\xi, \bs{\ddot{x}}_1) \, (x_2-x_1)^{1/2}
\; + \; o\big((x_2-x_1)^{1/2}\big), \qquad\textnormal{as }x_1,x_2\to\xi,
\end{align}
where 
\begin{align}\label{eqn::LU_times}
\LU_{\gamma/\times_1}(\xi, \bs{\ddot{x}}_1)
:= \prod_{3\le i \le 2N}(x_i-\xi)^{\lambda_{\gamma}(i,1)} \; \prod_{3\le i<j\le 2N}(x_j-x_i)^{\tfrac{1}{2}\lambda_{\gamma}(i,j)}. 
\end{align}
\end{itemize}
\end{lemma}

\begin{proof}
The second expansion~\eqref{eqn::LU_notpaired} is proved in~\cite[Lemma~3.3]{Liu-Wu:Scaling_limits_of_crossing_probabilities_in_metric_graph_GFF}. 
It thus remains to show~\eqref{eqn::LU_paired}. 
Assume that $\{1,2\}\in\gamma$.
Using the definition~\eqref{eqn::LU_def}, a short computation shows~\eqref{eqn::LU_paired}:
\begin{align*} 
\LU_{\gamma}(\bs{x}) 
= (x_2-x_1)^{-1/2} \, \underbrace{\prod_{3\le i \le 2N}\Big(\frac{x_i-x_2}{x_i-x_1}\Big)^{-\frac{1}{2}\lambda_{\gamma}(i,1)}}_{P}
\; \prod_{3\le i<j\le 2N}(x_j-x_i)^{\tfrac{1}{2}\lambda_{\gamma}(i,j)} ,
\end{align*}
where 
\begin{align*}
P = \; & \prod_{3\le i<j\le 2N} \Big(1-\frac{x_2-x_1}{x_i-x_1}\Big)^{-\frac{1}{2}\lambda_{\gamma}(i,1)}
\; = \; \prod_{3\le i<j\le 2N}\Big(1+\frac{\lambda_{\gamma}(i,1)}{2} \, \frac{x_2-x_1}{x_i-x_1}\Big)\; + \; o(x_2-x_1) \\
= \; & 1 \; + \; \sum_{i=3}^{2N}\frac{\lambda_{\gamma}(i,1)}{2} \, \frac{x_2-x_1}{x_i-x_1}
\; +\; o(x_2-x_1). 
\end{align*}
\end{proof}

The following binary relation on link patterns was introduced in~\cite{Kenyon-Wilson:Boundary_partitions_in_trees_and_dimers,Kenyon-Wilson:Double_dimer_pairings_and_skew_Young_diagrams} 
and discussed in detail in~\cite[Section~2]{KKP:Boundary_correlations_in_planar_LERW_and_UST}
(see also~\cite[Section~2.4]{Peltola-Wu:Global_and_local_multiple_SLEs_and_connection_probabilities_for_level_lines_of_GFF}).
It is known that the incidence matrix $\LK$ of this binary relation is invertible 
--- see~\cite[Theorems~1.5~\&~1.6]{Kenyon-Wilson:Double_dimer_pairings_and_skew_Young_diagrams} and~\cite[Theorem~2.9]{KKP:Boundary_correlations_in_planar_LERW_and_UST}. 

\begin{definition} \label{def:incidence}
Let $\alpha=\{\{a_1, b_1\}, \ldots, \{a_N, b_N\}\}\in\LP_N$ be ordered as in~\eqref{eq: link pattern ordering}. 
For a link pattern $\gamma\in\LP_N$, we write $\alpha \KWleq \gamma$ 
if there exists a permutation $\sigma$ of $\{1, 2,\ldots, N\}$ such that 
\begin{align*}
\gamma=\{\{a_1, b_{\sigma(1)}\}, \ldots, \{a_N, b_{\sigma(N)}\}\} .
\end{align*}
We define the \emph{incidence matrix} $\{\LK(\alpha,\gamma) \colon \alpha, \gamma\in\LP_N\}$ 
as $\LK(\alpha, \gamma)=\one\{\alpha \KWleq \gamma\}$. 
\end{definition}

The pure partition functions $\{\PartF^{(4)}_{\alpha} \colon \alpha \in \LP_N\}$
and the conformal block functions $\{\LU^{(4)}_{\alpha} \colon \alpha \in \LP_N\}$
with $\kappa=4$
are related through the incidence matrix $\LK$ 
(so they form bases of the same vector space). 

\begin{lemma}
\textnormal{\cite[Theorem~1.5]{Peltola-Wu:Global_and_local_multiple_SLEs_and_connection_probabilities_for_level_lines_of_GFF}}
\label{lem::PPF4_LU}
Fix $\kappa=4$. We have
\begin{align*}
\PartF^{(4)}_{\alpha}(\bs{x})=\sum_{\gamma\in\LP_N}\LK^{-1}(\alpha, \gamma) \, \LU_{\gamma}(\bs{x}),\qquad \textnormal{for all }\bs{x}\in\chamber_{2N} \textnormal{ and } \alpha\in\LP_N. 
\end{align*}
\end{lemma}

\begin{lemma} \label{lem:PPF4_Frob}
Fix $\kappa=4$ and $\alpha\in\LP_N$. 
Let $\bs{x}\in\chamber_{2N}$ and $\xi\in(-\infty,x_3)$, and use the notation~\eqref{eqn::bs_notation}. 
\begin{itemize}
\item If $\{1,2\}\in\alpha$, then
\begin{align}\label{eqn::PPF4_paired}
\begin{split}
\PartF^{(4)}_{\alpha}(\bs{x})
= & \; \PartF^{(4)}_{\alpha/\{1,2\}}(\bs{\ddot{x}}_1) \, (x_2-x_1)^{-1/2}
\; +\; \PartF^{(4)}_{\alpha/\vee_1}(\xi, \bs{\ddot{x}}_1) \, (x_2-x_1)^{1/2}\\
&\; \; +\; o\big((x_2-x_1)^{1/2}\big),\qquad\textnormal{as }x_1, x_2\to \xi, 
\end{split}
\end{align}
where
\begin{align}\label{eqn::PPF4_vee}
\begin{split}
\PartF^{(4)}_{\alpha/\vee_1}(\xi, \bs{\ddot{x}}_1)
= & \; \sum_{\substack{\gamma\in\LP_N \\ \{1,2\}\in\gamma}}\LK^{-1}(\alpha,\gamma) \, \LU_{\gamma/\{1,2\}}(\bs{\ddot{x}}_1) 
\sum_{i=3}^{2N}\frac{\frac{1}{2}\lambda_{\gamma}(i,1)}{x_i-\xi} + \sum_{\substack{\gamma\in\LP_N \\ \{1,2\}\not\in\gamma}}\LK^{-1}(\alpha,\gamma) \, \LU_{\gamma/\times_1}(\xi, \bs{\ddot{x}}_1). 
\end{split}
\end{align}

\item If $\{1,2\}\not\in\alpha$, then
\begin{align}\label{eqn::PPF4_notpaired}
\PartF^{(4)}_{\alpha}(\bs{x})
= & \; \PartF^{(4)}_{\alpha/\times_1}(\xi, \bs{\ddot{x}}_1) \, (x_2-x_1)^{1/2}
\; +\; o\big((x_2-x_1)^{1/2}\big), \qquad\textnormal{as }x_1,x_2\to\xi,
\end{align}
where 
\begin{align}\label{eqn::PPF4_times}
\begin{split}
\PartF^{(4)}_{\alpha/\times_1}(\xi, \bs{\ddot{x}}_1)
= & \; \sum_{\substack{\gamma\in\LP_N \\ \{1,2\}\in\gamma}}\LK^{-1}(\alpha,\gamma) \, \LU_{\gamma/\{1,2\}}(\bs{\ddot{x}}_1) 
\sum_{i=3}^{2N}\frac{\frac{1}{2}\lambda_{\gamma}(i,1)}{x_i-\xi}  + \sum_{\substack{\gamma\in\LP_N \\ \{1,2\}\not\in\gamma}}\LK^{-1}(\alpha,\gamma) \, \LU_{\gamma/\times_1}(\xi, \bs{\ddot{x}}_1). 
\end{split}
\end{align}
\end{itemize}
\end{lemma}

\begin{proof}
From Lemmas~\ref{lem::LU_Frobenius}~\&~\ref{lem::PPF4_LU}, we obtain
\begin{align*}
\PartF^{(4)}_{\alpha}(\bs{x})
=\;&\sum_{\substack{\gamma\in\LP_N\\\{1,2\}\in\gamma}}\LK^{-1}(\alpha, \gamma) \, \LU_{\gamma}(\bs{x})+\sum_{\substack{\gamma\in\LP_N\\\{1,2\}\not\in\gamma}}\LK^{-1}(\alpha, \gamma) \, \LU_{\gamma}(\bs{x})\\
=\;&\sum_{\substack{\gamma\in\LP_N\\\{1,2\}\in\gamma}}\LK^{-1}(\alpha, \gamma)  \, \LU_{\gamma/\{1,2\}}(\bs{\ddot{x}}_1) \, (x_2-x_1)^{-1/2}\\
\;&+\sum_{\substack{\gamma\in\LP_N\\\{1,2\}\in\gamma}}\LK^{-1}(\alpha, \gamma)  \,  \LU_{\gamma/\{1,2\}}(\bs{\ddot{x}}_1)  \bigg(\sum_{i=3}^{2N}\frac{\tfrac{1}{2}\lambda_{\gamma}(i,1)}{x_i-\xi}\bigg) \, (x_2-x_1)^{1/2}\\
\;&+\sum_{\substack{\gamma\in\LP_N\\\{1,2\}\not\in\gamma}}\LK^{-1}(\alpha, \gamma) \, \LU_{\gamma/\times_1}(\xi, \bs{\ddot{x}}_1) \, (x_2-x_1)^{1/2}+o((x_2-x_1)^{1/2}) , \qquad x_1, x_2\to \xi . 
\end{align*}
It thus remains to show that 
\begin{align}\label{eqn::PPF4_twocases}
\sum_{\substack{\gamma\in\LP_N\\\{1,2\}\in\gamma}}\LK^{-1}(\alpha, \gamma)  \, \LU_{\gamma/\{1,2\}}(\bs{\ddot{x}}_1)=\begin{cases}
\PartF^{(4)}_{\alpha/\{1,2\}}(\bs{\ddot{x}}_1), &\textnormal{if }\{1,2\}\in\alpha, \\
0,&\textnormal{if }\{1,2\}\not\in\alpha. 
\end{cases}
\end{align}
The first case in~\eqref{eqn::PPF4_twocases} follows from
Lemma~\ref{lem::PPF4_LU} and the following relation 
(see, e.g.,~\cite[Proof of Lemma~6.7]{Peltola-Wu:Global_and_local_multiple_SLEs_and_connection_probabilities_for_level_lines_of_GFF}) 
when $\{1,2\}\in\alpha$ and $\{1,2\}\in\gamma$, we have
\begin{align*}
\LK^{-1}(\alpha, \gamma)=\LK^{-1}(\alpha/\{1,2\}, \gamma/\{1,2\}). 
\end{align*}
The second case in~\eqref{eqn::PPF4_twocases} is proved in~\cite[Proof of Lemma~3.4]{Liu-Wu:Scaling_limits_of_crossing_probabilities_in_metric_graph_GFF}. 
\end{proof}

\begin{proof}[Proof of Proposition~\ref{prop::CGI4_Frobenius}]
As $x_1, x_2\to \xi$, from~(\ref{eqn::PPF4_notpaired},~\ref{eqn::PPF4_paired}) in Lemma~\ref{lem:PPF4_Frob}, we obtain 
\begin{align*}
\coulombGas^{(4)}_{\beta}(\bs{x})
= & \; \sum_{\alpha\in\LP_N}\meanderMat_2(\alpha,\beta) \, \PartF^{(4)}_{\alpha}(\bs{x})\\
= & \; \sum_{\substack{\alpha\in\LP_N \\ \{1,2\}\in\alpha}}\meanderMat_2(\alpha,\beta)\Big(\PartF^{(4)}_{\alpha/\{1,2\}}(\bs{\ddot{x}}_1) \, (x_2-x_1)^{-1/2}+\PartF^{(4)}_{\alpha/\vee_1}(\xi, \bs{\ddot{x}}_1) \, (x_2-x_1)^{1/2}\Big)
\\
& \; \;+\; \sum_{\substack{\alpha\in\LP_N \\ \{1,2\}\not\in\alpha}}\meanderMat_2(\alpha,\beta) \, \PartF^{(4)}_{\alpha/\times_1}(\xi, \bs{\ddot{x}}_1) \, (x_2-x_1)^{1/2}
\; + \; o\big((x_2-x_1)^{1/2}\big)\\
= & \; \underbrace{\bigg(\sum_{\substack{\alpha\in\LP_N \\ \{1,2\}\in\alpha}}\meanderMat_2(\alpha,\beta) \, \PartF^{(4)}_{\alpha/\{1,2\}}(\bs{\ddot{x}}_1)\bigg)}_{S_1} \, (x_2-x_1)^{-1/2}\\
& \; \;+\; 
\underbrace{\bigg(\sum_{\substack{\alpha\in\LP_N \\ \{1,2\}\in\alpha}}\meanderMat_2(\alpha,\beta) \, \PartF^{(4)}_{\alpha/\vee_1}(\xi, \bs{\ddot{x}}_1)+\sum_{\substack{\alpha\in\LP_N \\ \{1,2\}\not\in\alpha}}\meanderMat_2(\alpha,\beta) \, \PartF^{(4)}_{\alpha/\times_1}(\xi, \bs{\ddot{x}}_1)\bigg)}_{\coulombGas_{\beta/\vee_1}^{(4)}(\xi, \bs{\ddot{x}}_1)}
\, (x_2-x_1)^{1/2}\\
&\; \; +\; o\big((x_2-x_1)^{1/2}\big).
\end{align*}
Note that
\begin{align*}
S_1 = 
\begin{cases}
2\coulombGas^{(4)}_{\beta/\{1,2\}}(\bs{\ddot{x}}_1), &\textnormal{if }\{1,2\}\in\beta , \\
\coulombGas^{(4)}_{\wp_1(\beta)/\{1,2\}}(\bs{\ddot{x}}_1), &\textnormal{if }\{1,2\}\not\in\beta.
\end{cases}
\end{align*}
In particular, this implies~\eqref{eqn::CGI4_notpaired}. 
It then remains to show~\eqref{eqn::CGI4_paired}. 

Suppose $\{1,2\}\in\beta$. 
We will prove~\eqref{eqn::CGI4_paired} using the estimates derived in Lemma~\ref{lem::Frobenius_F_aux1}. 
Recall from~\eqref{eqn::Fro_pairing_aux1} and the discussion after it that, 
for $\kappa\in (8/3,8)$, we have
\begin{align}\label{eqn::CGI4_paired_aux1}
\coulombGas^{(\kappa)}_{\beta}(\bs{x}) 
- \fugacity(\kappa) \, \coulombGas^{(\kappa)}_{\beta/\{1,2\}}(\bs{\ddot{x}}_1) \, (x_2-x_1)^{-2h(\kappa)} 
= \; & (x_2-x_1)^{-2h(\kappa)+2} \, G_{\beta}^{(\kappa)}(\bs{x}) ,
\end{align}
where $G_{\beta}^{(\kappa)}(\bs{x})$ is uniformly bounded in $\kappa\in (8/3,8)$ and $x_1,x_2\in (\xi-(x_3-\xi)/3, \xi+(x_3-\xi)/3)$. 
Next, setting $\kappa=4$ in~\eqref{eqn::CGI4_paired_aux1} we obtain
\begin{align*}
\coulombGas^{(4)}_{\beta}(\bs{x}) 
- 2\coulombGas^{(4)}_{\beta/\{1,2\}}(\bs{\ddot{x}}_1) \, (x_2-x_1)^{-1/2} 
= \; & O\big((x_2-x_1)^{3/2}\big), \qquad\textnormal{as }x_1, x_2\to \xi,
\end{align*}
which completes the proof of~\eqref{eqn::CGI4_paired}.

It remains to show that the function $\smash{\coulombGas_{\beta/\vee_1}^{(4)}(\xi, \bs{\ddot{x}}_1)}$ 
defined in~\eqref{eqn::CGI4_fused} satisfies the third  order PDE~\eqref{eqn::third_order_PDE} with $\kappa=4$. 
\cite[Proof of Lemmas~3.6~\&~Lemma~3.8]{Liu-Wu:Scaling_limits_of_crossing_probabilities_in_metric_graph_GFF} 
show that both $\LU_{\gamma/\times_1}$ defined in~\eqref{eqn::LU_times}, and 
\begin{align*}
\LU_{\gamma/\{1,2\}}(\bs{\ddot{x}}_1) 
\sum_{i=3}^{2N}\frac{\frac{1}{2}\lambda_{\gamma}(i,1)}{x_i-\xi} ,
\end{align*}
satisfy the third order PDE~\eqref{eqn::third_order_PDE} with $\kappa=4$. 
Thus, $\PartF^{(4)}_{\alpha/\vee_1}(\xi, \bs{\ddot{x}}_1)$ in~\eqref{eqn::PPF4_vee}, $\PartF^{(4)}_{\alpha/\times_1}(\xi, \bs{\ddot{x}}_1)$ in~\eqref{eqn::PPF4_times} and $\smash{\coulombGas_{\beta/\vee_1}^{(4)}(\xi, \bs{\ddot{x}}_1)}$ 
in~\eqref{eqn::CGI4_fused} all satisfy~\eqref{eqn::third_order_PDE} with $\kappa=4$. 
This completes the proof of Proposition~\ref{prop::CGI4_Frobenius}. 
\end{proof}


\bigskip{}
\bibliographystyle{alpha}

{\small{
\newcommand{\etalchar}[1]{$^{#1}$}

}}

\end{document}